\documentclass[ejsv2,preprint,authoryear,noshowframe]{imsart}

\RequirePackage[numbers]{natbib}
\RequirePackage[colorlinks,citecolor=blue,urlcolor=blue]{hyperref}
\RequirePackage{graphicx}
\usepackage{xcolor}
\usepackage{mathtools}
\usepackage{caption}
\usepackage{xcolor}
\usepackage{enumitem}
\usepackage{multirow}
\usepackage{rotating}
\usepackage{pdflscape}
\usepackage{subfigure}
\usepackage{bm}
\usepackage{bbm}
\usepackage{float}
\usepackage{mathtools}
\usepackage{mathrsfs}
\usepackage{makecell}

\startlocaldefs

\theoremstyle{plain}
\newtheorem{theorem}{Theorem}[section]
\newtheorem{proposition}{Proposition}[section]
\newtheorem{lemma}[theorem]{Lemma}
\theoremstyle{remark}

\newtheorem{remark}[theorem]{Remark}

\newtheorem{asu}{Assumption}

\newcommand{\field}[1]{\mathbb{#1}}
\newcommand{\R}{\field{R}}

\newcommand{\E}{\field{E}}

\newcommand{\PP}{\mathbb{P}}

\def\EE{\mathbb E}

\def\var{\mathrm{Var}}

\DeclarePairedDelimiter\ceil{\lceil}{\rceil}

\def\bE{\mathbb{E}}

\newcommand{\toDistSt}{\overset{st}{\longrightarrow}}

\endlocaldefs
\usepackage{booktabs}
        \usepackage[utf8]{inputenc}
         \DeclareMathAlphabet{\mathscr}{U}{BOONDOX-cal}{m}{n}
         \SetMathAlphabet{\mathscr}{bold}{U}{BOONDOX-cal}{b}{n}
         \DeclareMathAlphabet{\mathbscr} {U}{BOONDOX-cal}{b}{n}
\usepackage{color}

\newcommand{\mc}[2]{\raisebox{-.4\height}[\height][\depth]{\shortstack{#1\\[-1.5pt]{\tiny $\pm\,#2$}}}}
\newcommand{\mcb}[2]{\raisebox{-.4\height}[\height][\depth]{\shortstack{\textbf{#1}\\[-1.5pt]{\tiny $\pm\,#2$}}}}

\begin{document}

\begin{frontmatter}
\pdfsubject{}
\title{Debiased Kernel Estimation of Spot Volatility in the Presence of Infinite Variation Jumps}
\runtitle{Debiased Kernel Estimation of Spot Volatility}

\begin{aug}
\author[A]{\fnms{B. Cooper}~\snm{Boniece}\ead[label=e1]{cooper.boniece@drexel.edu}}
\author[B]{\fnms{Jos\'e E.}~\snm{Figueroa-L\'opez}\ead[label=e2]{figueroa-lopez@wustl.edu}}
\author[B]{\fnms{Tianwei}~\snm{Zhou}\ead[label=e3]{tianweizhou@wustl.edu}}
\address[A]{Department of Mathematics, Drexel University\printead[presep={,\ }]{e1}}

\address[B]{Department of Statistics and Data Science, Washington University in St.~Louis\printead[presep={,\ }]{e2,e3}}
\end{aug}

\begin{abstract}
Volatility estimation is a central problem in financial econometrics, but becomes particularly challenging when jump activity is high, a phenomenon observed empirically in highly traded financial securities. In this paper, we revisit the problem of spot volatility estimation for an It\^o semimartingale with jumps of unbounded variation. We construct truncated kernel-based estimators and debiased variants that extend rate-optimal spot volatility estimation to a wider range of jump activity indices, from the previously available bound $Y<4/3$ to $Y<20/11$. Rate-suboptimal CLTs are also established for $Y>20/11$. Compared with earlier work, our approach achieves smaller asymptotic variances through the use of more general kernels and an optimal choice for the bandwidth convergence rate, and also has broader applicability under more flexible model assumptions. A comprehensive simulation study confirms that our procedures outperform competing methods in finite samples.
\end{abstract}

\begin{keyword}
\kwd{high-frequency data}
\kwd{spot volatility estimation}
\kwd{kernel smoothing}
\kwd{semimartingales}
\end{keyword}

\end{frontmatter}

\section{Introduction}

 Spot volatility estimation is a central problem in financial econometrics, underpinning numerous critical applications ranging from short-horizon risk management to derivatives pricing and hedging. Formally, the core task is to estimate the coefficient $\sigma_\tau$ of the Brownian component of an It\^o semimartingale $X$ at a fixed point of time $\tau$ given high-frequency observations of $X$. For recent developments, we refer to \cite{LIU2018,figueroa2020optimal,Figueroa-Lopez_Wu_2024,mancino2022asymptotic}, as well as the monographs \cite{jacod2012discretization,ait2014high} for in-depth treatments.

Jumps are a salient feature of price dynamics, separate from fluctuations captured by diffusive movements in $X$, and must be carefully accounted for when estimating volatility and related functionals of $X$. However, when jump activity is high, standard estimation methods face inherent limitations as the difficulty of disentangling the continuous and jump components of the process increases. 
Three qualitatively distinct regimes can be distinguished:
jumps of finite activity, infinite activity jumps of bounded
variation, and infinite activity jumps of unbounded variation.
{Small-jump activity is quantified by the
Blumenthal--Getoor index $Y$, the critical exponent above
which powers of absolute jump sizes are summable over a
finite time interval \citep{A_t_Sahalia_2009}. Larger $Y$
indicates greater small-jump activity. For volatility
estimation, the first two regimes pose similar statistical
challenges. When $Y>1$, the jumps have unbounded variation and pose greater
statistical difficulties; this regime is the focus of the
present work.}

A fundamental approach to handling jumps is through truncation, introduced by \cite{Mancini2004}.  Establishing jump-robustness for truncation-based methods classically amounts to two steps: suppose $\hat\theta_n(\Delta_1^nX,\dots,\Delta^n_n X)$ is an estimator of some functional $\theta$ of the continuous component $X^c$ of $X$, where as usual, $\Delta_i^nU:=U_{t_{i}}-U_{t_{i-1}}$ denotes the $i^{th}$ increment of a generic process $U$ sampled at times $0<t_0<\dots<t_n=T$ during a finite time period $[0,T]$.  In the first step, one shows that the (infeasible) estimator $\hat\theta^{\text{Cont}}_n:=\hat{\theta}_n(\Delta_1^nX^c,\dots,\Delta_n^n X^c)$ is consistent and satisfies a central limit theorem (CLT) with, say, rate  $r_n\to{}\infty$:
\begin{equation}\label{M1aTCLT}
   r_n\left(\hat{\theta}_n(\Delta_1^nX^c,\dots,\Delta_n^n X^c)-\theta\right)\stackrel{\mathcal{D}}{\longrightarrow}\mathcal{N}(0,V_\theta).
\end{equation}
Second, it is shown that, by choosing an appropriate threshold level $v_n\searrow{}0$, the truncated or thresholded version $\hat\theta^{Trunc}_n:=\hat{\theta}_n(\Delta_1^nX\mathbf{1}_{\{|\Delta_1^n X|\leq  v_n\}},\dots,{\Delta^n_n X}\mathbf{1}_{\{|\Delta_n^n X|\leq  v_n\}})$ is close enough to $\hat\theta^{Cont}$ so that 
\begin{equation}\label{M1bTCLT}
    { r_n}\left(\hat\theta_n^{\text{Trunc}}-\hat\theta_n^{\text{Cont}}\right)=o_P(1).
\end{equation}
An emblematic example of this approach arises in the estimation of the integrated variance or volatility $IV:=\int_0^T \sigma^2_sds$, where one passes from the quantity %
\begin{equation}
    \widehat{IV}^{\text{Cont}}:=\sum_{i=1}^n(\Delta_{i}^nX^{c})^2,
\end{equation}
to its truncated version
\begin{equation}\label{TRVEst}
 \widehat{IV}^{\text{Trunc}}:=\sum_{i=1}^n(\Delta_{i}^nX)^2{\bf 1}_{\{|\Delta_i^nX|\leq v_n\}}.
\end{equation}
By (essentially) employing the steps \eqref{M1aTCLT}-\eqref{M1bTCLT}, it has been established that in the {case of a bounded variation jump component (\cite{mancini2009non,cont2011nonparametric})}, an efficient CLT (see, e.g., \cite{jacod:reiss:2014}) is possible:
\[
    \sqrt{n}\left(\widehat{IV}^{\text{Trunc}}-IV\right)\stackrel{\mathcal{D}}{\longrightarrow}\mathcal{N}\left(0,2\int_0^T\sigma_s^4ds\right).
\]
However, with unbounded-variation jumps, this approach alone falls short.  It was proved by \cite{cont2011nonparametric} (see also \cite{mancini:2011}) that $\sqrt{n}\big(\widehat{IV}^{\text{Trunc}}-IV\big)\stackrel{P}{\longrightarrow}\infty$ in the presence of an infinite variation symmetric stable L\'evy process, precisely because the bias of $\widehat{IV}^{\text{Trunc}}$ -- and thus the difference $(\widehat{IV}^{\text{Trunc}} - \widehat{IV}^{\text{Cont}})$ -- vanishes too slowly.

Two methodological advances have since moved the field forward. The first is due {to \cite{Jacod2014,jacod:todorov:2016}}, who proposed an approach for integrated variance estimation based on the empirical characteristic function of the rescaled increments $\Delta_i^nX/\sqrt{\Delta_n}$, where $\Delta_n$ is the time span of the increment. After applying a suitable logarithmic transformation and a debiasing procedure, this approach admits a CLT with  optimal rate, even in the case of unbounded variation jumps.  However, with this method, full (rate and variance) efficiency is attainable only under specific conditions that require either the jump activity index $Y$ to satisfy $Y<3/2$, or under additional restrictive jump symmetry assumptions.  The second major advance was introduced more recently by    \cite{BONIECE2024}, who showed that a debiasing method similar to \cite{Jacod2014} applied to the realized truncated variations \eqref{TRVEst} leads to a fully efficient estimator for any $Y\leq 8/5$ without any symmetry conditions.

A natural next question is how these insights extend to the problem of spot volatility estimation. A standard and widely used approach to constructing spot volatility estimators is to apply kernel smoothing to an estimator of integrated variance, %
\[
    \hat\sigma^2_\tau=\int K_b(t-\tau)d \widehat{IV}_t,
\]
where $\widehat{IV}_t$ denotes an estimator of $IV_t=\int_0^t \sigma_s^2ds$ and $K_b(u):=K(u/b)/b$ for a chosen kernel $K$ and a bandwidth $b>0$.
For instance, using the realized quadratic variation $\widehat{IV}_t=\sum_{i=1}^{[nt]}(\Delta_i ^nX)^2$, this reduces to
\begin{equation}\label{KrnlSpotVol0}
    \hat\sigma^2_\tau=\sum_{i=1}^n K_b(t_{i-1}-\tau)(\Delta_i^n X)^2,
\end{equation}
as proposed in \cite{kristensen2010nonparametric,fan2008spot}.   

Spot volatility estimators differ markedly from their integrated volatility counterparts, with their local nature altering the efficiency picture in subtle but notable ways: on the one hand, optimal convergence rates are of order $n^{1/4}$ (compared with  $n^{1/2}$ in the integrated volatility case); on the other hand,  kernel-based estimators enjoy a degree of built-in robustness against jumps, since ``large'' jumps are relatively rare and are unlikely to fall in a narrow neighborhood of $\tau$ where $K_b$ places most of its mass.
For instance, %
\cite[Theorem 13.3.3]{jacod2012discretization} shows that in the case of uniform kernels (i.e., $K(u)={\bf 1}_{[0,1)}(u)$ or $K(u)={\bf 1}_{(-1,0]}(u)$) 
the estimator \eqref{KrnlSpotVol0} admits a rate-optimal CLT even in the case of unbounded variation jumps, provided $Y<4/3$.\footnote{Here and below, we use ``unbounded variation'' to refer to the regime $Y>1$. The case $Y=1$ is excluded from the moment expansions and debiasing analysis developed below because it gives rise to additional logarithmic terms; rate-optimal estimation without debiasing in this case is already covered by \cite[Theorem~2.1]{Figueroa-Lopez_Wu_2024}.} However, high levels of jump activity can still degrade performance, and the rate becomes suboptimal for  $Y\geq 4/3$: 
in general, for any $a\in(0,1/2]$ such that $Y<2/(1+a)$, one can pick a bandwidth sequence $ b=b_n$ such that $n^{a/2}(\hat\sigma_\tau^2-\sigma_\tau^2)\stackrel{\mathcal{D}}{\longrightarrow}\mathcal{N}(0,V_\sigma)$, which yields a rate that is slower than $n^{-\frac{2-Y}{2Y}}$. Thus, if, for instance, $Y\geq 3/2$, the rate is slower than $n^{-1/6}$. 
 Surprisingly, following the two-step approach \eqref{M1aTCLT}-\eqref{M1bTCLT}, even when thresholding is applied to \eqref{KrnlSpotVol0}, i.e.,
\begin{equation}\label{KrnlSpotVolThrholded}
    \hat\sigma^2_\tau=\sum_{i=1}^n K_b(t_{i-1}-\tau)(\Delta_i^nX)^2{\bf 1}_{\{|\Delta_i^nX|\leq v_n\}},
\end{equation}
there is no apparent asymptotic benefit: 
the optimal rate of $n^{1/4}$ remains attainable only if $Y<4/3$, and the best possible rate when $Y\geq 4/3$ remains bounded by $n^{\frac{2-Y}{2Y}}$, which is strictly slower.

Given these limitations, one may wonder if estimation performance can be improved by modifying the kernel $K$. This question was addressed in \cite{Figueroa-Lopez_Wu_2024}, which extended the approach in \cite{jacod2012discretization} to two-sided kernels of unbounded support.  Although the convergence rate remains unchanged, they established that the asymptotic variance can be substantially reduced. For example, using the kernel $K(u)={e^{-|u|}/2}$ leads to an asymptotic variance that is one-quarter of that obtained by a uniform kernel in the suboptimal rate regime $(Y\geq 4/3)$. Moreover, in the optimal rate regime $(Y< 4/3)$,  the same exponential kernel was shown to be variance-optimal,  highlighting the utility of unbounded kernels in general. However, we note that the results in both \cite{jacod2012discretization} and \cite{Figueroa-Lopez_Wu_2024} are based on the two-step approach \eqref{M1aTCLT}-\eqref{M1bTCLT}, which we depart from in the present work.

 These findings motivate a fundamental question: \textit{can one develop methods that attain better rates -- possibly optimal --  beyond the boundary $Y< 4/3$}? 
A recent step in this direction was taken by \cite{LIU2018}, who developed a localized variant of the characteristic function approach of \cite{Jacod2014}.  
 Under the assumption that the  infinite-variation jump component is of the form $\int_0^t\chi_{s}dL_s$, where $L$ is a strictly $Y$-stable symmetric L\'evy process, and restricting to  compactly-supported,  continuously differentiable kernels {$K$}, they establish that, {by choosing the bandwidth $b_n$ such that $\Delta_n\ll b_n\ll \sqrt{\Delta_n}$,} their estimator satisfies
\begin{equation}\label{LiuCLT0}
    \sqrt{\frac{b_n}{\Delta_n}}\frac{\hat\sigma^2_\tau-\sigma^2_\tau}{\sqrt{2\sigma_\tau^4}}\stackrel{\cal{D}}{\longrightarrow}\mathcal{N}\left(0,\int K^2(u)du\right),
\end{equation}
for any $Y<3/2$,  but with a suboptimal convergence rate (i.e.,  $\sqrt{\frac{b_n}{\Delta_n}}\ll n^{1/4}$). By applying an additional debiasing step similar to that in \cite{Jacod2014}, they further showed that associated tuning parameters can be chosen so that \eqref{LiuCLT0} holds for any $Y<2$, however, still with a suboptimal rate and under the jump symmetry condition described above. Consequently, the boundary $Y<4/3$ has long marked the limit of rate-optimality in spot volatility estimation, and the regime $Y\geq 4/3$ has remained largely underexplored without  a type of jump symmetry assumption.

In this work,  we extend the range of $Y$ over which rate-optimal estimation is available by introducing appropriately debiased versions of the estimator \eqref{KrnlSpotVolThrholded}. First, by exploiting high-order expansions of the truncated moments of $\Delta_i X$, rather than relying on the coarser two-step framework  \eqref{M1aTCLT}-\eqref{M1bTCLT}, which does not explicitly account for higher-order bias,  we show the estimator \eqref{KrnlSpotVolThrholded} actually achieves rate-optimality for any $Y<3/2$ (namely, $n^{1/4}(\hat\sigma_\tau^2-\sigma_\tau^2)\stackrel{\mathcal{D}}{\longrightarrow}\mathcal{N}(0,V_\sigma)$) under standard assumptions and without additional adjustment. Second, by introducing debiasing steps similar to those in \cite{Jacod2014,BONIECE2024}, we construct rate-optimal estimators that cover the range  $0<Y<20/11$. This extension is practically relevant, as empirical estimates of $Y$ for liquid stocks often fall in or near the range where the distinction between the $Y<4/3$ and $Y<20/11$ regimes matters. For example, {\cite{A_t_Sahalia_2009}} applied their jump activity estimator to 5-second and 15-second returns of INTC and MSFT from 2006 and reported estimates ranging from 1.46 to 1.76 for INTC and from 1.60 to 1.69 for MSFT; similar values were reported using 2005 data for the same stocks in \cite{figueroa12}. These estimates are argued to be relatively robust to microstructure noise because the underlying statistics are based on large increments, which are more likely to reflect genuine price movements than bid--ask bounce or other noise effects. More recently, \cite{BONIECE2024} applied the generalized method of moments estimator of \cite{Mies} to additional stocks, reporting estimates of 1.60 for 5-second returns of INTC, 1.37 for 1-minute returns of PFE, 1.82 for 5-second returns of SPY, 1.53 for 1-minute returns of AMGN, and 1.33 for 1-minute returns of MOT; evidence of jump asymmetry was also found for some stocks.

For $Y\geq 20/11$, we establish CLTs with suboptimal
convergence rates. {The boundary $20/11$ reflects the scope
of our two-step analysis and is not expected to be optimal.
Further debiasing would require higher-order moment
expansions and control of the additional correction errors,
as well as an assessment of numerical stability.
In practice, the jump activity index $Y$ would first need
to be estimated to determine the appropriate regime and
tuning choices.
Table~\ref{tab:spot-regimes} summarizes the
jump-activity regimes covered by our results and the
corresponding debiasing orders.}

\begin{table}[htbp]
\centering
\caption{Jump-activity regimes and sufficient debiasing orders for
spot variance estimation. The stated error rates hold under the
assumptions and tuning conditions of Section~\ref{SectMainResults}.}
\label{tab:spot-regimes}
\begin{tabular}{llll}
\toprule
Jump activity & Debiasing steps & Error rate & Detailed statement \\
\midrule
$0<Y<3/2$ & None & $n^{-1/4}$
& Remarks~\ref{SlyRmk1} and~\ref{finitevariation_remark} \\
$3/2\leq Y<12/7$ & One & $n^{-1/4}$
& Remark~\ref{onestep_remark} \\
$12/7\leq Y<20/11$ & Two & $n^{-1/4}$
& Theorem~\ref{main} \\
$20/11\leq Y<2$ & Two
& Slower than $n^{-3(2-Y)/(4-Y)}$
& Remark~\ref{twostep_remark} \\
\bottomrule
\end{tabular}
\end{table}

Compared with existing methods, our method  offers rate-optimal estimators over a wider range of $Y$, broader applicability through more flexible model assumptions, reduced asymptotic variances enabled by the use of unbounded kernels, and simpler estimators that are easier to implement. To validate our results, we perform a comprehensive simulation study and find that
our debiased estimators generally achieve lower mean squared error than those of \cite{LIU2018} across the settings considered.
Our findings further demonstrate that the advantages of kernels with unbounded support extend to the debiasing framework developed here, yielding improvements over the widely used uniform kernels that often appear in applications where spot volatility is estimated as a preliminary step, including jump testing, jump detection, and estimation of volatility functionals of the form
$\int_0^T g(\sigma^2_s)ds$ (see, e.g., \cite{jacod2012discretization}, \cite{LiXiu:2016}, \cite{Alvarez}, \cite{jacod2013quarticity}, \cite{PalmesWoerner}\footnote{The latter actually uses a localized Bipower estimator using uniform weights or kernels.})

The paper is organized as follows.  Section \ref{Sect2} introduces the setting, assumptions, and preliminary results. Section \ref{SectMainResults} presents the main results.  Section \ref{SectionSimul} contains simulations assessing the performance of our estimators. The proofs of the main results are deferred to Appendix \ref{MainPrfsa}. Some additional technical proofs are collected in Appendix \ref{ApTechPrfs}.

\section{Setting and background}\label{Sect2}
Consider a 1-dimensional It\^o semimartingale $X = \big(X_t\big)_{t\in \R_{+}}$ defined on a complete filtered probability space $\big(\Omega,\mathscr{F},(\mathscr{F})_{t\in \R_{+}},\PP\big)$ that can be decomposed as
\begin{align}\label{model}
    X_t =\int_0^t{\mu_s}ds +  \int_0^t\sigma_s dW_s + \int_0^t\chi_{s^-}dJ_s +\int_0^t\int\delta(s,z)\mathfrak{p}(ds,dz),\;\;t\in [0,\infty),
\end{align}
where $W:=(W_t)_{t \in \R_{+}}$ is a Wiener process, $\sigma$ is an adapted c\'adl\'ag process, $J:=(J_t)_{t \in \R_{+}}$ is an independent pure-jump L\'evy process with L\'evy triplet $(b,0,v)$, and $\mathfrak{p}$ is a Poisson random measure on $\R_+\times\R$ with an intensity $\mathfrak{q}(ds,dz) = ds \otimes \lambda(dz)$, such that the $\lambda$ is a $\sigma$-finite measure on $\R$, possibly dependent with $J$. {We will impose conditions below such that the last term in \eqref{model} is of finite variation, while the third term can induce jumps of infinite variation. More specifically,} the L\'evy measure $v$ of $J$ is assumed to admit a density $s:{\R\backslash\{0\}}\rightarrow \R_+$ of the form
\begin{align}
    s(x):=\frac{dv}{dx}:=\big(C_+\mathbf{1}_{(0,\infty)}(x) + C_-\mathbf{1}_{(-\infty,0)}(x)\big)q(x)|x|^{-1-Y}.
\end{align}
{We have the following standing assumptions:}
\begin{asu}\label{BscAs1}\hfill
\begin{enumerate}
    \item  $C_{\pm}>0$ and  $Y \in (0,2)\setminus\{1\}$.
    \item $q:{{\R\backslash\{0\}} \rightarrow \R_+}$ is a bounded Borel-measurable function such that:
    \begin{enumerate}
        \item $q(x) \rightarrow 1$, as $x\rightarrow 0$.
        \item There exist $a_{\pm}$ such that 
    \begin{equation}\label{CndqA}\int_0^1 \big|q(x)-1-a_+ x\big|x^{-Y-1}dx + \int_{-1}^0\big|q(x)-1-a_-x\big||x|^{-Y-1}dx<\infty.
        \end{equation}
    \end{enumerate}
    \item The process $\chi$ is given as
    \begin{equation}\label{CndChia}\chi_t = \chi_0 + \int_0^t{\mu_s^\chi} ds + \int_0^t \Sigma_s^\chi dB_s,
    \end{equation}
    {where $\chi_0$ is a constant, $\{\mu^{\chi}\}_{t\geq 0}$ and $(\Sigma^{\chi}_t)_{t\geq 0}$ are  adapted c\'adl\'ag locally bounded processes, and $W,B$ are standard Brownian motions with correlation $d\langle W,B\rangle_t = \rho_tdt$ for an adapted, locally bounded c\'adl\'ag process $\{\rho_t\}_{t\geq 0}$. Furthermore, we assume that $\chi_{\tau}\neq 0$ a.s. at the fixed estimation time $\tau$.} 
    \item $W,B$ are independent of $(J,\mathfrak{p})$; 
    $J$ is independent of $\sigma$.
    \item The spot variance process $c_t:=\sigma_t^2$ is assumed to follow the dynamics:
\begin{align}\label{CndCa}
    c_t &= c_0 + \int_0^t \tilde \mu_sds +\int_0^t\tilde \sigma_s dB_s+{\int_0^t\int\tilde{\delta}(s,z)%
    \bar{\mathfrak{p}}(ds,dz)},%
\end{align}
where $\bar{\mathfrak{p}}:=\mathfrak{p}-\mathfrak{q}$ and $B:=(B_t)_{t\geq 0}$ is as in point 3 above. Here {$c_0$ is a constant, and} $\{\tilde \mu\}_{t\geq 0}$ and $(\tilde\sigma_t)_{t\geq 0}$ are  adapted c\'adl\'ag locally bounded processes {such that $\{c_t\}_{0\leq t\leq{}1}$ is strictly positive} and $\widetilde \sigma$ has no fixed jump times.
    \item The {processes} $\Sigma^\chi$, $\tilde\sigma$, ${\mu}$, ${\mu^\chi}$, $\tilde \mu$ are c\'adl\'ag adapted, and {$\delta$} and $\tilde\delta$ are  predictable. {Furthermore, there exist} a sequence $\{\tau_n\}_{n\geq 1}$ of stopping times increasing to infinity, nonnegative {$\lambda(dz)$}-integrable function  {$H$}, {a nonnegative {$\lambda(dz)$}-square-integrable function {$\tilde{H}$}},  and a positive sequence $\{M_n\}_{n\geq 1}$ such that 
    \begin{equation*}
        t\leq \tau_n \Rightarrow \begin{cases}
            |\sigma_t| + |{\mu_t}| + |{\mu_t^\chi}| +|{\tilde\mu_t}|+ |\Sigma_t^\chi| \leq M_n,\\
            ({|\delta(t,z)|}\wedge 1)^{{r}} \leq M_n {H(z)},\\
            {{|\tilde\delta(t,z)|} \leq M_n \tilde{H}(z)},
        \end{cases}
    \end{equation*}
    for some ${r}\in[0,1\wedge Y)$.
    
\end{enumerate}
\end{asu}

\begin{remark}\label{CommntCond}
The conditions above are essentially the same as those in \cite{BONIECE2024}. \cite{LIU2018} considered a similar but more restrictive component of unbounded variation by taking $J$ to be a symmetric strictly stable L\'evy process, while here we consider a more general type of tempered stable L\'evy process {and allow for jumps in volatility}. The most technical condition is that in \eqref{CndqA}, which is used to apply a density transformation or change of probability measure under which the process $J$ becomes stable. We refer to \cite{BONIECE2024} for further details. {The process $c=\sigma^2$ is often constructed by first specifying the dynamics of $\sigma$ and then applying It\^o's formula to derive the dynamics in \eqref{CndCa}. This approach has the advantage of producing a non-negative process $c$, but given that the results involve $\tilde\sigma^2$, we prefer to spcify directly the dynamics on $c$.}
\end{remark}

We assume we have at our disposal $n$ evenly-spaced discrete observations $X_{t_i}$ during a fixed time interval $[0,T]$. For simplicity, we further assume $T=1$ and, thus, $t_i=i/n$. Denote $\Delta_n = 1/n$, $\mathcal{F}_{i}^n:=\mathcal{F}_{t_i}$, and $\Delta^n_iA = A_{i\Delta_n} - A_{(i-1)\Delta_n}$ for any process $A$, where we often omit the superscript $n$. Our estimation target is {$c_\tau=\sigma_\tau^2$} for a fixed time $\tau\in(0,T)$. 
We consider the spot volatility estimator
\begin{equation}\label{MEDHF}
\hat c_{{n,\tau}}(m_n, v_n) := \frac{\sum_{i=1}^{n}K_{m_n\Delta_n}(t_{i-1}-\tau)(\Delta_i^nX)^2\mathbf{1}_{\{|\Delta_i^nX|\leq  v_n\}}}{\Delta_n{\sum_{j=1}^nK_{m_n\Delta_n}(t_{j-1}-\tau)}},
\end{equation}
where $K$ is a kernel function, $m_n$ controls the  bandwidth {$b_n:=m_n\Delta_n$ such that} $m_n \rightarrow \infty$ and $m_n\Delta_n\rightarrow 0$, and  $K_b(x):=K(x/b)/b$. {We often omit the dependence on $\tau$ in $\hat c_{{n,\tau}}(m_n, v_n)$ and simply write $\hat c_{n}(m_n, v_n)$}. We {also} assume the following conditions on $K$, which are standard:
\begin{asu}\label{asu_kernel}
   The kernel function $K:\R\rightarrow [0,\infty)$ is bounded, Lipschitz, {and piecewise $C^1$ on its support $(\ell,{u})$, where $-\infty\leq\ell\leq 0\leq u\leq \infty$}. Additionally, we have $\int K(x)dx = 1$, $\int |K(x)x|dx<\infty$, $\int |K'(x)|dx<\infty$, $K(y)y^2{\rightarrow} 0$, as $|y|\rightarrow \infty$, {and $V_{-\infty}^{\infty} ({K^{\prime}}) < \infty$, where $V_{-\infty}^{\infty}(\cdot)$ is the total variation}.
\end{asu}

{Our analysis assumes observations without microstructure noise and is intended for applications at sampling frequencies where its effects can reasonably be
neglected (typically, minute frequencies). Noise can affect both the squared returns and the truncation, and is not addressed by our jump-bias corrections;  joint treatment would require
a separate analysis. For kernel methods combining
pre-averaging, noise-bias correction, and jump truncation,
we refer to \citet{Figueroa-Lopez_Wu_2024}.}

{The estimator \eqref{MEDHF} without the standardizing denominator, the truncation, or any debiasing was studied in \cite{fan2008spot} and  \cite{kristensen2010nonparametric} (see also additional references in the introduction). The denominator, motivated by the standard {Nadaraya-Watson} nonparametric regression estimator, and noted in passing in \cite{kristensen2010nonparametric}, serves a dual role: it provides edge correction in finite samples and simplifies several proofs,  notably leading to more relaxed technical {assumptions on} the optimal bandwidth regime (see Remark \ref{SlyRmk1} below).  In contrast, when the denominator is omitted, our results remain valid, but stricter conditions on the bandwidth $b_n = m_n \Delta_n$ and its relationship with $Y$ are required; however, in the suboptimal rate regime (see Remark \ref{SlyRmk1}, case (ii), below), the denominator can be dispensed of entirely without any loss of generality.  The truncated version in \eqref{MEDHF} was studied in \cite{figueroa2020optimal} and \cite{Figueroa-Lopez_Wu_2024}, also without debiasing. Finally, Assumption \ref{asu_kernel} permits kernels of unbounded support, which, as discussed in the introduction, enables estimators that are more variance efficient in both the optimal and suboptimal convergence regimes.}

\medskip
\noindent
\textbf{Notation:} Throughout, {$A_n\lesssim B_n$ and $a_n\ll b_n$ mean $A_n = O_P(B_n)$ and $a_n=o(b_n)$, respectively.}  For any process $V$, when unambiguous we write $\Delta_iV$ in place of $\Delta_i^nV$. Stable convergence in law is denoted as $\xrightarrow{ st}$ (see \cite{jacod2012discretization}). %

\section{Main results}\label{SectMainResults}
{The following moment expansions build on Lemmas~2--5 of
\cite{BONIECE2024}. The leading terms have the same form,
while the remainder estimates are adapted to the precision
required for spot estimation. The proof in
Appendix~\ref{PrfProp1a} checks the coefficient-freezing
estimates in the presence of jumps in volatility and
controls the contribution of the finite-variation jump
component. In particular, the revised remainder bounds
for the second-moment and threshold-difference expansions
permit a wider range of jump activity under the threshold
conditions stated below.}
\begin{proposition}\label{prop}
Suppose that $Y \in (0,2)\setminus\{1\}$. Let $\Delta_n^{\frac{1}{2}-s}\ll  v_n \ll \Delta_n^{\frac{1}{4-Y}}$ for a fixed $s\in(0,1/2)$. Then, the following statements hold:
\begin{enumerate}
    \item For any integer $p\geq 1$, we have
    \begin{equation}\label{GNMExp}
    \begin{aligned}
        \EE\big((\Delta_iX)^{2 p}\mathbf{1}_{\{|\Delta_iX|\leq  v_n\}}|\mathcal F_{i-1}\big)&=(2{p}-1)!!\sigma_{t_{i-1}}^{2{p}}\Delta_n^{ p}+C_{{p},i}\Delta_n v_n^{2{p}-Y}\\
        &\quad+o_P(\Delta_n v_n^{2{p}-Y})+o_P(\Delta_n^{p}),
    \end{aligned}
    \end{equation}
    where $C_{p,i} := \frac{(C_++C_-)|\chi_{t_{i-1}}|^Y}{2p-Y}$.
    \item Furthermore, if $p=1$, the following high-order expansion holds: 
        \begin{equation}\label{PartExp0}
\begin{aligned}
        \EE\big((\Delta_iX)^{2{}}\mathbf{1}_{\{|\Delta_iX|\leq  v_n\}}|\mathcal F_{i-1}\big)&=\sigma_{t_{i-1}}^2\Delta_n+C_{1,i}\Delta_n v_n^{2-Y}+D_{1,i} \Delta_n^2 v_n^{-Y}\\
        &\quad+O_P(\Delta_n^3 v_n^{-2-Y}) +o_P(\Delta_n^{5/4}),
\end{aligned}
\end{equation}
where $D_{1,i}={-}\frac{(C_++C_-)(Y+1)(Y+2)}{2Y}\sigma^2_{t_{i-1}}|\chi_{t_{i-1}}|^Y$. Additionally, {for any $\zeta>1$,} we have
\begin{align}
\nonumber
\EE\big((\Delta_iX)^{2{}}\mathbf{1}_{\{v_n<|\Delta_iX|\leq  \zeta v_n\}}|\mathcal F_{i-1}\big)
&=C_{1,i}\Delta_n(\zeta^{2-Y}-1) v_n^{2-Y}+D_{1,i} \Delta_n^2 (\zeta^{-Y}-1)v_n^{-Y}\\
\label{PartExp0_difference}
&\quad 
+O_P(\Delta_n^3 v_n^{-2-Y})
+o_P(\Delta_n^{\frac{3}{4}}v_n^{\frac{4-Y}{2}}).
\end{align}
    \item  For any $p>1$ (not necessarily an integer), we have
\begin{align}\label{prop1_non_integer}
        &\EE\big(|\Delta_iX|^{2p}\mathbf{1}_{\{|\Delta_iX|\leq  v_n\}}|\mathcal F_{i-1}\big) = O_P(\sigma_{t_{i-1}}^{2{p}}\Delta_n^{p})+O_P(C_{{p},t_{i-1}}\Delta_n v_n^{2{p}-Y}).
    \end{align}
\end{enumerate}
\end{proposition}
To analyze the asymptotic {behavior} of the estimator $\hat c_n(m_n, v_n)$ {defined in Eq.~\eqref{MEDHF}},  we consider the following decomposition:
\begin{align}\label{kernel_spot_estimator}
    \hat c_n(m_n, v_n)
    &=\left(\hat c_n(m_n, v_n) - \frac{\Delta_n\sum_{i=1}^{n}K_{m_n\Delta_n}(t_{i-1}-\tau)\sigma_{t_{i-1}}^2}{\Delta_n{\sum_{j=1}^nK_{m_n\Delta_n}(t_{j-1}-\tau)}}\right) \nonumber\\
    &\quad+ \left(\frac{\Delta_n\sum_{i=1}^{n}K_{m_n\Delta_n}(t_{i-1}-\tau)\sigma_{t_{i-1}}^2}{\Delta_n{\sum_{j=1}^nK_{m_n\Delta_n}(t_{j-1}-\tau)}}\right)\nonumber
    \\ 
    &=: \hat c_{n,1}(m_n, v_n) + \hat c_{n,2}(m_n).
\end{align}
Our next result is central to our analysis and underpins our main debiasing procedure. It establishes the joint asymptotic behavior of the terms in \eqref{kernel_spot_estimator}, accounting for appropriate bias terms, and provides the framework for separating qualitatively distinct asymptotic regimes for $m_n$.  In addition, on its own, it yields feasible CLTs for $\hat c_n(m_n,v_n)$ in certain settings that sharpen existing results in the literature (see Remark \ref{SlyRmk1} below). The proof is given in {Section \ref{31Ap}}.
\begin{theorem}\label{clt}
    Suppose that $\Delta_n^{\beta}\ll  v_n \ll \Delta_n^{\beta'}$ with $\frac{1}{4-Y}<\beta'<\beta<\frac{1}{2}$, $\sqrt{\log(n)} \ll m_n\ll \Delta_n^{-4}v_n^{4+2Y}$, and $m_n = O(\Delta_n^{-\frac{1}{2}})$. Let
    \begin{align}
    \tilde Z_n(m_n, v_n) &:= \sqrt{m_n}\Big(\hat c_{n,1}(m_n, v_n) -A(v_n,m_n)\Big),\nonumber\\
    \tilde Z_n'(m_n) &:= \frac{1}{\sqrt{m_n\Delta_n}}\Big(\hat c_{n,2}(m_n) -\sigma^2_\tau\Big),\nonumber
\end{align}
where, using the notation in \eqref{PartExp0},
\begin{align}\label{AnTerm}
   {A(v_n,m_n) := \frac{\sum_{i=1}^nK_{m_n\Delta_n}(t_{i-1}-\tau){C_{1,i}}\Delta_n v_n^{2-Y}}{\Delta_n\sum_{j=1}^nK_{m_n\Delta_n}(t_{j-1}-\tau)} + \frac{\sum_{i=1}^nK_{m_n\Delta_n}(t_{i-1}-\tau){D_{1,i}} \Delta_n^2 v_n^{-Y}}{\Delta_n\sum_{j=1}^nK_{m_n\Delta_n}(t_{j-1}-\tau)}.}
\end{align}
Then, under Assumptions \ref{BscAs1} and \ref{asu_kernel}, as $n\rightarrow \infty$,
\begin{align}\label{clt_spot}
    (\tilde Z_n(m_n, v_n),\tilde Z'_n(m_n))
\xrightarrow{ st}  (Z_1,Z_2):={\left(\sqrt{2}\sigma_\tau^2\| K\|_2 U_1, \tilde\sigma_{\tau}\|L\|_2 U_2\right)},
\end{align}
with $\|K\|_2=\sqrt{\int K^2(x)dx}$, $L(t) = \int_t^{\infty}K(u)du\mathbf{1}_{t\geq 0}-\int^t_{-\infty}K(u)du\mathbf{1}_{t<0}$, $\|L\|_2=\sqrt{\int L^2(t)dt}$,
and, {conditional on $\mathscr{F}$, $U_1$ and $U_2$ are independent standard normal variables}.
\end{theorem}
\begin{remark}\label{SlyRmk1} 
Taking $m_n\to\infty$ at different rates alters the balance of the errors appearing on the right-hand side of \eqref{clt_spot}.  Three distinct regimes can be distinguished:
(i) $m_n\sqrt{\Delta_n}\to\theta\in(0,\infty)$,
(ii) $m_n\sqrt{\Delta_n}\to 0$,
and (iii) $m_n\sqrt{\Delta_n}\to\infty$, which we remark on below.

\medskip
\noindent \textbf{Case (i): $m_n \sqrt{\Delta_n} \to \theta \in (0,\infty)$.} In this regime, the orders of $\sqrt{m_n}$ and $1/\sqrt{m_n\Delta_n}$ match, in which case Theorem \ref{clt} gives %
    \begin{equation}\label{ITODelta}
        \Delta_n^{-1/4}\left(\hat c_n(m_n, v_n)-A(v_n,m_n)-\sigma_{\tau}^2\right)\xrightarrow{ st} \theta^{-1/2}Z_1+\theta^{1/2} Z_2,
    \end{equation}
or, equivalently, 
    \begin{equation}\label{ITQRTm}
        \sqrt{m_n}\left(\hat c_n(m_n, v_n)-A(v_n,m_n)-\sigma_{\tau}^2\right)\xrightarrow{ st} Z_1+\theta Z_2.
    \end{equation}
Note that the bias term $A(v_n,m_n)$ is $O_P(v_n^{2-Y})$.  Thus, if  $\sqrt{m_n} v_n^{2-Y}\ll 1$, then \eqref{ITODelta} and \eqref{ITQRTm} remain true even with the term $A(v_n,m_n)$ omitted, in which case we obtain a valid feasible CLT for the estimation error $\hat c_n(m_n, v_n)-\sigma_{\tau}^2$ with convergence rate  $\Delta_n^{1/4}$ (which is optimal). Although condition $\sqrt{m_n} v_n^{2-Y}\ll 1$ can be met in principle by taking sufficiently small $v_n$,  Theorem \ref{clt} additionally requires $v_n\gg \Delta_n^{1/2}$, and thus implicitly $\Delta_n^{-1/4}\Delta_n^{(2-Y)/2}\ll 1$, which can happen only if $Y<3/2$. Nevertheless, this result improves significantly on \cite{Figueroa-Lopez_Wu_2024} (and also on {\cite{jacod2012discretization}}), where a rate-optimal CLT is shown to be valid only when $Y<4/3$. Theorem \ref{clt} also yields stronger results than the characteristic function approach taken in \cite{LIU2018}, in which a similar CLT is shown to be valid, under substantially more restrictive model assumptions\footnote{For instance, Liu et al.~assume that the process $J$ in \eqref{model} is a symmetric stable L\'evy process}, for $Y<3/2$ at a convergence rate strictly slower than $\Delta_n^{1/4}$ (due to an assumption that $m_n\sqrt{\Delta_n}\to{}0$). 

Furthermore, under case (i),  the conditions in $v_n$ take the form $\Delta_n^{\beta}\vee \Delta_n^{\frac{7}{4Y+8}}\ll v_n\ll \Delta_n^{\beta'}$ ($\frac{1}{4-Y}<\beta'<\beta<\frac{1}{2}$), which can hold only if $\frac{7}{4Y+8}>\frac{1}{4-Y}$, which is equivalent to $Y<\frac{20}{11}$. Therefore, this rate-optimal case (i) extends the previously available range $Y<4/3$ to $Y<20/11$, and ultimately reinforced our feasible rate-optimal CLT in this range (Theorem \ref{main}).
\medskip

\noindent \textbf{Case (ii): $m_n \sqrt{\Delta_n} \to  0$.}  This regime corresponds to setting $\theta=0$ in \eqref{ITQRTm}. Indeed, \begin{equation}\label{ITQRTmC2}
        \sqrt{m_n}\left(\hat c_n(m_n, v_n)-A(v_n,m_n)-\sigma_{\tau}^2\right)
        =\tilde Z_n(m_n, v_n)+m_n\sqrt{\Delta_n}\tilde Z'_n(m_n),
    \end{equation}
    and, from \eqref{clt_spot}, the first term on the right-hand side above converges to $Z_1$, while the second is $o_P(1)$. In this case, the attained convergence rate,  $\sqrt{m_n}$, is suboptimal, i.e., slower than the  $n^{1/4}$ rate achieved when $\theta\in(0,\infty)$.  However,  this asymptotic regime has the practical advantage that any eventual estimation of the volatility-of-volatility $\tilde\sigma_\tau^2$ is not needed when constructing confidence intervals, since $Z_2$ is absent from the limit. If we also require $\sqrt{m_n} v_n^{2-Y}\ll 1$, then the bias $A(v_n,m_n)$ can be omitted in \eqref{ITQRTmC2} leading to a feasible CLT for $\hat c_n(m_n, v_n)$ centered at $\sigma_{\tau}^2$.  Under this constraint, the condition  $v_n\gg \Delta_n^{\frac{1}{2}}$, and the asymptotic condition of case (ii) together imply that the convergence rate  $m_n^{-1/2}$ must be strictly slower than $\Delta_n^{\frac14 \wedge \frac {2-Y}2}$, yielding a convergence rate that can be taken arbitrarily close to  $n^{1/4}$  when $Y<3/2$ but deteriorates as $Y$ increases to $2$.
    \medskip

\noindent \textbf{Case (iii): $m_n \sqrt{\Delta_n} \to \infty$.}
 Although technically excluded from our hypotheses, our results can also cover the regime $m_n\sqrt{\Delta_n}\to\infty$, but at the expense of an upper bound of $Y< 8/5$ in the jump activity index. In this setting, the convergence rate achieved is $\sqrt{m_n\Delta_n}$, which is not as fast as the analogous rates when $\theta\in[0,\infty)$.
\end{remark}

\begin{remark} Beyond finite-sample edge correction, the denominator
in \eqref{MEDHF} (which converges to 1) relaxes the technical
conditions used to establish the joint convergence in
\eqref{clt_spot}; without it, our argument imposes additional
restrictions on $Y$ and $m_n$. 
   Specifically, define ${\bar c_n(m_n,
   v_n)} := \sum_{i=1}^{n} K_{m_n\Delta_n}(t_{i-1}-\tau)(\Delta_iX)^2 \mathbf{1}_{\{|\Delta_iX|\leq v_n\}}$, and let $\tilde Z_n(m_n, v_n)$, $\tilde Z'_n(m_n)$, and $A(v_n,m_n)$ denote the corresponding analogs without the denominator factor. Then the limit \eqref{clt_spot} still holds under the hypotheses of Theorem~\ref{clt}, provided also that $m_n \gg \Delta_n^{-1/3}$. Under this additional condition on $m_n$, \eqref{ITQRTm} remains valid for $Y<20/11$ when $m_n \sqrt{\Delta_n} \to \theta \in (0,\infty)$. However, a much more substantive difference arises when $m_n \sqrt{\Delta_n} \to 0$: if we also impose $\sqrt{m_n} v_n^{2-Y}\ll 1$ to obtain a feasible CLT (as discussed in Remark~\ref{SlyRmk1}), the constraint becomes $\Delta_n^{-1/3} \ll m_n \ll v_n^{2(Y-2)} \ll \Delta_n^{Y-2}$, which requires $Y<5/3$ {for this application of the joint limit
theorem. This additional restriction is not needed for the
$\sqrt{m_n}$-scaled CLT for the estimator itself: by
\eqref{kernel_convergence},
$\sqrt{m_n}(\bar c_n-\hat c_n)=O_P(m_n^{-1/2})=o_P(1)$.}
\end{remark}

The next result is the second key ingredient needed for our debiasing method.  For an arbitrary fixed $\zeta>1$, it establishes the convergence rate of the difference $\tilde Z_n(m_n, \zeta v_n)-\tilde Z_n(m_n, v_n)$, i.e. of $\hat c_{n,1}(m_n,\zeta v_n)-\hat c_{n,1}(m_n, v_n)$ subject to a  bias correction, which we exploit in our debiasing procedure. The proof is given in Section \ref{Ap3_4}.
\begin{theorem}\label{thm_clt_difference}
Suppose that $1<Y<2$, $ \sqrt{\log(n)} \ll m_n$, $m_n = O(\Delta_n^{-\frac{1}{2}})$,  $\Delta_n^{\beta}\ll  v_n \ll \Delta_n^{\beta'}$ with $\frac{1}{4-Y}<\beta'<\beta<\frac{1}{2}$, and $m_n \rightarrow \infty$ such that $v_n^{\gamma}\Delta_n^{-1}\ll m_n \ll \Delta_n^{-5} v_n^{8+Y}$ with some $\gamma<Y$. 
Then, under Assumptions \ref{BscAs1} and \ref{asu_kernel}, for an arbitrary fixed $\zeta > 1$, we have
\begin{align}\label{clt_difference}
    u_n^{-1}\big(\tilde Z_n(m_n, \zeta v_n)-\tilde Z_n(m_n, v_n)\big)
\xrightarrow{ st} 
{s U},
\end{align}
as $n\rightarrow \infty$, where $u_n:= \Delta_n^{-1/2} v_n^{2-Y/2}\rightarrow 0$, {$s^2:=\frac{(C_++C_-)|\chi_{\tau}|^Y}{4-Y}(\zeta^{4-Y}-1)\int K^2(x)dx$, and $U$ is a standard normal variable independent of $\chi$.}
\end{theorem}

\begin{remark}\label{INTNTL}
The conditions on $m_n$ in Theorem \ref{thm_clt_difference} are stricter than those in Theorem \ref{clt} due to the condition $\Delta_n^{-1}v_n^\gamma\ll m_n\ll \Delta_n^{-5}v_n^{8+Y}$ (note $\Delta_n^{-5}v_n^{8+Y} \ll \Delta_n^{-4}v_n^{4+2Y}$, which was the  upper bound for $m_n$ in Theorem \ref{clt}).
As we shall see, the optimal rate of convergence of the estimation error of our ultimate debiased estimator can be achieved by taking $m_n=\theta\Delta_n^{-1/2}$ for any constant $\theta>0$ (cf.~Remark \ref{SlyRmk1}). In this bandwidth regime, the additional constraint  on $m_n$ in Theorem \ref{thm_clt_difference} reduces to $\Delta_n^{\frac{9}{16+2Y}}\ll  v_n \ll  \min(\Delta_n^{\frac{1}{4-Y}},\Delta_n^{\frac{1}{2\gamma}})$, whose upper bound and lower bound {can be made} compatible only if ${Y}<20/11$.
\end{remark}

Finally we are ready to formally construct our debiasing estimators.  Our debiasing steps follow the same idea as in \cite{Jacod2014}. Write $\vec \zeta_k=(\zeta_1,\dots,\zeta_k)$, where $\zeta_1,\dots,\zeta_k>1$ are fixed constants. Let
\begin{align}\label{debias_formula}
\hat c^{(k)}_n(m_n,v_n,\vec \zeta_k) &:= \hat c^{(k-1)}_n(m_n, v_n,\vec \zeta_{k-1}) \\
&\quad- \frac{\big(\hat c^{(k-1)}_n(m_n,\zeta_k v_n,\vec \zeta_{k-1})-\hat c^{(k-1)}_n(m_n, v_n,\vec \zeta_{k-1})\big)^2}{\hat c^{(k-1)}_n(m_n,\zeta_k^2 v_n,\vec \zeta_{k-1})-2\hat c^{(k-1)}_n(m_n,\zeta_k v_n,\vec \zeta_{k-1})+\hat c^{(k-1)}_n(m_n, v_n,\vec \zeta_{k-1})},
\nonumber
\end{align}
for each $k\geq 1$, with {$\hat c^{(0)}_n(m_n, v_n,\vec \zeta_{0}) := \hat c_{{n,\tau}}(m_n, v_n)$ as defined in \eqref{MEDHF}}. Our next theorem is the second main result of this paper; it establishes a CLT for the  feasible estimator $\hat c^{(2)}_n(m_n,v_n,\vec \zeta_2)$. The proof is given in {Section \ref{Ap3_6}}.
\begin{theorem}\label{main}
     Let
    \begin{align}
    \tilde Z^{(2)}_n(m_n, v_n) &:= \sqrt{m_n}\Big(\hat c^{(2)}_n(m_n,v_n,\vec \zeta_2) -\hat c_{n,2}(m_n)\Big).\nonumber
\end{align}
Then, under the assumptions and notation of Theorems  \ref{clt} and \ref{thm_clt_difference}, as $n\rightarrow \infty$,
\begin{align}\label{clt_spot2}
    (\tilde Z^{(2)}_n(m_n, v_n),\tilde Z'_n(m_n))
\xrightarrow{ st}  (Z_1,{Z_2}).
\end{align}
Furthermore, 
if we set $$m_n\sqrt{\Delta_n} \rightarrow \theta, \quad \text{with} \quad\theta\in{[0,\infty)},$$
we have the following stable convergence in law, as $n\rightarrow \infty$:
\begin{equation}\label{MnDebRes2}
    \sqrt{m_n} (\hat c^{(2)}_n(m_n, v_n,\vec \zeta_2)-\sigma_\tau^2) \xrightarrow{ st} Z_1 + \theta{Z_2}.
\end{equation}
\end{theorem}

\begin{remark}\label{twostep_remark}
Recall that Remark \ref{SlyRmk1} identifies two bandwidth regimes of primary importance: (i) $0<\theta<\infty$, and (ii) $\theta=0$, both with a {convergence} rate dictated by  $m_n^{1/2}$. In case (i),  a rate-optimal (feasible) CLT for $\hat c_n^{(2)}(m_n, v_n,\vec \zeta_2)$ is attained in \eqref{MnDebRes2} since $m_n^{-1/2}\asymp \Delta_n^{1/4}$, provided $1< Y<20/11$, but no CLT is available in case (i) when $Y\geq 20/11$. In case (ii), the rate is strictly slower than  $n^{1/4}$ (see Remark \ref{SlyRmk1} for details), but a CLT still holds for $\hat c^{(2)}_n(m_n, v_n,\vec \zeta_2)$ even when $20/11\leq Y<2$ at the expense of a deteriorating rate.  Indeed, since $m_n\ll \Delta_n^{-5}v_n^{8+Y}$ and $v_n\ll\Delta_n^{\frac{1}{4-Y}}$, we have $m_n\ll \Delta_n^{(6Y-12)/(4-Y)}$ and hence $m_n^{-1/2}\gg \Delta_n^{(6-3Y)/(4-Y)}$, which tends to zero more slowly as $Y$ gets closer to $2$.
\end{remark}

\begin{remark}\label{onestep_remark}
In Theorem \ref{main}, we proceeded with a two-step debiasing procedure to iteratively remove bias terms in $A(v_n,m_n)$ one at a time. However, if the second-order term in \eqref{AnTerm} is negligible when multiplied by $u_n^{-1}\sqrt{m_n}$ (and, thus, \eqref{clt_difference} is valid replacing $A(v_n,m_n)$ with only the first-order term $\sum_{i=1}^nK_{m_n\Delta_n}(t_{i-1}-\tau)C_{1,i}\Delta_n v_n^{2-Y}$), then only one debiasing step is needed. Specifically, for one-step debiasing it suffices that
\[
u_n^{-1}\sqrt{m_n}\sum_{i=1}^n K_{m_n\Delta_n}(t_{i-1}-\tau)D_{1,i}\Delta_n^2 v_n^{-Y}\ll 1.
\]
Therefore, we only require that $u_n^{-1}\sqrt{m_n}\Delta_nv_n^{-Y}\ll 1$ or, equivalently, $m_n\ll \Delta_n^{-3} v_n^{Y+4}$. In the optimal-rate bandwidth regime $m_n\asymp \Delta_n^{-1/2}$, this imposes the relation $\Delta_n^{\frac{5}{2(Y+4)}}\ll v_n$, but since $v_n\ll \Delta_n^{\frac{1}{4-Y}}$, this yields the restriction $Y<12/7$. As in Remark \ref{twostep_remark}, if $Y\geq 12/7$ and $\theta=0$, the one-step debiased estimator $\hat c^{(1)}_n(m_n, v_n,\zeta_1)$ will still enjoy a CLT centered at $\sigma_\tau^2$, but with a rate strictly slower than $\Delta_n^{\frac{4-2Y}{4-Y}}$, which tends to zero more slowly as $Y$ gets closer to $2$.
\end{remark}

{\begin{remark}\label{finitevariation_remark}
Suppose $0<Y<1$, and that
\begin{equation}\label{FVtuning}
\begin{gathered}
\Delta_n^\beta\ll v_n\ll\Delta_n^{\beta'},\qquad \frac13<\beta'<\beta<\frac12,\qquad
\sqrt{\log(n)}\ll m_n,\\
m_n\sqrt{\Delta_n}\longrightarrow\theta\in[0,\infty).
\end{gathered}
\end{equation}
Then, under Assumptions \ref{BscAs1} and \ref{asu_kernel},
\begin{equation}\label{FVbaselineCLT}
\sqrt{m_n}\big(\hat c_{n,\tau}(m_n,v_n)-\sigma_\tau^2\big)\xrightarrow{st}Z_1+\theta Z_2,
\end{equation}
where $Z_1,Z_2$ are defined as in Theorem \ref{clt}.
Thus, \eqref{ITQRTm} holds with $A(v_n,m_n)=0$. In particular, the convergence rate is $n^{1/4}$ when $\theta>0$; see
\cite[Theorem~2.1(c)]{Figueroa-Lopez_Wu_2024}. We include the comparison estimates in Appendix \ref{AppFVbaseline}
for completeness.
\end{remark}

\begin{remark}\label{finitevariation_local_remark}
Under the assumptions and tuning conditions of
Remark \ref{finitevariation_remark}, suppose further that $K$ has bounded
support. In \eqref{debias_formula}, set the debiasing term to zero
whenever its denominator is zero. Then, for fixed $\zeta_1,\zeta_2>1$,
\begin{equation}\label{FVlocalequality}
\PP\left(
\hat c_n^{(j)}(m_n,v_n,\vec\zeta_j)
=\hat c_{n,\tau}(m_n,v_n),\ j=1,2
\right)\longrightarrow1.
\end{equation}
Hence, the one- and two-step local estimators also satisfy
\eqref{FVbaselineCLT}. The proof is given in Appendix \ref{AppFVlocal}.
\end{remark}
}

\section{Global-local debiased estimators} \label{SectionGlobLoc}
A common approach to implement the debiasing procedure \eqref{debias_formula} consists of splitting the correction term as follows: 
\begin{equation}\label{debias_formula2b}
\tilde c^{(1)}_{n,\tau}:=\tilde c^{(1)}_{n,\tau}(m_n,v_n,\zeta) := \hat c_{n,\tau}(m_n, v_n) - \widehat\Upsilon_n(\zeta) \left(\hat c_{n,\tau}(m_n,\zeta v_n)-\hat c_{n,\tau}(m_n,v_n)\right),
\end{equation}
where $\hat c_{n,\tau}(m_n, v_n)$ is defined as in \eqref{MEDHF} and $\widehat\Upsilon_n(\zeta)$ is a ``global'' estimator of the parameter:
\[
    \Upsilon(\zeta):=\frac{1}{\zeta^{2-Y}-1},
\]
which is the leading order term in the asymptotic expansion of the first-to-second-finite-difference ratio of $\hat c_{n,\tau}(m_n,v_n)$:
\begin{equation}\label{e:speed_to_accel}
    \frac{\hat c_{n,\tau}(m_n,\zeta v_n)-\hat c_{n,\tau}(m_n,v_n)}{\hat c_{n,\tau}(m_n,\zeta^2 v_n)-2\hat c_{n,\tau}(m_n,\zeta v_n)+\hat c_{n,\tau}(m_n,v_n)}.
\end{equation}
Here, global means that $\widehat\Upsilon_n$ is based on the estimates $\hat c_{n,\tau_\ell}(m_n,\zeta v_n)$ for $\tau_\ell$ ranging over the entire time horizon $[0,1]$ as follows:
\begin{equation}\label{EstSlp0}
\widehat\Upsilon_n(\zeta)
:=\frac{\sum_{\ell}\big(\hat c_{n,\tau_{\ell}}(m_n,\zeta pv_n)-\hat c_{n,\tau_{\ell}}(m_n,pv_n)\big)}{\sum_{\ell}\left(\hat c_{n,\tau_{\ell}}(m_n,\zeta^2 pv_n)-2\hat c_{n,\tau_{\ell}}(m_n,\zeta pv_n)+\hat c_{n,\tau_{\ell}}(m_n,pv_n)\right)},
\end{equation}
and $p\in(0,1)$ is an additional tuning parameter (for simplicity, the dependence of $\widehat\Upsilon_n(\zeta)$ on $p$ and $v_n$ is suppressed). This approach was originally suggested by \cite{Jacod2014} for the estimation of integrated volatility to reduce the noise in bias estimation, where it was argued that inclusion of the additional parameter $p\in(0,1)$ leads to better estimates. Interestingly, the asymptotic properties of the above approach have not been addressed in the literature. In this section, we shed some light on this approach for our spot volatility estimators. 

For simplicity, we take the kernel $K$ used to estimate the ratio $\widehat\Upsilon_n$ to be uniform in $(0,1]$ and $\tau_\ell=\ell \Delta_n m_n$ to avoid overlapped intervals and reduce the dependence between the estimates $\hat c_{n,\tau_{\ell}}(m_n,\zeta pv_n)$ in \eqref{EstSlp0} (these restrictions could be relaxed, but the technical details are beyond the scope of this article).
In that case, we can rewrite $\widehat\Upsilon_n$ as %
\begin{equation}\label{EstSlp0b}
\widehat\Upsilon_{n}(\zeta)=\frac{\widehat{C}_n(\zeta pv_n)-\widehat{C}_n(pv_n)}{\widehat{C}_n(\zeta^2 pv_n)-2\widehat{C}_n(\zeta pv_n)+\widehat{C}_n(pv_n)},
\end{equation}
where $\widehat{C}_n(v_n)$ is the usual integrated volatility estimator \eqref{TRVEst}:
\[
\widehat{C}_n(v_n):=\sum_{i=1}^n(\Delta_{i}^nX)^2{\bf 1}_{\{|\Delta_i^nX|\leq v_n\}}.
\]
The asymptotic behavior of $\widehat{C}_n(\zeta v_n)-\widehat{C}_n(v_n)$ was studied in detail in \cite{BONIECE2024}, where the following result was obtained (see Remark 2 therein):
\begin{align*}%
	&{v}_n^{-Y/2}\left(\frac{\widehat C_n(\zeta{v}_n)-\widehat C_n({v}_n)}{{v}_n^{2-Y}}-\frac{\bar{C}}{2-Y}(\zeta^{2-Y}-1)\int_0^1|\chi_s|^{Y}ds\right)\\
	&\qquad \toDistSt  \mathcal{N}\left(0,\frac{\bar{C}}{4-Y}(\zeta^{4-Y}-1)\int_{0}^{1}|\chi_s|^Yds\right),
\end{align*}
where $\bar{C}:=C_++C_-$. 
In particular,
\begin{align}\label{SmplImpFMLb}
	 \frac{\widehat C_n(\zeta{v}_n)-\widehat C_n({v}_n)}{{v}_n^{2-Y}}\stackrel{\mathbb{P}}{\longrightarrow}{} \frac1{\Upsilon(\zeta)}\frac{\overline C}{2-Y}\int_0^1|\chi_s|^{Y}ds.
\end{align}
Combining the result above from \cite{BONIECE2024} and the results in the present paper, we can then deduce the following (see Section \ref{Ap4_1} for its proof):
\begin{theorem}\label{main2b}
    Let ${\tilde c^{(1)}_{n,\tau}}$ be given as in \eqref{debias_formula2b} with $\widehat\Upsilon_n$ be given as in \eqref{EstSlp0b}. Set $m_n\sqrt{\Delta_n} \rightarrow \theta$, with $\theta\in{(0,\infty)}$ and, in addition to the constraints of Theorem \ref{thm_clt_difference}, suppose that $v_n\gg \Delta_n^{\frac{3}{4Y}}$ (in particular, since $v_n\ll \Delta_n^{1/(4-Y)}$, we necessarily have $Y<12/7$). Then, under the model assumptions of Section \ref{Sect2}, as $n\rightarrow \infty$,
\begin{equation}\label{MnDebRes2b}
    \sqrt{m_n} \left(\tilde c^{(1)}_{n,\tau}-\sigma_\tau^2\right) \xrightarrow{ st} Z_1 + \theta{Z_2}.
\end{equation}
\end{theorem}

For two-step global/local debiasing, {we take an alternative direction to \eqref{debias_formula}.} Let $\zeta_1,\zeta_2>1$ denote the
threshold multipliers used in the first and second steps,
respectively. After cancellation of the leading bias term, based on Proposition \ref{prop}, the second order bias term, as a function of the threshold $v_n$, is proportional to $v_n^{-Y}$.  The corresponding bias-correction
coefficient in the second step is therefore
$$
\frac{1}{\zeta_2^{-Y}-1}
=
\frac{\zeta_2^2\Upsilon(\zeta_2)}
{1-(\zeta_2^2-1)\Upsilon(\zeta_2)}.
$$
This identity suggests the two-step global-local estimator
\begin{equation}\label{Twicedebias_formula2b}
\tilde{c}^{(2)}_{n,\tau}:= \tilde c^{(1)}_{n,\tau}(m_n,v_n,\zeta_1) - \left(\frac{\zeta_2^2\widehat\Upsilon_n(\zeta_2)}{1-(\zeta_2^2-1)\widehat\Upsilon_n(\zeta_2)}\right) \left(\tilde c^{(1)}_{n,\tau}(m_n,\zeta_2 v_n,\zeta_1)-\tilde c^{(1)}_{n,\tau}(m_n,v_n,\zeta_1)\right).
\end{equation}

The following result shows the rate-optimality of the two-step-debiased estimators above. Its proof is given in Section \ref{Ap4_3}.
\begin{theorem}\label{maintwo-step2b}
    Set $m_n\sqrt{\Delta_n} \rightarrow \theta$, with $\theta\in{(0,\infty)}$ and, in addition to the constraints of Theorem \ref{thm_clt_difference}, suppose that $\Delta_n^{\frac{9}{2(8+Y)}}\ll v_n\ll {\Delta_n^{\frac{5}{2(4+Y)}}}$  (in particular, since we also require $v_n\ll \Delta_n^{1/(4-Y)}$, we necessarily have $Y<20/11$). Then, under the assumptions of Section \ref{Sect2}, as $n\rightarrow \infty$,
\begin{equation}\label{MnDebResTwo-step2c}
    \sqrt{m_n} \left(\tilde c^{(2)}_{n,\tau}-\sigma_\tau^2\right) \xrightarrow{ st} Z_1 + \theta{Z_2},
\end{equation}
\end{theorem}
{\begin{remark}\label{finitevariation_GL_remark}
Suppose that $0<Y<1$ and that the assumptions and tuning conditions \eqref{FVtuning} of
Remark \ref{finitevariation_remark} hold.
Let $\tilde c^{(1)}_{n,\tau}$ and $\tilde c^{(2)}_{n,\tau}$
be defined as in \eqref{debias_formula2b} and \eqref{Twicedebias_formula2b},
with $\widehat\Upsilon_n$ given by \eqref{EstSlp0b}.
For fixed $p\in(0,1)$ and $\zeta_1,\zeta_2>1$, we have
\begin{equation}\label{FVGLnegligible}
\sqrt{m_n}\big(\tilde c^{(j)}_{n,\tau}
-\hat c_{n,\tau}(m_n,v_n)\big)=o_P(1),\qquad j=1,2.
\end{equation}
Consequently,
$$
\sqrt{m_n}\big(\tilde c^{(j)}_{n,\tau}-\sigma_\tau^2\big)
\xrightarrow{st}Z_1+\theta Z_2,\qquad j=1,2.
$$
Here each ratio is set to zero when its denominator is zero. As shown in Appendix \ref{AppFVGL}, the probability of these events
tends to zero.
\end{remark}
}
\section{Simulation study} \label{SectionSimul}
In this section, we examine the finite-sample performance of the global/local {(GL)}
debiased estimators introduced in Section~\ref{SectionGlobLoc}.  These estimators retain the asymptotic guaranties of the fully local construction while being designed to provide greater finite-sample stability. We therefore
take them as our principal implementation in the simulations below.

\subsection{Implementation}

We compare our truncated-kernel estimator
$\hat c_{n,\tau}(m_n,v_n)$, the one-step global-local estimator
$\tilde c^{(1)}_{n,\tau}(m_n,v_n,\zeta_1)$, and two-step global-local estimator  $\tilde c ^{(2)}_{n,\tau}(m_n,v_n,\zeta_1,\zeta_2)$, in two variants each, 
one with the exponential kernel $K(x)=\exp(-|x|)/2$ and the other with the two-sided uniform
kernel $G(x)=\mathbf{1}_{\{|x|\leq1\}}/2$.  Throughout we set
$$
m_n\Delta_n=2\sqrt{\Delta_n},\qquad
v_n=2\sqrt{BV_\tau}\,\Delta_n^{0.45},
$$
where $BV_\tau$ is a centered local bipower-variation using a 2-calendar-day window. The selection of the parameters $\zeta_1,\zeta_2$ and $p$ are detailed in the next section.

For comparison, we consider two estimators provided by the characteristic-function approach of
\citet{LIU2018}: their initial estimator without jump-bias correction as well as their bias-corrected estimator, including the augmentations used in their numerical study.   Liu et al.\ introduce the augmentations specifically to improve finite-sample performance by combining their bias correction with
additional aggregation and sign adjustments.  Though it is not strictly covered by their asymptotic theory, we take their augmented implementation as our principal benchmark method, since these augmentations were introduced specifically to improve finite-sample performance of their approach.

Following the notation and settings of \citet{LIU2018}, let $h$ denote kernel bandwidth (the analog of $m_n\Delta_n$ in our estimator) and write $K_h(x)=h^{-1}K(x/h)$.  The quantity $u_n$ is the
argument at which the local empirical characteristic function is evaluated.
Set
\begin{align*}
\bar\sigma^2_{\tau,n}(u_n,h)
&= -\frac{2}{u_n^2}
\log\left(S_{\tau,n}(u_n,h)\vee\sqrt{\frac{\Delta_n}{h}}\right),\\
S_{\tau,n}(u_n,h)
&= \Delta_n\sum_{i=1}^n K_h(i\Delta_n-\tau)
\cos\left(\frac{u_n\Delta_iX}{\sqrt{\Delta_n}}\right).
\end{align*}
Their local bias-correction term and resulting bias-corrected estimator are
\begin{align}
\label{LiuEstDfn0}
\widehat\sigma^2_{\tau,n}(u_n,h)
&=\bar\sigma^2_{\tau,n}(u_n,h)
-\frac{2\Delta_n}{u_n^2h}
\left[\sinh\left(\frac{u_n^2}{2}
\bar\sigma^2_{\tau,n}(u_n,h)\right)\right]^2,\\
\widehat B_{\tau,n}(\lambda,u_n,h)
&=\frac{
\left[\widehat\sigma^2_{\tau,n}(\lambda u_n,h)
-\widehat\sigma^2_{\tau,n}(u_n,h)\right]^2
}{
\widehat\sigma^2_{\tau,n}(\lambda^2u_n,h)
-2\widehat\sigma^2_{\tau,n}(\lambda u_n,h)
+\widehat\sigma^2_{\tau,n}(u_n,h)
},\\
\label{LiuEstDfn1}
\widetilde\sigma^2_{\tau,n}(\lambda,u_n,h)
&=\widehat\sigma^2_{\tau,n}(u_n,h)
-\widehat B_{\tau,n}(\lambda,u_n,h),
\end{align}
where $\lambda>0, \lambda\neq 1$.
In their numerical implementation, \citet{LIU2018} replace
$\widehat B_{\tau,n}$ with the augmented correction
\begin{align*}
\widetilde B_{\tau,n}(\lambda,u_n,h)
&=\widetilde S_n(\lambda,u_n,h)
\left(\left[\widehat\sigma^2_{\tau,n}(\lambda u_n,h)
-\widehat\sigma^2_{\tau,n}(u_n,h)\right]\wedge0\right).
\end{align*}
Here $\widetilde S_n(\lambda,u_n,h)$ is an aggregated coefficient.%
\footnote{Although their simulation horizon is $T=1$ month,
\citet{LIU2018} estimate the corresponding coefficient
$\widetilde S_{n,k}$ using $k=12$ monthly blocks, or one year of data, requiring 11 months of data generated outside the 1-month observation period.  To
keep the observation period matched to the $T=1$ month time horizon, we compute $\widetilde S_n$ only on each one-month path.}
At each aggregation location $s$, write
$u_{s,n}=\Delta_n^{0.0025}/\sqrt{BV_{s,n}(h)}$, where $BV_{s,n}(h)$ is the
local bipower variation over the kernel window, and for compactness, write
$\widehat\sigma_i^2(a):=\widehat\sigma^2_{(2i-1)h,n}(a u_{(2i-1)h,n},h)$.  Then
$$
\widetilde S_{n}(\lambda,u_n,h)
=\frac{\displaystyle\sum_{i=1}^{\lfloor T/(2h)\rfloor}
[\widehat\sigma_i^2(\lambda p)-\widehat\sigma_i^2(p)]}
{\displaystyle\sum_{i=1}^{\lfloor T/(2h)\rfloor}
([\widehat\sigma_i^2(\lambda^2p)-2\widehat\sigma_i^2(\lambda p)
+\widehat\sigma_i^2(p)]\wedge0)}.
$$
Hence, our augmented comparator uses $\widetilde B_{\tau,n}$ in place of
$\widehat B_{\tau,n}$ in \eqref{LiuEstDfn1}. For the kernel, we use 
$
K_3(x)=\frac{15}{16}(1-x^2)^2\mathbf{1}_{\{|x|\leq1\}},
$ as specified in \cite[p.~1964]{LIU2018}\footnote{$K_3$ is the best-performing kernel among those considered in \cite{LIU2018} that satisfies the assumptions of their results (in particular, that is continuously differentiable as assumed by their main theorem).} and take $h=\Delta_n^{0.51}$.

We use the same pathwise loss and summary criteria for both estimator
families across the models described in the sections below.   We consider estimates on
each path at the 81 time points $l_i=i\lfloor n/100\rfloor$, $i=10,\ldots,90$.~Let $\widehat{\mathfrak c}_{i,j}$ denote a generic
estimator of the true spot variance
$c_{l_i,j}=\sigma^2_{l_i,j}$ on replication path $j$, $j=1,\ldots, M$.  With
$e_{i,j}=\widehat{\mathfrak c}_{i,j}-c_{l_i,j}$, define
$$
\widehat {\rm MSE}_j=\frac{1}{81}\sum_{i=10}^{90}e_{i,j}^2,
\qquad
\widehat{\mathrm{RMSE}}
=\left(\frac{1}{M}\sum_{j=1}^M \widehat {\rm MSE}_j\right)^{1/2}.
$$
We also report $\widehat{\mathrm{RMSE}}_{99}$, which is the same as $\widehat{\rm RMSE}$ but computed after excluding the
one percent of paths with the largest values of ${\rm MSE}_j$; we also report
the median absolute error and the mean signed relative error,
$$
\widehat{\mathrm{MAD}}=\mathop{\mathrm{median}}_{i,j}|e_{i,j}|,
\qquad
\widehat{\mathrm{RE}}=\frac{1}{81M}\sum_{j=1}^M\sum_{i=10}^{90}
\frac{e_{i,j}}{c_{l_i,j}}.
$$

Our finite-sample assessment is comprised of two distinct simulation studies for the estimators above. The first assesses the finite-sample effect of successive debiasing steps at a single common set of
tuning parameters across the considered values of $Y$. The second study examines a feasible implementation of our procedure when $Y$ is unknown, focusing on the effect of using an estimated jump-activity index to select the number of debiasing steps.  All other tuning parameters are held fixed at the values used in the first study.

\subsection{Finite-sample performance across $Y$}

In the first simulation study, we evaluate the performance of our debiased
estimators under the Heston-type volatility model and sampling scheme used by \cite{LIU2018}.  In particular, we consider the model
\begin{equation}\label{HestonModGenStandalone}
\begin{aligned}
X_t &= X_0+\int_0^t {\mu_s}\,ds+\int_0^t\sqrt{c_s}\,dW_s+J_t+J_t^{f},\\
c_t &= c_0+\int_0^t\kappa(\theta-c_s)\,ds
       +\xi\int_0^t\sqrt{c_s}\,dB_s,
\end{aligned}
\end{equation}
where $W$ and $B$ are correlated standard Brownian motions with correlation
$\rho$, $J$ is a strictly symmetric stable L\'evy process with
Blumenthal--Getoor index $Y$ and scale parameter $1$, and $J^f$ is an
independent compound Poisson process with an intensity of $3$ jumps per month and
$N(0,1)$ jump sizes.   We take their parameter values
$$
X_0=1,\qquad c_0=1.5,\qquad \rho=0,\qquad {\mu_t}\equiv1,
\qquad \kappa=0.03,\qquad \theta=1,\qquad \xi=0.15.
$$
We consider the values $Y\in\{0.8,1.2,1.6,1.75\}$, which correspond to the cases where,  no debiasing is needed ($Y\in\{0.8,1.2\}$), one step of debiasing is needed ($Y=1.6$), and two steps are needed $(Y=1.75)$ to attain rate-optimality using our method. 

The Heston variance process is simulated using an Euler discretization with absorption at zero (any negative variance proposal is replaced by zero before the next step), with five steps per observation interval. The continuous
price component is simulated on the same grid using left-endpoint variance values and the prescribed Brownian correlation, with
increments then aggregated to the observation grid. Stable increments over $\Delta_n$ are generated as $s_J\Delta_n^{1/Y}Z_i$, where the $Z_i$
are i.i.d.\ standard symmetric $Y$-stable variables with characteristic function $e^{-|u|^Y}$.

Following \citet{LIU2018}, we take $T=1$ month and $n=8580$, which they
describe as a five-minute sampling design.  We additionally consider the
fivefold finer grid $n=42900$, referred to below as the one-minute design.

Under this design, we select the tuning parameters of the two estimator
families using the same protocol.  Selection uses an independent sample of
$M_{\mathrm{tune}}=20000$ paths in each of the eight frequency/$Y$ cells.
Let $\mathcal F=\{5\text{-minute},1\text{-minute}\}$ and
$\mathcal Y=\{0.8,1.2,1.6,1.75\}$.  Define
\begin{equation}\label{e:r_q(y)}
r_{\rm GL}(y)=
\begin{cases}
0,&y\leq3/2,\\
1,&3/2<y<12/7,\\
2,&y\geq12/7,
\end{cases}
\qquad
r_{\rm CF}(y)=
\begin{cases}
0,&y\leq3/2,\\
1,&y>3/2.
\end{cases}
\end{equation}
These are the theory-prescribed debiasing orders for our global/local
estimators and the benchmark characteristic-function estimators,
respectively.  Let
$\vartheta_{\rm GL}=(\zeta_1,\zeta_2,p)$ and
$\vartheta_{\rm CF}=(\lambda,p)$ denote their respective tuning parameters.
For each family $q\in\{\mathrm{GL},\mathrm{CF}\}$, we select
$\vartheta_q$ to minimize
$\sum_{f\in\mathcal F}\sum_{Y\in\mathcal Y}
\widehat{\mathrm{MSE}}^{(r_q(Y))}_{q,f,Y}(\vartheta_q)/
\widehat{\mathrm{MSE}}^{(0)}_{\mathrm{GL},f,Y}
$,
where
$\widehat{\mathrm{MSE}}^{(\ell)}_{q,f,Y}(\vartheta_q)$ is the Monte Carlo
mean of the pathwise squared-error $\widehat{\rm MSE}_j$ for the order-$\ell$ estimator
from family $q$ with frequency/$Y$-specification  $(f,Y)$. The minimization for both $\zeta$ and $\lambda$ parameters is carried out over the finite grid
$\{1.05,1.10,1.20,\ldots ,1.90,2,2.25,2.5,3\}$, and
both procedures use the grid
$\{0.05,0.10,0.15,\ldots ,0.90\}$ for $p$. 

\begin{table}[htbp]
\centering
\caption{Pathwise estimation performance at the five-minute sampling
frequency under the simulation design of \citet{LIU2018}.  Boldface denotes the smallest  magnitude entry in each column. Small type reports $\pm1.96$ Monte Carlo standard errors. The tuning parameters are set to $(\zeta_1,\zeta_2,p)=(1.2,1.7,0.05)$ and
$(\lambda,p)=(1.9,0.9)$ in all configurations. }
\label{Table1}
\setlength{\tabcolsep}{2.8pt}
\begin{tabular}{@{}lcccccccc@{}}
\toprule
&\multicolumn{4}{c}{$Y=0.8$}&\multicolumn{4}{c}{$Y=1.2$}\\
\cmidrule(lr){2-5}\cmidrule(l){6-9}
Estimator&RMSE&RMSE$_{99}$&MAD&RE&RMSE&RMSE$_{99}$&MAD&RE\\
\midrule
$\widehat\sigma^2_{\tau,n}$
&\mc{2.1919}{1.3476}&\mc{0.2126}{0.0014}&\mc{0.1361}{0.0006}&\mc{0.0161}{0.0053}
&\mc{0.2382}{0.0065}&\mc{0.2282}{0.0007}&\mc{0.1455}{0.0006}&\mc{0.0466}{0.0006}\\
$\widetilde\sigma^2_{\tau,n}$
&\mc{3.2351}{2.3111}&\mc{0.2137}{0.0015}&\mc{0.1360}{0.0006}&\mc{0.0220}{0.0077}
&\mc{0.2445}{0.0089}&\mc{0.2312}{0.0007}&\mc{0.1468}{0.0006}&\mc{0.0516}{0.0006}\\
\addlinespace[2pt]
$\widetilde c^{(0)}_{n,\tau,\mathrm{exp}}$
&\mcb{0.0880}{0.0087}&\mcb{0.0818}{0.0004}&\mcb{0.0553}{0.0004}&\mcb{-0.0087}{0.0005}
&\mcb{0.0937}{0.0021}&\mcb{0.0915}{0.0005}&\mcb{0.0607}{0.0004}&\mc{0.0236}{0.0005}\\
$\widetilde c^{(1)}_{n,\tau,\mathrm{exp}}$
&\mc{1.6405}{1.5950}&\mc{0.1037}{0.0006}&\mc{0.0687}{0.0005}&\mc{-0.0257}{0.0077}
&\mc{0.1037}{0.0020}&\mc{0.1014}{0.0006}&\mc{0.0668}{0.0005}&\mcb{-0.0066}{0.0006}\\
$\widetilde c^{(2)}_{n,\tau,\mathrm{exp}}$
&\mc{0.4865}{0.1535}&\mc{0.2307}{0.0019}&\mc{0.1280}{0.0013}&\mc{-0.0899}{0.0026}
&\mc{0.9091}{0.7922}&\mc{0.2747}{0.0019}&\mc{0.1617}{0.0015}&\mc{-0.0995}{0.0053}\\
\addlinespace[2pt]
$\widetilde c^{(0)}_{n,\tau,\mathrm{unif}}$
&\mc{0.1129}{0.0016}&\mc{0.1110}{0.0005}&\mc{0.0748}{0.0004}&\mc{-0.0088}{0.0005}
&\mc{0.1232}{0.0030}&\mc{0.1206}{0.0005}&\mc{0.0798}{0.0004}&\mc{0.0236}{0.0005}\\
$\widetilde c^{(1)}_{n,\tau,\mathrm{unif}}$
&\mc{2.0929}{2.0365}&\mc{0.1331}{0.0006}&\mc{0.0875}{0.0005}&\mc{-0.0257}{0.0080}
&\mc{0.1387}{0.0027}&\mc{0.1360}{0.0006}&\mc{0.0896}{0.0005}&\mc{-0.0066}{0.0007}\\
$\widetilde c^{(2)}_{n,\tau,\mathrm{unif}}$
&\mc{0.6188}{0.2370}&\mc{0.2839}{0.0021}&\mc{0.1392}{0.0012}&\mc{-0.0903}{0.0028}
&\mc{1.2335}{1.0873}&\mc{0.3518}{0.0022}&\mc{0.1923}{0.0015}&\mc{-0.0993}{0.0053}\\
\midrule
&\multicolumn{4}{c}{$Y=1.6$}&\multicolumn{4}{c}{$Y=1.75$}\\
\cmidrule(lr){2-5}\cmidrule(l){6-9}
Estimator&RMSE&RMSE$_{99}$&MAD&RE&RMSE&RMSE$_{99}$&MAD&RE\\
\midrule
$\widehat\sigma^2_{\tau,n}$
&\mc{0.4718}{0.0047}&\mc{0.4678}{0.0010}&\mc{0.3668}{0.0012}&\mc{0.2587}{0.0008}
&\mc{0.7928}{0.0013}&\mc{0.7911}{0.0012}&\mc{0.7068}{0.0014}&\mc{0.4897}{0.0012}\\
$\widetilde\sigma^2_{\tau,n}$
&\mc{0.4927}{0.0054}&\mc{0.4882}{0.0011}&\mc{0.3879}{0.0012}&\mc{0.2736}{0.0008}
&\mc{0.8195}{0.0014}&\mc{0.8177}{0.0012}&\mc{0.7325}{0.0014}&\mc{0.5078}{0.0012}\\
\addlinespace[2pt]
$\widetilde c^{(0)}_{n,\tau,\mathrm{exp}}$
&\mc{0.3468}{0.0010}&\mc{0.3457}{0.0009}&\mc{0.3269}{0.0010}&\mc{0.2229}{0.0008}
&\mc{0.6759}{0.0012}&\mc{0.6746}{0.0011}&\mc{0.6604}{0.0012}&\mc{0.4483}{0.0012}\\
$\widetilde c^{(1)}_{n,\tau,\mathrm{exp}}$
&\mcb{0.2763}{0.0011}&\mcb{0.2750}{0.0011}&\mc{0.2452}{0.0013}&\mc{0.1639}{0.0009}
&\mc{0.5780}{0.0014}&\mc{0.5765}{0.0014}&\mc{0.5596}{0.0015}&\mc{0.3756}{0.0013}\\
$\widetilde c^{(2)}_{n,\tau,\mathrm{exp}}$
&\mc{0.3803}{0.0336}&\mc{0.3573}{0.0023}&\mcb{0.2188}{0.0016}&\mcb{-0.0261}{0.0023}
&\mcb{0.4695}{0.0030}&\mcb{0.4643}{0.0024}&\mcb{0.3270}{0.0023}&\mcb{0.1405}{0.0027}\\
\addlinespace[2pt]
$\widetilde c^{(0)}_{n,\tau,\mathrm{unif}}$
&\mc{0.3600}{0.0010}&\mc{0.3588}{0.0009}&\mc{0.3256}{0.0011}&\mc{0.2229}{0.0008}
&\mc{0.6854}{0.0012}&\mc{0.6841}{0.0011}&\mc{0.6590}{0.0012}&\mc{0.4483}{0.0012}\\
$\widetilde c^{(1)}_{n,\tau,\mathrm{unif}}$
&\mc{0.3009}{0.0011}&\mc{0.2996}{0.0010}&\mc{0.2482}{0.0013}&\mc{0.1638}{0.0009}
&\mc{0.5951}{0.0014}&\mc{0.5935}{0.0013}&\mc{0.5607}{0.0016}&\mc{0.3757}{0.0013}\\
$\widetilde c^{(2)}_{n,\tau,\mathrm{unif}}$
&\mc{0.5245}{0.0444}&\mc{0.4947}{0.0028}&\mc{0.3039}{0.0018}&\mc{-0.0263}{0.0024}
&\mc{0.6190}{0.0037}&\mc{0.6127}{0.0029}&\mc{0.4187}{0.0023}&\mc{0.1407}{0.0027}\\
\bottomrule
\end{tabular}
\end{table}
This procedure selects $(\zeta_1,\zeta_2,p)=(1.2,1.7,0.05)$ and
$(\lambda,p)=(1.9,0.9)$, respectively.  The former tuple, selected using the
exponential kernel, is also used for the uniform-kernel implementation. All reported results then use an independent validation sample of
$M_{\mathrm{val}}=5000$ paths using these fixed tuning values for each considered $Y$ and sampling frequency, with each estimator evaluated on the same  paths within each specification of $(f,Y)$.

Tables~\ref{Table1} and \ref{Table2} show that the value of debiasing is most
pronounced as $Y$ increases.  At one minute, the exponential-kernel RMSEs
exhibit successive improvement in the zero-, one-, and two-step ordering for both $Y=1.6$
and $Y=1.75$.  At five minutes, the two-step estimator gives the lowest RMSE
at $Y=1.75$; at $Y=1.6$, it gives the lowest MAD and nearly removes relative
bias, while the one-step estimator has the lowest RMSE.  For $Y=0.8$ and
$Y=1.2$, where debiasing is not required by the theory, applying unnecessary
higher-order corrections can come at a cost.

\begin{table}[htbp]
\centering
\caption{Pathwise estimation performance at the one-minute sampling
frequency under the simulation design of \citet{LIU2018}. Boldface denotes the smallest  magnitude entry in each column. Small type reports $\pm1.96$ Monte Carlo standard errors. The tuning parameters are the same as in Table~\ref{Table1}.}
\label{Table2}
\setlength{\tabcolsep}{2.8pt}
\begin{tabular}{@{}lcccccccc@{}}
\toprule
&\multicolumn{4}{c}{$Y=0.8$}&\multicolumn{4}{c}{$Y=1.2$}\\
\cmidrule(lr){2-5}\cmidrule(l){6-9}
Estimator&RMSE&RMSE$_{99}$&MAD&RE&RMSE&RMSE$_{99}$&MAD&RE\\
\midrule
$\widehat\sigma^2_{\tau,n}$
&\mc{0.3320}{0.2620}&\mc{0.1368}{0.0004}&\mc{0.0912}{0.0004}&\mcb{0.0027}{0.0007}
&\mc{0.1459}{0.0007}&\mc{0.1449}{0.0004}&\mc{0.0951}{0.0004}&\mc{0.0229}{0.0003}\\
$\widetilde\sigma^2_{\tau,n}$
&\mc{0.5117}{0.4596}&\mc{0.1368}{0.0004}&\mc{0.0911}{0.0004}&\mc{0.0036}{0.0011}
&\mc{0.1464}{0.0008}&\mc{0.1453}{0.0004}&\mc{0.0952}{0.0004}&\mc{0.0243}{0.0003}\\
\addlinespace[4pt]
$\widetilde c^{(0)}_{n,\tau,\mathrm{exp}}$
&\mcb{0.0964}{0.0603}&\mcb{0.0531}{0.0002}&\mcb{0.0358}{0.0002}&\mc{-0.0047}{0.0004}
&\mcb{0.0581}{0.0002}&\mcb{0.0578}{0.0002}&\mcb{0.0386}{0.0002}&\mc{0.0137}{0.0002}\\
$\widetilde c^{(1)}_{n,\tau,\mathrm{exp}}$
&\mc{0.1012}{0.0575}&\mc{0.0609}{0.0002}&\mc{0.0410}{0.0002}&\mc{-0.0151}{0.0005}
&\mc{0.0600}{0.0002}&\mc{0.0596}{0.0002}&\mc{0.0399}{0.0002}&\mcb{-0.0018}{0.0003}\\
$\widetilde c^{(2)}_{n,\tau,\mathrm{exp}}$
&\mc{1.8831}{1.8366}&\mc{0.1146}{0.0006}&\mc{0.0689}{0.0005}&\mc{-0.0486}{0.0060}
&\mc{0.2812}{0.2038}&\mc{0.1390}{0.0007}&\mc{0.0858}{0.0005}&\mc{-0.0496}{0.0010}\\
\addlinespace[4pt]
$\widetilde c^{(0)}_{n,\tau,\mathrm{unif}}$
&\mc{0.1094}{0.0507}&\mc{0.0735}{0.0002}&\mc{0.0494}{0.0002}&\mc{-0.0048}{0.0003}
&\mc{0.0782}{0.0003}&\mc{0.0778}{0.0003}&\mc{0.0517}{0.0002}&\mc{0.0137}{0.0002}\\
$\widetilde c^{(1)}_{n,\tau,\mathrm{unif}}$
&\mc{0.1151}{0.0482}&\mc{0.0812}{0.0003}&\mc{0.0543}{0.0002}&\mc{-0.0152}{0.0003}
&\mc{0.0825}{0.0003}&\mc{0.0821}{0.0003}&\mc{0.0548}{0.0002}&\mc{-0.0018}{0.0003}\\
$\widetilde c^{(2)}_{n,\tau,\mathrm{unif}}$
&\mc{3.2258}{3.1530}&\mc{0.1438}{0.0006}&\mc{0.0807}{0.0004}&\mc{-0.0498}{0.0085}
&\mc{0.4108}{0.3205}&\mc{0.1799}{0.0007}&\mc{0.1063}{0.0005}&\mc{-0.0496}{0.0011}\\
\midrule
&\multicolumn{4}{c}{$Y=1.6$}&\multicolumn{4}{c}{$Y=1.75$}\\
\cmidrule(lr){2-5}\cmidrule(l){6-9}
Estimator&RMSE&RMSE$_{99}$&MAD&RE&RMSE&RMSE$_{99}$&MAD&RE\\
\midrule
$\widehat\sigma^2_{\tau,n}$
&\mc{0.3194}{0.0006}&\mc{0.3187}{0.0005}&\mc{0.2650}{0.0006}&\mc{0.1834}{0.0005}
&\mc{0.6190}{0.0011}&\mc{0.6179}{0.0006}&\mc{0.5778}{0.0008}&\mc{0.3955}{0.0009}\\
$\widetilde\sigma^2_{\tau,n}$
&\mc{0.3283}{0.0006}&\mc{0.3277}{0.0005}&\mc{0.2748}{0.0006}&\mc{0.1902}{0.0005}
&\mc{0.6338}{0.0011}&\mc{0.6326}{0.0006}&\mc{0.5926}{0.0008}&\mc{0.4057}{0.0009}\\
\addlinespace[4pt]
$\widetilde c^{(0)}_{n,\tau,\mathrm{exp}}$
&\mc{0.2595}{0.0004}&\mc{0.2590}{0.0004}&\mc{0.2501}{0.0004}&\mc{0.1697}{0.0004}
&\mc{0.5681}{0.0005}&\mc{0.5676}{0.0005}&\mc{0.5615}{0.0005}&\mc{0.3802}{0.0008}\\
$\widetilde c^{(1)}_{n,\tau,\mathrm{exp}}$
&\mc{0.2090}{0.0005}&\mc{0.2084}{0.0004}&\mc{0.1954}{0.0005}&\mc{0.1316}{0.0004}
&\mc{0.4968}{0.0006}&\mc{0.4962}{0.0006}&\mc{0.4892}{0.0006}&\mc{0.3300}{0.0008}\\
$\widetilde c^{(2)}_{n,\tau,\mathrm{exp}}$
&\mcb{0.1973}{0.0008}&\mcb{0.1962}{0.0008}&\mcb{0.1269}{0.0006}&\mcb{0.0070}{0.0008}
&\mcb{0.3395}{0.0014}&\mcb{0.3376}{0.0011}&\mcb{0.2628}{0.0013}&\mc{0.1618}{0.0011}\\
\addlinespace[4pt]
$\widetilde c^{(0)}_{n,\tau,\mathrm{unif}}$
&\mc{0.2668}{0.0004}&\mc{0.2663}{0.0004}&\mc{0.2497}{0.0005}&\mc{0.1696}{0.0004}
&\mc{0.5729}{0.0005}&\mc{0.5723}{0.0005}&\mc{0.5608}{0.0006}&\mc{0.3801}{0.0008}\\
$\widetilde c^{(1)}_{n,\tau,\mathrm{unif}}$
&\mc{0.2212}{0.0005}&\mc{0.2206}{0.0004}&\mc{0.1959}{0.0006}&\mc{0.1316}{0.0004}
&\mc{0.5043}{0.0006}&\mc{0.5037}{0.0006}&\mc{0.4897}{0.0007}&\mc{0.3299}{0.0008}\\
$\widetilde c^{(2)}_{n,\tau,\mathrm{unif}}$
&\mc{0.2770}{0.0010}&\mc{0.2756}{0.0009}&\mc{0.1780}{0.0008}&\mc{0.0071}{0.0009}
&\mc{0.4141}{0.0015}&\mc{0.4121}{0.0011}&\mc{0.3014}{0.0012}&\mcb{0.1617}{0.0012}\\
\bottomrule
\end{tabular}
\end{table}

The exponential kernel has lower RMSE and $\mathrm{RMSE}_{99}$ than the
uniform kernel at every estimator stage, frequency, and value of $Y$.  At
$Y=1.6$ and $Y=1.75$, our theory-prescribed exponential-kernel estimator also has lower RMSE, $\mathrm{RMSE}_{99}$, and MAD than the corresponding
bias-corrected characteristic-function estimator.

\subsection{Feasible estimation with unknown debiasing order}
The number of debiasing steps required by our theory depends on $Y$, which is typically unknown in applications. We therefore use the jump-activity estimator of \citet{bull:2016} and compare the resulting feasible estimator with a theory-guided oracle that uses the true value of $Y$. As a separate robustness exercise below, we assess dependence on model-specific tuning by carrying the same tuning parameters, without re-optimization, to a substantially different Heston specification.

 For $q\in \{\rm{GL},\rm{CF}\}$, denoting the global-local and ch.f.-based estimators, let  $r_q(Y)$ be as in \eqref{e:r_q(y)} above. The orcale-informed estimators use debiasing order $r_q(Y)$, whereas feasible estimators use the debiasing order
$r_q(\widehat Y)$, respectively.  Thus, the feasible estimators do not assume prior knowledge of the jump index $Y$.

We estimate jump activity using the method of \citet{bull:2016}, which depends
on a threshold constant $C_{\rm Bull}$, a scale ratio $\rho_{\rm Bull}$, and
a multiscale order $m_{\rm Bull}$.  Since $C_{\rm Bull}$ is calibrated
relative to a unit diffusive scale, we apply the estimator to the increments
after normalizing them by $\sqrt{BV_n}$, where $BV_n$ is the full-path
bipower-variation estimate of integrated variance.  We use the recommended default
values $(C_{\rm Bull},\rho_{\rm Bull},m_{\rm Bull})=(0.05,2,3)$.

We keep the tuning parameters the same from the previous section, and compare our implementation with the exponential kernel  against the ch.f.-based benchmarks.  

\begin{table}[htbp]
\centering
\caption{Oracle-informed and feasible RMSE and $\mathrm{RMSE}_{99}$ under
the simulation design of \citet{LIU2018}.  Results are based on $M=5000$ independent paths for each frequency/$Y$ configuration. Boldface denotes the smallest  magnitude entry in each column. Small type reports $\pm1.96$ Monte Carlo standard errors. The tuning parameters are the same as in Table~\ref{Table1}.}
\label{Table3}
\footnotesize
\setlength{\tabcolsep}{4pt}
\begin{tabular}{@{}cc|cccc|cccc@{}}
\toprule
&&\multicolumn{4}{c}{Oracle-informed}&\multicolumn{4}{c}{Feasible}\\
\cmidrule(lr){3-6}\cmidrule(l){7-10}
&&\multicolumn{2}{c}{Ours}&\multicolumn{2}{c}{Liu et al.}
&\multicolumn{2}{c}{Ours}&\multicolumn{2}{c}{Liu et al.}\\
\cmidrule(lr){3-4}\cmidrule(lr){5-6}\cmidrule(lr){7-8}\cmidrule(l){9-10}
Frequency&$Y$&RMSE&RMSE$_{99}$&RMSE&RMSE$_{99}$&RMSE&RMSE$_{99}$&RMSE&RMSE$_{99}$\\
\midrule
\multirow{4}{*}{5 minutes}
&$0.8$&\mcb{0.0880}{0.0087}&\mcb{0.0818}{0.0004}&\mc{2.1919}{1.3476}&\mc{0.2126}{0.0014}&\mcb{0.1560}{0.0056}&\mcb{0.1486}{0.0024}&\mc{2.1920}{1.3476}&\mc{0.2129}{0.0014}\\
&$1.2$&\mcb{0.0937}{0.0021}&\mcb{0.0915}{0.0005}&\mc{0.2382}{0.0065}&\mc{0.2282}{0.0007}&\mcb{0.1939}{0.0034}&\mcb{0.1882}{0.0028}&\mc{0.2397}{0.0065}&\mc{0.2296}{0.0007}\\
&$1.6$&\mcb{0.2763}{0.0011}&\mcb{0.2750}{0.0011}&\mc{0.4927}{0.0054}&\mc{0.4882}{0.0011}&\mcb{0.3507}{0.0021}&\mcb{0.3468}{0.0018}&\mc{0.4839}{0.0046}&\mc{0.4798}{0.0011}\\
&$1.75$&\mcb{0.4695}{0.0030}&\mcb{0.4643}{0.0024}&\mc{0.8195}{0.0014}&\mc{0.8177}{0.0012}&\mcb{0.5604}{0.0036}&\mcb{0.5568}{0.0034}&\mc{0.8094}{0.0014}&\mc{0.8076}{0.0013}\\
\addlinespace[4pt]
\multirow{4}{*}{1 minute}
&$0.8$&\mcb{0.0964}{0.0603}&\mcb{0.0531}{0.0002}&\mc{0.3320}{0.2620}&\mc{0.1368}{0.0004}&\mcb{0.1030}{0.0564}&\mcb{0.0629}{0.0008}&\mc{0.3320}{0.2620}&\mc{0.1368}{0.0004}\\
&$1.2$&\mcb{0.0581}{0.0002}&\mcb{0.0578}{0.0002}&\mc{0.1459}{0.0007}&\mc{0.1449}{0.0004}&\mcb{0.0784}{0.0012}&\mcb{0.0763}{0.0011}&\mc{0.1460}{0.0007}&\mc{0.1450}{0.0004}\\
&$1.6$&\mcb{0.2090}{0.0005}&\mcb{0.2084}{0.0004}&\mc{0.3283}{0.0006}&\mc{0.3277}{0.0005}&\mcb{0.2440}{0.0009}&\mcb{0.2433}{0.0009}&\mc{0.3219}{0.0006}&\mc{0.3212}{0.0005}\\
&$1.75$&\mcb{0.3395}{0.0014}&\mcb{0.3376}{0.0011}&\mc{0.6338}{0.0011}&\mc{0.6326}{0.0006}&\mcb{0.5143}{0.0025}&\mcb{0.5132}{0.0025}&\mc{0.6237}{0.0011}&\mc{0.6226}{0.0007}\\\addlinespace[2pt]
\bottomrule
\end{tabular}
\end{table}

However, we note that while the model of \cite{LIU2018} serves as a useful benchmark, it has the overall annualized volatility set to $1$, which can be considered relatively high for many financial applications. Moreover, in the sampling design, the number of 1- and 5-minute observations assumed available over a 1-month period is somewhat nonstandard.  Hence, by way of robustness check, we include a second model in our feasible estimation comparison study with parameters that are arguably more realistic for financial applications.  First, for the Heston parameters in \eqref{HestonModGenStandalone}, we borrow the values from \cite{ZhangMk}:
\begin{equation}\label{Kawai}
\rho = -0.5,\quad \kappa = 5,\quad\xi=0.5,\quad\sigma_0^2=\theta=0.16,
\end{equation}
and here we take with $b\equiv J^f\equiv 0$ and $J$ and independent symmetric $Y$-stable component having scale
parameter $0.5$.  

 Assuming 252 trading days per year, and a 6.5 hour trading day, we consider $\Delta_n = (252*6.5*60)^{-1}$ as the 1-minute sampling frequency. We consider a time horizon of 3 months  ($T = \frac{1}{4}$), which yields a sample size of $n=24570$ at 1-minute sampling, and $ n = 4914$ at $5$-minute sampling. The model is simulated with the same scheme as the previous section.  We keep the tuning parameters the same from the previous section; i.e., no re-tuning is done for this second Heston design.

\begin{table}[htbp]
\centering
\caption{Oracle-informed and feasible RMSE and $\mathrm{RMSE}_{99}$ under
the second Heston design with mean  volatility  level of $0.4$. Boldface denotes the smallest  magnitude entry in each column. Small type reports $\pm1.96$ Monte Carlo standard errors.  The tuning parameters are the same as in Table~\ref{Table1}.}
\label{Table4}
\footnotesize
\setlength{\tabcolsep}{4pt}
\begin{tabular}{@{}cc|cccc|cccc@{}}
\toprule
&&\multicolumn{4}{c}{Oracle-informed}&\multicolumn{4}{c}{Feasible}\\
\cmidrule(lr){3-6}\cmidrule(l){7-10}
&&\multicolumn{2}{c}{Ours}&\multicolumn{2}{c}{Liu et al.}
&\multicolumn{2}{c}{Ours}&\multicolumn{2}{c}{Liu et al.}\\
\cmidrule(lr){3-4}\cmidrule(lr){5-6}\cmidrule(lr){7-8}\cmidrule(l){9-10}
Frequency&$Y$&RMSE&RMSE$_{99}$&RMSE&RMSE$_{99}$&RMSE&RMSE$_{99}$&RMSE&RMSE$_{99}$\\
\midrule
\multirow{4}{*}{5 minutes}
&$0.8$&\mcb{0.0175}{0.0061}&\mcb{0.0138}{0.0001}&\mc{0.0559}{0.0268}&\mc{0.0194}{0.0002}&\mcb{0.0210}{0.0051}&\mcb{0.0173}{0.0002}&\mc{0.0559}{0.0268}&\mc{0.0194}{0.0002}\\
&$1.2$&\mcb{0.0149}{0.0001}&\mcb{0.0147}{0.0001}&\mc{0.0223}{0.0003}&\mc{0.0219}{0.0002}&\mc{0.0237}{0.0005}&\mc{0.0227}{0.0003}&\mcb{0.0224}{0.0003}&\mcb{0.0220}{0.0002}\\
&$1.6$&\mcb{0.0535}{0.0002}&\mcb{0.0533}{0.0002}&\mc{0.0808}{0.0007}&\mc{0.0801}{0.0002}&\mcb{0.0598}{0.0006}&\mcb{0.0588}{0.0003}&\mc{0.0781}{0.0003}&\mc{0.0777}{0.0002}\\
&$1.75$&\mc{0.4713}{0.4113}&\mcb{0.0879}{0.0005}&\mcb{0.1634}{0.0003}&\mc{0.1630}{0.0003}&\mcb{0.1285}{0.0029}&\mcb{0.1266}{0.0008}&\mc{0.1593}{0.0003}&\mc{0.1590}{0.0003}\\
\addlinespace[4pt]
\multirow{4}{*}{1 minute}
&$0.8$&\mcb{0.0093}{0.0001}&\mcb{0.0093}{{<}0.0001}&\mc{0.0623}{0.0416}&\mc{0.0128}{0.0001}&\mcb{0.0107}{0.0001}&\mcb{0.0105}{0.0001}&\mc{0.0623}{0.0416}&\mc{0.0128}{0.0001}\\
&$1.2$&\mcb{0.0099}{0.0001}&\mcb{0.0098}{{<}0.0001}&\mc{0.0141}{0.0001}&\mc{0.0139}{0.0001}&\mcb{0.0129}{0.0001}&\mcb{0.0126}{0.0001}&\mc{0.0141}{0.0001}&\mc{0.0139}{0.0001}\\
&$1.6$&\mcb{0.0410}{0.0001}&\mcb{0.0409}{0.0001}&\mc{0.0551}{0.0001}&\mc{0.0549}{0.0001}&\mcb{0.0382}{0.0003}&\mcb{0.0379}{0.0003}&\mc{0.0541}{0.0001}&\mc{0.0540}{0.0001}\\
&$1.75$&\mcb{0.0723}{0.0036}&\mcb{0.0698}{0.0003}&\mc{0.1295}{0.0001}&\mc{0.1293}{0.0001}&\mcb{0.0940}{0.0007}&\mcb{0.0936}{0.0007}&\mc{0.1277}{0.0001}&\mc{0.1275}{0.0001}\\\addlinespace[2pt]
\bottomrule
\end{tabular}
\end{table}

The feasible comparison is broadly stable across the two Heston designs and sampling
frequencies.  Our exponential-kernel feasible estimator has lower pathwise RMSE and $\mathrm{RMSE}_{99}$ than the feasible
characteristic-function estimator in 15 of the 16 model/frequency/$Y$ cells;
the only exception under either criterion is the five-minute second Heston
configuration with $Y=1.2$.  At $Y=1.6$ and $Y=1.75$, our feasible estimator has lower
RMSE and $\mathrm{RMSE}_{99}$ in every case.

The results also shed light on the oracle-informed performance.  In
the five-minute realistic Heston design with $Y=1.75$, a small fraction of realizations of the pathwise RMSE dramatically inflate the value of the theory-prescribed two-step
estimator.  Its $\mathrm{RMSE}_{99}$ is $0.0879$, compared with $0.1630$ for
the characteristic-function estimator, so the reversal in raw RMSE is
tail-driven.  Interestingly, the feasible Bull rule sometimes selects fewer debiasing steps in this configuration
and consequently attains a lower pathwise RMSE than the oracle-informed estimator. 

Thus, although estimating  $Y$ to select the debiasing steps may lead to finite-sample loss relative to oracle order selection, the loss is generally moderate and does not alter the broad relative performance comparison with with the characteristic-function benchmark.  In view of the performance displayed in both simulation designs, we recommend $(\zeta_1,\zeta_2,p)=(1.2,1.7,0.05)$ as practical defaults.

\subsection{Finite-sample sensitivity to $\zeta$}

We next study how the finite-sample mean squared error varies with {the tuning parameter} $\zeta$ for the global-local one-step estimator
\eqref{debias_formula2b} and for {the corresponding tuning parameter} $\lambda$ in the augmented characteristic-function
estimator of \cite{LIU2018} in \eqref{LiuEstDfn1}. In each case, the
auxiliary parameter $p$ is held fixed, at $0.05$ and $0.9$,
respectively, while the multiplier ($\zeta$ or $\lambda$) varies over the same 14-point grid
from $1.05$ to $3$.

We consider two models: first, the model of \citet{LIU2018} (which is also used in Tables \ref{Table1}-\ref{Table3} above) and the {more} realistic Heston design \eqref{Kawai} (used in Table \ref{Table4} above), both at one-minute sampling frequency, with $Y=1.6$ and
$Y=1.75$. The target is the spot variance $c_\tau=\sigma_\tau^2$ at
$\tau=T/2$: $\tau=1/2$ for the Liu design and $\tau=1/8$ for the
Heston design. 

Within each setting, both methods' MSEs are divided
by the same denominator: the smallest MSE observed across the two
methods and all grid values. This normalization retains the comparison
of their MSE levels while making the multiplier profiles easier to
compare across different models.  The error bars show the pointwise Monte Carlo standard error of the MSE, multiplied by $1.96$ and divided by that same denominator.

\begin{figure}[hbpt]
  \centering
  \includegraphics[width=0.8\linewidth]
    {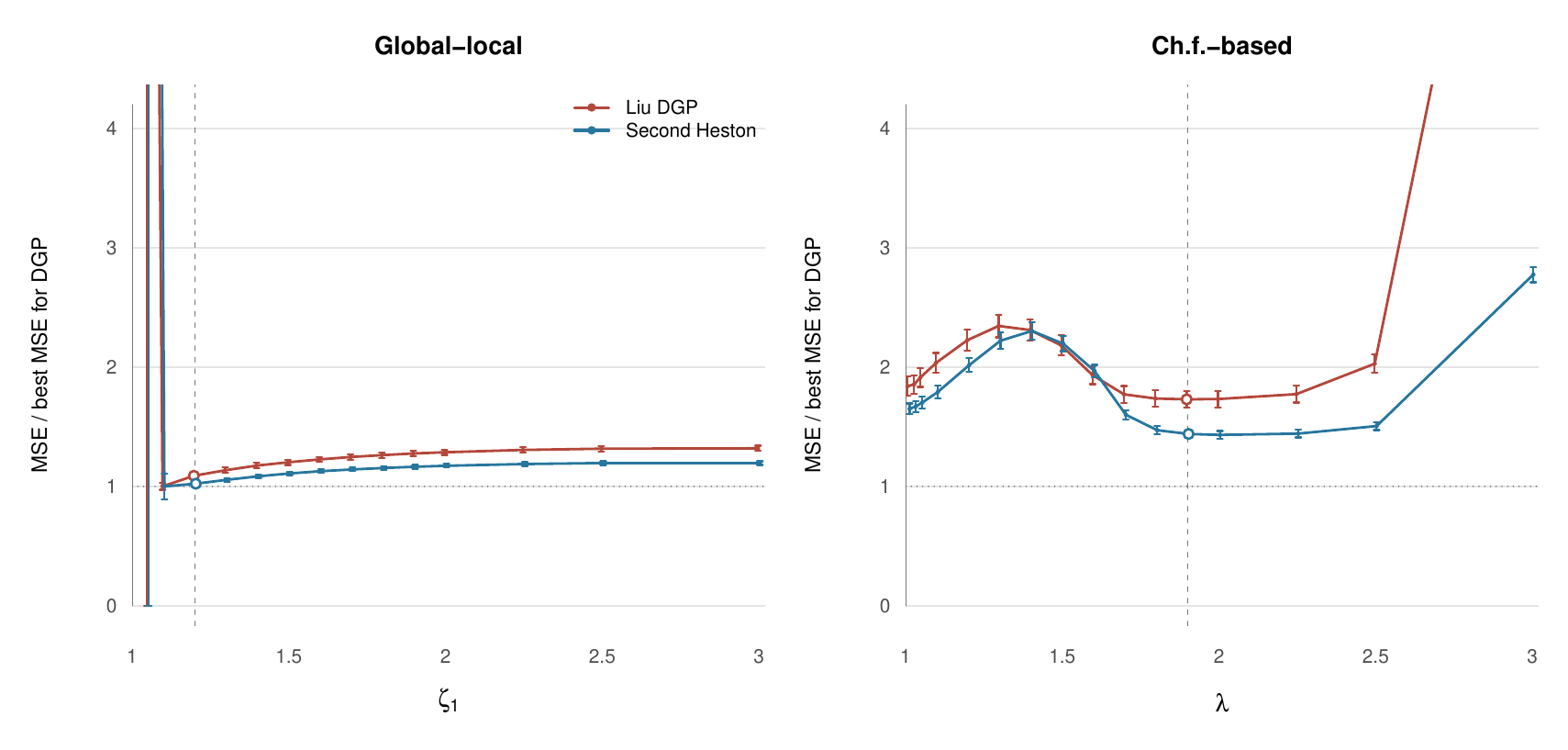}
  \caption{\label{MSEDiffp170}
    Sensitivity of one-step MSE to the tuning parameters when $Y=1.75$.
    The left panel shows the global-local estimator as a function of
    $\zeta_1$; the right panel shows the augmented
    characteristic-function estimator \eqref{LiuEstDfn1} as a function of $\lambda$.
    Red curves correspond to the Heston model  \eqref{HestonModGenStandalone} with parameters from \cite{LIU2018}, and the blue curves correspond to the second Heston design with the settings \eqref{Kawai}.  The vertical lines mark $\zeta_1=1.2$ and $\lambda=1.9$.
    Both panels use a common vertical scale.}
\end{figure}

Figure \ref{MSEDiffp170} gives the results for $Y=1.75$. Over the
central part of the grid, the MSE of the global-local estimator changes
gradually with $\zeta_1$. The characteristic-function estimator has a
more pronounced trough near $\lambda=1.9$--$2$, with higher MSE at
several smaller and larger multipliers. Although the global-local MSE rises
sharply when $\zeta_1$ is very close to one, it is generally more stable than the ch.f-based comparison.  Figure \ref{MSEDiffp160} repeats the comparison for $Y=1.6$. The
qualitative pattern is similar. 

Across all four $(Y,\text{design})$
settings, the MSE at $\zeta_1=1.2$ is within approximately $10\%$ of
the lowest global-local MSE on the displayed grid. Even at its own
best grid value, the characteristic-function estimator has a higher
MSE than the global-local estimator at $\zeta_1=1.2$. Overall, the comparison suggests that the truncation-based debiasing procedures considered here are less sensitive to $\zeta$ than the characteristic-function-based alternatives in these settings.

\begin{figure}[hbpt]
  \centering
  \includegraphics[width=0.8\linewidth]
    {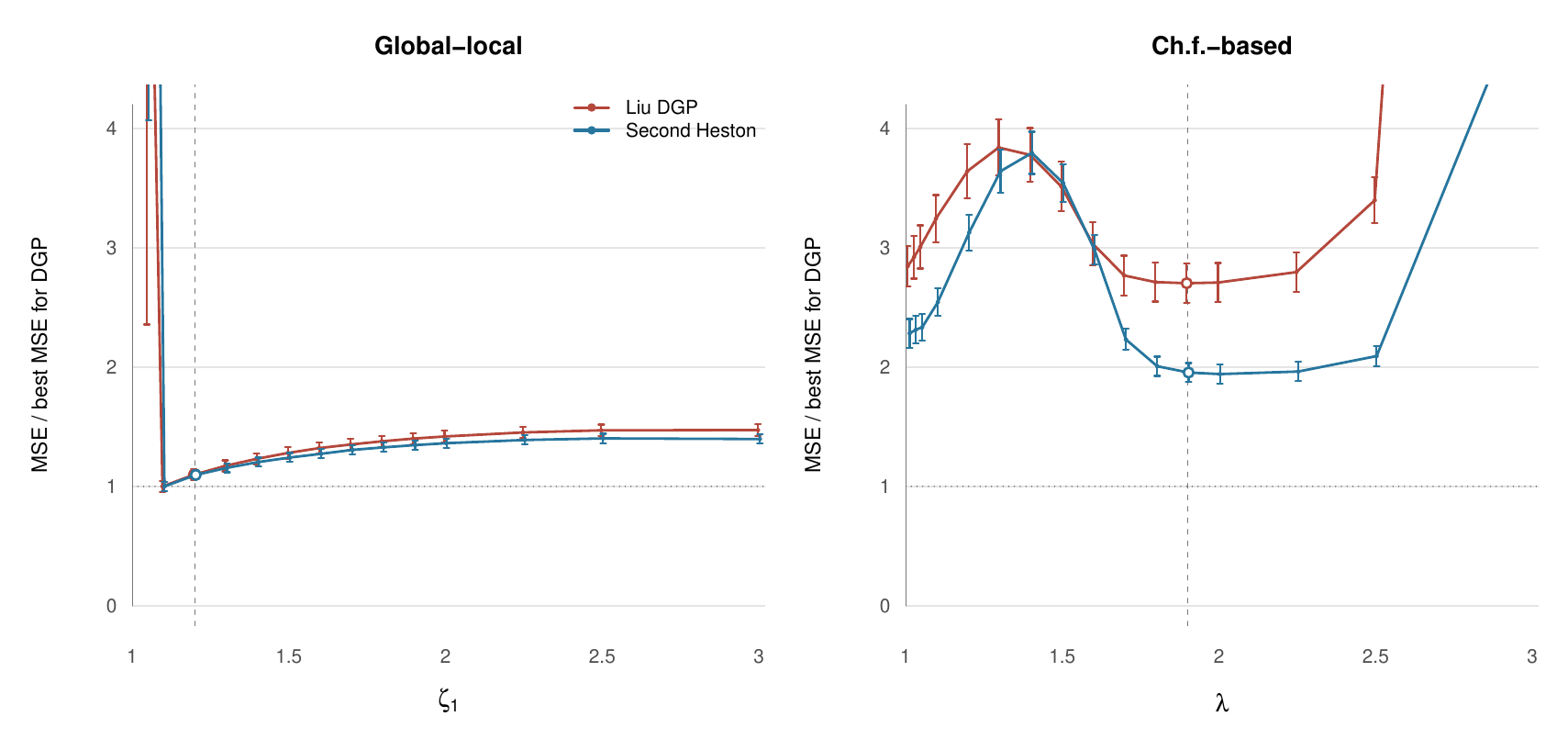}
  \caption{\label{MSEDiffp160}
    Sensitivity of one-step MSE to the scale multiplier when $Y=1.6$.
    The estimators, normalization, colors, reference multipliers, and
    Monte Carlo error bars are as in Figure \ref{MSEDiffp170}.}
\end{figure}

\newpage
\appendix

\section{Main proofs}\label{MainPrfsa}
 For completeness, we first provide two results that will often be used throughout the proofs.
\begin{lemma}[\textbf{Lemma 3.1 in \cite{FIGUEROALOPEZ20204693}}]
Let $K$ be a kernel function satisfying the conditions in Assumption \ref{asu_kernel} with supports $(\ell, u)$ ($-\infty\leq \ell\leq 0\leq u\leq \infty$). For $\gamma > 0$, assume the following for a function $f:\mathbb{R} \rightarrow \mathbb{R}$:
\begin{enumerate}
\item[(i)] 
$f(\tau+h) - f(\tau) = C_{\gamma} (h;\tau) + o (|h|^{\gamma})$, as $h\rightarrow 0$ for a fixed  given $ \tau \in (0,T)$, where $C_\gamma :\mathbb{R} \times (0,T) \rightarrow \mathbb{R}$ is a function such that
$$C_{\gamma}(hs;\tau) = h^{\gamma} C_{\gamma}(s;\tau),\quad s \in \mathbb{R}, h > 0, \tau \in (0, T).$$
\item[(ii)]
$f \in C([0,T])$.
\end{enumerate}	
Let
\begin{equation*}
\begin{split}
D_1 (f)
:= & 
\sum _{i=1}^n
K_{b}(t_{i-1} - \tau)
\int _{t_{i-1}}^{t_{i}}f(s) ds- 
\int _{[0,T]} 
K_{b}(s - \tau)
f(s) ds,
\end{split}
\end{equation*}
where {$K_{b}(x) := K(x/b)/b$}. Then, for any $b_n\to{}0$ such that $\Delta_n/b_n \rightarrow 0$, we have the following:
\begin{equation*}
\begin{split}
D_1 (f)
=
\frac{1}{2} f(\tau) (K(u^+)- K(\ell^-)) \frac{\Delta_n}{b_n} + o \left( \frac{\Delta_n}{b_n} \right) .
\end{split}
\end{equation*}
\end{lemma}
The next lemma gives a stable central limit
theorem suited to spot estimation. Its hypotheses are analogous to those of Theorem~2.2.15 in \cite{jacod2012discretization}. Here, convergence of the conditional covariance is required only for its terminal sum, and the martingale orthogonality condition is imposed as an absolute summability requirement.  We provide a self-contained proof
in Appendix~\ref{s:proof_of_lem_clt}.
\begin{lemma}
\label{lem:terminal_stable_clt}
Let $(\Omega,\mathcal F,(\mathcal F_t)_{t\geq0},\PP)$ satisfy
the usual conditions. Fix $T>0$, and let
$t_i=i\Delta_n$, $i=0,\ldots,n$, where $n\Delta_n=T$.
Write
$$
\EE_{i-1}(\cdot)
:=\EE(\cdot\mid\mathcal F_{t_{i-1}}),
\qquad
\Delta_iM:=M_{t_i}-M_{t_{i-1}}.
$$
Let $\xi_i^n$, $i=1,\ldots,n$, be square-integrable,
$\mathbb R^d$-valued random variables such that $\xi_i^n$
is $\mathcal F_{t_i}$-measurable.
Suppose that
\begin{align}
&\sum_{i=1}^n
\EE_{i-1}[\xi_i^n]\xrightarrow{P}A,
\label{terminal_clt_drift}\\
&\sum_{i=1}^n\left(
\EE_{i-1}\big[\xi_i^n(\xi_i^n)^\top\big]
-\big(\EE_{i-1}\xi_i^n\big)
 \big(\EE_{i-1}\xi_i^n\big)^\top\right)
\xrightarrow{P} C,
\label{terminal_clt_variance}\\
&\sum_{i=1}^n
\EE_{i-1}\big[\|\xi_i^n\|^p\big]
\xrightarrow{P}0
\qquad\text{for some }p>2,
\label{terminal_clt_moments}
\end{align}
where $A$ is an almost surely finite, $\mathcal F$-measurable
$\mathbb R^d$-valued random variable and $C$ is an almost
surely finite, $\mathcal F$-measurable, symmetric nonnegative
definite matrix.
For each $j=1,\ldots,d$, let $B^j$ be a standard
real-valued Brownian motions relative to $(\mathcal F_t)$. 
Writing $\xi_i^{n,j}$ for the $j$th coordinate of $\xi_i^n$,
suppose further that
\begin{equation}
\sum_{i=1}^n
\left|
\EE_{i-1}\big[\xi_i^{n,j}\Delta_iM\big]
\right|
\xrightarrow{P}0
\label{terminal_clt_orthogonality}
\end{equation}
for $M=B^j$ and for every bounded real-valued martingale
$M$ orthogonal to $B^j$, for each $j=1,\ldots,d$.
Then
$$
\sum_{i=1}^n\xi_i^n
\xrightarrow{st} A+C^{1/2}Z,
$$
where $Z$ is a standard $d$-dimensional normal random vector,
independent of $\mathcal F$, defined on an extension of
the original probability space.
\end{lemma}

\subsection{Proof of Theorem \ref{clt}}\label{31Ap}
Denote
\begin{align*}
\bar{K}_{b_n}(\tau)&:=\Delta_n\sum_{j=1}^nK_{b_n}(t_{j-1}-\tau),\\
Y_i &:= \sqrt{m_n}K_{m_n\Delta_n}(t_{i-1}-\tau)(\Delta_iX)^2\mathbf{1}_{\{|\Delta_iX|\leq  v_n\}},\\
    Y'_i &:= \frac{\Delta_n}{\sqrt{m_n\Delta_n}}\times\left\{
    \begin{aligned}
        &0, \;&\text{if } i=1,\\
        -&\Big(\sum_{l=1}^{i-1}K_{m_n\Delta_n}(t_{l-1}-\tau)\Big)\Delta_iB,\; &\text{if } 2\leq i\leq \ceil{\tau/\Delta_n},\\
       &\Big(\sum_{l=i}^{n}K_{m_n\Delta_n}(t_{l-1}-\tau)\Big)\Delta_iB,\; &\text{if } \ceil{\tau/\Delta_n} < i \leq n.
    \end{aligned}\right.
\end{align*}
Consider the following decompositions:
\begin{align}\label{TrmT1T2}
    \tilde Z_n(m_n, v_n)  
    &=\Big(\frac{\sum_{i=1}^n Y_i}{\bar{K}_{m_n\Delta_n}(\tau)} - \frac{\sum_{i=1}^n \EE(Y_i|\mathcal{F}_{i-1})}{\bar{K}_{m_n\Delta_n}(\tau)}\Big) \nonumber\\
    &\qquad+ \Big(\frac{\sum_{i=1}^n \EE(Y_i|\mathcal{F}_{i-1})}{\bar{K}_{m_n\Delta_n}(\tau)} - \sqrt{m_n}\Delta_n\sum_{i=1}^n\frac{K_{m_n\Delta_n}(t_{i-1}-\tau)}{\bar{K}_{m_n\Delta_n}(\tau)}\sigma_{t_{i-1}}^2 -\sqrt{m_n}A(v_n,m_n)\Big)\nonumber\\
    &=:T_1 + T_2,
\end{align}
and
\begin{align}\label{TrmT3T4a}
    \tilde Z'_n(m_n) &= \tilde \sigma_{\tau}\frac{\sum_{i=1}^nY_i'}{\bar{K}_{m_n\Delta_n}(\tau)} + \Big({\tilde Z'_n(m_n)}- \frac{\tilde\sigma_{\tau}\sum_{i=1}^n Y_i'}{\bar{K}_{m_n\Delta_n}(\tau)}
    \Big)\nonumber\\
    &=\tilde\sigma_{\tau} \frac{\sum_{i=1}^nY_i'}{\bar{K}_{m_n\Delta_n}(\tau)}  \nonumber\\
    &\quad+ \Big(\frac{\Delta_n}{\sqrt{m_n\Delta_n}}\sum_{i=1}^{n}\frac{K_{m_n\Delta_n}(t_{i-1}-\tau)(\sigma^2_{t_{i-1}} - \sigma^2_\tau)}{\bar{K}_{m_n\Delta_n}(\tau)} - \tilde\sigma_{\tau} \frac{\sum_{i=1}^nY_i'}{\bar{K}_{m_n\Delta_n}(\tau)}\Big)\nonumber\\
    &=:T_3 + T_4.
\end{align}
We now proceed to show $(T_1,T_3)\xrightarrow{st} (Z_1, Z_2)$ by verifying the conditions of Lemma~\ref{lem:terminal_stable_clt}, and in addition that $T_2$ and $T_4$ are asymptotically negligible.

\smallskip
\noindent
{\bf a)} We start with $T_3$. Clearly, $\EE(Y'_i|\mathcal F_{i-1})  = 0$ and, using the moment formula for Gaussian distributions {and Lemma 3.1 in \cite{FIGUEROALOPEZ20204693}}, we have
\begin{align*}
  &  \sum_{i=1}^n {\var(Y'_i|\mathcal F_{i-1})}\\
  &= \frac{\Delta_n^2}{m_n\Delta_n} \sum_{i=1}^n \left(\Big(-\sum_{l=1}^{i-1}K_{m_n\Delta_n}(t_{l-1}-\tau)\Big)^2\Delta_n\mathbf{1}_{\{i\leq\ceil{\tau/\Delta_n}\}} + \Big(\sum_{l=i}^{n}K_{m_n\Delta_n}(t_{l-1}-\tau)\Big)^2\Delta_n\mathbf{1}_{\{i>\ceil{\tau/\Delta_n}\}}\right)\nonumber\\
    &\quad=  \frac{1}{m_n\Delta_n}\int_0^1 \Big(\int_{{v}}^{1}K_{m_n\Delta_n}({s}-\tau)ds\mathbf{1}_{\tau< {v}}-\int^{{v}}_{0}K_{m_n\Delta_n}({s}-\tau)ds\mathbf{1}_{\tau\geq {v}}\Big)^2 dv + o_P(1)\nonumber\\
    &\quad \xrightarrow{P} \int \Big(\int_t^{\infty}K(u)du\mathbf{1}_{t\geq 0}-\int^t_{-\infty}K(u)du\mathbf{1}_{t<0}\Big)^2 dt,\nonumber
\end{align*}
and
 \begin{align*}
    &\sum_{i=1}^n\EE({Y'}_i^4|\mathcal F_{i-1})\\
    &= \frac{3\Delta_n^4}{m_n^2\Delta_n^2} \sum_{i=1}^n \left(\Big(\sum_{l=1}^{i-1}K_{m_n\Delta_n}(t_{l-1}-\tau)\Big)^4\Delta_n^2\mathbf{1}_{\{i\leq\ceil{\tau/\Delta_n}\}} + \Big(\sum_{l=i}^{n}K_{m_n\Delta_n}(t_{l-1}-\tau)\Big)^4\Delta_n^2\mathbf{1}_{\{i>\ceil{\tau/\Delta_n}\}}\right)\nonumber\\
    &\quad =\frac{3}{m_n}\int \left(\int_t^{\infty}K(u)du\mathbf{1}_{t\geq 0}-\int^t_{-\infty}K(u)du\mathbf{1}_{t<0}\right)^4 dt + o_P(1)\xrightarrow{P} 0.\nonumber
\end{align*}
On the other hand, by Assumption \ref{asu_kernel} and Lemma 3.1 in \cite{FIGUEROALOPEZ20204693}, with $f\equiv 1$ and $m=1$ therein, we have
\begin{align}\label{kernel_convergence}
    \Delta_n\sum_{j=1}^nK_{m_n\Delta_n}(t_{j-1}-\tau) - 1 = %
    {O(m_n^{-1})},
\end{align}
which vanishes as $n\rightarrow \infty$. Therefore,
\begin{gather*}
\frac{\sum_{i=1}^n {\var(Y'_i|\mathcal F_{i-1})}}{{\big(\Delta_n\sum_{j=1}^nK_{m_n\Delta_n}(t_{j-1}-\tau)\big)^2}} \xrightarrow{P} \int \Big(\int_t^{\infty}K(u)du\mathbf{1}_{t\geq 0}-\int^t_{-\infty}K(u)du\mathbf{1}_{t<0}\Big)^2 dt,\\
\text{and }\;\frac{\sum_{i=1}^n\EE({Y'}_i^4|\mathcal F_{i-1})}{\big(\Delta_n\sum_{j=1}^nK_{m_n\Delta_n}(t_{j-1}-\tau)\big)^4}\rightarrow 0.
\end{gather*}
For the martingale orthogonality check, we establish the
following stronger condition:
\begin{align}\label{initial_technicalYprime}
    \sum_{i=1}^n
    \left|\EE_{i-1}\big(Y'_i\Delta_iM\big)\right|
    \xrightarrow{P}0,
\end{align}
where $\EE_{i-1}(\cdot)=\EE(\cdot|\mathcal F_{i-1})$
and $M$ is either $W$ or $B$ or belongs to the set
$\mathcal N$ of bounded martingales orthogonal to both
$W$ and $B$. In fact, we verify this condition for every
bounded martingale orthogonal to $B$. When $M=B$, using $K\geq0$ and interchanging the sums,
we have
\begin{align*}
    &\sum_{i=1}^n
    \left|\EE_{i-1}\big(Y'_i\Delta_iB\big)\right|\\
    &=\frac{\Delta_n^2}{\sqrt{m_n\Delta_n}}
    \sum_{l=1}^n
    \left|l-\ceil{\tau/\Delta_n}\right|
    K_{m_n\Delta_n}(t_{l-1}-\tau)\\
    &\leq\frac{\Delta_n}{\sqrt{m_n\Delta_n}}
    \sum_{l=1}^n
    \big(|t_{l-1}-\tau|+\Delta_n\big)
    K_{m_n\Delta_n}(t_{l-1}-\tau)\\
    &\leq
    \sqrt{m_n\Delta_n}\int_{\mathbb R}|u|K(u)\,du
    +O(m_n^{-1})
    +O\left(\frac{\Delta_n}{\sqrt{m_n\Delta_n}}\right)
    \longrightarrow0.
\end{align*}
For the last inequality, recalling $b_n=m_n\Delta_n$,
the kernel-sum approximation satisfies
\begin{align*}
    &\left|
    \Delta_n\sum_{l=1}^n
    |t_{l-1}-\tau|K_{b_n}(t_{l-1}-\tau)
    -\int_0^1|s-\tau|K_{b_n}(s-\tau)\,ds
    \right|\\
    &\qquad\leq
    \Delta_n\int_{-\tau}^{1-\tau}
    \left|\frac{d}{dx}\big(|x|K_{b_n}(x)\big)\right|dx
    \lesssim \frac{\Delta_n}{\sqrt{b_n}}.
\end{align*}
Indeed, Assumption \ref{asu_kernel} implies
$\int u^2|K'(u)|^2du<\infty$, so the last bound follows
by Cauchy--Schwarz. Dividing by $\sqrt{b_n}$ gives
the $O(m_n^{-1})$ error above.
When $M=W$,  Cauchy--Schwarz gives
$$
\left|\EE_{i-1}(\Delta_iB\Delta_iW)\right|
\leq
\EE_{i-1}\big[(\Delta_iB)^2\big]^{1/2}
\EE_{i-1}\big[(\Delta_iW)^2\big]^{1/2}
=\Delta_n.
$$
Consequently,
$$
\sum_{i=1}^n
\left|\EE_{i-1}\big(Y'_i\Delta_iW\big)\right|
\leq
\sum_{i=1}^n
\left|\EE_{i-1}\big(Y'_i\Delta_iB\big)\right|
\longrightarrow0.
$$
Finally, if $M$ is a bounded martingale orthogonal to $B$,
then
$$
\EE_{i-1}(\Delta_iB\Delta_iM)=0,
\qquad
\EE_{i-1}(Y'_i\Delta_iM)=0
$$
for every $i$, which completes the verification of
\eqref{initial_technicalYprime}.
\medskip
\noindent

{\bf b)} Next, we work with the term $T_4$ of \eqref{TrmT3T4a}. 
Let $i'\in\{1,\dots,n\}$ be such that $\tau \in (t_{i'-1},t_{i'}]$. Notice that
\begin{align*}
    {\frac{\sqrt{m_n\Delta_n}}{\Delta_n}}\sum_{i=1}^n Y_i' &= -\sum_{l=1}^{i'-1}\big(\sum_{i=l+1}^{i'}\Delta_iB\big)K_{m_n\Delta_n}(t_{l-1}-\tau) + \sum_{l=i'+1}^n\big(\sum_{i=i'+1}^{l}\Delta_iB\big)K_{m_n\Delta_n}(t_{l-1}-\tau)\\
    &=\sum_{l=1}^{i'-1}K_{m_n\Delta_n}(t_{l-1}-\tau)(B_{t_l}-B_{t_{i'}}) + \sum_{l=i'+1}^nK_{m_n\Delta_n}(t_{l-1}-\tau)(B_{t_l}-B_{t_{i'}}).
\end{align*}
In light of \eqref{kernel_convergence} {and using the notation $M_t:=\int_0^t\int\tilde{\delta}(s,z)\bar{\mathfrak{p}}(ds,dz)$}, we can then write
\begin{align}
    |T_4|
    &\lesssim \Big|\frac{\Delta_n}{\sqrt{m_n\Delta_n}}\sum_{i=1}^{n}K_{m_n\Delta_n}(t_{i-1}-\tau)\Big(\int_{\tau}^{t_{i-1}}\tilde\sigma_s\,dB_s
-\tilde\sigma_\tau(B_{t_{i-1}}-B_\tau)+\int_{\tau}^{t_{i-1}}\tilde \mu_sds+{M_{t_{i-1}}-M_{\tau}}\Big)\Big|\nonumber\\
    &\quad + \Big| \frac{\Delta_n}{\sqrt{m_n\Delta_n}} \sum_{i=1}^{n}K_{m_n\Delta_n}(t_{i-1}-\tau)\Big(\tilde\sigma_\tau
\big(B_{t_{i'}}-B_\tau-\Delta_iB\big)\Big)\Big|
    \nonumber\\
    &\quad=:|T_{4,1}|+|T_{4,2}|.
\end{align}
Since  $\Delta_n\sum_{i=1}^nK_{m_n\Delta_n}(t_{i-1}-\tau)=O(1)$,
we have
\begin{align*}
    \EE|T_{4,2}|
    &\lesssim
    \frac{\Delta_n}{\sqrt{m_n\Delta_n}}
    \sum_{i=1}^nK_{m_n\Delta_n}(t_{i-1}-\tau)
    \left(\EE|\Delta_iB|+\EE|B_{t_{i'}}-B_\tau|\right)\\
    &\lesssim
    \frac{\sqrt{\Delta_n}}{\sqrt{m_n\Delta_n}}
    =m_n^{-1/2}\longrightarrow0.
\end{align*}
Thus, $T_{4,2}=o_P(1)$. It remains to control the first term. {For a given positive sequence $N_n \rightarrow \infty$ with  $m_n\Delta_nN_n\rightarrow 0$ as $n\rightarrow \infty$ and $\varepsilon>0$, let 
\begin{equation*}\begin{aligned}
    \eta_i&:=\eta_i^c+\eta_i^j
    :=\left(
    \int_\tau^{t_{i-1}}\tilde\sigma_s\,dB_s
    -\tilde\sigma_\tau(B_{t_{i-1}}-B_\tau)
    +\int_\tau^{t_{i-1}}\tilde\mu_s\,ds
    \right)
    +\left(M_{t_{i-1}}-M_\tau\right),\\
    \rho(i)&:=
    \left(
    \EE\left[
    \sup_{s\in[t_{i-1}\wedge\tau,t_{i-1}\vee\tau]}
    |\tilde\sigma_s-\tilde\sigma_\tau|^4
    \right]
    \right)^{1/2}.
\end{aligned}
\end{equation*}
We follow} similar ideas to those in the proof of (13.3.37) for $j=6$ in \cite{jacod2012discretization} (see also the proof of (A.6) for $l=4$ in \cite{Figueroa-Lopez_Wu_2024}), {with some adjustments due to the unboundedness of the support of $K$.} First, we bound $\eta_i^c$ and $\eta_i^j$ separately. We have
\begin{equation*}
\begin{aligned}
&{\EE\Big((\eta^c_i)^2\mathbf{1}_{\{i\in\mathcal{I}_n\}}\Big)}\leq c|t_{i-1}-\tau|\rho(i)+c|t_{i-1}-\tau|^2,\\
&\EE|\eta_i^j|
\leq C|t_{i-1}-\tau|^{1/2}
\left(\int_{\{\tilde H\leq\varepsilon\}}
\tilde H^2(z)\lambda(dz)\right)^{1/2}
+C_\varepsilon|t_{i-1}-\tau|,\\
&\sup_{i\in\mathcal{I}_n}\rho(i)  \rightarrow 0,\quad
\EE \eta_i^2 \lesssim |\tau - {t_{i-1}}|,
\end{aligned}
\end{equation*}
where $\mathcal{I}_n:=\{i\in\{1,\dots,n\}:{t_{i-1}} \in (\tau-N_n m_n\Delta_n,\tau + N_n m_n\Delta_n)\}$ and
\[
C_\varepsilon
:=C\left(1+\int_{\{\tilde H>\varepsilon\}}
\tilde H(z)\lambda(dz)\right)<\infty.
\]
To justify the bound for $\eta_i^c$ when
$a:=t_{i-1}<\tau$, write
$$
\int_a^\tau\tilde\sigma_s\,dB_s
-\tilde\sigma_\tau(B_\tau-B_a)
=
\int_a^\tau(\tilde\sigma_s-\tilde\sigma_a)\,dB_s
+(\tilde\sigma_a-\tilde\sigma_\tau)(B_\tau-B_a).
$$
It\^o's isometry and Cauchy--Schwarz give
\begin{align*}
&\EE\left|
\int_a^\tau\tilde\sigma_s\,dB_s
-\tilde\sigma_\tau(B_\tau-B_a)
\right|^2\\
&\quad\leq
2\int_a^\tau
\EE|\tilde\sigma_s-\tilde\sigma_a|^2\,ds
+2\left(
\EE|\tilde\sigma_a-\tilde\sigma_\tau|^4
\EE|B_\tau-B_a|^4
\right)^{1/2}\\
&\quad\leq C(\tau-a)\rho(i).
\end{align*}
For $t_{i-1}\geq\tau$, the same bound follows directly
from It\^o's isometry. Boundedness of $\tilde\mu$ then
gives the displayed estimate for $\eta_i^c$.

Since $\tilde\sigma$ is bounded (under localization), c\'adl\'ag, and has no fixed jump at $\tau$,
dominated convergence yields $\sup_{i\in\mathcal I_n}\rho(i)\to0$ whenever $N_nm_n\Delta_n\to0$.
The bound for $\eta_i^j$ follows by splitting the compensated Poisson integral over
$\{\tilde H\leq\varepsilon\}$ and $\{\tilde H>\varepsilon\}$. 

Thus, using that $\EE(|U|)\leq\sqrt{\EE U^2}$,
\begin{align*}
    \EE\big|T_{4,1}\big|  
     &\lesssim\frac{\Delta_n}{\sqrt{m_n\Delta_n}}\sum_{i=1}^nK_{m_n\Delta_n}(t_{i-1}-\tau){\left(\mathbf{1}_{i\in\mathcal{I}_n}+
     \mathbf{1}_{i\notin\mathcal{I}_n}\right)}\EE\Big(|\eta_i|\Big)\\
     &\lesssim\frac{\Delta_n}{\sqrt{m_n\Delta_n}}\sum_{i=1}^{n}\frac{1}{m_n\Delta_n}\mathbf{1}_{i\in\mathcal{I}_n}\sqrt{|t_{i-1}-\tau|\rho(i)}\\
     &\quad+{\frac{\Delta_n}{\sqrt{m_n\Delta_n}}\sum_{i=1}^{n}K_{m_n\Delta_n}(t_{i-1}-\tau)\sqrt{|t_{i-1}-\tau|}
     \left(\int_{|\tilde{H}|\leq{}\varepsilon}\tilde{H}^2(z)\lambda(dz)\right)^{1/2}}\\
     &\quad+ C_\varepsilon\frac{\Delta_n}{\sqrt{m_n\Delta_n}}\sum_{i=1}^{n}K_{m_n\Delta_n}(t_{i-1}-\tau)
     \mathbf{1}_{i\in\mathcal{I}_n}|t_{i-1}-\tau|\\
    &\quad+ \frac{1}{\sqrt{m_n\Delta_n}}\Big(\Delta_n\sum_{i=1}^{n}K_{m_n\Delta_n}(t_{i-1}-\tau)\mathbf{1}_{i\notin\mathcal{I}_n}\sqrt{|t_{i-1}-\tau|}\Big).
\end{align*}
For the term involving $C_\varepsilon$, we use
\[
\mathbf{1}_{i\in\mathcal{I}_n}|t_{i-1}-\tau|
\leq \sqrt{N_nm_n\Delta_n}\sqrt{|t_{i-1}-\tau|}.
\]
The kernel assumptions imply $K(u)\leq C/(1+u^2)$, and hence the corresponding kernel sums give
\begin{align*}
    \EE\big|T_{4,1}\big|  
     &\lesssim
     \frac{m_n\Delta_n}{\sqrt{m_n\Delta_n}m_n\Delta_n}N_n\sqrt{N_n m_n\Delta_n \sup_{i\in\mathcal{I}_n}\rho(i)} \\
     &\qquad +{\left(\int_{|\tilde{H}|\leq{}\varepsilon}\tilde{H}^2(z)\lambda(dz)\right)^{1/2}}
     +C_\varepsilon\sqrt{N_nm_n\Delta_n}+\frac{1}{\sqrt{N_n}}.
\end{align*}
{Then, we see that, by taking $N_n\to\infty$ slow enough such that $N_n^{3/2}\left(\sup_{i\in\mathcal{I}_n}\rho(i)\right)^{1/2} \rightarrow 0$ and $N_nm_n\Delta_n\to0$, and letting first $n\to\infty$ and then $\varepsilon\to0$, we have
\[
    \limsup_{n}\EE\big|T_{4,1}\big|=0.
    \]
We conclude that $T_{4,1}=o_P(1)$.} 
 Combining the convergence of $T_3$ with $T_4=o_P(1)$ gives
\[
    \tilde Z'_n(m_n)\xrightarrow{st} Z_2 .
\]

\medskip
\noindent
{\bf c)}
We now proceed to consider $T_1$ and $T_2$ of Eq.~\eqref{TrmT1T2}. First, by Proposition \ref{prop}, for any $p>2$, we can write
\begin{align*}
    \EE(Y_i|\mathcal F_{i-1}) & = \sqrt{m_n}K_{m_n\Delta_n}(t_{i-1}-\tau)\Big(\sigma_{t_{i-1}}^2\Delta_n+C_{1,i}\Delta_n v_n^{2-Y}+D_{1,i} \Delta_n^2 v_n^{-Y}\\
    &\qquad\qquad\qquad\qquad\qquad+{O_P(\Delta_n^3 v_n^{-2-Y})+o_P(\Delta_n^{5/4})\Big)},\nonumber\\
    \EE(Y_i^p|\mathcal F_{i-1}) & \lesssim m_n^{p/2} K^p_{m_n\Delta_n}(t_{i-1}-\tau)\big(\Delta_n^{p}\sigma_{t_{i-1}}^{2p} + \Delta_n v_n^{2p-Y}\big),\nonumber\\
    \var(Y_i|\mathcal F_{i-1}) &=m_n K^2_{m_n\Delta_n}(t_{i-1}-\tau)\big( 2\Delta_n^2\sigma_{t_{i-1}}^4 + O_P(\Delta_n^2 v_n^{2-Y}) + O(\Delta_n v_n^{4-Y})\big) \nonumber\\
    &=m_n K^2_{m_n\Delta_n}(t_{i-1}-\tau)\big( 2\Delta_n^2\sigma_{t_{i-1}}^4 + O_P(\Delta_n v_n^{4-Y})\big).
\end{align*}
Then, 
$$s_{1n}^2 :=\sum_{i=1}^n\var(Y_i|\mathcal F_{i-1}) = m_n\sum_{i=1}^nK^2_{m_n\Delta_n}(t_{i-1}-\tau)\big( 2\Delta_n^2\sigma_{t_{i-1}}^4 + O_P(\Delta_n v_n^{4-Y})\big).$$
Given $ v_n \ll \Delta_n^{\frac{1}{4-Y}}$ and \eqref{kernel_convergence}, we obtain
\begin{align*}
    \frac{s_{1n}^2}{(\Delta_n\sum_{j=1}K_{m_n\Delta_n}(t_{j-1}-\tau))^2}\rightarrow2\sigma^4_\tau\int K^2(x)dx.
\end{align*}
Moreover, we have
\begin{align*}
    \sum_{i=1}^n\EE(Y_i^p|\mathcal F_{i-1}) \lesssim \frac{1}{m_n^{p/2-1}\Delta_n^{p-1}}\int K^p(x)dx (\sigma^{2p}_{\tau}\Delta_n^{p-1} + O_P(v_n^{2p-Y})).
\end{align*}
If we rewrite $p = (2 \ell +1)/ \ell$ for some ${\ell}>0$, the right-hand side of the above inequality vanishes if and only if
\begin{align*}
    m_n^{-\frac{1}{2\ell}}\Delta_n^{-1-\frac{1}{\ell}}v_n^{\frac{4\ell+2}{\ell}-Y} = \Big((\Delta_n^{-1}v_n^{4-Y})^\ell (m_n^{-\frac{1}{2}}\Delta_n^{-1}v_n^{2})\Big)^{\frac{1}{\ell}} = \Delta_n^{-1}v_n^{4-Y}  \Big(m_n^{-\frac{1}{2}}\Delta_n^{-1}v_n^{2}\Big)^{\frac{1}{\ell}}\rightarrow 0.
\end{align*}
Since $\Delta_n^{-1}v_n^{4-Y}\ll \Delta_n^{\beta'(4-Y)-1}\ll \Delta_n^s$ for small enough $s>0$, the above vanishes by picking $\ell$ large enough. 
In light of \eqref{kernel_convergence}, we also have $\frac{\sum_{i=1}^n\EE(Y_i^p|\mathcal F_{i-1})}{(\Delta_n\sum_{j=1}K_{m_n\Delta_n}(t_{j-1}-\tau))^p}\rightarrow 0$ as $n\rightarrow \infty$ for such {an $\ell$}.  Taking $\ell$ larger if necessary, we may suppose that
$2<p<4$. By Jensen's inequality,
$$
\sum_{i=1}^n
\EE_{i-1}\left|Y_i-\EE_{i-1}(Y_i)\right|^p
\leq 2^p\sum_{i=1}^n\EE_{i-1}(Y_i^p)
\xrightarrow{P}0.
$$
Moreover, by H\"older's inequality and the second- and
fourth-moment bounds in part a,
$$
\sum_{i=1}^n\EE_{i-1}|Y'_i|^p
\leq
\left(\sum_{i=1}^n
\EE_{i-1}\big[(Y'_i)^2\big]\right)^{(4-p)/2}
\left(\sum_{i=1}^n
\EE_{i-1}\big[(Y'_i)^4\big]\right)^{(p-2)/2}
\xrightarrow{P}0.
$$
Consequently, for the same $p>2$,
$$
\sum_{i=1}^n
\EE_{i-1}\left[
\left\|
\begin{pmatrix}
Y_i-\EE_{i-1}(Y_i)\\
Y'_i
\end{pmatrix}
\right\|^p
\right]
\xrightarrow{P}0.
$$
To show  $T_2=o_P(1)$, recall $m_n = O(\Delta_n^{-\frac{1}{2}})$ and, thus, by Proposition \ref{prop} and \eqref{kernel_convergence},
\begin{align}
   \sum_{i=1}^n\EE(Y_i|\mathcal F_{i-1}) &=\sqrt{m_n}\sum_{i=1}^n K_{m_n\Delta_n}(t_{i-1}-\tau)\Delta_n\sigma_{t_{i-1}}^2 +\sqrt{m_n}\sum_{i=1}^n K_{m_n\Delta_n}(t_{i-1}-\tau)C_{1,i} \Delta_nv_n^{2-Y} \nonumber\\
    &+ \sqrt{m_n}\sum_{i=1}^nK_{m_n\Delta_n}(t_{i-1}-\tau)D_{1,i} \Delta_n^2 v_n^{-Y} + O_P(\sqrt{m_n}\Delta_n^2 v_n^{-2-Y})+o_P(1).
\end{align}
Hence,%
\begin{align*}
T_2&=\frac{\sum_{i=1}^n \EE(Y_i|\mathcal{F}_{i-1})}{\Delta_n\sum_{j=1}^nK_{m_n\Delta_n}(t_{j-1}-\tau)} - \sqrt{m_n}\Delta_n\sum_{i=1}^n\frac{K_{m_n\Delta_n}(t_{i-1}-\tau)}{\Delta_n\sum_{j=1}^nK_{m_n\Delta_n}(t_{j-1}-\tau)}\sigma_{t_{i-1}}^2 -\sqrt{m_n}A(v_n,m_n)\\
&= O_P(\sqrt{m_n}\Delta_n^2 v_n^{-2-Y}) + o_P(1)
\end{align*}
which vanishes due to the condition $m_n\ll \Delta_n^{-4}v_n^{4+2Y}$. To obtain $T_1\xrightarrow{st} Z_1$, by Lemma~\ref{lem:terminal_stable_clt}, it suffices to check the condition:
\begin{align}\label{initial_technical}
    \sum_{i=1}^n\left|\EE_{i-1}\big(Y_i\Delta_iM\big)\right| \rightarrow 0,
\end{align}
where $M$ is either $W$ or $B$ or is in the set $\mathcal N$ containing all bounded martingales orthogonal to $W$ and $B$. (In fact, the proof applies to every bounded martingale orthogonal to $W$.) The proof of \eqref{initial_technical} is technical and is deferred to Appendix \ref{TecnPrf1}.
This concludes the proof for $\tilde Z_n(m_n, v_n)
\xrightarrow{ st} Z_1$.

\medskip
\noindent
\textbf{d)}
To establish the joint convergence \eqref{clt_spot},
it remains to verify that the conditional cross-covariances
vanish. It suffices to show that
\begin{align}\label{cov_summand}
    &\frac{1}{\sqrt{\Delta_n}}
    \sum_{i=1}^n K_{m_n\Delta_n}(t_{i-1}-\tau)
    \left|
    \EE_{i-1}\left[
    (\Delta_iX)^2
    \mathbf{1}_{\{|\Delta_iX|\leq v_n\}}\Delta_iB
    \right]
    \right|
    \xrightarrow{P}0,
\end{align}
which is verified in Appendix \ref{TecnPrf1}.
Indeed, since $\EE_{i-1}(Y'_i)=0$, the definition of
$Y'_i$ and the nonnegativity of $K$ give
\begin{align*}
    &\sum_{i=1}^n
    \left|
    \EE_{i-1}\left[
    \big(Y_i-\EE_{i-1}(Y_i)\big)Y'_i
    \right]
    \right|\\
    &=\sum_{i=1}^n
    \left|\EE_{i-1}(Y_iY'_i)\right|\\
    &\leq
    \frac{\sqrt{m_n}\Delta_n}{\sqrt{m_n\Delta_n}}
    \left(\sum_{l=1}^n
    K_{m_n\Delta_n}(t_{l-1}-\tau)\right)
    \sum_{i=1}^n K_{m_n\Delta_n}(t_{i-1}-\tau)
    \left|
    \EE_{i-1}\left[
    (\Delta_iX)^2
    \mathbf{1}_{\{|\Delta_iX|\leq v_n\}}\Delta_iB
    \right]
    \right|\\
    &=
    \frac{\bar K_{m_n\Delta_n}(\tau)}{\sqrt{\Delta_n}}
    \sum_{i=1}^n K_{m_n\Delta_n}(t_{i-1}-\tau)
    \left|
    \EE_{i-1}\left[
    (\Delta_iX)^2
    \mathbf{1}_{\{|\Delta_iX|\leq v_n\}}\Delta_iB
    \right]
    \right|
    \xrightarrow{P}0,
\end{align*}
where we used \eqref{kernel_convergence}.
The same conclusion holds after division by
$\bar K_{m_n\Delta_n}(\tau)^2$.
   In view of \eqref{initial_technicalYprime}, \eqref{initial_technical}, and the moment and covariance limits above, Lemma~\ref{lem:terminal_stable_clt} applies to $(Y_i-\EE_{i-1}(Y_i),Y'_i)^\top$, with $B^1=W$ and $B^2=B$. Together with \eqref{kernel_convergence}, this yields the required joint convergence.
    
    \subsection{Proof of Theorem \ref{thm_clt_difference}}\label{Ap3_4}
    
    Denote
$$Z_i =  \frac{m_n^{1/2}}{\Delta_n^{-1/2} v_n^{2-Y/2}} K_{m_n\Delta_n}(t_{i-1}-\tau)(\Delta_iX)^2\left(\mathbf{1}_{\{|\Delta_iX|\leq \zeta  v_n\}}-\mathbf{1}_{\{|\Delta_iX|\leq  v_n\}}\right),$$
and note that
\begin{align*}
&u_n^{-1}\big(\tilde Z_n(m_n, \zeta v_n)-\tilde Z_n(m_n, v_n)\big)\\
&=\frac{\sum_{i=1}^n Z_i}{\Delta_n\sum_{j=1}^nK_{m_n\Delta_n}(t_{j-1}-\tau)}-{u_n^{-1}\sqrt{m_n}\left(A(\zeta v_n,m_n)-A(v_n,m_n)\right)}\\
&=:\frac{\sum_{i=1}^n\xi_i}{\Delta_n\sum_{j=1}^nK_{m_n\Delta_n}(t_{j-1}-\tau)}\\
&\quad+\Big(\frac{\sum_{i=1}^n\EE(Z_i|\mathcal{F}_{i-1})}{\Delta_n\sum_{j=1}^nK_{m_n\Delta_n}(t_{j-1}-\tau)}-u_n^{-1}\sqrt{m_n}\left(A(\zeta v_n,m_n)-A(v_n,m_n)\right)\Big),
\end{align*}
{where $\xi_i:=Z_i-\EE(Z_i|\mathcal{F}_{i-1})$}. The second term above is $o_P(1)$. Indeed, since $m_n = O(\Delta_n^{-\frac{1}{2}})$, {by \eqref{PartExp0_difference} in Proposition \ref{prop},
\begin{align*}
    &\frac{\sum_{i=1}^n\EE(Z_i|\mathcal{F}_{i-1})}{\Delta_n\sum_{j=1}^nK_{m_n\Delta_n}(t_{j-1}-\tau)} - u_n^{-1}\sqrt{m_n}\big(A({\zeta v_n},m_n) - A({v_n},m_n)\big)\\
    &\quad\lesssim \frac{m_n^{1/2}\Delta_n^2 v_n^{-2-Y}}{\Delta_n^{-1/2} v_n^{2-Y/2}}
    +\frac{m_n^{1/2}o_P(\Delta_n^{-1/4} v_n^{2-Y/2})}{\Delta_n^{-1/2} v_n^{2-Y/2}}\\
    &\quad= m_n^{1/2}\Delta_n^{5/2}v_n^{-4-Y/2}+o_P(1),
\end{align*}
which} vanishes given our assumption $m_n \ll \Delta_n^{-5}v_n^{8+Y}$.  Next we obtain the asymptotic behavior for $\sum_{i=1}^n {\xi_i}$ by verifying the conditions of Lemma~\ref{lem:terminal_stable_clt}.
For the conditional variance, we use the following
expansion, justified at the end of this proof:
\begin{equation}\label{GNMExp_difference}
\EE_{i-1}\left[(\Delta_iX)^4
\mathbf{1}_{\{v_n<|\Delta_iX|\leq\zeta v_n\}}\right]
=C_{2,i}\Delta_nv_n^{4-Y}(\zeta^{4-Y}-1)
+o_P(\Delta_nv_n^{4-Y}).
\end{equation}
Together with \eqref{PartExp0_difference}, this yields
\begin{align*}
   & \var({\xi_i}|\mathcal{F}_{i-1})\\ & = \frac{m_n}{\Delta_n^{-1} v_n^{4-Y}} K^2_{m_n\Delta_n}(t_{i-1}-\tau)\Big((\zeta^{4-Y}-1){C_{2,i}}\Delta_n v_n^{4-Y} + {o_P(\Delta_n v_n^{4-Y})}\\
     &\quad\qquad\qquad\qquad- \big\{(\zeta^{2-Y}-1){C_{1,i}}\Delta_nv_n^{2-Y}+(\zeta^{-Y}-1){D_{1,i}}\Delta_n^2v_n^{-Y}+{O_P(\Delta_n^3 v_n^{-2-Y})}\big\}^2 \Big)\nonumber\\
    & = \frac{m_n}{\Delta_n^{-1} v_n^{4-Y}} K^2_{m_n\Delta_n}(t_{i-1}-\tau)\big( (\zeta^{4-Y}-1){C_{2,i}}\Delta_n v_n^{4-Y} + o_P(\Delta_n v_n^{4-Y})
    \big),
\end{align*}
where $C_{2,i} = \frac{(C_++C_-)|\chi_{t_{i-1}}|^Y}{4-Y}$.
{Therefore,} with $1<Y<2$ and $ v_n \ll \Delta_n^{1/(4-Y)}$, we obtain
\begin{align*}
    s_{2n}^2 &:= \sum_{i=1}^n \var({\xi_i}|\mathcal{F}_{i-1}) \\
    &=m_n\Delta_n^2 \sum_{i=1}^nK^2_{m_n\Delta_n}(t_{i-1}-\tau)\frac{(C_++C_-)|\chi_{t_{i-1}}|^Y}{4-Y}(\zeta^{{4-Y}}-1) + o_P(1) \\
    &\longrightarrow
    \frac{(C_++C_-)|\chi_{\tau}|^Y}{4-Y}(\zeta^{{4-Y}}-1)\int K^2(x)dx.
\end{align*}
And consequently, with \eqref{kernel_convergence}, we have 
\begin{equation}\label{th2limitvar}
\frac{s_{2n}^2}{\big(\Delta_n\sum_{j=1}^nK_{m_n\Delta_n}(t_{j-1}-\tau)\big)^2} \rightarrow \frac{(C_++C_-)|\chi_{\tau}|^Y}{4-Y}(\zeta^{{4-Y}}-1)\int K^2(x)dx.
\end{equation}
For the higher moments, boundedness and nonnegativity of
$K$, together with the definition of $Z_i$, give
$$
0\leq Z_i
\leq \zeta^2\|K\|_\infty
\left(\frac{v_n^Y}{m_n\Delta_n}\right)^{1/2}.
$$
Since $\xi_i=Z_i-\EE_{i-1}(Z_i)$, the same
bound applies to $|\xi_i|$. Thus, using $|\xi_i|^p=|\xi_i|^{p-2}\xi_i^2$ and the bound above, for any fixed $p>2$ we have
\begin{align}\label{th2_momentvanish}
    &\frac{\sum_{i=1}^n\EE_{i-1}|\xi_i|^p}
    {\big(\Delta_n\sum_{j=1}^n
    K_{m_n\Delta_n}(t_{j-1}-\tau)\big)^p}
    \nonumber\\
    &\quad\lesssim
    \left(\frac{v_n^Y}{m_n\Delta_n}\right)^{(p-2)/2}
    \frac{s_{2n}^2}
    {\big(\Delta_n\sum_{j=1}^n
    K_{m_n\Delta_n}(t_{j-1}-\tau)\big)^p}
    \xrightarrow{P}0.
\end{align}
Indeed, $s_{2n}^2=O_P(1)$, and
$$
\frac{v_n^Y}{m_n\Delta_n}
=
v_n^{Y-\gamma}
\frac{v_n^\gamma}{m_n\Delta_n}
\longrightarrow0
$$
under the assumed condition
$v_n^\gamma\Delta_n^{-1}\ll m_n$ with $\gamma<Y$.

In view of \eqref{th2_momentvanish} and
\eqref{th2limitvar}, it remains to verify martingale
orthogonality. By \eqref{kernel_convergence}, it suffices
to establish the condition
$$
\sum_{i=1}^n
\left|\EE_{i-1}\left[\xi_i\Delta_iM\right]\right|
=
\sum_{i=1}^n
\left|\EE_{i-1}\left[Z_i\Delta_iM\right]\right|
\xrightarrow{P}0,
$$
when $M=W$ or $M$ is a bounded martingale orthogonal
to $W$. The equality follows from
$\EE_{i-1}(\Delta_iM)=0$. Equivalently, we must show that
$$
\frac{m_n^{1/2}}{\Delta_n^{-1/2}v_n^{2-Y/2}}
\sum_{i=1}^n K_{m_n\Delta_n}(t_{i-1}-\tau)
\left|
\EE_{i-1}\left[
(\Delta_iX)^2 f(\Delta_iX)\Delta_iM
\right]
\right|
=o_P(1),
$$
where $f(x):=\mathbf{1}_{\{v_n<|x|\leq\zeta v_n\}}$.
Similar to the proof of Lemma 12 in \cite{BONIECE2024}, there exists a \(C^2\) smooth approximation \(f_n\) of \(f\) such that for any \(\eta > 0\),
\[
\mathbf{1}_{\{v_n(1+ \frac{2}{3}v_n^\eta) < |x| < \zeta v_n(1 - \frac{2}{3}v_n^\eta)\}} \leq f_n(x) \leq \mathbf{1}_{\{v_n(1 + \frac{1}{3}v_n^\eta) < |x| < \zeta v_n(1 - \frac{1}{3}v_n^\eta)\}},
\]
\[
|f'_n(x)| \leq \frac{C}{v_n^{1+\eta}}, \quad |f''_n(x)| \leq \frac{C}{v_n^{2+2\eta}}.
\]
Then, the result follows from the following two limits
\begin{align}\label{diff_fn_approximation}
    &\frac{m_n^{1/2}}{\Delta_n^{-1/2}v_n^{2-Y/2}}
    \sum_{i=1}^n K_{m_n\Delta_n}(t_{i-1}-\tau)
    \left|
    \EE_{i-1}\left[
    (\Delta_iX)^2(f-f_n)(\Delta_iX)\Delta_iM
    \right]
    \right|
    =o_P(1)
\end{align}
and 
\begin{align}\label{diff_fn_technical}
    &\frac{m_n^{1/2}}{\Delta_n^{-1/2}v_n^{2-Y/2}}
    \sum_{i=1}^n K_{m_n\Delta_n}(t_{i-1}-\tau)
    \left|
    \EE_{i-1}\left[
    (\Delta_iX)^2 f_n(\Delta_iX)\Delta_iM
    \right]
    \right|
    =o_P(1).
\end{align}
{While the proof of \eqref{diff_fn_approximation} is relatively straightforward, the proof of \eqref{diff_fn_technical} is  nontrivial and lengthy. 
Both statements are shown in Appendix \ref{TecnPrf2}.}

Finally, we justify \eqref{GNMExp_difference}.
For any fixed $0<a<b$, apply
\eqref{PartExp0_difference} with $v_n$ replaced by
$av_n$ and $\zeta$ replaced by $b/a$. This gives
$$
\begin{aligned}
&\EE_{i-1}\left[(\Delta_iX)^2
\mathbf{1}_{\{av_n<|\Delta_iX|\leq bv_n\}}\right]\\
&\qquad=C_{1,i}\Delta_nv_n^{2-Y}(b^{2-Y}-a^{2-Y})
+o_P(\Delta_nv_n^{2-Y}).
\end{aligned}
$$
The remaining terms are negligible since
$\Delta_n/v_n^2\to0$ and $v_n^Y/\sqrt{\Delta_n}\to0$;
the latter follows from
$v_n^\gamma\ll m_n\Delta_n=O(\sqrt{\Delta_n})$
and $\gamma<Y$. Now, fix a partition $1=a_0<\cdots<a_L=\zeta$.
Consider the bands
$\mathcal I_j=\{a_{j-1}^2v_n^2<(\Delta_iX)^2\leq a_j^2v_n^2\}$. Since
$$v_n^2a_{j-1}^2\EE_{i-1}[(\Delta_iX)^2
\mathbf{1}_{\mathcal I_j}]\leq \EE_{i-1}[(\Delta_iX)^4
\mathbf{1}_{\mathcal I_j}]\leq v_n^2 a_{j}^2\EE_{i-1}[(\Delta_iX)^2
\mathbf{1}_{\mathcal I_j}],$$
summing over $j=1,\ldots,L$ and dividing by
$\Delta_nv_n^{4-Y}$ gives
$$
\begin{aligned}
C_{1,i}\sum_{j=1}^L
a_{j-1}^2(a_j^{2-Y}-a_{j-1}^{2-Y})+o_P(1)
&\leq
\frac{\EE_{i-1}[(\Delta_iX)^4
\mathbf{1}_{\{v_n<|\Delta_iX|\leq\zeta v_n\}}]}
{\Delta_nv_n^{4-Y}}\\
&\leq
C_{1,i}\sum_{j=1}^L
a_j^2(a_j^{2-Y}-a_{j-1}^{2-Y})+o_P(1).
\end{aligned}
$$
Both Riemann sums converge, as the mesh tends to zero, to
$\int_1^\zeta x^2\,d(x^{2-Y})
=(2-Y)(\zeta^{4-Y}-1)/(4-Y)$.
Since $C_{1,i}$ is bounded (due to localization),
taking $n\to\infty$ for each fixed partition and then
letting the mesh tend to zero proves
\eqref{GNMExp_difference}, using
$C_{2,i}=C_{1,i}(2-Y)/(4-Y)$.

\subsection{Proof of Theorem \ref{main}}\label{Ap3_6}
Denote $\mathcal{I} = \{(1,1),(2,1)\}$ and recall that 
\begin{align}\label{original}
    \tilde Z_n(m_n, v_n) = \sqrt {m_n} \left(\hat c_n(m_n, v_n) - \frac{\Delta_n\sum_{i=1}^{n}K_{m_n\Delta_n}(t_{i-1}-\tau)\sigma_{t_{i-1}}^2}{\Delta_n{\sum_{j=1}^nK_{m_n\Delta_n}(t_{j-1}-\tau)}} -{A(v_n,m_n)}\right),
\end{align}
where
\begin{align*}
    A(v_n,{m_n}) &= \frac{\sum_{i=1}^nK_{m_n\Delta_n}(t_{i-1}-\tau){C_{1,i} \Delta_n} v_n^{2-Y}}{\Delta_n{\sum_{j=1}^nK_{m_n\Delta_n}(t_{j-1}-\tau)}} + \frac{\sum_{i=1}^nK_{m_n\Delta_n}(t_{i-1}-\tau){D_{1,i}} {\Delta_n^2} v_n^{-Y}}{{\Delta_n{\sum_{j=1}^nK_{m_n\Delta_n}(t_{j-1}-\tau)}}}\\
    &=: \sum_{(i,j)\in \mathcal I} a_{i,j}( v_n).
\end{align*}
We introduce the following notation:
\begin{equation}\label{INNH}
\begin{aligned}
&\eta_{i,j}(\zeta) := \zeta^{2-2(i-j)-jY}-1,\\
    &\Psi_n:= u_n^{-1}\big(\tilde Z_n( m_n, \zeta_{1}^2 v_n) - 2\tilde Z_n( m_n, \zeta_{1}v_n) + \tilde Z_n( m_n, v_n)\big)
    = O_P(1),\\
    &\Phi_n := u_n^{-1}\big(\tilde Z_n( m_n, \zeta v_n)-\tilde Z_n( m_n, v_n)\big) = O_P(1),
\end{aligned}
\end{equation}
where $O_P(1)$ is a consequence of Theorem \ref{thm_clt_difference}.
Note that, for any $(i,j)$ we have
\begin{equation}\label{difference}
    \begin{aligned}
    &a_{i,j}(\zeta v_n) - a_{i,j}( v_n) = \eta_{i,j}(\zeta)a_{i,j}( v_n),\\
   & a_{i,j}(\zeta^2 v_n) - 2a_{i,j}(\zeta v_n)  + a_{i,j}( v_n) = \eta_{i,j}^2(\zeta)a_{i,j}( v_n).
\end{aligned}
\end{equation}
Recall $\Delta_nv_n^{-2}\ll1$ gives the order comparison:
\begin{align}
    a_{1,1}( v_n) \gg  a_{2,1}( v_n).
\end{align}
In the first stage of debiasing, the leading order term $a_{1,1}( v_n)$ is removed. %
\subsubsection*{First step: removal of $a_{1,1}( v_n)$}
Let
\begin{align}\label{debias1}
    \hat c_n^{(1)}(m_n,v_n,\zeta_{1}) &= \hat c_n(m_n, v_n) - \frac{\Big(\hat c_n(m_n,\zeta_{1} v_n)-\hat c_n(m_n, v_n)\Big)^2}{{\hat c_n(m_n,\zeta_{1}^2 v_n)} - 2\hat c_n(m_n,\zeta_{1} v_n) + \hat c_n(m_n, v_n)}\nonumber\\
    &=:\hat c_n(m_n, v_n) - \hat a_{1,1}( v_n).
\end{align}
Note that, due to \eqref{original}-\eqref{INNH},
\begin{align}\nonumber
\hat{c}_n(m_n,\zeta v_n)-\hat{c}_n(m_n,v_n)
&=m_n^{-1/2}\left(\tilde Z_n(m_n, \zeta v_n)-\tilde Z_n(m_n, v_n)\right)+A_n(\zeta v_n,m_n)-A_n(v_n,m_n)\\
&=m_n^{-1/2}u_n\Phi_n+A_n(\zeta v_n,m_n)-A_n(v_n,m_n).\label{DiffOfcFor}
\end{align}
Similarly, we can write:
\begin{align*}
&\hat{c}_n(m_n,\zeta^2 v_n)-2\hat{c}_n(m_n,\zeta v_n)+\hat{c}_n(m_n,v_n)\\
&\quad=m_n^{-1/2}u_n\Psi_n+A_n(\zeta^2 v_n,m_n)-2A_n(\zeta v_n,m_n)+A_n(v_n,m_n).
\end{align*}
{Then, in light of \eqref{INNH} and %
\eqref{difference}, we may expand $\hat a_{1,1}( v_n)$ as follows:
\begin{equation}\label{hata110}
\begin{aligned}
    \hat a_{1,1}( v_n)
    &= \frac{\big(\sum_{\substack{(r,s)\in \mathcal I}} \eta_{r,s}a_{r,s}( v_n)+m_n^{-1/2}u_n\Phi_n\big)^2}{\sum_{\substack{(r,s)\in \mathcal I}} \eta^2_{r,s}a_{r,s}( v_n)+m_n^{-1/2}u_n\Psi_n}\\
    &=\frac{\eta_{1,1}^2a^2_{1,1}+2\eta_{1,1}\eta_{2,1}a_{1,1}a_{2,1}+\eta_{2,1}^2a^2_{2,1}}{\sum_{\substack{(r,s)\in \mathcal I}} \eta^2_{r,s}a_{r,s}( v_n)+m_n^{-1/2}u_n\Psi_n}\\
    &\quad+\frac{2m_n^{-1/2}u_n\Phi_n\sum_{\substack{(r,s)\in \mathcal I}} \eta_{r,s}a_{r,s}( v_n)+\big(m_n^{-1/2}u_n\Phi_n\big)^2}{\sum_{\substack{(r,s)\in \mathcal I}} \eta^2_{r,s}a_{r,s}( v_n)+m_n^{-1/2}u_n\Psi_n}.
\end{aligned}
\end{equation}
Note that the last term above is $O_P(m_n^{-1/2}u_n)$ because $a_{1,1}\gg a_{2,1}$, $a_{1,1}\gg m_n^{-1/2}u_n\Phi_n$ and $a_{1,1}\gg m_n^{-1/2}u_n\Psi_n$. Then, we can write:
\begin{equation}\label{hata11}
\begin{aligned}
    \hat a_{1,1}( v_n)
    &=\frac{\eta_{1,1}^2a^2_{1,1}+2\eta_{1,1}\eta_{2,1}a_{1,1}a_{2,1}+\eta_{2,1}^2a^2_{2,1}}{\sum_{\substack{(r,s)\in \mathcal I}} \eta^2_{r,s}a_{r,s}( v_n)+m_n^{-1/2}u_n\Psi_n}+O_P(m_n^{-1/2}u_n)\\
     &=a_{1,1}( v_n) +O_P(m_n^{-1/2}u_n) \\
     &\quad+ \eta_{2,1}a_{2,1}( v_n)\frac{(2\eta_{1,1} - \eta_{2,1})a_{1,1} + \eta_{2,1}a_{2,1}( v_n)}{\eta^2_{1,1}a_{1,1}( v_n) + \eta^2_{2,1}a_{2,1}( v_n)+ m_n^{-1/2}u_n \Psi_n},
\end{aligned}
\end{equation}
where $O_P(m_n^{-1/2}u_n)$ now also includes  $m_n^{-1/2}u_n \Psi_n a_{1,1}/(\sum_{\substack{(r,s)\in \mathcal I}} \eta^2_{r,s}a_{r,s}( v_n)+m_n^{-1/2}u_n\Psi_n)$.} %

To simplify the last term in \eqref{hata11}, %
note that for any variable $N = o_p\big(a_{1,1}^2( v_n)\big)$,
\begin{align}\label{bigopapprox}
    &\frac{N}{\sum_{\substack{(r,s)\in \mathcal I}} a_{r,s}( v_n) + O_p( m_n^{-1/2}u_n)} -\frac{N}{\sum_{\substack{(r,s)\in \mathcal I}} a_{r,s}( v_n)}\nonumber\\
    &\quad= -\frac{N}{\sum_{\substack{(r,s)\in \mathcal I}} a_{r,s}( v_n)} \cdot \frac{O_P(m_n^{-1/2}u_n)}{\big(\sum_{\substack{(r,s)\in \mathcal I}} a_{r,s}( v_n) + O_P(m_n^{-1/2}u_n)\big)}\nonumber\\
    &\quad=\frac{-N}{a_{1,1}^2( v_n)(1+o_p(1))}\cdot O_P(m_n^{-1/2}u_n) \nonumber\\
    &\quad =o_P(m_n^{-1/2}u_n),
\end{align}
{where above, for simplicity, we omitted the coefficients associated with $a_{i,j}( v_n)$. The above shows that we can drop the term $m_n^{-1/2}u_n \Psi_n$ in the denominator and can write \eqref{hata11} as follows
\begin{equation*}%
\begin{aligned}
    \hat a_{1,1}( v_n)
     &=a_{1,1}( v_n) +O_P(m_n^{-1/2}u_n) \\
     &\quad+ \eta_{2,1}a_{2,1}( v_n)\frac{(2\eta_{1,1} - \eta_{2,1})a_{1,1} + \eta_{2,1}a_{2,1}( v_n)}{\eta^2_{1,1}a_{1,1}( v_n) + \eta^2_{2,1}a_{2,1}( v_n)}\\
     &=a_{1,1}( v_n)+\tilde{\eta}_{2,1}a_{2,1}+\breve{\eta}_{2,1}\frac{a_{2,1}^2(v_n)}{\eta^2_{1,1}a_{1,1}( v_n) + \eta^2_{2,1}a_{2,1}( v_n)}+O_P(m_n^{-1/2}u _n),
\end{aligned}
\end{equation*}
for some constants $\tilde{\eta}_{2,1}$ and $\breve{\eta}_{2,1}$. Note that the third term above is $O_P(a_{2,1}^2/a_{1,1})$ and, thus, this term is $o_P(m_n^{-1/2}u_n)$ due to our assumption  $m_n \ll \Delta_n^{-5} v_n^{8+Y}$. We conclude that 
\begin{equation}\label{hata112}
    \hat a_{1,1}(v_n)
     =a_{1,1}(v_n)+\tilde{\eta}_{2,1}a_{2,1}(v_n)
     +O_P(m_n^{-1/2}u _n).
\end{equation}
We now show the first debiasing step eliminates $a_{1,1}(v_n)$. Indeed, denoting $\tilde{a}_{2,1}(v_n):=(1-\tilde{\eta}_{2,1})a_{2,1}(v_n)$, in light of \eqref{original}, \eqref{debias1}, and \eqref{hata112} and recalling that $A(v_n,m_n)=a_{1,1}+a_{2,1}$, we see that
\begin{align*}
    &\tilde Z^{(1)}_n(m_n, v_n) \\
    &:= \sqrt{m_n} \left(\hat c_n^{(1)}(m_n, v_n,\zeta_{1}) - \tilde{a}_{2,1}(v_n)-{\frac{\Delta_n\sum_{i=1}^{n}K_{m_n\Delta_n}(t_{i-1}-\tau)\sigma_{t_{i-1}}^2}{\Delta_n{\sum_{j=1}^nK_{m_n\Delta_n}(t_{j-1}-\tau)}}}\right)\nonumber\\
    &=\sqrt{m_n} \left(\hat c_n(m_n, v_n) -\hat{a}_{1,1}(v_n)-(1-\tilde{\eta}_{2,1})a_{2,1}(v_n)-{\frac{\Delta_n\sum_{i=1}^{n}K_{m_n\Delta_n}(t_{i-1}-\tau)\sigma_{t_{i-1}}^2}{\Delta_n{\sum_{j=1}^nK_{m_n\Delta_n}(t_{j-1}-\tau)}}}\right)\nonumber\\
    &= \tilde Z_n(m_n, v_n) + O_P(u_n).
\end{align*}
Therefore, \eqref{clt_spot} implies that}
\begin{align}\label{clt_spot2b}
    (\tilde Z_n^{(1)}(m_n, v_n),\tilde Z'_n(m_n))
\xrightarrow{ st} (Z_1,Z_2).
\end{align}

\subsubsection*{Second step: removal of $a_{2,1}( v_n)$}
Similar to \eqref{debias1}, we can define
\begin{align*}
    \hat c_n^{(2)}(m_n, v_n,\zeta_{2},\zeta_{1}) &:= \hat c_n^{(1)}(m_n, v_n,\zeta_{1})\\
    &\quad- \frac{\Big(\hat c_n^{(1)}(m_n,\zeta_{2} v_n,\zeta_{1})-\hat c_n^{(1)}(m_n, v_n,\zeta_{1})\Big)^2}{\hat c_n^{(1)}(m_n,\zeta_{2}^2 v_n,\zeta_{1}) - 2\hat c_n^{(1)}(m_n,\zeta_{2} v_n,\zeta_{1}) + \hat c_n^{(1)}(m_n, v_n,\zeta_{1})}\nonumber\\
    &=:\hat c_n^{(1)}(m_n, v_n,\zeta_{1}) - \hat a_{2,1}( v_n).
\end{align*}
If $\Delta_n v_n^{-Y}\gg m_n^{-1/2}u_n$, the same
ratio expansion as in \eqref{hata11} gives
$\hat a_{2,1}(v_n)-\tilde a_{2,1}(v_n)
=O_P(m_n^{-1/2}u_n)$. 

Suppose next that $\Delta_n v_n^{-Y}/(m_n^{-1/2}u_n)\to\lambda\in[0,\infty)$.
Then $\tilde a_{2,1}(v_n)=O_P(m_n^{-1/2}u_n)$, so it suffices to show that $\hat a_{2,1}(v_n)=O_P(m_n^{-1/2}u_n)$.
By \eqref{hata112}, we have that 
$\hat c_n^{(1)}(m_n,\zeta_2v_n,\zeta_1)
-\hat c_n^{(1)}(m_n,v_n,\zeta_1)=O_P(m_n^{-1/2}u_n)$, and hence its square (the numerator of $\hat a_{2,1}(v_n)$) is of order 
$O_P(m_n^{-1}u_n^2)$. Unlike the preceding case, however, the bias term in the denominator may not dominate the denominator. For simplicity write the second-step denominator as $D_n:=\hat c_n^{(1)}(m_n,\zeta_2^2v_n,\zeta_1)
-2\hat c_n^{(1)}(m_n,\zeta_2v_n,\zeta_1)
+\hat c_n^{(1)}(m_n,v_n,\zeta_1)$.

For fixed $a\geq1$, write
$R_n(a):=u_n^{-1}\{\tilde Z_n(m_n,a\zeta_2^2v_n)
-2\tilde Z_n(m_n,a\zeta_2v_n)+\tilde Z_n(m_n,av_n)\}$.
Since $a_{2,1}(v_n)=O_P(m_n^{-1/2}u_n)$,
the first-order expansion of \eqref{hata110} has
remainder of order $O_P(m_n^{-1}u_n^2/a_{1,1}(v_n))
=o_P(m_n^{-1/2}u_n)$.
Applying this expansion at the three thresholds
$v_n,\zeta_2v_n,\zeta_2^2v_n$ gives
$$
\frac{D_n}{m_n^{-1/2}u_n}
=
(\zeta_2^{-Y}-1)^2
\frac{\tilde a_{2,1}(v_n)}{m_n^{-1/2}u_n}
+
\frac{\zeta_1^{4-2Y}R_n(1)
-2\zeta_1^{2-Y}R_n(\zeta_1)+R_n(\zeta_1^2)}
{(\zeta_1^{2-Y}-1)^2}
+o_P(1).
$$
The argument proving \eqref{clt_difference}, applied
to finite linear combinations of threshold differences,
gives joint stable convergence of
$(R_n(1),R_n(\zeta_1),R_n(\zeta_1^2))$
to a centered conditionally nondegenerate Gaussian vector.

By the definition of $\tilde a_{2,1}(v_n)$, 
$(\zeta_2^{-Y}-1)^2\tilde a_{2,1}(v_n)/(m_n^{-1/2}u_n)
\xrightarrow{P}
-\lambda(\zeta_2^{-Y}-1)^2(1-\tilde\eta_{2,1})
\frac{(C_++C_-)(Y+1)(Y+2)}{2Y}
\sigma_\tau^2|\chi_\tau|^Y$.  
Thus $D_n/(m_n^{-1/2}u_n)$ has a limiting distribution
with no atom at zero, yielding $m_n^{-1/2}u_n/D_n=O_P(1)$.
Together with the numerator bound, this gives
$\hat a_{2,1}(v_n)
=O_P(m_n^{-1}u_n^2)D_n^{-1}
=O_P(m_n^{-1/2}u_n)$, and consequently
$\hat a_{2,1}(v_n)-\tilde a_{2,1}(v_n)
=O_P(m_n^{-1/2}u_n)$.

Recalling the definition of $\hat c_{n,2}(m_n)$ in
\eqref{kernel_spot_estimator}, we then obtain
\begin{align*}
    \tilde Z^{(2)}_n(m_n, v_n)
    &:= \sqrt{m_n}\Big(
    \hat c_n^{(2)}(m_n,v_n,\zeta_2,\zeta_1)
    -\hat c_{n,2}(m_n)\Big)\\
    &= \tilde Z^{(1)}_n(m_n,v_n)
    +\sqrt{m_n}\Big(
    \tilde a_{2,1}(v_n)-\hat a_{2,1}(v_n)\Big)\\
    &= \tilde Z^{(1)}_n(m_n,v_n)+O_P(u_n).
\end{align*}
Again, \eqref{clt_spot2b} now implies that
\begin{align*}%
    (\tilde Z_n^{(2)}(m_n, v_n),\tilde Z'_n(m_n))
\xrightarrow{ st}  (Z_1,Z_2).
\end{align*}
To conclude the last two statements of the theorem, we follow the same arguments as in Remark \ref{SlyRmk1}.

\subsection{Proof of Theorem \ref{main2b}}\label{Ap4_1}
 We need the following notation:
 \begin{equation*}
 \begin{aligned}
 \check{Z}_n(v) :=  \sqrt{n}\left(\widehat C_n(v) - \int_0^1\sigma_s^2ds - \check{A}_n(v)  \right),\label{e:def_Zn}
\end{aligned}
\end{equation*}
where $\widehat C_n(v)$ is the integrated volatility estimator \eqref{TRVEst} and  
 \begin{align*}\label{e:def_a(eps,h)0}
     \check{A}_n(v) 
     &:=v^{2-Y}\frac{\bar{C}}{2-Y}\int_0^1|\chi_s|^{Y}ds - \Delta_nv^{-Y}\frac{\bar{C}(Y+1)(Y+2)}{2Y} \int_0^1|\chi_s|^Y\sigma^2_sds\\
     &=:\kappa_{1,1}v^{2-Y}+\kappa_{2,1}\Delta_nv^{-Y}.
 \end{align*}
 In \cite{BONIECE2024} (Proposition 1) it was shown that: 
 \begin{equation*}
 \begin{aligned}
     \check\Phi_n&:= u_n^{-1}\left(\check{Z}_n(\zeta pv_n)-\check{Z}_n(pv_n) \right) = O_P(1),\\
        \check\Psi_n&:={u_n^{-1}\left(\check{Z}_n(\zeta^2pv_n)-2\check{Z}_n(\zeta pv_n)+\check{Z}_n(pv_n)\right)} = O_P(1),
\end{aligned}
\end{equation*}
where, as before $u_n=\Delta_n^{-\frac{1}{2}}v_n^{\frac{4-Y}{2}}$. However, the result therein required $Y<8/5$. In Section \ref{PrfNwEsthPhi}, we show that under our constraints and the expanded range of $Y$-values, we can use the expansions \eqref{GNMExp} and \eqref{PartExp0_difference} to show that
\begin{equation}\label{GlobalLocalNoiseOrderPre}
\sqrt{m_n\Delta_n}\,\check\Phi_n=o_P(1), \qquad \sqrt{m_n\Delta_n}\,\check\Psi_n=o_P(1).
\end{equation}
In that case, 
\begin{equation}
\label{GlobalLocalNoiseOrder}
\sqrt{\Delta_n}u_n\check\Phi_n =o_P(m_n^{-1/2}u_n),\qquad
\sqrt{\Delta_n}u_n\check\Psi_n =o_P(m_n^{-1/2}u_n).
\end{equation}
Furthermore, the condition $v_n^\gamma\Delta_n^{-1}\ll m_n$, with
$\gamma<Y$, in Theorem \ref{thm_clt_difference} implies that
\begin{equation}
\label{GlobalLocalRelativeNoise}
\frac{v_n^Y}{m_n\Delta_n}
=v_n^{Y-\gamma}\frac{v_n^\gamma}{m_n\Delta_n}
\longrightarrow0.
\end{equation}
Recall that $\Upsilon=\eta_{1,1}^{-1}$, where
$\eta_{1,1}=\zeta^{2-Y}-1$ and $\eta_{2,1}=\zeta^{-Y}-1$
are defined as in \eqref{INNH}.
By \eqref{EstSlp0b}, \eqref{GlobalLocalNoiseOrder}, and \eqref{GlobalLocalRelativeNoise}, we can write
\begin{equation}
\label{ThetaAsympExp}
\begin{aligned}
\widehat\Upsilon_n
&=\frac{
\eta_{1,1}\kappa_{1,1}(pv_n)^{2-Y}+\eta_{2,1}\kappa_{2,1}\Delta_n(pv_n)^{-Y}+\sqrt{\Delta_n}u_n\check\Phi_n}
{\eta_{1,1}^2\kappa_{1,1}(pv_n)^{2-Y}
+\eta_{2,1}^2\kappa_{2,1}\Delta_n(pv_n)^{-Y}
+\sqrt{\Delta_n}u_n\check\Psi_n
}\\
&=\frac{
\eta_{1,1}\kappa_{1,1}
+\eta_{2,1}\kappa_{2,1}\Delta_n(pv_n)^{-2}
+o_P\!\left(v_n^{Y/2}/\sqrt{m_n\Delta_n}\right)
}{\eta_{1,1}^2\kappa_{1,1}
+\eta_{2,1}^2\kappa_{2,1}\Delta_n(pv_n)^{-2}
+o_P\!\left(v_n^{Y/2}/\sqrt{m_n\Delta_n}\right)
}\\
&=\Upsilon+O_P(\Delta_nv_n^{-2})
+o_P\!\left(\frac{v_n^{Y/2}}{\sqrt{m_n\Delta_n}}\right),
\end{aligned}
\end{equation}
where we used that  $\Delta_nv_n^{-2}\to0$. In particular, $\widehat\Upsilon_n\stackrel{\mathbb P}{\to}\Upsilon$
and $\widehat\Upsilon_n=O_P(1)$.
Similarly, as shown in \eqref{DiffOfcFor},
\begin{equation}
\label{GlobalLocalSpotDifference}
\begin{aligned}
&\hat c_{n,\tau}(m_n,\zeta v_n)-\hat c_{n,\tau}(m_n,v_n)\\
&\quad=A(\zeta v_n,m_n)-A(v_n,m_n)+m_n^{-1/2}u_n\Phi_n\\
&\quad=\Upsilon^{-1}a_{1,1}(v_n)
+O_P(\Delta_nv_n^{-Y})+O_P(m_n^{-1/2}u_n).
\end{aligned}
\end{equation}
Since $a_{1,1}(v_n)=O_P(v_n^{2-Y})$, \eqref{ThetaAsympExp} and 
\eqref{GlobalLocalSpotDifference} gives
\begin{equation*}
\begin{aligned}
&\widehat\Upsilon_n
\left(\hat c_{n,\tau}(m_n,\zeta v_n)
-\hat c_{n,\tau}(m_n,v_n)\right)\\
&=a_{1,1}(v_n)+O_P(\Delta_nv_n^{-Y})+O_P(m_n^{-1/2}u_n).
\end{aligned}
\end{equation*}
We can then write, by \eqref{debias_formula2b},
\begin{equation}
\label{GlobalLocalOneStepExpansion}
\begin{aligned}
\sqrt{m_n}\left(\tilde c^{(1)}_{n,\tau}-\sigma_\tau^2\right)
&=\sqrt{m_n}\left(
\hat c_{n,\tau}(m_n,v_n)-\sigma_\tau^2-a_{1,1}(v_n)
\right)\\
&\quad+O_P(\sqrt{m_n}\Delta_nv_n^{-Y})+O_P(u_n).
\end{aligned}
\end{equation}
Note that, under our assumptions
$m_n\sqrt{\Delta_n}\to\theta\in(0,\infty)$,
$v_n\gg\Delta_n^{3/(4Y)}$, and $v_n\ll\Delta_n^{1/(4-Y)}$,
\begin{equation*}
\begin{aligned}
\sqrt{m_n}\Delta_nv_n^{-Y}
=O(\Delta_n^{3/4}v_n^{-Y})
=O\!\left(\left(\frac{\Delta_n^{3/(4Y)}}{v_n}\right)^Y\right)
=o(1),\quad
u_n=\left(\frac{v_n^{4-Y}}{\Delta_n}\right)^{1/2}=o(1).
\end{aligned}
\end{equation*}
Finally, since $A(v_n,m_n)=a_{1,1}(v_n)+a_{2,1}(v_n)$ and
$a_{2,1}(v_n)=O_P(\Delta_nv_n^{-Y})$, note that we have
$\sqrt{m_n}a_{2,1}(v_n)=o_P(1)$. Therefore, by \eqref{ITQRTm},
\begin{equation}
\label{GlobalLocalOneStepLimit}
\begin{aligned}
\sqrt{m_n}\left(\tilde c^{(1)}_{n,\tau}-\sigma_\tau^2\right)
&=\sqrt{m_n}\left(
\hat c_{n,\tau}(m_n,v_n)-\sigma_\tau^2-A(v_n,m_n)
\right)+o_P(1)\\
&\xrightarrow{\mathrm{st}}Z_1+\theta Z_2,
\end{aligned}
\end{equation}
which concludes the proof.

\subsection{Proof of Theorem \ref{maintwo-step2b}}\label{Ap4_3}
Note that \eqref{ThetaAsympExp} implies that 
\begin{equation}\label{ThetaAsympExp2}
\widehat{\Upsilon}_n(\zeta)=\Upsilon(\zeta)+\frac{(\eta_{1,1}\eta_{2,1}-\eta_{2,1}^2) \kappa_{2,1}}{\eta_{1,1}^3\kappa_{1,1}}{p^{-2}}\Delta_nv_n^{-2}+O_P(\Delta_n^2v_n^{-4})+o_P\left(\frac{v_n^{Y/2}}{\sqrt{m_n\Delta_n}}\right).%
\end{equation}
Furthermore, since $ v_n\ll {\Delta_n^{\frac{5}{2(4+Y)}}}$ and $m_n\sim \theta \Delta_n^{-1/2}$, we have ${v_n^{Y/2}/\sqrt{m_n\Delta_n}}\ll\Delta_n v_n^{-2}$. 
Together with \eqref{DiffOfcFor}, we obtain that, for a constant $\check \kappa_{2,1}$,
\begin{align*}
    \widehat{\Upsilon}_{n}(\zeta)(\hat{c}_{n,\tau}(m_n,\zeta v_n)-\hat{c}_{n,\tau}(m_n,v_n))&=
    a_{1,1}+\check\kappa_{2,1}a_{2,1}+O_P(m_n^{-1/2}u_n)+O_P(\Delta_n^2v_n^{-Y-2})\\
    &=
    a_{1,1}+\check\kappa_{2,1}a_{2,1}+O_P(m_n^{-1/2}u_n),
\end{align*}
where in the last equality, we used our assumption $v_n\gg \Delta_n^{\frac{9}{2(8+Y)}}$ to neglect $O_P(\Delta_n^2v_n^{-Y-2})$. 
Therefore, 
\begin{align*}
    \tilde c_n^{(1)}(m_n, v_n,\zeta)&=\hat c_n(m_n, v_n)- \widehat{\Upsilon}_{n}(\zeta)(\hat{c}_{n,\tau}(m_n,\zeta v_n)-\hat{c}_{n,\tau}(m_n,v_n))\\
    &=\hat c_n(m_n, v_n) -a_{1,1}(v_n)-\check\kappa_{2,1}a_{2,1}(v_n)+O_P(m_n^{-1/2}u_n).
\end{align*}
Setting
$\check a_{2,1}(v_n):=(1-\check\kappa_{2,1})a_{2,1}(v_n)$,
we may equivalently write
\begin{align}\label{RefEst2de}
    &\tilde c_n^{(1)}(m_n, v_n,\zeta)\notag\\
    &\quad=\check{a}_{2,1}(v_n)+\Delta_n\sum_{i=1}^{n}K_{m_n\Delta_n}(t_{i-1}-\tau)\sigma_{t_{i-1}}^2+m_{n}^{-1/2}\tilde Z_n(m_n, v_n) + O_P(m_{n}^{-1/2}u_n),
\end{align}
and 
\begin{align}\nonumber
    \tilde c_n^{(1)}(m_n, \zeta_2 v_n,\zeta_1)-\tilde c_n^{(1)}(m_n, v_n,\zeta_1)
   & =(\zeta_2^{-Y}-1)\check{a}_{2,1}(v_n)\\
   \nonumber
    &\quad +m_{n}^{-1/2}(\tilde Z_n(m_n, \zeta_2 v_n)-\tilde Z_n(m_n, v_n))+O_P(m_{n}^{-1/2}u_n)\\
&=\eta_{2,1}(\zeta_2) \check{a}_{2,1}(v_n)+m_{n}^{-1/2}u_{n}\check{\Phi}^{(1)}, 
\label{RefDifEst2de}
\end{align}
where, due to Theorem \ref{thm_clt_difference},  $\check{\Phi}^{(1)} =O_P(1)$. 
Next, let $f(x)=\zeta_2^2x/[1-(\zeta_2^2-1)x]$.
Since $1-(\zeta_2^2-1)\Upsilon(\zeta_2)\neq0$,
$f$ has a bounded derivative in a neighborhood of
$\Upsilon(\zeta_2)$.
Taylor's formula and
$\widehat\Upsilon_n(\zeta_2)-\Upsilon(\zeta_2)
=O_P(\Delta_nv_n^{-2})$ therefore give
$$
f(\widehat\Upsilon_n(\zeta_2))
=f(\Upsilon(\zeta_2))+O_P(\Delta_nv_n^{-2})
=\frac{1}{\zeta_2^{-Y}-1}+O_P(\Delta_nv_n^{-2}).
$$
Together with \eqref{RefDifEst2de}, we conclude that
\[
\left(\frac{\zeta_2^2\widehat\Upsilon_n(\zeta_2)}{1-(\zeta_2^2-1)\widehat\Upsilon_n(\zeta_2)}\right) \left(\tilde c^{(1)}_{n,\tau}(m_n,\zeta_2 v_n,\zeta_1)-\tilde c^{(1)}_{n,\tau}(m_n,v_n,\zeta_1)\right)
=\check{a}_{2,1}(v_n)+O_P(m_n^{-1/2}
u_n).\]
where again we used $v_n\gg \Delta_n^{\frac{9}{2(8+Y)}}$ to neglect $\Delta_nv_n^{-2}\check{a}_{2,1}=O_P(\Delta_n^2v_n^{-Y-2})\ll m_n^{-1/2}
u_n$. 
Combining with \eqref{RefEst2de}, 
\begin{align*}
    {\hat{Z}}^{(2)}_n(m_n, v_n):=m_{n}^{1/2}\left({\tilde{c}}^{(2)}_{n,\tau}-\Delta_n\sum_{i=1}^{n}K_{m_n\Delta_n}(t_{i-1}-\tau)\sigma_{t_{i-1}}^2\right)=\tilde Z_n(m_n, v_n) + O_P(u_n).
\end{align*}
In particular, by Theorem \ref{clt}, 
\begin{align}\label{clt_spothh}
    ({\hat{Z}}^{(2)}_n(m_n, v_n),\tilde Z'_n(m_n))
\xrightarrow{ st}  (Z_1,Z_2),
\end{align}
and 
\[
    m_{n}^{1/2}\left({\tilde{c}}^{(2)}_{n,\tau}-\sigma^2_{\tau}\right)={\hat{Z}}^{(2)}_n(m_n, v_n)+\theta\tilde Z'_n(m_n)\xrightarrow{ st} Z_1 + \theta{Z_2}.
\]

\section{Additional technical proofs}\label{ApTechPrfs}
Throughout the proofs below, we make use of the following notation:
\begin{align}\nonumber
    x_i &:={\mu_{t_{i-1}}}\Delta_n + \sigma_{t_{i-1}}\Delta_iW + \chi_{t_{i-1}}\Delta_iJ,\\
    \label{discretizing}
    \hat x_i&:=\sigma_{t_{i-1}}\Delta_iW + \chi_{t_{i-1}}\Delta_iJ,\\
    \nonumber
    \varepsilon_i &:= \Delta_iX - x_i=\int_{t_{i-1}}^{t_i}\sigma_tdW_t-\sigma_{t_{i-1}}\Delta_iW + \int_{t_{i-1}}^{t_i}\chi_tdJ_t-\chi_{t_{i-1}}\Delta_iJ + \int_{t_{i-1}}^{t_i}{\mu}_tdt-{\mu}_{t_{i-1}}\Delta_n.
\end{align}
For the reduced process $X'$ defined in \eqref{DfnOfXpr0},
with the notation in \eqref{discretizing} understood
accordingly, we use the following standard moment estimate
(see Lemma 3 in \cite{BONIECE2024supplement}):
\begin{equation}\label{KME0}
    \mathbb{E}\left[|\varepsilon_i|^p\big|\mathcal{F}_{t_{i-1}}\right]\leq C
    \Delta_{n}^{\frac{(2\wedge p)+p}{2}},
\end{equation}
for any $p>0$ and a constant $C$.
\subsection{Proof of  Proposition \ref{prop}}\label{PrfProp1a}
{Similar to \cite{BONIECE2024} and \cite{BONIECE2024supplement}, we approximate the original process $X_t$ by
\begin{equation}\label{DfnOfXpr0}
    X_t' := \int_0^t {\mu}_s \, ds + \int_0^t \sigma_s \, dW_s + \int_0^t \chi_s \, dJ_s^\infty,
\end{equation}
where, with $\delta_0 \in (0,1)$ being chosen such that $q(x) > 0$ for all $|x| \leq \delta_0$ and denoting the jump measure of $J$ and its compensated version as $N$ and $\bar{N}$, respectively, the process $J^\infty$ is defined as
\begin{equation}\label{DfnJunfty0}
    J_t^\infty := \left( {b - \int_{\delta_0 < |x| \leq 1} x \, \nu(dx)} \right)t + \int_0^t \int_{|x| \leq \delta_0} x \, \bar{N}(ds,dx) + \breve{J}_t.
\end{equation}
The term $\breve J_t$ is a compound Poisson process,
independent of $J$, with no deterministic drift and
with L\'evy measure $\breve\nu$ given by
\begin{equation*}
    \breve{\nu}(dx) := e^{-|x|^p} \left( C_+ \mathbf{1}_{(0, \infty)}(x) + C_- \mathbf{1}_{(-\infty, 0)}(x) \right) \mathbf{1}_{|x| > \delta_0} |x|^{-1-Y} \, dx,
\end{equation*}
for a fixed constant $0<p < 1 \wedge Y$. Additionally, we write
\begin{equation}\label{e:def_J0}
    J_t^0 := J_t - J_t^{\infty} = \int_0^t \int_{|x|>\delta_0}xN(ds,dx) - \breve J_t,
    \end{equation}
and note that
$$X_t - X_t' = \int_0^t\int\delta(s,z)\mathfrak{p}(ds,dz) + \int_0^t \chi_sdJ_s^0.$$
In conclusion, the process $X'$ contains the continuous part and the infinite-variation jump component, and $X-X'$ contains only bounded-variation jump terms.}

{As mentioned before, the expansions  of Proposition \ref{prop} are based on arguments and analogous results ({Lemmas} 2, 3, 4 and 5) in \cite{BONIECE2024}. It is not difficult to verify that the conclusions of the lemmas therein are valid even in the presence of jumps in the volatility. Indeed, in addition to the approximation of the L\'evy process $J$ above by $J^{\infty}$, the other main ingredient is the approximation of the stochastic integrals involved in the short-time increments $\Delta_i X$ by $x_i$ as in \eqref{discretizing}. As it turns out (see, e.g., Lemma 3 in \cite{BONIECE2024supplement} and Eq.(6.24) in \cite{Jacod2014}), the bound on the error $\Delta_i X-x_i$ stated in \eqref{KME0} remains valid even in the presence of jumps in the volatility.}

\subsubsection{Proof of \eqref{GNMExp}}
This expansion follows from the one obtained in \cite[Lemma 4]{BONIECE2024}. Indeed, the condition $Y<8/5$ stated in the hypotheses of that lemma is not used in its proof.

\subsubsection{Proof of \eqref{PartExp0}}
The proof of \eqref{PartExp0} follows along the lines of the proof of Lemma 2 and Lemma 5 in \cite{BONIECE2024}. The main difference is that \cite{BONIECE2024} has the constraint $v_n\gg \Delta_n^{\frac{3}{2(2+Y)}}$ as well as the condition $Y\leq 8/5$. However, these restrictions were imposed to achieve an error term of order $o_P(\Delta_n^{3/2})$ in the expansion therein, while, in our case, it suffices to have an error of order $o_P(\Delta_n^{\frac{5}{4}})+O_P(\Delta_n^3 v_n^{-2-Y})$. Under this more relaxed constraint, we will show the statement \eqref{PartExp0} holds under the entire range $Y\in (0,1)\cup (1,2)$.%

Note that the proof of Lemma 2 in \cite{BONIECE2024} remains valid for $0<Y<1$, without relying on the condition $v_n\gg \Delta_n^{\frac{3}{2(2+Y)}}$. Hence, Lemma 2 holds for $0<Y<1$ under the weaker condition $\Delta_n^{1/2}\ll v_n\ll \Delta_n^{1/(4-Y)}$, which in turn implies \eqref{PartExp0}. In the following, we focus on the proof for the case $1<Y<2$.

First, we show that the expansion \eqref{PartExp0} holds {when $X$ is a sum of a Brownian motion and the tempered stable L\'evy process $J^\infty$ defined in \eqref{DfnJunfty0}. This fact actually follows directly from} Proposition 2 in \cite{BONIECE2024supplement}. Indeed, it was shown there that {for such a process} the error, excluding exponentially decaying terms, is of order $O_P(\Delta_n^{3}v_n^{-Y-2})+O_P(\Delta_n^{2}v_n^{2-2Y})+O_P(\Delta_nv_n^{2-\bar\delta})+O_P(\Delta_n^{3/2}v_n^{1-Y/2})$. However, we observe that both $\Delta_nv_n^{2-\bar\delta}$ and $\Delta_n^{\frac{3}{2}}v_n^{1-Y/2}$ are $o_P(\Delta_n^{\frac{5}{4}})$ and $\Delta_n^3v_n^{-Y-2}\gg\Delta_n^{2}v_n^{2-2Y}$ {since $v_n\ll\Delta_n^{1/(4-Y)}$}. Therefore, the expansion of Proposition 2 in \cite{BONIECE2024supplement} is valid for any $Y\in (1,2)$ with the errors being of order $O_P(\Delta_n^3v_n^{-2-Y})+o_P(\Delta_n^{5/4})$.

{As in \cite{BONIECE2024}, we next extend the expansion for the process $X'$ defined in \eqref{DfnOfXpr0}, which broadly replaces $J$ with $J^\infty$ and omits the last bounded variation Poisson term in $X$. This is done in Lemma 2 in \cite{BONIECE2024}. To conclude that the expansion therein holds, %
but with an adjusted error of order $O_P(\Delta_n^3v_n^{-2-Y})+o_P(\Delta_n^{5/4})$,} it suffices to show the terms $R_{i,1}$ and $R_{i,4}$ appearing in the proof of \cite[Lemma 2]{BONIECE2024} are of order $o_P(\Delta_n^{5/4})$. We only show $R_{i,4}=o_P(\Delta_n^{5/4})$ since the argument for $R_{i,1}$ is similar. It suffices to bound $D_{i,1}$ and $D_{i,2}$ as defined on \cite[p.20]{BONIECE2024}.
To obtain $D_{i,1}=o_P(\Delta_n^{5/4})$, under the constraint $1<Y<2$, the analog of \cite[eq.~(B.12)]{BONIECE2024} reads
\begin{align}\label{b.12}
    \Delta_nv_n^{1-Y}\delta \ll \Delta_n^{5/4} \quad\iff \quad\delta\ll \Delta_n^{1/4}v_n^{Y-1}.
\end{align}
To obtain $D_{i,2}=o_P(\Delta_n^{5/4})$ for $1<Y<2$, it suffices to bound the terms $\mathcal V_{i,2}$ and $\mathcal V_{i,3}$ defined on \cite[p.~21]{BONIECE2024}. Following the same argument leading to \cite[eq.~(B.13)]{BONIECE2024}, to obtain $\mathcal V_{i,1}=o_P(\Delta_n^{5/4}),$ we require, for arbitrarily small $s',s''>0$, %
\begin{align}\label{b.13}
    \Delta_n^{5/2-s'}v_n^{1-Y-s''}\delta^{-1} \ll \Delta_n^{5/4} \quad\iff \quad \delta\gg \Delta_n^{5/4-s'}v_n^{1-Y-s''}.
\end{align}
The conditions \eqref{b.12} and \eqref{b.13} are consistent for any $1<Y<2$, since
$$\Delta_n^{5/4-s'}v_n^{1-Y-s''}\ll \Delta_n^{1/4}v_n^{Y-1}\Leftrightarrow\Delta_n^\frac{1-s'}{2Y-2+s''}\ll v_n,$$
which holds for sufficiently small $s',s''$ under our condition $\Delta_n^{\frac{1}{2}-s}\ll v_n$. Finally, for the term $\mathcal V_{i,3},$ since $v_n\ll\Delta_n^{\frac{1}{4-Y}}$, the condition $\delta\gg \Delta_n^{\frac{1}{2}}v_n^{\frac{Y}{2}}$ imposed at the end of the proof of \cite[Lemma 2]{BONIECE2024} is consistent with \eqref{b.12}. {The condition $\delta\gg \Delta_n^{\frac{1}{2}}v_n^{\frac{Y}{2}}$ enables us to apply \cite[Lemma 4]{BONIECE2024supplement} and conclude that $\EE(V_{i,3}) = O(\Delta_n^2v_n^{2-2Y})$, which is also $O(\Delta_n^3v_n^{-Y-2})$, as desired.}

To  {conclude the proof},
with the updated error $O_P(\Delta_n^3v_n^{-2-Y})+o_P(\Delta_n^{5/4})$ and the more general process $X$, it remains to establish an analog of Lemma 5 in \cite{BONIECE2024}, {which allows us to add the last bounded variation term to $X$ and go from $J^\infty$ to the more general tempered stable L\'evy process $J$}. This is detailed in Section \ref{YWeCAL5} below. %

\subsubsection{Proof of \eqref{PartExp0_difference}}
We first note {that} the expansion \eqref{PartExp0_difference} holds when $X$ is a sum of a Brownian motion and {the independent tempered stable L\'evy process $J^\infty$} on account of \cite[Proposition 2]{BONIECE2024supplement}, which implies that {that the error terms in \eqref{PartExp0_difference} are $O_P(\Delta_n^{3}v_n^{-Y-2})+O_P(\Delta_n^{2}v_n^{2-2Y})+O_P(\Delta_nv_n^{2-\bar\delta})+O_P(\Delta_n^{3/2}v_n^{1-Y/2})$. As with the proof of \eqref{PartExp0},
we can similarly check that $O_P(\Delta_n^3v_n^{-2-Y})$ and $o_P(\Delta_n^{\frac{3}{4}}v_n^{\frac{4-Y}{2}})$ dominate the higher order terms $\Delta_n^2v_n^{2-2Y}$, $\Delta_n^{\frac{3}{2}}v_n^{1-\frac{Y}{2}}$ and $\Delta_nv_n^{2-\bar\delta}$ therein.}

Next, {similarly to \cite[Lemma 3]{BONIECE2024}, we argue that the expansion \eqref{PartExp0_difference} holds for the process $X'$ defined in \eqref{DfnOfXpr0}}.
{Indeed, following the proof of \cite[Lemma 3]{BONIECE2024}}, it suffices to obtain bounds for the terms $\overline{\mathcal D}_{i,1}$ and $\overline{\mathcal D}_{i,2}$ on \cite[p.22]{BONIECE2024}. First to show $\overline{\mathcal D}_{i,1}=o_P(\Delta_n^{\frac{3}{4}}v_n^{\frac{4-Y}{2}})$, we arrive at the constraint analogous to \cite[eq.~(B.17)]{BONIECE2024}:
\begin{align}\label{b.17}
    \Delta_nv_n^{1-Y}\delta \ll \Delta_n^{\frac{3}{4}}v_n^{\frac{4-Y}{2}} \quad\iff\quad \delta\ll \Delta_n^{-\frac{1}{4}} v_n^{1+\frac{Y}{2}}.
\end{align}
Similarly, to obtain $\overline{\mathcal D}_{i,2}=o_P(\Delta_n^{\frac{3}{4}}v_n^{\frac{4-Y}{2}})+ O_P(\Delta_n^{3}v_n^{-Y-2})$, we first recall that
(see \cite[eq.~(B.19)]{BONIECE2024}), under the constraint $\delta\gg \Delta_n^{1/2} v_n^{Y/2}$,
\[
    \overline{\mathcal D}_{i,2}=O_P(\Delta_n^{5/2-s'}v_n^{1-Y-s''}\delta^{-1})
    +O_P(\Delta^2 v_n^{2-2Y}).
\]
The second term above is $O_P(\Delta_n^{3}v_n^{-Y-2})$ because $v_n\ll \Delta_n^{1/(4-Y)}$. For the first term above to be $o_P(\Delta_n^{\frac{3}{4}}v_n^{\frac{4-Y}{2}})$, we need:
\begin{align}\label{b.19}
     \Delta_n^{5/2-s'}v_n^{1-Y-s''}\delta^{-1} \ll \Delta_n^{\frac{3}{4}}v_n^{\frac{4-Y}{2}}  \quad\iff\quad  \Delta_n^{\frac{7}{4}-s'} v_n^{-1-\frac{Y}{2}-s''}\ll\delta,
\end{align}
which is compatible with \eqref{b.17} if $\Delta_n^{\frac{7}{4}} v_n^{-1-\frac{Y}{2}}\ll \Delta_n^{-\frac{1}{4}} v_n^{1+\frac{Y}{2}}$ or, equivalently, $\Delta_n^{\frac{2}{2+Y}}\ll v_n$. The latter obviously holds for 
all $0<Y<2$ since $v_n\gg \Delta_n^{\frac{1}{2}}$. On the other hand, \eqref{b.17} is also compatible with the constraint $\delta\gg \Delta_n^{1/2} v_n^{Y/2}$ required for \cite[eq.~(B.19)]{BONIECE2024} since $\Delta_n^{1/2}\ll v_n$. The remaining part of the proof is the same as Lemma 3 in \cite{BONIECE2024}. 
To carry the expansion of \cite[Lemma 3]{BONIECE2024}, with the updated error $O_P(\Delta_n^3 v_n^{-2-Y})
+o_P(\Delta_n^{\frac{3}{4}}v_n^{\frac{4-Y}{2}})$, to the most general process $X$, it remains to establish an appropriate analog of Lemma 5 in \cite{BONIECE2024}. This is detailed in Section \ref{YWeCAL5} below.

\subsubsection{Proof of the analog of Lemma 5 in \cite{BONIECE2024}}\label{YWeCAL5}

With $r = r_0\in (0,Y\wedge1)$, Lemma 5 in \cite{BONIECE2024} states that for any $\alpha\geq1$,
\begin{align}
    &|(\Delta_iX)^{2p}\mathbf{1}_{\{|\Delta_iX|\leq v_n\}} - (\Delta_iX')^{2p}\mathbf{1}_{\{|\Delta_iX'|\leq v_n\}}|\notag\\
    &\quad=O_P(\Delta_nv_n^{2p-\alpha r_0}) + O_P(\Delta_nv_n^{2p-Y-1+\alpha}).\label{lemma5in3}
\end{align}
The above equation still holds under the hypotheses of  Proposition \ref{prop}.  To establish  \eqref{GNMExp}, it is further required that the above terms are 
\(o_P(\Delta_n v_n^{2p-Y})\), which can be ensured by choosing 
\(\alpha\) close to 1 from above, since $ r_0 < Y$, $\Delta_nv_n^{-r_1\alpha}\ll\Delta_nv_n^{-2}\ll 1$ under our condition $r_0\in[0,Y\wedge1)$. 

For \eqref{PartExp0} and \eqref{PartExp0_difference}, 
the corresponding requirements for the errors in \eqref{lemma5in3} are \(o_P(\Delta_n^{5/4})\) and 
\(o_P(\Delta_n^{3/4}v_n^{2-Y/2})\), respectively. Since 
$\Delta_n^{3/4} v_n^{2-Y/2} \ll \Delta_n^{5/4}$ under $v_n \ll \Delta_n^{1/(4-Y)}$,
it suffices to establish that the terms are \(o_P(\Delta_n^{3/4}v_n^{2-Y/2})\), which occurs if the following hold: 

\begin{enumerate}

    \item $\Delta_nv_n^{2-\alpha r_0} \ll \Delta_n^{3/4}v_n^{2-Y/2}$, {which holds} if $\alpha < (Y+1)/(2r_0)$. This is obtained by observing $\Delta_nv_n^{2-\alpha r_0} \ll \Delta_n^{3/4}v_n^{2-Y/2} \Leftrightarrow \Delta_n\ll v_n^{4\alpha r_0-2Y}$ {and, since we required $\Delta_n^{1/2}\ll v_n$, we only need to impose that $v_n^2\ll v_n^{4\alpha r_0-2Y}$, which is satisfied if $\alpha < (Y+1)/(2r_0)$.}
    \item $\Delta_nv_n^{2-Y-1+\alpha} \ll \Delta_n^{3/4}v_n^{2-Y/2}$ if $\alpha > 1/2 + Y/2$. This is obtained by observing $\Delta_nv_n^{2-Y-1+\alpha} \ll \Delta_n^{3/4}v_n^{2-Y/2} \Leftrightarrow \Delta_n^{1/(
    2Y+4-4\alpha)}\ll v_n$ and comparing $\Delta_n^{1/(2Y+4-4\alpha)}\ll \Delta_n^{1/2}$.
\end{enumerate}
In conclusion, we need to pick $\alpha$ such that $$\frac{1+Y}{2}<\alpha<\frac{1+Y}{2r_0},$$
which is always possible given $r_0\in[0,1)$. %

\subsubsection{Proof of \eqref{prop1_non_integer}}\label{FPNH0}

{As with the proofs of \eqref{PartExp0} and \eqref{PartExp0_difference}, we proceed in two steps: first we show \eqref{prop1_non_integer} for the process $X'$ defined in \eqref{DfnOfXpr0} and then, we show the analog of Lemma 5 in \cite{BONIECE2024};} i.e., we show that, for any real number $p>1$,
\begin{equation}\label{finite_jump_difference}
    ||\Delta_iX|^{2p}\mathbf{1}_{\{|\Delta_iX|\leq v_n\}} - |\Delta_iX'|^{2p}\mathbf{1}_{\{|\Delta_iX'|\leq v_n\}}| =O_P(\sigma_{t_{i-1}}^{2p}\Delta_n^p)+ O_P(\Delta_nv_n^{2p-Y}).
\end{equation}
Note that \eqref{finite_jump_difference} allows non-integer moment and does not directly follow from Lemma 5 in \cite{BONIECE2024}.

For simplicity, we re-use the notation in \eqref{discretizing}, except with $J_t^{\infty}$ in place of $J_t$. To {prove} \eqref{prop1_non_integer}, %
    we need to show that, for any $p>1$, the following two estimates hold:
    \begin{align}\label{prop1_non_integer_levy_drift}
        \EE_{i-1}|x_i\mathbf{1}_{\{|x_i|\leq v_n\}}|^{2p} = O_P(\sigma_{t_{i-1}}^{2p}\Delta_n^p) + O_P(\Delta_nv_n^{2p-Y}),
    \end{align}
    and
    \begin{align}\label{prop1_non_integer_semidiff}
        \EE_{i-1}\big|(\Delta_iX)\mathbf{1}_{\{|\Delta_iX|\leq v_n\}}|^{2p}-\EE_{i-1}|x_i\mathbf{1}_{\{|x_i|\leq v_n\}}|^{2p} = O_P(\sigma_{t_{i-1}}^{2p}\Delta_n^p) + O_P(\Delta_nv_n^{2p-Y}).
    \end{align}
We begin with \eqref{prop1_non_integer_levy_drift} and proceed along the lines of Proposition 3 in \cite{BONIECE2024supplement}. It suffices to show
\begin{align}\label{prop1_non_integer_levy}
    \EE_{i-1}\big(|
    \sigma_{t_{i-1}}\Delta_iW +\chi_{t_{i-1}} \Delta_iJ^\infty|^{2p}
    \mathbf{1}_{\{|\sigma_{t_{i-1}}\Delta_iW +\chi_{t_{i-1}} \Delta_iJ^\infty|\leq v_n\}}\big)= O_P(\sigma_{t_{i-1}}^{2p}\Delta_n^p) + O_P(\Delta_nv_n^{2p-Y}).
\end{align}
Note that 
\begin{align*}
    &\EE_{i-1}\big(\big|
    \sigma_{t_{i-1}}\Delta_iW + \chi_{t_{i-1}}\Delta_iJ^\infty\big|^{2p}\mathbf{1}_{\{|\sigma_{t_{i-1}}\Delta_iW +\chi_{t_{i-1}} \Delta_iJ^\infty|\leq v_n\}}\big) \\
    &\quad\lesssim \EE_{i-1}| \sigma_{t_{i-1}}\Delta_iW|^{2p} + \EE_{i-1}\big(| \Delta_iJ^\infty|^{2p}\mathbf{1}_{\{|\sigma_{t_{i-1}}\Delta_iW +\chi_{t_{i-1}} \Delta_iJ^\infty|\leq v_n\}}\big),
\end{align*}
and, by self-similarity, $\EE_{i-1}| \sigma_{t_{i-1}}\Delta_iW|^{2p} =O_P( \sigma_{t_{i-1}}^{2p}\Delta_n^{p})$. Thus, to show \eqref{prop1_non_integer_levy}, it suffices that
$$ 
A_n:=\EE_{i-1}\big( |\chi_{t_{i-1}}\Delta_iJ^\infty|^{2p}\mathbf{1}_{\{|\sigma_{t_{i-1}}\Delta_iW +\chi_{t_{i-1}} \Delta_iJ^\infty|\leq v_n\}}\big)\lesssim \Delta_nv_n^{2p-Y}+\Delta_n^p.$$
Note that $\mathbf{1}_{\{|\sigma_{t_{i-1}}\Delta_iW +\chi_{t_{i-1}} \Delta_iJ^\infty|\leq v_n\}} \leq \mathbf{1}_{|\sigma_{t_{i-1}}\Delta_iW|\leq v_n}\mathbf{1}_{\{|\chi_{t_{i-1}}\Delta_iJ^\infty|\leq 2v_n\}} + \mathbf{1}_{|\sigma_{t_{i-1}}\Delta_iW|> v_n}$. Then we may further write
\begin{align*}
    &A_n\leq \EE_{i-1}\big( |\chi_{t_{i-1}}\Delta_iJ^\infty|^{2p}\mathbf{1}_{\{|\chi_{t_{i-1}}\Delta_iJ^\infty|\leq 2v_n\}}\big) \\
    &\qquad+ \EE_{i-1}\big( |\chi_{t_{i-1}}\Delta_iJ^\infty|^{2p}\mathbf{1}_{\{|\sigma_{t_{i-1}}\Delta_iW +\chi_{t_{i-1}} \Delta_iJ^\infty|\leq v_n\}}\cdot\mathbf{1}_{|\sigma_{t_{i-1}}\Delta_iW|> v_n}\big)\\
    &\quad=: V_1+V_2.
\end{align*}
For $V_1$, by {(B.38) in \cite{BONIECE2024supplement}}, we obtain that
\begin{align}\label{prop1_v1}
    V_1 \lesssim \Delta_nv_n^{2p-Y} + \Delta_n^p.
\end{align}
For $V_2$, let $K$ be the smallest positive integer such that $K> p$. And further assume $\sigma_{t_{i-1}}\neq 0$, otherwise $V_2=0$. By H\"older's inequality and {(B.37) in \cite{BONIECE2024supplement}}, we have 
\begin{align}
    V_2 &\leq \EE_{i-1}\big( |\chi_{t_{i-1}}\Delta_iJ^\infty|^{2K}\mathbf{1}_{\{|\sigma_{t_{i-1}}\Delta_iW +\chi_{t_{i-1}} \Delta_iJ^\infty|\leq v_n\}}\big)^{\frac{p}{K}}P\big(|\sigma_{t_{i-1}}\Delta_iW|> v_n\big)^{\frac{K-p}{K}}\notag\\
    &\lesssim (\Delta_nv_n^{2K-Y}+\Delta_n^2v_n^{2K-Y-2}+\Delta_n^{2K-Y/2})^{\frac{p}{K}}\exp\Big(-\frac{K-p}{2K{\sigma^2_{t_{i-1}}}}(\Delta_n^{-1}v_n^2)\Big),\label{prop1_v2}
\end{align}
which vanishes faster than any power of $\Delta_n$ given $\Delta_n^{-1}v_n^2\gg {\Delta_n^{-2s}}$ and hence $V_2 = o_P(\Delta_nv_n^{2p-Y})$. Combining \eqref{prop1_v1} and \eqref{prop1_v2} yields \eqref{prop1_non_integer_levy}.

Now we proceed to establish \eqref{prop1_non_integer_semidiff}. Note that for any $q>1$ and real numbers $a,b$, we have\footnote{{Indeed, for some $\theta\in(0,1)$, $||b|^q-|a|^q|=q ||b|-|a|||(\theta |a|+(1-\theta)|b|)|^{q-1}\leq q||b|-|a||(|a|^{q-1}+|b|^{q-1})$ because clearly $\theta |a|+(1-\theta)|b|\leq \max\{|a|,|b|\}$.}}
\begin{align}\label{prop1_momentdiff_bound}
\big||a|^q-|b|^q\big|\leq q\big||a|-|b|\big|\big(|a|^{q-1}+|b|^{q-1}\big).
\end{align}
We utilize the inequality \eqref{prop1_momentdiff_bound} to bound \eqref{prop1_non_integer_semidiff} by taking $q=2p$, $a = (\Delta_iX)\mathbf{1}_{\{|\Delta_iX|\leq v_n\}}$, and $b = x_i\mathbf{1}_{\{|x_i|\leq v_n\}}$. Also, recalling that $\Delta_iX = x_i+\varepsilon_i$, we have $$\big||a|-|b|\big|\leq|a-b| \leq |\varepsilon_i|\mathbf{1}_{\{|x_i+\varepsilon_i|\leq v_n\}} + |x_i|(\mathbf{1}_{\{|x_i+\varepsilon_i|\leq v_n< |x_i|\}}+\mathbf{1}_{\{|x_i|\leq v_n< |x_i+\varepsilon_i|\}}).$$
{Using the two inequalities above together with $|x+\varepsilon|^{2p-1}\leq 2^{2p-2}(|x|^{2p-1}+|\varepsilon|^{2p-1})$}, we obtain:
\begin{align}
    & \EE_{i-1}\big|(\Delta_iX)\mathbf{1}_{\{|\Delta_iX|\leq v_n\}}|^{2p}-\EE_{i-1}|x_i\mathbf{1}_{\{|x_i|\leq v_n\}}|^{2p} \nonumber\\
    \quad&\lesssim\EE_{i-1}\Big\{\big[|\varepsilon_i\mathbf{1}_{\{|x_i+\varepsilon_i|\leq v_n\}}| +|x_i|(\mathbf{1}_{\{|x_i+\varepsilon_i|\leq v_n< |x_i|\}}+\mathbf{1}_{\{|x_i|\leq v_n< |x_i+\varepsilon_i|\}}
    )\big]\nonumber\\
    &\qquad\qquad\cdot
    \big[{|x_i+\varepsilon_i|}^{2p-1}\mathbf{1}_{\{|x_i+\varepsilon_i| \leq v_n\}} + {|x_i|}^{2p-1} \mathbf{1}_{\{|x_i| \leq v_n\}}\big]\Big\}\nonumber\\
    &\lesssim \EE_{i-1}|\varepsilon_i^{2p}\mathbf{1}_{\{|x_i+\varepsilon_i|\leq v_n\}}| + \EE_{i-1}\big|\varepsilon_ix_i^{2p-1}(\mathbf{1}_{\{|x_i+\varepsilon_i|\leq v_n\}}+ \mathbf{1}_{\{|x_i+\varepsilon_i|\leq v_n,|x_i|\leq v_n\}})\big|\nonumber\\
    &\quad + \EE_{i-1}\big|x_i^{2p}(\mathbf{1}_{\{|x_i+\varepsilon_i|\leq v_n< |x_i|\}}+\mathbf{1}_{\{|x_i|\leq v_n\leq |x_i+\varepsilon_i|\}})(\mathbf{1}_{\{|x_i+\varepsilon_i|\leq v_n\}} + \mathbf{1}_{\{|x_i|\leq v_n\}})\big|\nonumber\\
    &= \EE_{i-1}|\varepsilon_i^{2p}\mathbf{1}_{\{|x_i+\varepsilon_i|\leq v_n\}}| + \EE_{i-1}\big|\varepsilon_ix_i^{2p-1}(\mathbf{1}_{|x_i+\varepsilon_i|\leq v_n<|x_i|}+ \mathbf{1}_{\{|x_i+\varepsilon_i|\leq v_n,|x_i|\leq v_n\}})\big|\nonumber\\
    &\quad + \EE_{i-1}\big|x_i^{2p}(\mathbf{1}_{\{|x_i+\varepsilon_i|\leq v_n< |x_i|\}}+\mathbf{1}_{\{|x_i|\leq v_n\leq |x_i+\varepsilon_i|\}})\big|\nonumber\\
    \label{SmplDcmp0}
    &=:R_{i,1} +R_{i,2}+R_{i,3}.
\end{align}
By \eqref{KME0}, we obtain $R_{i,1} = O_P(\Delta_n^{1+p}) = o_P(\Delta_nv_n^{2p-Y})$ given $\Delta_n^{1/2}\ll v_n$ and $Y>0$. 
Now we analyze $R_{i,3}$. For $\EE_{i-1}|x_i|^{2p}\mathbf{1}_{\{|x_i|\leq v_n\leq |x_i+\varepsilon_i|\}}$ we may directly apply \eqref{prop1_non_integer_levy_drift} and consequently $\EE_{i-1}|x_i|^{2p}\mathbf{1}_{\{|x_i|\leq v_n\leq |x_i+\varepsilon_i|\}}$ is of order ${O_P(\sigma_{t_{i-1}}^{2p}\Delta_n^p) + O_P(\Delta_nv_n^{2p-Y})}$. For the remaining term in $R_{i,3}$,  we consider the following decomposition:
\begin{align*}
    \EE_{i-1}|x_i\mathbf{1}_{\{|x_i+\varepsilon_i|\leq v_n\leq |x_i|\}}|^{2p}&\leq \EE_{i-1}|x_i\mathbf{1}_{\{|x_i+\varepsilon_i|\leq v_n\leq |x_i|\}}\cdot\mathbf{1}_{\{|\varepsilon_i|\leq|v_n|\}}|^{2p}\\
    &\quad +\EE_{i-1}|x_i\mathbf{1}_{\{|x_i+\varepsilon_i|\leq v_n\leq |x_i|\}}\cdot\mathbf{1}_{\{|\varepsilon_i|>|v_n|\}}|^{2p}\\
    & =: \tilde D_{i,1} + \tilde D_{i,2}.
\end{align*}
Note that $|x_i+\varepsilon_i|\leq v_n$ and $|\varepsilon_i|\leq|v_n|$ implies $|x_i| \leq 2v_n$ and, we may again, simply apply \eqref{prop1_non_integer_levy_drift} to obtain 
$\tilde D_{i,1}\lesssim \Delta_nv_n^{2p-Y} + \Delta_n^p$. Regarding $\tilde D_{i,2}$, on
$\{|x_i+\varepsilon_i|\leq v_n,\ |\varepsilon_i|>v_n\}$
we have $|x_i|\leq v_n+|\varepsilon_i|<2|\varepsilon_i|$.
Hence, by \eqref{KME0},
$$
\tilde D_{i,2}
\leq 2^{2p}\EE_{i-1}|\varepsilon_i|^{2p}
=O_P(\Delta_n^{1+p})
=o_P(\Delta_nv_n^{2p-Y}),
$$
since $v_n\gg\sqrt{\Delta_n}$ and $2p-Y>0$.
Together with the preceding bounds, this yields
$R_{i,3}=O_P(\Delta_n^p+\Delta_nv_n^{2p-Y})$.

Lastly, we turn to $R_{i,2}$ in \eqref{SmplDcmp0}. By H\"older's inequality,
\begin{align*}
    R_{i,2}&\leq\big(\EE_{i-1}|\varepsilon_i|^{2p}\big)^{\frac{1}{2p}}\big(\EE_{i-1}|x_i|^{2p}\mathbf{1}_{|x_i+\varepsilon_i|\leq v_n<|x_i|}\big)^{\frac{2p-1}{2p}}\\
    &\quad+\big(\EE_{i-1}|\varepsilon_i|^{2p}\big)^{\frac{1}{2p}}\big(\EE_{i-1}|x_i|^{2p}\mathbf{1}_{\{|x_i+\varepsilon_i|\leq v_n,|x_i|\leq v_n\}}\big)^{\frac{2p-1}{2p}}\\
       &\lesssim
    \Delta_n^{\frac{1+p}{2p}}
    (\Delta_n^p+\Delta_nv_n^{2p-Y})^{\frac{2p-1}{2p}}.
\end{align*}
Here we used \eqref{KME0},
\eqref{prop1_non_integer_levy_drift}, and the preceding
bounds. Since
$\Delta_n^{1+p}=o(\Delta_n^p+\Delta_nv_n^{2p-Y})$,
we conclude that
$R_{i,2}=o_P(\Delta_n^p+\Delta_nv_n^{2p-Y})$.

Combining the results for $R_{i,1}$, $R_{i,2}$, and $R_{i,3}$, we establish \eqref{prop1_non_integer_semidiff} and conclude the proof of \eqref{prop1_non_integer} for $X_t'$. More specifically, we have obtained
\begin{equation}\label{non_integer_bound_wo_finite_jump}
    \EE\big(|\Delta_iX'|^{2p}\mathbf{1}_{\{|\Delta_iX'|\leq  v_n\}}|\mathcal F_{i-1}\big) = O_P(\sigma_{t_{i-1}}^{2{p}}\Delta_n^{p})+O_P(C_{{p},t_{i-1}}\Delta_n v_n^{2{p}-Y}).
\end{equation}

Now we will show \eqref{finite_jump_difference}. Following the idea of Lemma 5 in \cite{BONIECE2024}, for any $p>1$, let \[ D_{2p} := \left| |\Delta_i X|^{2p} \mathbf{1}_{\{|\Delta_i X| \leq v_n\}} - |\Delta_i X'|^{2p} \mathbf{1}_{\{|\Delta_i X'| \leq v_n\}} \right|.\] 
Denote \( V_t = X_t - X_t' \), where recalling $J^0$ as defined in \eqref{e:def_J0},
\[
V_t := \int_0^t \int \delta(s, z) \mathfrak{p}(ds, dz) + \int_0^t\chi_sdJ_s^0 =:Y_t + \int_0^t\chi_sdJ_s^0.
\]

For any fixed \( p \geq 1 \), by decomposing the indicator $\mathbf{1}_{\{|\Delta_i X| \leq v_n\}}$ and $\mathbf{1}_{\{|\Delta_i X'| \leq v_n\}}$, we obtain
\begin{align*}
D_{2p} &= \left| |\Delta_i X|^{2p} - |\Delta_i X'|^{2p} \right| \mathbf{1}_{\{|\Delta_i X| \leq v_n, |\Delta_i X'| \leq v_n\}} \nonumber\\
&\quad + |\Delta_i X|^{2p} \mathbf{1}_{\{|\Delta_i X| \leq v_n, |\Delta_i X'| > v_n\}} \nonumber\\
&\quad + |\Delta_i X'|^{2p} \mathbf{1}_{\{|\Delta_i X| > v_n, |\Delta_i X'| \leq v_n\}} \\
&=: T_1 + T_2 + T_3.
\end{align*}
We start with $T_1$. By Taylor expansion of \( |x' + v|^{2p} \) at \( v = 0 \), we obtain
\[
||x' + v|^{2p} - |x'|^{2p}| \leq K|v| \left( |x'|^{2p-1} + |v|^{2p-1} \right),
\]
which applied to $T_1$ yields
\begin{equation*}
T_1 \leq K \left( |\Delta_i V|^{2p} + |\Delta_i V| |\Delta_i X'|^{2p-1} \right) \mathbf{1}_{\{|\Delta_i X| \leq v_n, |\Delta_i X'| \leq v_n\}}.
\end{equation*}
We first bound $\Delta_i V$. By Corollary 2.1.9 in \cite{jacod2012discretization}, for each \( l\geq 1 \) and some constant $K>0$, we have:
\begin{equation}\label{jacod_2.1.9}
\mathbb{E}_{i-1} \left( \left| \frac{\Delta_iY}{v_n} \right|\wedge 1 \right)^l \leq K\Delta_n v_n^{-r_0}.
\end{equation}
On the other hand, since $J^0$ is compound Poisson, if we let $N^0(ds,dx)$ be its jump measure, for some $K>0$, we have
$$\EE_{i-1}\left( \left| \frac{1}{v_n} \int_{t_{i-1}}^{t_i}\chi_sdJ_s^0 \right| \wedge 1 \right)^p \leq \PP(N^0([t_{i-1},t_i],\R)>0) \leq K \Delta_n.$$
Together with \eqref{jacod_2.1.9}, we obtain
\[
\mathbb{E}_{i-1} \left( |\Delta_i V|^l \mathbf{1}_{\{|\Delta_i V| \leq v_n\}} \right) \leq K\Delta_n v_n^{l - r_0}.
\]

Since \( |x + v| \leq v_n \) and \( |x| \leq v_n \) imply \( |v| \leq 2v_n \), applying the above inequality to $T_1$ we get:
\begin{equation}\label{e:T1bound}
\mathbb{E}_{i-1} T_1 \leq K \left( \mathbb{E}_{i-1}[|\Delta_i V|^{2p} \mathbf{1}_{\{|\Delta_i V| \leq 2v_n\}}] + v_n^{2p-1} \mathbb{E}_{i-1}[|\Delta_i V| \mathbf{1}_{\{|\Delta_i V| \leq 2v_n\}}] \right) \leq K\Delta_n v_n^{2p - r_0}.
\end{equation}

Now we move to \( T_2 \). For some $\alpha\geq 1$ to be later specified, we may decompose $T_2$ as follows:
\begin{equation*}
T_2 = |\Delta_i X|^{2p} \mathbf{1}_{\{|\Delta_i X| \leq v_n, |\Delta_i X'| > v_n + v_n^\alpha\}} + |\Delta_i X|^{2p} \mathbf{1}_{\{|\Delta_i X| \leq v_n < |\Delta_i X'| \leq v_n + v_n^\alpha\}} =: T_2' + T_2''
\end{equation*}

Notice that \( |x' + v| \leq v_n \) and \( |x'| > v_n + v_n^\alpha \) implies \( |v| > v_n^\alpha \). For $T_2'$ we obtain the bound 
\begin{align}
    \mathbb{E}_{i-1} T_2' &\leq v_n^{2p} \mathbb{P}_{i-1}(|\Delta_i V| > v_n^\alpha, |\Delta_i X'| > v_n + v_n^\alpha)\nonumber\\
    &\leq K v_n^{2p} \left( \Delta_nv_n^{-\alpha r_0} \right),\label{e:T2primebound}
\end{align}
which follows from the bound $\mathbb{P}_{i-1}(|\Delta_i V| > u) \leq \mathbb{E}_{i-1} \left( \left| \frac{\Delta_iV}{u}  \right|\wedge 1 \right)^l \leq K\Delta_n u^{-r_0}.$
For \( T_2'' \), notice that  $|x' + v| \leq v_n$  and $|x'| \leq v_n + v_n^\alpha$ implies $|v| \leq 2v_n + v_n^\alpha$. We have:
\begin{align*}
    \mathbb{E}_{i-1} T_2'' &= \mathbb{E}_{i-1} \left[ |\Delta_i X|^{2p} \mathbf{1}_{\{|\Delta_i X| \leq v_n < |\Delta_i X'| \leq v_n + v_n^\alpha\}} \right]\nonumber\\
    &\leq K \mathbb{E}_{i-1} \left[ \big(|\Delta_i X'|^{2p} + |\Delta_i V|^{2p} \big) \mathbf{1}_{\{|\Delta_i X| \leq v_n < |\Delta_i X'| \leq v_n + v_n^\alpha\}}\right]\nonumber\\
    &\leq \mathbb{E}_{i-1} |\Delta_i X'|^{2p}\mathbf{1}_{\{|\Delta_i X'| \leq v_n + v_n^\alpha\}}  + O_P(\Delta_nv_n^{2p - r_0}).
\end{align*}

Since \( \alpha \geq 1 \), by \eqref{non_integer_bound_wo_finite_jump}, we obtain
\[
\mathbb{E}_{i-1} |\Delta_i X'|^{2p}\mathbf{1}_{\{|\Delta_i X'| \leq v_n + v_n^\alpha\}} = O_P(\sigma_{t_{i-1}}^{2p}\Delta_n^p)+O_P(\Delta_n v_n^{2p - Y}),
\]
which combined with \eqref{e:T2primebound} implies $$\EE_{i-1}T_2=O_P(\sigma_{t_{i-1}}^{2p}\Delta_n^p)+ O_P(\Delta_nv_n^{2p-Y}) + O_P(\Delta_nv_n^{2p-\alpha r_0}).$$
Lastly for \( T_3 \), similarly we may write
\[
T_3 \leq |\Delta_i X'|^{2p} \mathbf{1}_{\{|\Delta_i X| > v_n, |\Delta_i X'| \leq v_n - v_n^\alpha\}} + |\Delta_i X'|^{2p} \mathbf{1}_{\{v_n - v_n^\alpha < |\Delta_i X'| \leq v_n\}} = T_3' + T_3''.
\]

Since \( |x' + v| > v_n \) and \( |x'| \leq v_n - v_n^\alpha \) imply \( |v| > v_n^\alpha \), the same arguments as for \( T_2', T_2'' \) apply to \( T_3', T_3'' \). In conclusion, we have the following bound
\begin{align*}
    \EE_{i-1} T_1 + \EE_{i-1} T_2 + \EE_{i-1} T_3\ \lesssim \Delta_n v_n^{2p - r_0} +    \Delta_nv_n^{2p-\alpha r_0}  + \sigma_{t_{i-1}}^{2p}\Delta_n^p+\Delta_n v_n^{2p - Y}.
\end{align*}
To show \eqref{finite_jump_difference}, it suffices to show $$\Delta_n v_n^{2p - r_0} +    \Delta_nv_n^{2p-\alpha r_0} \lesssim \sigma_{t_{i-1}}^{2p}\Delta_n^p+\Delta_n v_n^{2p - Y}.$$
Taking $\alpha=1$, both terms are
$o(\Delta_nv_n^{2p-Y})$ since $r_0<Y$.

\subsection{Technical identities in Theorem \ref{clt}}\label{TecnPrf1}
In view of \eqref{prop1_non_integer} and Lemma 5 in \cite{BONIECE2024}, in the following, we assume that the discontinuous part of $X$ is  comprised only of L\'evy-driven term of infinite variation and may omit the jumps of finite variation. Specifically, we may assume that $X$ takes the form:
$$X_t = \int_0^t{\mu}_sds +  \int_0^t\sigma_s dW_s + \int_0^t\chi_{s}dJ_s.$$ Here and below, following this reduction, we write $X$ for $X'$  and $J$ for $J^\infty$.We also absorb the drift coefficient of $J$, multiplied by $\chi$, into $\mu$, and henceforth use $J$ for the resulting L\'evy process with zero drift coefficient relative to the truncation $x\mathbf{1}_{\{|x|\leq1\}}$. In particular, the driving L\'evy process has moments of every order, and the discretization estimate \eqref{KME0} applies to this reduced process.

\begin{proof}[Proof of \eqref{initial_technical}]
As in \cite{BONIECE2024}, the key idea is to discretize the integrals involved in ${\Delta_iX}$. Indeed, {recalling the notation in \eqref{discretizing},} 
we  treat each term in the expansion $(\Delta_iX)^2=(x_i+\varepsilon_i)^2 = x_i^2 + 2x_i\varepsilon_i + \varepsilon_i^2$ separately. For the third term $\varepsilon_i^2$, by Cauchy-Schwarz, we obtain
\begin{align}\label{discretizing_third}
    &m_n \Big(\sum_{i=1}^nK_{m_n\Delta_n}(t_{i-1}-\tau)\EE_{i-1}\varepsilon_i^2\mathbf{1}_{\{|\Delta_iX|\leq  v_n\}}|\Delta_iM|\big)\Big)^2 \nonumber\\
    &\quad\leq m_n \Big(\sum_{i=1}^nK_{m_n\Delta_n}(t_{i-1}-\tau)\EE_{i-1}(\varepsilon_i^4)^{1/2}\EE_{i-1}(\Delta_iM^2)^{1/2}\Big)^2\nonumber\\
    &\quad\leq m_n\Big(\sum_{i=1}^nK_{m_n\Delta_n}(t_{i-1}-\tau)\EE_{i-1}(\varepsilon_i^4)\Big)\Big(\sum_{i=1}^nK_{m_n\Delta_n}(t_{i-1}-\tau)\EE_{i-1}(\Delta_iM^2)\Big)\nonumber\\
    &\quad\leq m_n \Delta_n ^{-1}O_P(\Delta_n^3)O_P\big(\frac{1}{m_n\Delta_n}\big) = O_P(\Delta_n) \rightarrow 0,
\end{align}
where the second inequality follows from Cauchy-Schwarz concerning the summation and the last inequality is due to \eqref{KME0} and
\begin{align}\label{M2_kernel_average}
    &\sum_{i=1}^n K_{m_n\Delta_n}(t_{i-1}-\tau)
    \EE_{i-1}[(\Delta_iM)^2] \leq
    \frac{\|K\|_\infty}{m_n\Delta_n}
    \sum_{i=1}^n\EE_{i-1}[(\Delta_iM)^2]
    =O_P\left(\frac{1}{m_n\Delta_n}\right).
\end{align}
Indeed, since $M$ is square integrable,
$
\EE\left[\sum_{i=1}^n
\EE_{i-1}[(\Delta_iM)^2]\right]
=
\sum_{i=1}^n\EE[(\Delta_iM)^2]
=
\EE[(M_1-M_0)^2]<\infty,
$
giving the last bound in \eqref{M2_kernel_average}.

For the term $2x_i\varepsilon_i$, when $|\varepsilon_i|>v_n$, since $\Delta_n^{1/2}/v_n\ll \Delta_n^{1/2 - \beta} \ll 1$, we may use Markov's inequality to conclude that
\[
P(|\varepsilon_i| > v_n) \leq K \Delta_n^{1 + (1/2 - \beta)r},
\]
for all $r\geq2$ according to {\eqref{KME0}}. By taking large enough $r$ and applying H\"older's inequality, for some $q_1,q_2,q_3>1$ and $\sum_{i=1}^31/q_i = 1/2$, as $\EE_{i-1} [ |x_i|^{q_1}]$ and $\EE_{i-1} [ |\varepsilon_i|^{q_2}]$ is finite, we have that
\begin{align}\label{holder}
    &\sqrt{m_n} \sum_{i=1}^{n}K_{m_n\Delta_n}(t_{i-1}-\tau) \EE_{i-1} \left[ |x_i \varepsilon_i| \mathbf{1}_{\{|\Delta_i X| \leq v_n\}} \mathbf{1}_{\{|\varepsilon_i| > v_n\}} |\Delta_i M| \right] \nonumber\\
    &\quad\leq \sqrt{m_n} \sum_{i=1}^{n}K_{m_n\Delta_n}(t_{i-1}-\tau) \EE_{i-1} [ |x_i|^{q_1}]^{1/q_1} \EE_{i-1} [ |\varepsilon_i|^{q_2}]^{1/q_2} \EE_{i-1} (\Delta_i M^2)^{1/2} \PP_{i-1}(|\varepsilon_i| > v_n)^{1/{q_3}} \nonumber\\
    &\quad\leq \sqrt{m_n} \sum_{i=1}^{n}K_{m_n\Delta_n}(t_{i-1}-\tau) \EE_{i-1} [ |x_i|^{q_1}]^{1/q_1} \EE_{i-1} [ |\varepsilon_i|^{q_2}]^{1/q_2} \EE_{i-1} (\Delta_i M^2)^{1/2}\Delta_n^{1/q_{3} + (1/2 - \beta)r/q_3}\nonumber\\
    &\quad = o_P(1).
\end{align}
When \( |\varepsilon_i| \leq v_n \), note \( |\Delta_i X| = |x_i + \varepsilon_i| \leq v_n \) implies that \( |x_i| \leq 2 v_n \). Then by applying Cauchy-Schwarz twice, we obtain:
\begin{align*}
    &m_n \left| \sum_{i=1}^{n}K_{m_n\Delta_n}(t_{i-1}-\tau) \EE_{i-1} \left[ |\varepsilon_i x_i| \mathbf{1}_{\{|x_i| \leq 2v_n\}} \mathbf{1}_{\{|\varepsilon_i| \leq v_n\}} |\Delta_i M| \right] \right|^2 \nonumber\\
    &\leq m_n \sum_{i=1}^{n}K_{m_n\Delta_n}(t_{i-1}-\tau)  \EE_{i-1} \left[ \varepsilon_i^2 x_i^2 \mathbf{1}_{\{|x_i| \leq 2v_n\}} \right] \sum_{i=1}^{n}K_{m_n\Delta_n}(t_{i-1}-\tau)  \EE_{i-1} \left[ (\Delta_i M)^2 \right] \nonumber\\
    &\leq m_n \sum_{i=1}^{n}K_{m_n\Delta_n}(t_{i-1}-\tau)  O_P(\Delta_n^{3/2}) O_P(\Delta_n) \sum_{i=1}^{n}K_{m_n\Delta_n}(t_{i-1}-\tau) \EE_{i-1} \left[ (\Delta_i M)^2 \right] \nonumber\\
     &=m_nO_P(\Delta_n^{3/2})O_P\big(\frac{1}{m_n\Delta_n}\big) = O_P(\Delta_n^{1/2})\xrightarrow{P} 0,
\end{align*}
{where in the second inequality we used Proposition \ref{prop} and \eqref{KME0}, and in the last line, we used \eqref{M2_kernel_average}}. All that remains is the term $x_i^2$. Next, we show that 
\begin{align}\label{cross_technical}
    \sqrt{m_n}\sum_{i=1}^n
    K_{m_n\Delta_n}(t_{i-1}-\tau)
    \left|
    \EE_{i-1}\left[
    x_i^2\mathbf{1}_{\{|x_i+\varepsilon_i|\leq v_n\}}
    \Delta_iM
    \right]
    \right|
    \xrightarrow{P}0.
\end{align}
Firstly we start with replacing $x_i$ with ${\mu_{t_{i-1}}}\Delta_n$ in \eqref{cross_technical} and arbitrary $M$. By the boundness of ${\mu_t}$, $K$ and Cauchy-Schwarz, it follows that
\begin{align*}
     &\sqrt{m_n} \sum_{i=1}^{n}K_{m_n\Delta_n}(t_{i-1}-\tau) \left|\EE_{i-1} \left[ {\mu}_{t_{i-1}}^2\Delta_n^2 \mathbf{1}_{\{|x_i + \varepsilon_i| \leq v_n\}} (\Delta_i M) \right]\right|\\
     \quad&\lesssim \frac{\sqrt{m_n}\Delta_n^2}{m_n\Delta_n} \sum_{i=1}^{n}\EE_{i-1} \left[(\Delta_i M)^2 \right]^{1/2}\\
     \quad&\leq \frac{\sqrt{m_n}\Delta_n^2}{m_n\Delta_n^{3/2}} \Big(\sum_{i=1}^{n}\EE_{i-1}\left[(\Delta_i M)^2 \right]\Big)^{1/2} \xrightarrow{P} 0,
\end{align*}
where we used $\sum_{i=1}^{n}\EE_{i-1} \left[(\Delta_i M)^2 \right]^{1/2} \leq \Delta_n^{-1/2}\Big(\sum_{i=1}^{n}\EE_{i-1} \left[(\Delta_i M)^2 \right]\Big)^{1/2}$. Besides, by Cauchy-Schwarz twice, we also have
\begin{align*}
    &\sqrt{m_n}\sum_{i=1}^n
    K_{m_n\Delta_n}(t_{i-1}-\tau)
    \left|
    \EE_{i-1}\left[\mu_{t_{i-1}}\Delta_n\hat x_i
    \mathbf{1}_{\{|x_i+\varepsilon_i|\leq v_n\}} \Delta_iM
    \right]
    \right|\\
     \quad&\lesssim  \frac{\sqrt{m_n}\Delta_n}{m_n\Delta_n} \Big(\sum_{i=1}^{n}\EE_{i-1} \left[\hat x_i^2 \mathbf{1}_{\{|x_i + \varepsilon_i| \leq v_n\}}\right]\Big)^{1/2}\Big(\sum_{i=1}^{n}\EE_{i-1} \left[(\Delta_i M)^2 \right]\Big)^{1/2}\xrightarrow{P} 0,
\end{align*}
Consequently, it is sufficient to prove \eqref{cross_technical} with $x_i = \sigma_{t_{i-1}}\Delta_iW + \chi_{t_{i-1}}\Delta_iJ$.

Now we prove \eqref{cross_technical} when \( M = W \) and we start with the case when $x_i$ is replaced with $\sigma_{t_{i-1}}\Delta_iW$. Similar to last inequality, by combining Cauchy-Schwarz and {Proposition \ref{prop}}, we can write
\begin{align}\label{W_technical1}
    &\sqrt{m_n}  \sum_{i=1}^{n}K_{m_n\Delta_n}(t_{i-1}-\tau) \sigma^2_{t_{i-1}}\left|  \EE_{i-1} \left[ (\Delta_i W)^2 \mathbf{1}_{\{|x_i + \varepsilon_i| \leq v_n\}} \Delta_i M \right] \right| \nonumber\\
    &\quad= \sqrt{m_n} \sum_{i=1}^{n}K_{m_n\Delta_n}(t_{i-1}-\tau) \sigma^2_{t_{i-1}}\left|  \EE_{i-1} \left[ (\Delta_i W)^3 \mathbf{1}_{\{|x_i + \varepsilon_i| > v_n\}} \right] \right|\nonumber\\
    &\quad\leq \sqrt{m_n} \sum_{i=1}^{n}K_{m_n\Delta_n}(t_{i-1}-\tau) \sigma^2_{t_{i-1}} \EE_{i-1} \left[ |\Delta_i W|^3 |x_i + \varepsilon_i| \right]/ v_n \nonumber\\
    & \quad \leq C \sqrt{m_n} \sum_{i=1}^{n}K_{m_n\Delta_n}(t_{i-1}-\tau) \EE_{i-1} \left[ (\Delta_i W)^6 \right]^{1/2} \EE_{i-1} \left[ |x_i + \varepsilon_i|^2 \right]^{1/2}/ v_n\nonumber\\
& \quad\leq C\sqrt{m_n} n \Delta_n^{3/2} \Delta_n^{1/2}/v_n \rightarrow 0,
\end{align}
where the last limit follows from $m_n\Delta_n\rightarrow 0$ and $\Delta_n ^{1/2} \ll v_n$.

When \( x_i \) and $M$ are replaced with \( \chi_{t_{i-1}}\Delta_i J \) and $M=W$ in \eqref{cross_technical}, respectively, we need to show that
\begin{align}\label{jump_technical}
    \sqrt{m_n} \sum_{i=1}^{n}K_{m_n\Delta_n}(t_{i-1}-\tau) \EE_{i-1} \left[ (\chi_{t_{i-1}}\Delta_i J)^2 \mathbf{1}_{\{|x_i + \varepsilon_i| \leq v_n\}} |\Delta_i W| \right] \xrightarrow{P} 0.
\end{align}
Note that %
\( P(|\varepsilon_i| > v_n) \) and \( P(|\Delta_i W| > v_n) \) decay faster than any power of \( \Delta_n \) since \( \Delta_n^{1/2}/v_n \ll \Delta_n^{1/2 - \beta} \). Then, by similar techniques as in \eqref{holder}, 
we obtain
\begin{align}\label{jump_technical_1}
    \sqrt{m_n} \sum_{i=1}^{n}K_{m_n\Delta_n}(t_{i-1}-\tau) \EE_{i-1} \left[ (\chi_{t_{i-1}}\Delta_i J)^2 \mathbf{1}_{\{|x_i + \varepsilon_i| \leq v_n\}} |\Delta_i W| \left( \mathbf{1}_{\{|\Delta_i W| > v_n\}} + \mathbf{1}_{\{|\varepsilon_i| > v_n\}} \right) \right] = o_P(1).
\end{align}
So we only focus on the case when \( |\Delta_i W| \leq v_n \) and \( |\varepsilon_i| \leq v_n \) with $M=W$. Note that \( |x_i + \varepsilon_i| \leq v_n \) implies that \( |\chi_{t_{i-1}}\Delta_i J| \leq Cv_n \), for some positive constant \( C \). In that case, for \eqref{jump_technical}, it suffices to show
\[
\sqrt{m_n} \sum_{i=1}^{n}K_{m_n\Delta_n}(t_{i-1}-\tau) \EE_{i-1} \left[ (\chi_{t_{i-1}}\Delta_i J)^2 \mathbf{1}_{\{|\chi_{t_{i-1}}\Delta_i J| \leq C v_n\}} |\Delta_i W| \right] \xrightarrow{P} 0.
\]
By H\"older's inequality and Lemma 17 in \cite{BONIECE2024}, for any \( p, q > 1 \) such that \( 1/p + 1/q = 1 \),
\begin{align*}
         &\sqrt{m_n} \sum_{i=1}^{n}K_{m_n\Delta_n}(t_{i-1}-\tau) \EE_{i-1} \left[ (\chi_{t_{i-1}}\Delta_i J)^2 \mathbf{1}_{\{|\chi_{t_{i-1}}\Delta_i J| \leq C v_n\}} |\Delta_i W| \right]\nonumber\\
    &\quad\leq \sqrt{m_n} \sum_{i=1}^{n}K_{m_n\Delta_n}(t_{i-1}-\tau) \EE_{i-1} \left[ |\chi_{t_{i-1}}\Delta_i J|^{2p} \mathbf{1}_{\{|\chi_{t_{i-1}}\Delta_i J| \leq C v_n\}} \right]^{1/p} \EE_{i-1}\left[ |\Delta_i W|^q \right]^{1/q}\nonumber\\
    &\quad\lesssim  \sqrt{m_n} \sum_{i=1}^{n}K_{m_n\Delta_n}(t_{i-1}-\tau) {(\Delta_n v_n^{2p-Y})^{1/p}}\Delta_n^{1/2}
\lesssim \sqrt{m_n} n \Delta_n^{1/p + 1/2} v_n^{2 - Y/p}.
\end{align*}
So, making \( p \) close to 1 (and \( q \) large) and combining with \eqref{jump_technical_1}, we get the convergence to 0 for \eqref{jump_technical} as $m_n\Delta_n \rightarrow 0$ and $\Delta_n^{1/2}\ll v_n$. Combining \eqref{W_technical1} and \eqref{jump_technical}, we establish \eqref{cross_technical} when \( M = W \).

For $M=B$, the stronger estimate \eqref{cov_summand}
(proved below), together with $m_n\Delta_n\to0$, gives
$$
\sqrt{m_n}\sum_{i=1}^n
K_{m_n\Delta_n}(t_{i-1}-\tau)
\left|\EE_{i-1}\left[
(\Delta_iX)^2\mathbf{1}_{\{|\Delta_iX|\leq v_n\}}
\Delta_iB
\right]\right|
=\sqrt{m_n\Delta_n}\,o_P(1)=o_P(1).
$$

Lastly, we want to show \eqref{cross_technical} for any bounded martingale $M$ which is orthogonal to $W$. We first consider the case when $x_i$ is replaced with $ \chi_{t_{i-1}}\Delta_iJ$. As before, we can deal with the cases $|\Delta_iW|>v_n$ or $|\varepsilon_i|>v_n$, since \( P(|\varepsilon_i| > v_n) \) and \( P(|\Delta_i W| > v_n) \) decays sufficiently fast. %
So, we assume that \( |\Delta_i W| \leq {v_n} \) and \( |\varepsilon_i| \leq {v_n} \), {so that, together with \( |x_i + \varepsilon_i| \leq v_n \), we have \( |\chi_{t_{i-1}}\Delta_i J| \leq Cv_n \). Then,} we proceed as follows:
\begin{align*}
    &\sqrt{m_n} \sum_{i=1}^{n}K_{m_n\Delta_n}(t_{i-1}-\tau) \EE_{i-1} \left[ |\chi_{t_{i-1}}\Delta_i J|^2 \mathbf{1}_{\{|\chi_{t_{i-1}}\Delta_i J| \leq Cv_n\}} |\Delta_i M| \right]\nonumber\\
&\quad\leq \sqrt{m_n} \sum_{i=1}^{n}K_{m_n\Delta_n}(t_{i-1}-\tau) \EE_{i-1} \left[ |\chi_{t_{i-1}}\Delta_i J|^4 \mathbf{1}_{\{|\chi_{t_{i-1}}\Delta_i J| \leq Cv_n\}} \right]^{1/2} \EE_{i-1} \left[ |\Delta_i M|^2 \right]^{1/2} \nonumber\\
 &\quad \leq\sqrt{m_n} \left( \sum_{i=1}^{n}K_{m_n\Delta_n}(t_{i-1}-\tau) \EE_{i-1}\left[ |\chi_{t_{i-1}}\Delta_i J|^4 \mathbf{1}_{\{|\chi_{t_{i-1}}\Delta_i J| \leq Cv_n\}} \right]  \sum_{i=1}^{n}K_{m_n\Delta_n}(t_{i-1}-\tau) \EE_{i-1} \left[ |\Delta_i M|^2 \right] \right)^{1/2} \nonumber\\
& \quad\lesssim\sqrt{m_n} (n\Delta_nv_n^{4 - Y})^{1/2} \Big(\frac{1}{m_n^{1/2}\Delta_n^{1/2}}\Big) = (\Delta_n^{-1} v_n^{4 - Y})^{1/2} \ll 1,
\end{align*}
which follows from \( v_n \ll \Delta_n^{1/(4 - Y)} \). For the cross term $(\sigma_{t_{i-1}}\Delta_iW)(\chi_{t_{i-1}}\Delta_iJ)$, when expanding \( x_i^2 \), and again assuming \( |\chi_{t_{i-1}}\Delta_i J| \leq Cv_n \), we have, for any $p>1$:
\begin{align*}
    &\sqrt{m_n} \sum_{i=1}^{n}K_{m_n\Delta_n}(t_{i-1}-\tau) \EE_{i-1} \left[ |\chi_{t_{i-1}}\Delta_i J||\Delta_i W| \mathbf{1}_{\{|\chi_{t_{i-1}}\Delta_i J| \leq Cv_n\}} |\Delta_i M| \right]\nonumber\\
&\quad\leq \sqrt{m_n} \sum_{i=1}^{n}K_{m_n\Delta_n}(t_{i-1}-\tau) \EE_{i-1} \left[ |\chi_{t_{i-1}}\Delta_i J|^{2p} \mathbf{1}_{\{|\chi_{t_{i-1}}\Delta_i J| \leq Cv_n\}} \right]^{\frac{1}{2p}} \EE_{i-1} \left[ |\Delta_i W|^{\frac{2p}{p-1}} \right]^{\frac{p-1}{2p}} \EE_{i-1} \left[ |\Delta_i M|^{2} \right]^{\frac{1}{2}}\nonumber\\
&\quad\leq \sqrt{m_n}\Delta_n^{1/2}\left(\sum_{i=1}^{n}K_{m_n\Delta_n}(t_{i-1}-\tau) \big(\EE_{i-1} \left[ |\chi_{t_{i-1}}\Delta_i J|^{2p} \mathbf{1}_{\{|\chi_{t_{i-1}}\Delta_i J| \leq Cv_n\}} \right]\big)^{1/p} \right)^{1/2} \nonumber\\
&\qquad\times\left(\sum_{i=1}^{n}K_{m_n\Delta_n}(t_{i-1}-\tau) \EE_{i-1} \left[ |\Delta_i M|^2 \right] \right)^{1/2}\nonumber\\
&\quad\lesssim m_n^{1/2}\Delta_n^{1/2}(n\Delta_n^{1/p}v_n^{2 - Y/p})^{1/2} \Big(\frac{1}{m_n^{1/2}\Delta_n^{1/2}}\Big).
\end{align*}
By taking $p$ close to 1, the above vanishes given $v_n \rightarrow 0$. Finally, we consider the case when \( x_i \) is replaced with \( \sigma_{t_{i-1}} \Delta_i W \). Note that
\[
(\Delta_i W)^2 = 2 \int_{t_{i-1}}^{t_i} (W_s - W_{t_{i-1}}) dW_s + \Delta_n,
\]
and since \( M \) is orthogonal to \( W \), 
\[
\EE_{i-1} \left[ \left( (\Delta_i W)^2 - \Delta_n \right) \Delta_i M \right] = 0,
\]
which implies that
$\EE_{i-1} \left[ (\Delta_i W)^2 \Delta_i M \right] = 0$. 
Thus, for any \( p, q > 0 \) such that \( 1/p + 1/q = 1/2 \), by orthogonality between $M$ and $W$, we have
\begin{align}
    &\sqrt{m_n}  \sum_{i=1}^{n}K_{m_n\Delta_n}(t_{i-1}-\tau) \sigma^2_{t_{i-1}} \left|\EE_{i-1} \left[ (\Delta_i W)^2 \mathbf{1}_{\{|x_i + \varepsilon_i| \leq v_n\}} \Delta_i M \right] \right|\nonumber\\
 &\quad =\sqrt{m_n} \sum_{i=1}^{n}K_{m_n\Delta_n}(t_{i-1}-\tau) \sigma^2_{t_{i-1}} \left|\EE_{i-1} \left[ (\Delta_i W)^2 \mathbf{1}_{\{|x_i + \varepsilon_i|>v_n\}} \Delta_i M \right] \right|\nonumber\\
&\quad \leq  \sqrt{m_n} v_n^{-1} \sum_{i=1}^{n}K_{m_n\Delta_n}(t_{i-1}-\tau) \EE_{i-1} \left[ |\Delta_i W|^2 |x_i + \varepsilon_i| |\Delta_i M| \right]\nonumber\\
&\quad \leq  \sqrt{m_n} v_n^{-1} \sum_{i=1}^{n}K_{m_n\Delta_n}(t_{i-1}-\tau) \EE_{i-1} \left[ |\Delta_i W|^{2p} \right]^{1/p} \EE_{i-1} \left[ |x_i + \varepsilon_i|^q \right]^{1/q} \EE_{i-1} \left[ |\Delta_i M|^2 \right]^{1/2}\nonumber\\
&\quad\leq  \sqrt{m_n} v_n^{-1} \left( \sum_{i=1}^{n}K_{m_n\Delta_n}(t_{i-1}-\tau) \EE_{i-1} \left[ |\Delta_i W|^{2p} \right]^{2/p} \EE_{i-1} \left[ |x_i + \varepsilon_i|^q \right]^{2/q} \right)^{1/2}\nonumber\\
&\qquad\times \left( \sum_{i=1}^{n}K_{m_n\Delta_n}(t_{i-1}-\tau) \EE_{i-1} \left[ |\Delta_i M|^2 \right] \right)^{1/2}\nonumber\\
&\quad \lesssim  \sqrt{m_n} v_n^{-1} \left( n\Delta_n^2 \Delta_n^{2/q} \right)^{1/2} \times \Big(\frac{1}{m_n^{1/2}\Delta_n^{1/2}}\Big)\\
&\quad\leq (m_n\Delta_n)^{1/2}  \Delta_n^{1/q}v_n^{-1} \times \Big(\frac{1}{m_n^{1/2}\Delta_n^{1/2}}\Big)=\Delta_n^{1/q}v_n^{-1}, \nonumber
\end{align}
which vanishes by taking \( q \) close enough to 2 and \( p \) large enough given that $v_n \gg \Delta_n^{\beta}$ with $\beta<1/2$. 
\end{proof}

\begin{proof}[Proof of \eqref{cov_summand}]
Recall that $(\Delta_iX)^2 = x_i^2 + 2x_i\varepsilon_i + \varepsilon_i^2.$ We analyze the conditional expectation 
$\EE_{i-1}\big((x^2_i+2x_i\varepsilon_i+\varepsilon_i^2)\mathbf{1}_{\{|\Delta_iX|\leq v_n\}}\Delta_iB\big)$ term by term. For the third term, by Cauchy-Schwarz {and \eqref{KME0},} we have
\begin{align*}
   \sup_i {|\EE_{i-1}\big(\varepsilon_i^2\mathbf{1}_{\{|\Delta_iX|\leq v_n\}}\Delta_iB\big)|}\leq \sup_i \big(\EE_{i-1}\varepsilon_i^4 \E_{i-1}\Delta_iB^2\big)^{1/2} = O_P(\Delta_n^{3/2})(\Delta_n^{1/2}) = o_P(\Delta_n^{3/2}).
\end{align*}
For the second term, when $|\varepsilon_i|>v_n$ and
$|\Delta_iX|\leq v_n$, we have
$|x_i|\leq|\Delta_iX|+|\varepsilon_i|\leq2|\varepsilon_i|$.
Thus, by Cauchy--Schwarz and \eqref{KME0},
\begin{align*}
    &\sup_i\EE_{i-1}\left[
    |\varepsilon_i x_i|
    \mathbf{1}_{\{|\Delta_iX|\leq v_n,\,
                   |\varepsilon_i|>v_n\}}
    |\Delta_iB|
    \right]\\
    &\quad\leq
    2\sup_i
    \left(\EE_{i-1}[\varepsilon_i^4]\right)^{1/2}
    \left(\EE_{i-1}[(\Delta_iB)^2]\right)^{1/2}
    =O_P(\Delta_n^2)
    =o_P(\Delta_n^{3/2}).
\end{align*}
When $|\varepsilon_i|\leq v_n$ and $|\Delta_iX|\leq v_n$,
we have $|x_i|\leq2v_n$. Thus, by Cauchy--Schwarz
and \eqref{KME0},
\begin{align*}
    &\sup_i\EE_{i-1}\left[
    |\varepsilon_i x_i|
    \mathbf{1}_{\{|\Delta_iX|\leq v_n,\,
                   |\varepsilon_i|\leq v_n\}}
    |\Delta_iB|
    \right]\\
    &\quad\leq
    2v_n\sup_i
    \left(\EE_{i-1}[\varepsilon_i^2]\right)^{1/2}
    \left(\EE_{i-1}[(\Delta_iB)^2]\right)^{1/2}\\
    &\quad\lesssim v_n\Delta_n^{3/2}
    =o(\Delta_n^{3/2}).
\end{align*}
The preceding estimates and \eqref{kernel_convergence}
give
\begin{align*}
    &\frac{1}{\sqrt{\Delta_n}}
    \sum_{i=1}^n K_{m_n\Delta_n}(t_{i-1}-\tau)
    \left|\EE_{i-1}\left[
    (2x_i\varepsilon_i+\varepsilon_i^2)
    \mathbf{1}_{\{|\Delta_iX|\leq v_n\}}\Delta_iB
    \right]\right|\\
    &\quad\leq
    \frac{\Delta_n\sum_{i=1}^n
    K_{m_n\Delta_n}(t_{i-1}-\tau)}{\Delta_n^{3/2}}
    \sup_i\left|\EE_{i-1}\left[
    (2x_i\varepsilon_i+\varepsilon_i^2)
    \mathbf{1}_{\{|\Delta_iX|\leq v_n\}}\Delta_iB
    \right]\right|
    \xrightarrow{P}0.
\end{align*}
It therefore remains to show that
\begin{align}\label{covariance_x2}
    &\frac{1}{\sqrt{\Delta_n}}
    \sum_{i=1}^n K_{m_n\Delta_n}(t_{i-1}-\tau)
    \left|\EE_{i-1}\left[
    x_i^2\mathbf{1}_{\{|x_i+\varepsilon_i|\leq v_n\}}
    \Delta_iB
    \right]\right|
    \xrightarrow{P}0.
\end{align}
We start by replacing $x_i^2$ with
$(\sigma_{t_{i-1}}\Delta_iW)^2$. We consider the decomposition
\begin{align*}
    &\EE_{i-1}\left[
    \sigma_{t_{i-1}}^2(\Delta_iW)^2
    \mathbf{1}_{\{|x_i+\varepsilon_i|\leq v_n\}}\Delta_iB
    \right]\\
    &\quad=
    \EE_{i-1}\left[
    \sigma_{t_{i-1}}^2(\Delta_iW)^2\Delta_iB
    \right]
    -
    \EE_{i-1}\left[
    \sigma_{t_{i-1}}^2(\Delta_iW)^2
    \mathbf{1}_{\{|x_i+\varepsilon_i|>v_n\}}\Delta_iB
    \right]\\
    &\quad=:T_1-T_2.
\end{align*}
For $T_1$, recall that
$d\langle W,B\rangle_s=\rho_sds$.
By It\^o's formula and
$\EE_{i-1}(W_s-W_{t_{i-1}})=0$,
$$
\EE_{i-1}\left[(\Delta_iW)^2\Delta_iB\right]
=
2\int_{t_{i-1}}^{t_i}
\EE_{i-1}\left[
(W_s-W_{t_{i-1}})
(\rho_s-\rho_{t_{i-1}})
\right]ds.
$$
Cauchy--Schwarz, first for the integral and then for
the sum, gives
\begin{align*}
    &\frac{1}{\sqrt{\Delta_n}}
    \sum_{i=1}^n K_{b_n}(t_{i-1}-\tau)
    \left|\EE_{i-1}\left[
    (\Delta_iW)^2\Delta_iB
    \right]\right|\\
    &\quad\leq
    \left(
    2\Delta_n\sum_{i=1}^nK_{b_n}(t_{i-1}-\tau)
    \right)^{1/2}
    \left(
    \sum_{i=1}^nK_{b_n}(t_{i-1}-\tau)
    \int_{t_{i-1}}^{t_i}
    \EE_{i-1}\left[
    |\rho_s-\rho_{t_{i-1}}|^2
    \right]ds
    \right)^{1/2},
\end{align*}
where $b_n=m_n\Delta_n$.

To bound the second factor, fix
$0<\delta<\frac12\min(\tau,1-\tau)$.
For all sufficiently large $n$, boundedness of $\rho$
and $K$ yields
\begin{align*}
    &\EE\left[
    \sum_{i=1}^nK_{b_n}(t_{i-1}-\tau)
    \int_{t_{i-1}}^{t_i}
    \EE_{i-1}\left[
    |\rho_s-\rho_{t_{i-1}}|^2
    \right]ds
    \right]\\
    &\quad\leq C\EE\left[
    \sup_{\tau-2\delta\leq s<\tau}
    |\rho_s-\rho_{\tau-}|^2
    +
    \sup_{\tau\leq s\leq\tau+2\delta}
    |\rho_s-\rho_\tau|^2
    \right]
    +\frac{C}{m_n}+\frac{Cb_n}{\delta^2}.
\end{align*}
Here, the first term bounds intervals within
$\delta$ of $\tau$ that do not cross $\tau$. The interval
crossing $\tau$, if any, contributes at most
$C\Delta_n/b_n=C/m_n$. For the remaining intervals,
we used $K_{b_n}(x)\leq Cb_n/x^2$, which follows from
Assumption \ref{asu_kernel}.

Letting $n\to\infty$ and then $\delta\downarrow0$, the
right-hand side tends to zero since $\rho$ is bounded
and c\'adl\'ag. Thus, by Markov's inequality,
\eqref{kernel_convergence}, and boundedness of $\sigma$,
$$
\frac{1}{\sqrt{\Delta_n}}
\sum_{i=1}^nK_{b_n}(t_{i-1}-\tau)\sigma_{t_{i-1}}^2
\left|\EE_{i-1}\left[
(\Delta_iW)^2\Delta_iB
\right]\right|
\xrightarrow{P}0.
$$

For $T_2$, H\"older's inequality gives
\begin{align*}
    |T_2|
    &\leq C
    \left(\EE_{i-1}[(\Delta_iW)^8]\right)^{1/4}
    \left(\EE_{i-1}[(\Delta_iB)^2]\right)^{1/2}
    \mathbb \PP_{i-1}(|x_i+\varepsilon_i|>v_n)^{1/4}\\
    &\lesssim
    \Delta_n^{3/2}
    \mathbb P_{i-1}(|x_i+\varepsilon_i|>v_n)^{1/4}
    =o_P(\Delta_n^{3/2}),
\end{align*}
uniformly in $i$, by the uniform conditional tail bounds
used above. By \eqref{kernel_convergence}, this also gives
an $o_P(1)$ contribution to the weighted sum in
\eqref{covariance_x2}.

Now replace $x_i^2$ with $(\chi_{t_{i-1}}\Delta_iJ)^2$ in \eqref{covariance_x2}. We show
\begin{align*}
     \sup_i\big|\EE_{i-1} \left[ (\chi_{t_{i-1}}\Delta_iJ)^2\mathbf{1}_{\{|x_i+\varepsilon| \leq v_n\}} \Delta_i B \right]\big| = o_P(\Delta_n^{3/2}).
\end{align*}
Note that when \( |\Delta_i W| > v_n \) or \( |\varepsilon_i| > v_n \), by H\"older's inequality, and that \( P(|\varepsilon_i| > v_n) \) and \( P(|\Delta_i W| > v_n) \) decay faster than any power of \( \Delta_n \) since \( \Delta_n^{1/2}/v_n \ll \Delta_n^{1/2 - \beta} \), by similar techniques as in \eqref{jump_technical_1}, we  obtain
\begin{align*}
    \sup_i\EE_{i-1} \left[ (\chi_{t_{i-1}}\Delta_i J)^2 \mathbf{1}_{\{|x_i + \varepsilon_i| \leq v_n\}} |\Delta_i B| \left( \mathbf{1}_{\{|\Delta_i W| > v_n\}} + \mathbf{1}_{\{|\varepsilon_i| > v_n\}} \right) \right] = o_P(\Delta_n^{3/2}).
\end{align*}
And further notice that when $|\Delta_iW|\leq v_n$ and
$|\varepsilon_i|\leq v_n$, $|x_i+\varepsilon_i|\leq v_n$
implies $|\chi_{t_{i-1}}\Delta_iJ|\leq Cv_n$ for some
positive constant $C$. Then, by H\"older's inequality
and the truncated-moment bound of Lemma 17 in
\cite{BONIECE2024}, for any $p,q>1$ such that
$1/p+1/q=1$, we obtain
\begin{align*}
    &\sup_i\EE_{i-1}\left[
    |\chi_{t_{i-1}}\Delta_iJ|^2
    \mathbf{1}_{\{|\chi_{t_{i-1}}\Delta_iJ|\leq Cv_n\}}
    |\Delta_iB|\right]\\
    &\quad\lesssim \sup_i
    \EE_{i-1}\left[
    |\chi_{t_{i-1}}\Delta_iJ|^{2p}
    \mathbf{1}_{\{|\chi_{t_{i-1}}\Delta_iJ|\leq Cv_n\}}
    \right]^{1/p}
    \EE_{i-1}\left[|\Delta_iB|^q\right]^{1/q}\\
    &\quad\lesssim
    (\Delta_nv_n^{2p-Y})^{1/p}\Delta_n^{1/2}.
\end{align*}
Choose a fixed $p>1$ sufficiently close to $1$ that
$1/p-1+\beta'(2-Y/p)>0$, which is possible since $\beta'(2-Y)>0$. Then, using $v_n\ll\Delta_n^{\beta'}$,
$$
\Delta_n^{1/p-1}v_n^{2-Y/p}
=o\left(\Delta_n^{1/p-1+\beta'(2-Y/p)}\right)
\longrightarrow0.
$$
Thus the preceding bound is $o_P(\Delta_n^{3/2})$.

What remains to show is to replace $x_i^2$ with $\Delta_iW\Delta_iJ$ in \eqref{covariance_x2}. Similar to the case for $( \chi_{t_{i-1}}\Delta_iJ)^2$. By H\"older's inequality, and that \( P(|\varepsilon_i| > v_n) \) and \( P(|\Delta_i W| > v_n) \) decay faster than any power of \( \Delta_n \) since \( \Delta_n^{1/2}/v_n \ll \Delta_n^{1/2 - \beta} \), with similar techniques as in \eqref{jump_technical_1}, we obtain
\begin{align*}
    \sup_i \sigma_{t_{i-1}}\EE_{i-1} \left[ |\chi_{t_{i-1}}\Delta_i J||\Delta_iW| \mathbf{1}_{\{|x_i + \varepsilon_i| \leq v_n\}} |\Delta_i B| \left( \mathbf{1}_{\{|\Delta_i W| > v_n\}} + \mathbf{1}_{\{|\varepsilon_i| > v_n\}} \right) \right] = o_P(\Delta_n^{3/2}).
\end{align*}
Further notice that when $|\Delta_iW|\leq v_n$ and
$|\varepsilon_i|\leq v_n$, $|x_i+\varepsilon_i|\leq v_n$ implies $|\chi_{t_{i-1}}\Delta_iJ|\leq Cv_n$.
By boundedness of $\sigma$, H\"older's inequality, and the preceding truncated-moment
bound, we obtain
\begin{align*}
    &\sup_i\EE_{i-1}\left[
    |\chi_{t_{i-1}}\Delta_iJ||\Delta_iW|
    \mathbf{1}_{\{|\chi_{t_{i-1}}\Delta_iJ|\leq Cv_n\}}
    |\Delta_iB|\right]\\
    &\quad\leq \sup_i
    \EE_{i-1}\left[
    |\chi_{t_{i-1}}\Delta_iJ|^2
    \mathbf{1}_{\{|\chi_{t_{i-1}}\Delta_iJ|\leq Cv_n\}}
    \right]^{1/2}
    \EE_{i-1}\left[|\Delta_iW|^4\right]^{1/4}
    \EE_{i-1}\left[|\Delta_iB|^4\right]^{1/4}\\
    &\quad\lesssim
    (\Delta_nv_n^{2-Y})^{1/2}\Delta_n
    =\Delta_n^{3/2}v_n^{1-Y/2}
    =o(\Delta_n^{3/2}),
\end{align*}
since $Y<2$. Finally, recall that
$x_i=\mu_{t_{i-1}}\Delta_n+\hat x_i$.
On $\{|\Delta_iX|\leq v_n\}$, we have
$|\hat x_i|\leq v_n+|\varepsilon_i|+C\Delta_n$.
Hence, by boundedness of $\mu$, Cauchy--Schwarz,
and \eqref{KME0},
\begin{align*}
    \sup_i\left|\EE_{i-1}\left[
    (x_i^2-\hat x_i^2)
    \mathbf{1}_{\{|\Delta_iX|\leq v_n\}}
    \Delta_iB
    \right]\right|&\quad\lesssim
    \Delta_n\sup_i\EE_{i-1}\left[
    (v_n+|\varepsilon_i|+\Delta_n)|\Delta_iB|
    \right]\\
    &\quad\lesssim
    \Delta_n^{3/2}(v_n+\Delta_n)
    =o(\Delta_n^{3/2}).
\end{align*}
By \eqref{kernel_convergence}, these terms also give
an $o_P(1)$ contribution to \eqref{covariance_x2}.
Thus we have proved \eqref{cov_summand}.%
\end{proof}
\subsection{Auxiliary estimates for Theorem \ref{thm_clt_difference}}\label{TecnPrf2}
\begin{proof}[Proof of \eqref{diff_fn_approximation}]
It suffices to prove:
\[
A_{1} := \frac{m_n^{1/2}}{\Delta_n^{-1/2} v_n^{2-Y/2}} \sum_{i=1}^nK_{m_n\Delta_n}(t_{i-1}-\tau) \EE_{i-1} \left[ (\Delta_i X)^2 \mathbf{1}_{\{v_n \leq |\Delta_i X| \leq v_n(1+v_n^\eta)\}} |\Delta_i M| \right] = o_P(1),
\]
and
\[
A_{2} := \frac{m_n^{1/2}}{\Delta_n^{-1/2} v_n^{2-Y/2}} \sum_{i=1}^nK_{m_n\Delta_n}(t_{i-1}-\tau) \EE_{i-1} \left[ (\Delta_i X)^2 \mathbf{1}_{\{\zeta v_n(1-v_n^\eta) \leq |\Delta_i X| \leq \zeta v_n\}} |\Delta_i M| \right] = o_P(1).
\]

We first consider $A_1$. By Cauchy-Schwarz and \eqref{GNMExp_difference}, we obtain:
\begin{align}
    \nonumber|A_{1}| &\leq \frac{m_n^{1/2}}{\Delta_n^{-1/2} v_n^{2-Y/2}} \sum_{i=1}^nK_{m_n\Delta_n}(t_{i-1}-\tau) \EE_{i-1} \left[ (\Delta_i X)^4 \mathbf{1}_{\{v_n \leq |\Delta_i X| \leq v_n(1+v_n^\eta)\}} \right]^{\frac{1}{2}} \EE_{i-1} \left[ (\Delta_i M)^2 \right]^{\frac{1}{2}}\nonumber\\
&\lesssim\frac{m_n^{1/2}}{\Delta_n^{-1/2} v_n^{2-Y/2}} o_P((\Delta_n v_n^{4-Y})^{1/2}) \sum_{i=1}^n K_{m_n\Delta_n}(t_{i-1}-\tau)\EE_{i-1} \left[ (\Delta_i M)^2 \right]^{\frac{1}{2}}.\label{e:A1_bd}
\end{align}
Since $M$ is square integrable, 
$\sum_{i=1}^n\EE_{i-1}[(\Delta_iM)^2]=O_P(1)$.
By Cauchy--Schwarz, we obtain
\begin{align}\label{kernel_sqrt_m}
    &\sum_{i=1}^n K_{m_n\Delta_n}(t_{i-1}-\tau)
    \EE_{i-1}\left[(\Delta_iM)^2\right]^{1/2}
    \nonumber\\
    &\quad\leq
    \left(\sum_{i=1}^n
    K_{m_n\Delta_n}^2(t_{i-1}-\tau)\right)^{1/2}
    \left(\sum_{i=1}^n
    \EE_{i-1}\left[(\Delta_iM)^2\right]\right)^{1/2}
    \nonumber\\
    &\quad\lesssim \frac{1}{\sqrt{m_n}\Delta_n},
\end{align}
where, by boundedness of $K$ and
\eqref{kernel_convergence},
$$
\sum_{i=1}^n K_{m_n\Delta_n}^2(t_{i-1}-\tau)
\leq
\frac{\|K\|_\infty}{m_n\Delta_n}
\sum_{i=1}^n K_{m_n\Delta_n}(t_{i-1}-\tau)
=O\left(\frac{1}{m_n\Delta_n^2}\right).
$$
Substituting into the preceding bound \eqref{e:A1_bd} gives
$
|A_1|=o_P(1).$ A similar argument shows that \(A_{2} = o_P(1)\).
\end{proof}
\begin{proof}[Proof of \eqref{diff_fn_technical}]
    For convenience, we rewrite \eqref{diff_fn_technical} here: 
\begin{align}\label{diff_fn_technical2b}
    &\frac{m_n^{1/2}}{\Delta_n^{-1/2}v_n^{2-Y/2}}
    \sum_{i=1}^n K_{m_n\Delta_n}(t_{i-1}-\tau)
    \left|\EE_{i-1}\left[
    (\Delta_iX)^2 f_n(\Delta_iX)\Delta_iM
    \right]\right|
    =o_P(1).
\end{align}
The proof follows along the lines of the proof of an analogous identity in \cite{BONIECE2024} (see Lemma 7 and Eq.~(D.10) in the supplemental material of \cite{BONIECE2024supplement} posted online). More specifically, it was proved there in that 
\begin{align}\label{diff_fn_technical2before}
    \frac{1}{v_n^{2-Y/2}}\sum_{i=1}^n \EE_{i-1} \left[ (\Delta_i X)^2 f_n(\Delta_i X) \Delta_i M \right] = o_P(1).
\end{align}
Note that \eqref{diff_fn_technical2b} does not directly follow from \eqref{diff_fn_technical2before} and the boundedness of $K$ because $m_n\Delta_n\to{}0$. We therefore need to establish absolute summability and account for the localizing effect of the kernel.

    Using the notation introduced in \eqref{discretizing}, consider the decomposition  
\[
(\Delta_i X)^2 = (x_i + \varepsilon_i)^2 = x_i^2 + 2x_i\varepsilon_i + \varepsilon_i^2.
\]  
We begin by analyzing the second and third terms on the right-hand side by Cauchy-Schwarz. In particular, applying \eqref{KME0}, we obtain that
\begin{align*}
   & \frac{m_n^{1/2}}{\Delta_n^{-1/2} v_n^{2-Y/2}} \sum_{i=1}^nK_{m_n\Delta_n}(t_{i-1}-\tau) \EE_{i-1} \left[ \varepsilon_i^2 f_n(\Delta_i X) |\Delta_i M| \right] \nonumber\\
   &\leq \frac{m_n^{1/2}}{\Delta_n^{-1/2} v_n^{2-Y/2}} \sum_{i=1}^nK_{m_n\Delta_n}(t_{i-1}-\tau) \EE_{i-1} \left[ \varepsilon_i^4 \right]^{\frac{1}{2}} \EE_{i-1} \left[ (\Delta_i M)^2 \right]^{\frac{1}{2}}\nonumber\\
   &\lesssim\frac{m_n^{1/2}}{\Delta_n^{-1/2} v_n^{2-Y/2}} \Delta_n^{\frac{3}{2}} \sum_{i=1}^n K_{m_n\Delta_n}(t_{i-1}-\tau)\EE_{i-1} \left[ (\Delta_i M)^2 \right]^{\frac{1}{2}} = o_P(1),
\end{align*}
since, by \eqref{kernel_sqrt_m}, \(\sum_{i=1}^n K_{m_n\Delta_n}(t_{i-1}-\tau) \EE_{i-1} \left[ (\Delta_i M)^2 \right]^{\frac{1}{2}} \lesssim \frac{1}{\sqrt{m_n}\Delta_n}\) and, recalling that $v_n\gg\Delta_n^{1/2}$,
\[\frac{m_n^{1/2}}{\Delta_n^{-1/2} v_n^{2-Y/2}} \Delta_n^{3/2}\frac{1}{\sqrt{m_n}\Delta_n} = \Delta_n v_n^{Y/2-2} \ll \Delta_n^{Y/4} \ll 1.\]

For the term involving \(x_i\varepsilon_i\), note that when \(|\varepsilon_i| > v_n\), by Lemma 3 in \cite{BONIECE2024supplement}, for some constant $K>0$, we have
\begin{align}\label{varepsilon_tail}
    \PP_{i-1}(|\varepsilon_i| > v_n) \leq \frac{\EE_{i-1}[|\varepsilon_i|^r]}{v_n^r} \leq K\Delta_n^{1+r\left(\frac{1}{2}-\beta\right)},
\end{align}
which holds for sufficiently large \(r\) due to {\eqref{KME0}} and \(v_n \gg \Delta_n^\beta\). On the support of $f_n(\Delta_iX)$, if
$|\varepsilon_i|>v_n$, we have
$$
|x_i|
\leq|\Delta_iX|+|\varepsilon_i|
\leq\zeta v_n+|\varepsilon_i|
\leq(\zeta+1)|\varepsilon_i|.
$$
Consequently, the preceding estimate for the
$\varepsilon_i^2$ term gives
\begin{align*}
    &\frac{m_n^{1/2}}{\Delta_n^{-1/2}v_n^{2-Y/2}}
    \sum_{i=1}^nK_{m_n\Delta_n}(t_{i-1}-\tau)
    \EE_{i-1}\left[
    |\varepsilon_i x_i|f_n(\Delta_iX)
    \mathbf{1}_{\{|\varepsilon_i|>v_n\}}|\Delta_iM|
    \right]\\
    &\quad\leq
    \frac{(\zeta+1)m_n^{1/2}}
    {\Delta_n^{-1/2}v_n^{2-Y/2}}
    \sum_{i=1}^nK_{m_n\Delta_n}(t_{i-1}-\tau)
    \EE_{i-1}\left[
    \varepsilon_i^2f_n(\Delta_iX)|\Delta_iM|
    \right]
    =o_P(1).
\end{align*}

When \(|\varepsilon_i| \leq v_n\), observe that \(|\Delta_i X| \leq \zeta v_n\) implies \(|x_i| = |\Delta_i X - \varepsilon_i| \leq (1 + \zeta)v_n =: Cv_n\). Then, by Proposition 3 in \cite{BONIECE2024supplement}, {\eqref{KME0}}, and \eqref{kernel_sqrt_m}, we get:
\begin{align*}
    &\frac{m_n^{1/2}}{\Delta_n^{-1/2} v_n^{2-Y/2}} \sum_{i=1}^nK_{m_n\Delta_n}(t_{i-1}-\tau) \EE_{i-1} \left[ |\varepsilon_ix_i| \mathbf{1}_{\{|x_i| \leq Cv_n\}} |\Delta_i M| \right]\\
\quad&\leq \frac{m_n^{1/2}}{\Delta_n^{-1/2} v_n^{2-Y/2}} \sum_{i=1}^nK_{m_n\Delta_n}(t_{i-1}-\tau) \EE_{i-1} \left[ \varepsilon_i^4 \right]^{\frac{1}{4}} \EE_{i-1} \left[ x_i^4 \mathbf{1}_{\{|x_i| \leq Cv_n\}} \right]^{\frac{1}{4}} \EE_{i-1} \left[ (\Delta_i M)^2 \right]^{\frac{1}{2}} \\
\quad&\lesssim \frac{m_n^{1/2}}{\Delta_n^{-1/2} v_n^{2-Y/2}} \Delta_n^{\frac{3}{4}} \Delta_n^{\frac{1}{2}} \sum_{i=1}^n K_{m_n\Delta_n}(t_{i-1}-\tau)\EE_{i-1} \left[ (\Delta_i M)^2 \right]^{\frac{1}{2}} \nonumber\\
\quad&= \frac{m_n^{1/2}}{\Delta_n^{-1/2} v_n^{2-Y/2}} \Delta_n^{\frac{3}{4}} \Delta_n^{\frac{1}{2}}\frac{1}{\sqrt{m_n}\Delta_n} = \Delta_n^{\frac{3}{4}}v_n^{Y/2-2} \ll \Delta_n^{\frac{Y-1}{4}} \ll 1,
\end{align*}
since \(\Delta_n^{1/2} \ll v_n\) and $Y>1$.

Then we just need to tackle on the last term involving \(x_i^2\). We need to show that
\begin{align}\label{x2_difference}
    &\frac{m_n^{1/2}}{\Delta_n^{-1/2}v_n^{2-Y/2}}
    \sum_{i=1}^n K_{m_n\Delta_n}(t_{i-1}-\tau)
    \left|\EE_{i-1}\left[
    x_i^2 f_n(\Delta_iX)\Delta_iM
    \right]\right|
    =o_P(1).
\end{align}
We start by replacing $x_i$ with $\mu_{t_{i-1}}\Delta_n$ in \eqref{x2_difference} with arbitrary $M$. By the boundedness of $K$, ${\mu}$ and Cauchy-Schwarz, we have
\begin{align*}
    &\frac{m_n^{1/2}\Delta_n^2}{\Delta_n^{-1/2} v_n^{2-Y/2}} \sum_{i=1}^n K_{m_n\Delta_n}(t_{i-1}-\tau) \left|\EE_{i-1}\left[
\mu_{t_{i-1}}^2 f_n(\Delta_iX)\Delta_iM
\right]\right|\\
    &\quad\lesssim\frac{m_n^{1/2}\Delta_n^2}{\Delta_n^{-1/2} v_n^{2-Y/2}m_n\Delta_n} \sum_{i=1}^n \EE_{i-1} \left[(\Delta_i M)^2 \right]^{1/2} \nonumber\\
    &\quad\lesssim\frac{m_n^{1/2}\Delta_n^{3/2}}{ v_n^{2-Y/2}m_n\Delta_n}\Big( \sum_{i=1}^n \EE_{i-1} \left[(\Delta_i M)^2 \right]\Big)^{1/2} \nonumber\\
    &\quad\lesssim \frac{\Delta_n^{1/2}}{ v_n^{2-Y/2}\Delta_n^{-1/2}v_n^{Y/2}} \rightarrow 0
\end{align*}
where the last line follows from $m_n\gg \Delta_n^{-1}v_n^\gamma$ with $\gamma<Y$ and $v_n\gg\Delta_n^{1/2}$. Besides, recall that $\hat x_i = \sigma_{t_{i-1}}\Delta_iW + \chi_{t_{i-1}}\Delta_iJ$ and by applying Cauchy-Schwarz twice and boundedness of $K$, we obtain
\begin{align*}
    &\frac{m_n^{1/2}\Delta_n}{\Delta_n^{-1/2} v_n^{2-Y/2}} \sum_{i=1}^n K_{m_n\Delta_n}(t_{i-1}-\tau)\left|\EE_{i-1}\left[
    \mu_{t_{i-1}}\hat x_i f_n(\Delta_iX)\Delta_iM
    \right]\right|\\
    \quad&\quad\lesssim\frac{m_n^{1/2}\Delta_n}{\Delta_n^{-1/2} v_n^{2-Y/2}m_n\Delta_n} \Big(\sum_{i=1}^n \EE_{i-1} \left[(\Delta_i M)^2 \right]\Big)^{1/2}\Big(\sum_{i=1}^n \EE_{i-1} \left[\hat x_i^2f_n^2(\Delta_i X) \right]\Big)^{1/2}\nonumber\\
    \quad&\quad\lesssim\frac{\Delta_n^{1/2}}{ v_n^{2-Y/2}m_n^{1/2}} \lesssim \frac{\Delta_n^{1/2}}{ v_n^{2-Y/2}\Delta_n^{-1/2}v_n^{Y/2}} \rightarrow 0,
\end{align*}
where the last line follows from $m_n\gg \Delta_n^{-1}v_n^\gamma$ with $\gamma<Y$ and $v_n\gg\Delta_n^{1/2}$.

It remains to treat the term involving $\hat x_i^2$. We first prove this for \(M = W\). Replacing \(x_i\) with $\hat x_i= \sigma_{t_{i-1}} \Delta_i W$ in \eqref{x2_difference}, for any \(q, p > 1\) such that \(\frac{1}{q} + \frac{1}{p} = 1\), we obtain:
\begin{align}\label{x2_difference_brownian}
    &\frac{m_n^{1/2}}{\Delta_n^{-1/2} v_n^{2-Y/2}} \sum_{i=1}^n \sigma_{t_{i-1}}^2 K_{m_n\Delta_n}(t_{i-1}-\tau) \left| \EE_{i-1}\left[ (\Delta_i W)^2 f_n(\Delta_i X) \Delta_i M \right] \right| \nonumber\\
    \quad&\quad\lesssim \frac{m_n^{1/2}}{\Delta_n^{-1/2} v_n^{2-Y/2}} \sum_{i=1}^n K_{m_n\Delta_n}(t_{i-1}-\tau)\EE_{i-1} \left[ |\Delta_i W|^3 \mathbf{1}_{\{|\Delta_i X| > v_n\}} \right]\nonumber\\
    \quad&\quad\lesssim \frac{m_n^{1/2}}{\Delta_n^{-1/2} v_n^{2-Y/2}} \sum_{i=1}^n K_{m_n\Delta_n}(t_{i-1}-\tau)\EE_{i-1} \left[ |\Delta_i W|^{3p} \right]^{\frac{1}{p}} \PP_{i-1}\left[ |\Delta_i X| > v_n \right]^{\frac{1}{q}}\nonumber\\
    \quad&\quad\lesssim \frac{m_n^{1/2}}{\Delta_n^{-1/2} v_n^{2-Y/2}} \Delta_n^{-1} \Delta_n^{\frac{3}{2}} (\Delta_nv_n^{-Y})^{\frac{1}{q}},
\end{align}
where we used Lemma 2 in \cite{BONIECE2024supplement} and H\"older's inequality in the third line. Notice that
$$
\frac{m_n^{1/2}}{\Delta_n^{-1/2}v_n^{2-Y/2}}
\Delta_n^{1/2}(\Delta_nv_n^{-Y})^{1/q}
\ll
\Delta_n^{1/2}v_n^{-2+Y/2}
(\Delta_nv_n^{-Y})^{1/q}.
$$
Moreover, since \(\frac{3}{4+Y} > \frac{1}{2}\), we have
\[
\Delta_n^{3/2} v_n^{-\frac{4+Y}{2}} \ll 1 \quad \Leftrightarrow  \quad  \Delta_n^{3/(4 + Y)} \ll v_n.
\]
Then \eqref{x2_difference_brownian} vanishes by taking $q$ close to 1. 

Now, replacing \(x_i\) with \( \chi_{t_{i-1}}\Delta_i J\) with $M=W$, we need to show that:
\begin{align}\label{x2_difference_jump}
    \frac{m_n^{1/2}}{\Delta_n^{-1/2} v_n^{2-Y/2}} \sum_{i=1}^n  K_{m_n\Delta_n}(t_{i-1}-\tau) \EE_{i-1} \left[ (\chi_{t_{i-1}}\Delta_i J)^2 f_n(\Delta_i X) |\Delta_i W| \right] = o_P(1). 
\end{align}
When \(|\Delta_i W| > v_n\) or \(|\varepsilon_i| > v_n\), recall \eqref{varepsilon_tail} that \(P(|\varepsilon_i| > v_n)\) and \(P(|\Delta_i W| > v_n)\) both decay faster than any power of \(\Delta_n\). By H\"older's inequality, we derive:
\[
\frac{m_n^{1/2}}{\Delta_n^{-1/2} v_n^{2-Y/2}} \sum_{i=1}^n K_{m_n\Delta_n}(t_{i-1}-\tau)\EE_{i-1}\left[ (\chi_{t_{i-1}}\Delta_i J)^2 f_n(\Delta_i X) |\Delta_i W| (\mathbf{1}_{\{|\Delta_i W| > v_n\}} + \mathbf{1}_{\{|\varepsilon_i| > v_n\}}) \right] = o_P(1).
\]
Thus, when \(|\Delta_i W| \leq v_n\), \(|\varepsilon_i| \leq v_n\), and \(|\Delta_i X| = |x_i + \varepsilon_i| \leq \zeta v_n\), we have \(|\chi_{t_{i-1}}\Delta_i J| \leq C v_n\) for some constant $C$. Therefore, for \eqref{x2_difference_jump} to hold, it suffices to show that:
\[
\frac{m_n^{1/2}}{\Delta_n^{-1/2} v_n^{2-Y/2}} \sum_{i=1}^n K_{m_n\Delta_n}(t_{i-1}-\tau)\EE_{i-1}\left[ (\chi_{t_{i-1}}\Delta_i J)^2 \mathbf{1}_{\{|\chi_{t_{i-1}}\Delta_i J| \leq C v_n\}} |\Delta_i W| \right] \overset{P}{\to} 0.
\]

Using H\"older's inequality and Lemma 17 in \cite{BONIECE2024}, for any \(p, q > 1\) such that \(\frac{1}{p} + \frac{1}{q} = 1\), we have:
\begin{align*}
    &\frac{m_n^{1/2}}{\Delta_n^{-1/2} v_n^{2-Y/2}} \sum_{i=1}^n K_{m_n\Delta_n}(t_{i-1}-\tau)\EE_{i-1}\left[ |\chi_{t_{i-1}}\Delta_i J|^{2p} \mathbf{1}_{\{|\chi_{t_{i-1}}\Delta_i J| \leq C v_n\}} \right]^{\frac{1}{p}} \EE_{i-1}\left[ |\Delta_i W|^q \right]^{\frac{1}{q}}\nonumber\\
    &\lesssim  \frac{m_n^{1/2}}{\Delta_n^{-1/2} v_n^{2-Y/2}} \sum_{i=1}^n K_{m_n\Delta_n}(t_{i-1}-\tau) \left( \Delta_n v_n^{2p - Y} \right)^{\frac{1}{p}} \Delta_n^{1/2}.
\end{align*}
Choose a fixed $1<p<2$ and $q=p/(p-1)$.
By \eqref{kernel_convergence}, the preceding bound is
of order
$$
\sqrt{m_n}\Delta_n^{1/p}v_n^{Y/2-Y/p}
=
\sqrt{m_n\Delta_n}
\left(\Delta_nv_n^{-Y}\right)^{1/p-1/2}
\longrightarrow0,
$$
since $m_n\Delta_n\to0$ and $\Delta_nv_n^{-Y}\to0$.

Now, we prove \eqref{x2_difference} for any bounded
martingale $M$ orthogonal to $W$. The following estimate
for the cross-product term also applies to $M=W$.
As above, after treating the exceptional events
$|\Delta_iW|>v_n$ and $|\varepsilon_i|>v_n$, it suffices
to consider $|\chi_{t_{i-1}}\Delta_iJ|\leq Cv_n$.
Using H\"older's inequality and \eqref{kernel_sqrt_m},
for $p,q>2$ such that $1/p+1/q=1/2$, we obtain
\begin{align*}
    &\frac{m_n^{1/2}}{\Delta_n^{-1/2}v_n^{2-Y/2}}
    \sum_{i=1}^n K_{m_n\Delta_n}(t_{i-1}-\tau)
    \EE_{i-1}\left[
    |\chi_{t_{i-1}}\Delta_iJ||\Delta_iW|
    \mathbf{1}_{\{|\chi_{t_{i-1}}\Delta_iJ|\leq Cv_n\}}
    |\Delta_iM|
    \right]\\
    &\quad\leq
    \frac{m_n^{1/2}}{\Delta_n^{-1/2}v_n^{2-Y/2}}
    \sum_{i=1}^n K_{m_n\Delta_n}(t_{i-1}-\tau)
    \EE_{i-1}\left[
    |\chi_{t_{i-1}}\Delta_iJ|^p
    \mathbf{1}_{\{|\chi_{t_{i-1}}\Delta_iJ|\leq Cv_n\}}
    \right]^{1/p}\\
    &\qquad\times
    \EE_{i-1}\left[|\Delta_iW|^q\right]^{1/q}
    \EE_{i-1}\left[(\Delta_iM)^2\right]^{1/2}\\
    &\quad\lesssim
    \frac{m_n^{1/2}}{\Delta_n^{-1/2}v_n^{2-Y/2}}
    (\Delta_nv_n^{p-Y})^{1/p}\Delta_n^{1/2}
    \frac{1}{\sqrt{m_n}\Delta_n}\\
    &\quad=
    \Delta_n^{1/p}v_n^{-1+Y/2-Y/p}.
\end{align*}
Choose a fixed $p>2$ sufficiently close to $2$ that
$
\frac1p-\beta\big(1-\frac Y2+\frac Yp\big)>0$
(this is possible since the expression equals
$1/2-\beta>0$ at $p=2$). Since
$v_n\gg\Delta_n^\beta$, the preceding bound tends to zero.

Next we consider the term involving
$\sigma_{t_{i-1}}^2(\Delta_iW)^2$. It suffices to show
$$
D_n:=
\frac{m_n^{1/2}}{\Delta_n^{-1/2}v_n^{2-Y/2}}
\sum_{i=1}^n K_{m_n\Delta_n}(t_{i-1}-\tau)
\sigma_{t_{i-1}}^2
\EE_{i-1}\left[
|\Delta_iW|^2 f_n(\Delta_iX)|\Delta_iM|
\right]
=o_P(1).
$$
Since $0\leq f_n(\Delta_iX)\leq
\mathbf{1}_{\{|\Delta_iX|\geq v_n\}}$, boundedness of
the coefficients and $\Delta_n=o(v_n)$ imply that,
for some constant $c>0$ and all sufficiently large $n$,
$$
f_n(\Delta_iX)
\leq
\mathbf{1}_{\{|\Delta_iW|>cv_n\}}
+\mathbf{1}_{\{|\varepsilon_i|>cv_n\}}
+\mathbf{1}_{\{|\chi_{t_{i-1}}\Delta_iJ|>cv_n\}}.
$$
The first two indicators give $o_P(1)$ contributions
to $D_n$, by H\"older's inequality, the Gaussian
moment bounds, and \eqref{KME0}. For the last indicator,
take $p,q>2$ such that $1/p+1/q=1/2$. Using
H\"older's inequality and \eqref{kernel_sqrt_m},
we obtain
\begin{align*}
D_n
&\lesssim
\frac{m_n^{1/2}}{\Delta_n^{-1/2}v_n^{2-Y/2}}
\sum_{i=1}^n K_{m_n\Delta_n}(t_{i-1}-\tau)
\EE_{i-1}\left[|\Delta_iW|^{2p}\right]^{1/p}\\
&\qquad\times
\PP_{i-1}\left(
|\chi_{t_{i-1}}\Delta_iJ|>cv_n
\right)^{1/q}
\EE_{i-1}\left[(\Delta_iM)^2\right]^{1/2}
+o_P(1)\\
&\lesssim
\frac{m_n^{1/2}}{\Delta_n^{-1/2}v_n^{2-Y/2}}
\Delta_n(\Delta_nv_n^{-Y})^{1/q}
\sum_{i=1}^n K_{m_n\Delta_n}(t_{i-1}-\tau)
\EE_{i-1}\left[(\Delta_iM)^2\right]^{1/2}
+o_P(1)\\
&=
O_P\left(
\Delta_n^{1/2+1/q}v_n^{-2+Y/2-Y/q}
\right)+o_P(1).
\end{align*}
Choose a fixed $q>2$ sufficiently close to $2$ that
$
\frac12+\frac1q
-\beta\big(2-\frac Y2+\frac Yq\big)>0$ (this is possible since the expression equals
$1-2\beta>0$ at $q=2$.) Since
$v_n\gg\Delta_n^\beta$, the preceding bound tends to zero.

Finally, we consider the case where \(x_i\) is replaced with \(\chi_{t_{i-1}}\Delta_i J\) in \eqref{x2_difference}. It requires showing:
\begin{align*}
    \frac{m_n^{1/2}}{\Delta_n^{-1/2} v_n^{2-Y/2}} \sum_{i=1}^n K_{m_n\Delta_n}(t_{i-1}-\tau) {\chi_{t_{i-1}}^2}\left|\EE_{i-1}\left[
(\Delta_iJ)^2 f_n(\Delta_iX)\Delta_iM
\right]\right| = o_P(1).
\end{align*}
Note that
\begin{equation}\label{diff_indicator_approx}
\begin{aligned}
    &\frac{m_n^{1/2}}{\Delta_n^{-1/2} v_n^{2-Y/2}} \sum_{i=1}^n K_{m_n\Delta_n}(t_{i-1}-\tau) {\chi_{t_{i-1}}^2}\EE_{i-1}\left[ (\Delta_i J)^2 \left| f_n(\Delta_i X) - f_n(\chi_{t_{i-1}} \Delta_i J) \right| |\Delta_i M| \right]\\
    &\quad= o_P(1),
\end{aligned}
\end{equation}
which is shown at the end of the appendix.
Hence it suffices to prove:
\begin{align}\label{x2_difference_jump_final}
    \frac{m_n^{1/2}}{\Delta_n^{-1/2} v_n^{2-Y/2}} \sum_{i=1}^n K_{m_n\Delta_n}(t_{i-1}-\tau) \chi_{t_{i-1}}^2\left|\EE_{i-1}\left[
(\Delta_iJ)^2
f_n(\chi_{t_{i-1}}\Delta_iJ)\Delta_iM
\right]\right| = o_P(1).
\end{align}
Note that, conditional on $\mathscr{F}_{t_{i-1}}$, $\chi_{t_{i-1}}$ is constant. In other words, note that 
\[
    \EE_{i-1}\left[ (\Delta_i J)^2 f_n(\chi_{t_{i-1}} \Delta_i J) \Delta_i M \right]=G_i(\chi_{t_{i-1}}),
\]
where for a deterministic constant $\chi$ (possibly $0$), 
\[
    G_i(\chi):=\EE_{i-1}\left[ (\Delta_i J)^2 f_n(\chi \Delta_i J) \Delta_i M \right].
\]
Let us now evaluate $G_i(\chi)$. 
Suppose that \({J}\) has the following L\'evy-It\^o decomposition:
\[
J_t = {\int_0^t \int{x \mathbf{1}_{\{|x| \leq 1\}}} \bar{N}(ds, dx) + \int_0^t \int{x \mathbf{1}_{\{|x| > 1\}}} N(ds, dx)},
\]
where \(\bar{N}(ds, dx) = N(ds, dx) - \nu(dx)ds\) is the compensated jump measure of \(J\). Additionally, similar to expression (80) in \cite{A_t_Sahalia_2009}, \(M\) admits the decomposition:
\[
M_t = M'_t + \int_0^t \int \delta(s, x) \bar{N}(ds, dx),
\]
where $M'$ is a square-integrable martingale orthogonal
to $W$ and to every square-integrable stochastic integral
with respect to $\bar N$, and $\delta$ is a bounded
predictable function satisfying:
\[
E \left[ \int_t^u \int \delta^2(s, x) \nu(dx) ds \middle| \mathcal{F}_t \right] \leq E \left[ (M_u - M_t)^2 \middle| \mathcal{F}_t \right] < \infty.
\]
For future reference, define \(N^i(ds, dx) := N(t_i + ds, dx)\), \(\delta_i(s, x) := \delta(t_i + s, x)\), and:
\[
p_n(x) = x^2 f_n({\chi}x), \; t_n(y, x) = p_n(x+y) - p_n(y), \; w_n(y, x) = p_n(x+y) - p_n(y) - x p'_n(y) \mathbf{1}_{\{|x| \leq 1\}}.
\]

Applying It\^o's formula to \(Y_{u}^{i-1} := {J_{t_{i-1}+u} - J_{t_{i-1}}}\), we obtain:
\begin{align*}
    p_n(Y_{u}^{i-1}) =& \int_0^u p'_n(Y_{s^-}^{i-1}) dY_s^{i-1} + \int_0^u \int \left( p_n(Y_{s^-}^{i-1} + x) - p_n(Y_{s^-}^{i-1}) - p'_n(Y_{s-}^{i-1}) x \right) N^{i-1}(ds, dx)\nonumber\\
    =& \int_0^u \int_{|x| \leq 1} p'_n(Y_{s-}^{i-1}) x \bar{N}^{i-1}(ds, dx) + \int_0^u \int w_n(Y_{s-}^{i-1}, x) N^{i-1}(ds, dx).
\end{align*}

Now, let \(Z_u^{i-1} := M_{t_{i-1}+u} - M_{t_{i-1}}\). Applying It\^o's formula to \(p_n(Y_u^{i-1}) Z_u^{i-1}\), we get:
\begin{align}\label{e20}
    p_n(Y_u^{i-1}) Z_u^{i-1} &= \int_0^u p_n(Y_{s-}^{i-1}) dZ_s^{i-1} + \int_0^u Z_{s-}^{i-1} d\left( p_n(Y_s^{i-1}) \right) + \sum_{s \leq u} \Delta(p_n(Y_s^{i-1})) \Delta Z_s^{i-1}\nonumber\\
    \nonumber
    &= {\text{Mrtg}} + \int_0^u \int \delta_{i-1}(s,x) t_n(Y_s^{i-1}, x) \nu(dx) ds \\
    &\quad+ \int_0^u \int w_n(Y_s^{i-1}, x) Z_s^{i-1} \nu(dx) ds,
\end{align}
where $ {\text{Mrtg}}$ denotes the martingale part generated from the It\^o's formula.
Fixing \(u = \Delta_n\) in \eqref{e20} and taking expectations, we obtain:
\begin{equation}\label{e21}
\begin{aligned}
    \EE_{i-1} \left[ p_n(\Delta_i J) \Delta_i M \right] =& \EE_{i-1} \left[ \int_0^{\Delta_n} \int \delta_{i-1}(s,x) t_n(Y_s, x) \nu(dx) ds \right]\\
    &+ \EE_{i-1} \left[ \int_0^{\Delta_n} \int w_n(Y_s, x) Z_s \nu(dx) ds \right]. 
\end{aligned}
\end{equation}
In the above equation, for simplicity we omit the superscript $i-1$ in $Y$ and $Z$.
Applying the triangle inequality in \eqref{e21} term
by term, it suffices to bound the two terms
with absolute values inside the integrals.

{We first consider the second term in \eqref{e21}. Note that if $\chi=0$, then $w_n(x,y)=0$ and the second term vanishes. Let us suppose that $\chi\neq 0$. We can then follow the steps in \cite{BONIECE2024supplement} to derive their formulas (D.22) and (D.23) to deduce that for some positive constant powers $\tilde{\eta}_1,\tilde\eta_2$ and a constant $K$, we have:
\begin{align*}
\EE_{i-1} \left[ \int_0^{\Delta_n} \int |w_n(Y_s, x)| |Z_s| \nu(dx) ds \right] &\leq \frac{K}{|\chi_{t_{i-1}}|^{2-\tilde\eta_1}} (v_n^{2 - Y(1+\eta)} \vee v_n) \Delta_n \EE_{i-1} (|\Delta_i M|)\\
&\quad + \frac{K}{|\chi_{t_{i-1}}|^{2-\tilde\eta_2}}  v_n^{2 - Y(1+\eta)} \Delta_n^2 v_n^{\eta - Y}.
\end{align*}
Combining the above equation with the second term in \eqref{e21} and {plugging it back in \eqref{x2_difference_jump_final}}, {we need to show the negligibility of the two terms below}:
\begin{align}
    &\frac{m_n^{1/2}}{\Delta_n^{-1/2} v_n^{2-Y/2}} \left( v_n^{2 - Y(1+\eta)} \vee v_n \right) \Delta_n \sum_{i=1}^n K_{m_n\Delta_n}(t_{i-1}-\tau) |\chi_{t_{i-1}}|^{\tilde\eta_1}{\bf 1}_{\{\chi_{t_{i-1}}\neq{}0\}}\EE_{i-1}(|\Delta_i M|) \nonumber\\
& \quad+\frac{m_n^{1/2}}{\Delta_n^{-1/2} v_n^{2-Y/2}} v_n^{2 - Y(1+\eta)} \Delta_n^2 v_n^{\eta - Y}  \sum_{i=1}^n K_{m_n\Delta_n}(t_{i-1}-\tau) |\chi_{t_{i-1}}|^{\tilde\eta_2}{\bf 1}_{\{\chi_{t_{i-1}}\neq{}0\}}.
\label{NSTH}
\end{align}
For the first term above, first note that
\begin{align*}
    &\Delta_n^{1/2}\sum_{i=1}^n K_{m_n\Delta_n}(t_{i-1}-\tau) |\chi_{t_{i-1}}|^{\tilde\eta_1}{\bf 1}_{\{\chi_{t_{i-1}}\neq{}0\}}\EE_{i-1}(|\Delta_i M|)\\
    &\leq \Big(\Delta_n\sum_{i=1}^n K^2_{m_n\Delta_n}(t_{i-1}-\tau)|\chi_{t_{i-1}}|^{2\tilde\eta_1}{\bf 1}_{\{\chi_{t_{i-1}}\neq{}0\}}\Big)^{1/2} \Big(\sum_{i=1}^n\EE_{i-1}(|\Delta_i M|^2)\Big)^{1/2}.
\end{align*}
The second factor above is $O_P(1)$, while the first factor is $O_P(m_n^{-1/2}\Delta_n^{-1/2})$ because
\[
\Delta_n\sum_{i=1}^n \frac{1}{m_n\Delta_n}K^2\left(\frac{t_{i-1}-\tau}{m_n\Delta_n}\right)|\chi_{t_{i-1}}|^{2\tilde\eta_1}{\bf 1}_{\{\chi_{t_{i-1}}\neq{}0\}}\stackrel{P}{\longrightarrow}|\chi_{\tau}|^{2\tilde\eta_1}\int K^2(x)dx.
\]
We then conclude that the first term of \eqref{NSTH} converges to 0 since $v_n^{Y/2-2}v_n^{2 - Y(1+\eta)}\Delta_n^{1/2} \ll 1$ and $v_n^{Y/2-2}v_n\Delta_n^{1/2} \ll 1$ with $Y>1$ if $\eta$ is small enough.}
Similarly, the second term in \eqref{NSTH} is of order $m_n^{1/2}v_n^{-3Y/2+(1-Y)\eta}\Delta_n^{3/2}$, which vanishes by taking $\eta$ small enough due to the condition $m_n\ll \Delta_n^{-5}v_n^{8+Y}$, since $\Delta_n^{-5}v_n^{8+Y} \ll v_n^{3Y}\Delta_n^{-3}$ with $v_n \ll \Delta_n^{1/(4-Y)}$.
Thus, the two terms {in \eqref{NSTH}} converge to zero, and we conclude that:
\[
\frac{m_n^{1/2}}{\Delta_n^{-1/2} v_n^{2-Y/2}} \sum_{i=1}^n K_{m_n\Delta_n}(t_{i-1}-\tau)\chi_{t_{i-1}}^2 \EE_{i-1} \left( \int_0^{\Delta_n} \int |w_n(Y_s, x)| |Z_s| \nu(dx) ds \right) = o_P(1).
\]

{It remains to show that the contribution of the first term in \eqref{e21} is negligible:}
\begin{align*}
F_n &:= \frac{m_n^{1/2}}{\Delta_n^{-1/2} v_n^{2-Y/2}} \sum_{i=1}^n K_{m_n\Delta_n}(t_{i-1}-\tau)\chi_{t_{i-1}}^2 \EE_{i-1} \left[ \int_0^{\Delta_n} \int| \delta_{i-1}(s, x) t_n(Y_s, x) |\nu(dx) ds \right]\\
&= o_P(1).
\end{align*}
{{To evaluate the expectation within the summation above, assume that $\chi_{t_{i-1}}\neq{}0$, otherwise $t_n(Y_s, x)=0$ and the expectation is $0$.} 
As in (D.25) in \cite{BONIECE2024supplement},}
 we separately consider the cases where \( |\chi_{t_{i-1}} Y_s| \leq v_n \), \( |\chi_{t_{i-1}} Y_s| > \zeta v_n \), \( v_n(1 + v_n^\eta) < |\chi_{t_{i-1}} Y_s| < \zeta v_n(1 - v_n^\eta) \), \( v_n \leq |\chi_{t_{i-1}} Y_s| \leq v_n(1 + v_n^\eta) \), and $ \zeta v_n (1 - v_n^\eta) \leq |\chi_{t_{i-1}} Y_s| \leq \zeta v_n$

First, suppose that \( |\chi_{t_{i-1}} Y_s| \leq v_n \), so that \( f_n(\chi_{t_{i-1}} Y_s) = 0 \), and thus \( t_n(Y_s, x) = (Y_s + x)^2 f_n(\chi_{t_{i-1}}(Y_s + x)) \), with \( \frac{1}{3} v_n^{1 + \eta}\leq |{\chi_{t_{i-1}}} x| \leq (1+\zeta)v_n=: Cv_n \) (otherwise \( f_n(\chi_{t_{i-1}}(Y_s + x)) = 0 \)). Notice that \( |t_n(Y_s, x)| \leq Y_s^2 + 2 |x Y_s| + x^2 \).
For \( Y_s^2 \), {using Cauchy-Schwarz and the bound 
\begin{equation}\label{BNDFYC}
\EE_{i-1} \left[ Y_s^{2k} 1_{\{|{\chi_{t_{i-1}}}Y_s| \leq v_n\}} \right] \leq K \frac{1}{{|\chi_{t_{i-1}}|^{2k-Y}}}\Delta_n v_n^{2k - Y},
\end{equation}
valid for all \( 0 < s < \Delta_n \) (see Lemma 12 in \cite{BONIECE2024supplement}),}
\begin{align*}
    &\sum_{i=1}^n K_{m_n\Delta_n}(t_{i-1}-\tau)\chi_{t_{i-1}}^2 \EE_{i-1} \left( \int_0^{\Delta_n} \int |\delta_{i-1}(s, x)| Y_s^2 1_{\{|{\chi_{t_{i-1}}}Y_s| \leq v_n\}} 1_{\{|{\chi_{t_{i-1}}}x| \geq \frac{1}{3} v_n^{1 + \eta}\}} \nu(dx) ds \right)\nonumber\\
    &\leq  \sum_{i=1}^n K_{m_n\Delta_n}(t_{i-1}-\tau) \chi_{t_{i-1}}^2\EE_{i-1} \left( \int_0^{\Delta_n} \int Y_s^4 1_{\{|{\chi_{t_{i-1}}}Y_s| \leq v_n\}} 1_{\{|{\chi_{t_{i-1}}}x| \geq \frac{1}{3} v_n^{1 + \eta}\}} \nu(dx) ds \right)^{\frac{1}{2}}\nonumber\\
&\quad\qquad\times \EE_{i-1} \left( \int_0^{\Delta_n} \int |\delta_{i-1}(s, x)|^2 \nu(dx) ds \right)^{\frac{1}{2}}\nonumber\\
&\lesssim  \left(v_n^{-Y(1+\eta)}\Delta_n^2 v_n^{4-Y}\right)^{1/2}\sum_{i=1}^n K_{m_n\Delta_n}(t_{i-1}-\tau)\chi_{t_{i-1}}^2\frac{{\bf 1}_{{\chi_{t_{i-1}}}\neq 0}}{{|\chi_{t_{i-1}}|^{2-Y}}}\EE_{i-1} \left( (\Delta_iM)^2 \right)^{\frac{1}{2}}\nonumber\\
&\lesssim  \left(v_n^{-Y(1+\eta)}\Delta_n^2 v_n^{4-Y}\right)^{1/2}\left(\sum_{i=1}^n K^2_{m_n\Delta_n}(t_{i-1}-\tau){|\chi_{t_{i-1}}|^{2Y}}\sum_{i=1}^{n}\EE_{i-1} \left( (\Delta_iM)^2 \right)\right)^{\frac{1}{2}}\nonumber\\
&\lesssim  \left(v_n^{-Y(1+\eta)}\Delta_n^2 v_n^{4-Y}\right)^{1/2}\frac{1}{\sqrt{m_n}\Delta_n}\times \|K\|_2|\chi_\tau|^{Y}.\nonumber
\end{align*}
Therefore, the contribution of the term $Y_s^2$ to $F_n$ in the case \( |\chi_{t_{i-1}} Y_s| \leq v_n \) is of the order $O_P(\Delta_n^{1/2} v_n^{-Y(1+\eta)/2})$. This tends to zero under our assumption $v_n\gg \Delta_n^{\beta}$ with small enough $\eta$. Next, we focus on $2Y_sx$. {Using again \eqref{BNDFYC} and Cauchy-Schwarz},
\begin{align*}
    &\sum_{i=1}^n K_{m_n\Delta_n}(t_{i-1}-\tau)\chi_{t_{i-1}}^2{\bf 1}_{{\chi_{t_{i-1}}}\neq 0}\EE_{i-1} \left[ \int_0^{\Delta_n} \int |\delta_{i-1}(s, x)||Y_s x| 1_{\{|{\chi_{t_{i-1}}}Y_s| \leq v_n,|{\chi_{t_{i-1}}}x| \leq Cv_n\}} \nu(dx) ds \right]\nonumber\\
    &\leq \sum_{i=1}^n K_{m_n\Delta_n}(t_{i-1}-\tau)\chi_{t_{i-1}}^2\times {\bf 1}_{{\chi_{t_{i-1}}}\neq 0}\times \EE_{i-1} \left( \int_0^{\Delta_n} \int |\delta_{i-1}(s, x)|^2 \nu(dx) ds \right)^{\frac{1}{2}} \nonumber\\
    &\qquad\times \EE_{i-1} \left( \int_0^{\Delta_n} \int Y_s^2 x^2 1_{\{|{\chi_{t_{i-1}}}Y_s| \leq v_n\}} 1_{\{|{\chi_{t_{i-1}}}x| \leq C v_n\}} \nu(dx) ds \right)^{\frac{1}{2}}\nonumber\\
    &\lesssim \left( \Delta_n^2 v_n^{2(2 - Y)} \right)^{\frac{1}{2}} \sum_{i=1}^n K_{m_n\Delta_n}(t_{i-1}-\tau){\bf 1}_{{\chi_{t_{i-1}}}\neq 0}{|\chi_{t_{i-1}}|^{Y}}\EE_{i-1} \left[ \int_{t_{i-1}}^{t_i} \int |\delta(u, x)|^2 \nu(dx) du \right]^{\frac{1}{2}}\nonumber\\
   & \lesssim \left( \Delta_n^2 v_n^{2(2 - Y)} \right)^{\frac{1}{2}}\left(\sum_{i=1}^n K^2_{m_n\Delta_n}(t_{i-1}-\tau){\bf 1}_{{\chi_{t_{i-1}}}\neq 0}{|\chi_{t_{i-1}}|^{2Y}}\sum_{i=1}^{n}\EE_{i-1} \left( (\Delta_iM)^2 \right)\right)^{\frac{1}{2}}
    \\
    & \lesssim \left( \Delta_n^2 v_n^{2(2 - Y)} \right)^{\frac{1}{2}}\frac{\|K\|_2}{\sqrt{m_n}\Delta_n}\times {|\chi_{\tau}|^{Y}}.
\end{align*}
Therefore, after multiplying by $\frac{m_n^{1/2}}{\Delta_n^{-1/2} v_n^{2-Y/2}} $, the contribution of this term to $F_n$ is of order $\Delta_n^{1/2}v_n^{-Y/2}$, which is \(o_P(1)\) since \( v_n \gg \Delta_n^{\frac{1}{2}} \).

{It remains to analyze the term corresponding to \(x^2\)}. %
For any sequence \(\alpha_n > 0\) such that \(\alpha_n \to 0\),
\begin{align*}
    &\sum_{i=1}^n K_{m_n\Delta_n}(t_{i-1}-\tau)\chi_{t_{i-1}}^2{\bf 1}_{{\chi_{t_{i-1}}}\neq 0}\EE_{i-1} \left[ \int_0^{\Delta_n} \int_{|\delta_{i-1}(s,x)| \leq \alpha_n} |\delta_{i-1}(s, x)| x^2 1_{\{|{\chi_{t_{i-1}}}x| \leq C v_n\}} \nu(dx) ds \right] \nonumber\\
    & \leq\sum_{i=1}^n K_{m_n\Delta_n}(t_{i-1}-\tau)\chi_{t_{i-1}}^2\times {\bf 1}_{{\chi_{t_{i-1}}}\neq 0}\times \EE_{i-1} \left[ \int_0^{\Delta_n} \int_{|{\chi_{t_{i-1}}}x| \leq C v_n} x^4 \nu(dx) ds \right]^{\frac{1}{2}} \nonumber\\
    &\qquad\qquad\qquad\qquad\times \EE_{i-1} \left[ \int_0^{\Delta_n} \int_{|\delta_{i-1}(s,x)| \leq \alpha_n} |\delta_{i-1}(s, x)|^2 \nu(dx) ds \right]^{\frac{1}{2}}\nonumber\\
    &\lesssim  \left( \Delta_n v_n^{4 - Y} \right)^{\frac{1}{2}} \sum_{i=1}^n K_{m_n\Delta_n}(t_{i-1}-\tau){\bf 1}_{{\chi_{t_{i-1}}}\neq 0}{|\chi_{t_{i-1}}|^{Y/2}}\EE_{i-1} \left[ \int_{t_{i-1}}^{t_i} \int_{|\delta(s,x)| \leq \alpha_n} |\delta(s, x)|^2 \nu(dx) ds \right]^{1/2}\nonumber\\
    &\lesssim \left( \Delta_n v_n^{4 - Y} \right)^{\frac{1}{2}} \Big[\sum_{i=1}^n K^2_{m_n\Delta_n}(t_{i-1}-\tau){|\chi_{t_{i-1}}|^{Y}}\sum_{i=1}^n\EE_{i-1} \Big[ \int_{t_{i-1}}^{t_i} \int_{|\delta(s,x)| \leq \alpha_n} |\delta(s, x)|^2 \nu(dx) ds \Big]\Big]^{1/2}\nonumber\\
    & \lesssim \frac{\left( \Delta_n v_n^{4 - Y} \right)^{\frac{1}{2}}}{\sqrt{m_n}\Delta_n}{\left [\sum_{i=1}^n\EE_{i-1} \Big[ \int_{t_{i-1}}^{t_i} \int_{|\delta(s,x)| \leq \alpha_n} |\delta(s, x)|^2 \nu(dx) ds \Big]\right]^{1/2}} = o_P\left(\frac{\Delta_n^{-1/2} v_n^{2-Y/2}}{m_n^{1/2}}\right),
\end{align*}
since as $\alpha_n\to 0$,
\begin{align*}
    &\EE \left[\sum_{i=1}^n\EE_{i-1} \left( \int_{t_{i-1}}^{t_i} \int_{|\delta(s,x)| \leq \alpha_n} |\delta(s, x)|^2 \nu(dx) ds \right) \right] =   \EE \left[\int_{0}^{T} \int_{|\delta(s,x)| \leq \alpha_n} |\delta(s, x)|^2 \nu(dx) ds \right]  \to 0. 
\end{align*}
{On the complementary region, we proceed as follows:}
\begin{align*}
    &\sum_{i=1}^n K_{m_n\Delta_n}(t_{i-1}-\tau)\chi_{t_{i-1}}^2{\bf 1}_{{\chi_{t_{i-1}}}\neq 0}\EE_{i-1} \left[ \int_0^{\Delta_n} \int_{|\delta_{i-1}(s,x)| > \alpha_n} |\delta_{i-1}(s, x)| x^2 1_{\{|{\chi_{t_{i-1}}}x| \leq C v_n\}} \nu(dx) ds \right]\nonumber\\
    &\leq \frac{C}{\alpha_n} v_n^2\sum_{i=1}^n K_{m_n\Delta_n}(t_{i-1}-\tau){\bf 1}_{{\chi_{t_{i-1}}}\neq{}0}\EE_{i-1} \left[ \int_0^{\Delta_n} \int_{|\delta_{i-1}(s,x)| > \alpha_n} |\delta_{i-1}(s, x)|^2 \nu(dx) ds \right]\nonumber\\
    &\leq \frac{C}{\alpha_n}  v_n^2\sum_{i=1}^n K_{m_n\Delta_n}(t_{i-1}-\tau)\EE_{i-1} \left[ \int_{t_{i-1}}^{t_i} \int |\delta(s, x)|^2 \nu(dx) ds \right]\nonumber\\
    &\lesssim \frac{C}{\alpha_n}  v_n^2 \frac{1}{m_n\Delta_n} =  %
    o_P\left(\frac{\Delta_n^{-1/2} v_n^{2-Y/2}}{m_n^{1/2}}\right),
\end{align*}
under our assumption $m_n\Delta_n \gg v_n^{\gamma}$ with $\gamma < Y$, since
$\frac{v_n^{\frac{Y}{2}}}{\alpha_n m_n^{1/2}\Delta_n^{1/2}} \ll \frac{v_n^{\frac{Y-\gamma}{2}}}{\alpha_n}$ and we may take \(\alpha_n \to 0\) such that \(\frac{v_n^{(Y-\gamma)/2}}{\alpha_n} \to 0\). The proof for the negligibility of \(F_n\) in the case \(|\chi_{t_{i-1}}Y_s| \leq v_n\) is complete.

For the case of \(|\chi_{t_{i-1}}Y_s| \geq \zeta v_n\), we have \(f_n(\chi_{t_{i-1}}Y_s) = 0\) and \(t_n(Y_s, x) = (Y_s + x)^2 f_n(\chi_{t_{i-1}}(Y_s + x)) \leq (\zeta v_n)^2/\chi_{t_{i-1}}^2\). We observe that when \(f_n(\chi_{t_{i-1}}(Y_s + x)) > 0\), we have \(|Y_s + x| \leq \zeta v_n(1 - \frac{1}{3}v_n^\eta)/|\chi_{t_{i-1}}|\), implying \(|\chi_{t_{i-1}}x| > \zeta v_n^{1 + \eta}/3\). We then may write
\begin{align*}
    &\frac{m_n^{1/2}}{\Delta_n^{-1/2} v_n^{2-Y/2}}\sum_{i=1}^n K_{m_n\Delta_n}(t_{i-1}-\tau)\chi_{t_{i-1}}^2\nonumber\\
    &\quad\times\EE_{i-1} \Bigg[ \int_0^{\Delta_n} \int |\delta_{i-1}(s, x)| (Y_s + x)^2\\
    & \qquad \qquad \qquad\qquad\qquad f_n\big(\chi_{t_{i-1}}(Y_s + x)\big) 1_{\{|\chi_{t_{i-1}}Y_s| \geq \zeta v_n\}} 1_{\{|\chi_{t_{i-1}}x| \geq C v_n^{1 + \eta}\}} \nu(dx) ds \Bigg]\nonumber\\
    &\lesssim \frac{m_n^{1/2}}{\Delta_n^{-1/2} v_n^{2-Y/2}} v_n^2\sum_{i=1}^n K_{m_n\Delta_n}(t_{i-1}-\tau)\EE_{i-1} \left[ \int_0^{\Delta_n} \int |\delta_{i-1}(s, x)|^2 \nu(dx) ds \right]^{\frac{1}{2}} \nonumber\\
    &\qquad\qquad\qquad\qquad\qquad\times\EE_{i-1} \bigg[ \int_0^{\Delta_n} \int 1_{\{|\chi_{t_{i-1}}Y_s| \geq \zeta v_n\}} 1_{\{|\chi_{t_{i-1}}x| \geq C v_n^{1 + \eta}\}} \nu(dx) ds \bigg]^{\frac{1}{2}}\nonumber\\
     &\lesssim  \frac{m_n^{1/2}}{\Delta_n^{-1/2} v_n^{2-Y/2}} v_n^2 \left( \Delta_n^2 v_n^{-Y - (1 + \eta)Y} \right)^{\frac{1}{2}}\sum_{i=1}^n K_{m_n\Delta_n}(t_{i-1}-\tau)\EE_{i-1} \left[ \int_{t_{i-1}}^{t_i} \int |\delta(u, x)|^2 \nu(dx) du \right]^{\frac{1}{2}}\\&\stackrel{\eqref{kernel_sqrt_m}}{\lesssim}
    \Delta_n^{\frac{1}{2}} v_n^{-\frac{(1 + \eta)Y}{2}},
\end{align*}
which vanishes by taking small \(\eta > 0\) due to \(v_n \gg \Delta_n^\beta\). {We have then proved the negligibility of \(F_n\) in the case \(|\chi_{t_{i-1}}Y_s| \geq \zeta v_n\).}

For the case when \(v_n(1 + v_n^\eta) < |\chi_{t_{i-1}}Y_s| < \zeta v_n(1 - v_n^\eta)\), when \(|\chi_{t_{i-1}}x| \geq \frac{1}{3} v_n^{1+\eta}\), we may use  the bound \(|t_n(Y_s, x)| \leq C v_n^2/\chi_{t_{i-1}}^2\) and follow a similar proof as above. When \(|\chi_{t_{i-1}}x| \leq \frac{1}{3} v_n^{1+\eta}\), \(f_n(\chi_{t_{i-1}}(x + Y_s)) = 1 = f_n(\chi_{t_{i-1}}Y_s)\), implying \(t_n(Y_s, x) = x^2 + 2xY_s\). We consider the contribution of each term separately below. For the case of \(x^2\), we have:
\begin{align*}
    &\frac{m_n^{1/2}}{\Delta_n^{-1/2} v_n^{2-Y/2}}\sum_{i=1}^n K_{m_n\Delta_n}(t_{i-1}-\tau)\chi_{t_{i-1}}^2\nonumber\\
    &\qquad\times\EE_{i-1} \left[ \int_0^{\Delta_n} \int |\delta_{i-1}(s, x)| x^2 1_{v_n(1 + v_n^\eta) < |\chi_{t_{i-1}}Y_s| < \zeta v_n(1 - v_n^\eta)} 1_{\{|\chi_{t_{i-1}}x| \leq C v_n^{1 + \eta}\}} \nu(dx) ds \right]\nonumber\\
    &\leq \frac{m_n^{1/2}}{\Delta_n^{-1/2} v_n^{2-Y/2}}\sum_{i=1}^n K_{m_n\Delta_n}(t_{i-1}-\tau)\chi_{t_{i-1}}^2 \EE_{i-1} \left[ \int_0^{\Delta_n} \int |\delta_{i-1}(s, x)|^2 \nu(dx) ds \right]^{\frac{1}{2}}\nonumber\\
    &\qquad\times E_{i-1} \left[ \int_0^{\Delta_n} \int x^4 1_{\{|\chi_{t_{i-1}}Y_s| \geq v_n,|\chi_{t_{i-1}}x| \leq C v_n^{1 + \eta}\}}\nu(dx) ds \right]^{\frac{1}{2}}\nonumber\\
    &\lesssim \frac{m_n^{1/2}}{\Delta_n^{-1/2} v_n^{2-Y/2}} \left( v_n^{(4 - Y)(1 + \eta)} \Delta_n^2 v_n^{-Y} \right)^{\frac{1}{2}} \sum_{i=1}^n K_{m_n\Delta_n}(t_{i-1}-\tau)\left( E_{i-1}[(\Delta_i M)^2] \right)^{\frac{1}{2}}
    \lesssim v_n^{\frac{\eta (4 - Y) - Y}{2}} \Delta_n^{\frac{1}{2}},\nonumber
\end{align*}
which vanishes 
for all \(\eta > 0\), where the last line follows from \eqref{kernel_sqrt_m}. For the term corresponding to \(2|x Y_s|\), since \(|\chi_{t_{i-1}}Y_s| < \zeta v_n(1 - v_n^\eta)\), for some \(C\),
\begin{align*}
    &\frac{m_n^{1/2}}{\Delta_n^{-1/2} v_n^{2-Y/2}}\sum_{i=1}^n K_{m_n\Delta_n}(t_{i-1}-\tau)\chi_{t_{i-1}}^2\nonumber\\
   & \qquad\times\EE_{i-1} \left[ \int_0^{\Delta_n} \int |\delta_{i-1}(s, x)| |x Y_s| 1_{v_n(1 + v_n^\eta) < |\chi_{t_{i-1}}Y_s| < \zeta v_n(1 - v_n^\eta)} 1_{\{|\chi_{t_{i-1}}x| \leq C v_n^{1 + \eta}\}} \nu(dx) ds \right]\nonumber\\
   & \leq \frac{m_n^{1/2}}{\Delta_n^{-1/2} v_n^{2-Y/2}}\sum_{i=1}^n K_{m_n\Delta_n}(t_{i-1}-\tau)\chi_{t_{i-1}}^2\EE_{i-1} \left[ \int_0^{\Delta_n} \int |\delta_{i-1}(s, x)|^2 \nu(dx) ds \right]^{\frac{1}{2}} \nonumber\\
    &\qquad\qquad\times\EE_{i-1} \left[ \int_0^{\Delta_n} \int x^2 Y_s^2 1_{\{|\chi_{t_{i-1}}Y_s| \leq \zeta v_n\}}1_{\{|\chi_{t_{i-1}}x| \leq C v_n^{1 + \eta}\}} \nu(dx) ds \right]^{\frac{1}{2}}\nonumber\\
   & \lesssim \frac{m_n^{1/2}}{\Delta_n^{-1/2} v_n^{2-Y/2}} \left( v_n^{(2 - Y)(1 + \eta)} \Delta_n^2 v_n^{2 - Y} \right)^{\frac{1}{2}} \sum_{i=1}^n K_{m_n\Delta_n}(t_{i-1}-\tau) \left( \EE_{i-1}[(\Delta_i M)^2] \right)^{\frac{1}{2}} \nonumber\\
    &= v_n^{-\frac{Y}{2} + \frac{\eta (2 - Y)}{2}} \Delta_n^{\frac{1}{2}} O_P(1) = o_P(1),
\end{align*}
where again the last line follows from the boundedness of $\chi$ and \eqref{kernel_sqrt_m} for all \(\eta > 0\). This proves the negligibility of \(F_n\) for the case \(v_n(1 + v_n^\eta) < |\chi_{t_{i-1}}Y_s| < \zeta v_n(1 - v_n^\eta)\).

The only cases left are when \(v_n \leq |\chi_{t_{i-1}}Y_s| \leq v_n(1+v_n^\eta)\) and \(\zeta v_n(1-v_n^\eta) \leq |\chi_{t_{i-1}}Y_s| \leq \zeta v_n\). It suffices to consider the first case as the other case is proved similarly. If \(|\chi_{t_{i-1}}x| \geq v_n^{1 + \eta}\), consider the same argument as in (D.26) in \cite{BONIECE2024supplement} and suppose that \(|\chi_{t_{i-1}}x| \leq v_n^{1 + \eta}\). %
Note that, for $\chi_{t_{i-1}}\neq0$, $|p'_n(y)| \leq 2|y f_n(\chi_{t_{i-1}}y)| +|\chi_{t_{i-1}}y^2 f'_n(\chi_{t_{i-1}}y)| \leq \frac{C v_n^{1-\eta}}{|\chi_{t_{i-1}}|}.$
Consequently, $|t_n(Y_s,x)|\leq \frac{C v_n^{1-\eta}}{|\chi_{t_{i-1}}|}|x|.$ Thus, \(F_n\) is bounded by
\begin{align}\nonumber
    &\frac{m_n^{1/2}v_n^{1-\eta}}{\Delta_n^{-1/2}v_n^{2-Y/2}}  \sum_{i=1}^n K_{m_n\Delta_n}(t_{i-1}-\tau)|\chi_{t_{i-1}}|\\
    &\qquad \times \EE_{i-1} \left[ \int_0^{\Delta_n} \int |\delta_{i-1}(s, x)||x| 1_{{v_n \leq |\chi_{t_{i-1}}Y_s| \leq v_n(1 + v_n^\eta)}} 1_{{|\chi_{t_{i-1}}x| \leq v_n^{1 + \eta}}} \nu(dx) ds \right]\nonumber\\
   & \leq \frac{m_n^{1/2}v_n^{1-\eta}}{\Delta_n^{-1/2}v_n^{2-Y/2}}  \sum_{i=1}^n K_{m_n\Delta_n}(t_{i-1}-\tau)|\chi_{t_{i-1}}|\EE_{i-1} \left[ \int_0^{\Delta_n} \int |\delta_{i-1}(s, x)|^2 \nu(dx) ds \right]^{\frac{1}{2}} \nonumber\\
    &\qquad\times\EE_{i-1} \left[ \int_0^{\Delta_n} \int_{|\chi_{t_{i-1}}x| \leq C v_n^{1 + \eta}} x^2 1_{{v_n \leq |\chi_{t_{i-1}}Y_s| \leq v_n(1 + v_n^\eta)}} \nu(dx) ds \right]^{\frac{1}{2}}\nonumber\\
    & \lesssim \frac{m_n^{1/2}v_n^{1-\eta}}{\Delta_n^{-1/2}v_n^{2-Y/2}} \left( v_n^{(2 - Y)(1 + \eta)} \Delta_n^2 v_n^{\eta - Y} \right)^{\frac{1}{2}} \sum_{i=1}^n K_{m_n\Delta_n}(t_{i-1}-\tau)
|\chi_{t_{i-1}}|^Y
\left(\EE_{i-1}[(\Delta_iM)^2]\right)^{1/2}\nonumber\\
&\lesssim
\left(\sup_{0\leq t\leq1}|\chi_t|^Y\right)
\Delta_n^{1/2}v_n^{-Y/2+\eta(1-Y)/2}O_P(1)
\nonumber\\
&=\Delta_n^{1/2}v_n^{-Y/2+\eta(1-Y)/2}O_P(1).
\end{align}
Here the last line follows from the boundedness of $\chi$ and \eqref{kernel_sqrt_m}.
Since $v_n\gg\Delta_n^\beta$ with $\beta<1/2$ and $1<Y<2$, we may choose $\eta>0$ sufficiently small that
$\beta\big(Y+\eta(Y-1)\big)<1$. The final bound is therefore $o_P(1)$. This concludes the proof.
\end{proof}
\begin{proof}[Proof of \eqref{diff_indicator_approx}]
We start by noting that 
\begin{align}\label{diff_indicator_largejump}
        &\frac{m_n^{1/2}}{\Delta_n^{-1/2} v_n^{2-Y/2}} \sum_{i=1}^n \bigg(K_{m_n\Delta_n}(t_{i-1}-\tau)\\
        & \qquad \qquad \times \,\chi_{t_{i-1}}^2 \EE_{i-1}\left[ (\Delta_i J)^2\mathbf{1}_{|\chi_{t_{i-1}}\Delta_iJ|>4\zeta v_n} \left| f_n(\Delta_i X) - f_n(\chi_{t_{i-1}} \Delta_i J) \right| |\Delta_i M| \right]\bigg)=o_P(1).
    \end{align}
    Indeed, note that    \begin{align*}
        \mathbf{1}_{|\chi_{t_{i-1}}\Delta_iJ|>4\zeta v_n} \left| f_n(\Delta_i X) - f_n(\chi_{t_{i-1}} \Delta_i J) \right| &\leq \mathbf{1}_{|\chi_{t_{i-1}}\Delta_iJ|>4\zeta v_n} f_n(\Delta_i X)\\
        &\leq f_n(\Delta_iX)
\mathbf{1}_{\{|\mu_{t_{i-1}}\Delta_n
+\sigma_{t_{i-1}}\Delta_iW+\varepsilon_i|
>3\zeta v_n\}}.
    \end{align*}
    On the support of $f_n(\Delta_iX)$,
$$
|\chi_{t_{i-1}}\Delta_iJ|
\leq \zeta v_n+
|\mu_{t_{i-1}}\Delta_n
+\sigma_{t_{i-1}}\Delta_iW+\varepsilon_i|.
$$
Consequently, for sufficiently large $n$,
$$
|\chi_{t_{i-1}}\Delta_iJ|^2f_n(\Delta_iX)
\leq C\left(v_n^2+|\Delta_iW|^2+|\varepsilon_i|^2\right).
$$
    Thus, by H\"older's and Markov's inequalities along with $v_n\gg \Delta_n^{\beta}$ and \eqref{KME0}, for large enough $r$, we have that
    \begin{align*}
        &\frac{m_n^{1/2}}{\Delta_n^{-1/2} v_n^{2-Y/2}} \sum_{i=1}^n K_{m_n\Delta_n}(t_{i-1}-\tau) \chi_{t_{i-1}}^2\\
        &\qquad\times \EE_{i-1}\left[ (\Delta_i J)^2 f_n(\Delta_iX)|\Delta_i M| (\mathbf{1}_{|\mu_{t_{i-1}}\Delta_n|>\zeta v_n}+\mathbf{1}_{|\sigma_{t_{i-1}}\Delta_iW|>\zeta v_n}+\mathbf{1}_{|\varepsilon_i|>\zeta v_n})\right],
    \end{align*}
is $o_P(1)$,    
    which implies \eqref{diff_indicator_largejump}.
    Consequently, it suffices to prove that
\begin{align*}%
        &\frac{m_n^{1/2}}{\Delta_n^{-1/2} v_n^{2-Y/2}} \sum_{i=1}^n \bigg(K_{m_n\Delta_n}(t_{i-1}-\tau)\\
        &\qquad\qquad\times\,\chi_{t_{i-1}}^2\EE_{i-1}\left[ (\Delta_i J)^2\mathbf{1}_{|\chi_{t_{i-1}}\Delta_iJ|\leq4\zeta v_n} \left| f_n(\Delta_i X) - f_n(\chi_{t_{i-1}} \Delta_i J) \right| |\Delta_i M| \right]\bigg)=o_P(1). 
    \end{align*}
First, we observe that
\begin{align}\label{diff_indicatorapprox_decomp}
        |f_n(x+y)-f_n(y)| \leq \mathbf{1}_{|x|\geq v_n^{1+\eta}/3} + \mathbf{1}_{v_n\leq |y| \leq v_n(1+v_n^\eta) \text{ or } \zeta v_n(1-v_n^\eta)\leq |y| \leq \zeta v_n}C \frac{|x|}{v_n^{1+\eta}},
    \end{align}
    where $y=\chi_{t_{i-1}}\Delta_iJ$ and $x = \Delta_iX - \chi_{t_{i-1}}\Delta_iJ = \mu_{t_{i-1}}\Delta_n+\sigma_{t_{i-1}}\Delta_iW+\varepsilon_i =: z_{i,1}+z_{i,2}+z_{i,3}$.

    We start with the first term in \eqref{diff_indicatorapprox_decomp}.  We proceed to verify that
    $$B_l := \frac{m_n^{1/2}}{\Delta_n^{-1/2} v_n^{2-Y/2}} \sum_{i=1}^n K_{m_n\Delta_n}(t_{i-1}-\tau)\chi_{t_{i-1}}^2 \EE_{i-1}\left[ (\Delta_i J)^2\mathbf{1}_{|\chi_{t_{i-1}}\Delta_iJ|\leq4\zeta v_n} \mathbf{1}_{|z_{i,l}|\geq v_n^{1+\eta}/9} |\Delta_i M| \right],$$
    is $o_P(1)$ for $l=1,2,3$. Take the case $l=2$ for instance. By H\"older's inequality we can bound $B_2$ by
    \begin{align*}
    &\frac{m_n^{1/2}}{\Delta_n^{-1/2} v_n^{2-Y/2}} \sum_{i=1}^n K_{m_n\Delta_n}(t_{i-1}-\tau)\chi_{t_{i-1}}^2\EE_{i-1}\big(\Delta_iJ^8\mathbf{1}_{|\chi_{t_{i-1}}\Delta_iJ|\leq4\zeta v_n}\big)^{1/4}\\
    &\qquad\qquad \times \PP_{i-1}\big(|\sigma_{t_{i-1}}\Delta_iW|\geq v_n^{1+\eta}/9\big)^{1/4}\EE_{i-1}\big(|\Delta_iM|^2\big)^{1/2}.
    \end{align*}
    Notice that for any $\alpha>0$, for some constant $C$, we have $\PP_{i-1}(|\sigma_{t_{i-1}}\Delta_iW|\geq v_n^{1+\eta}/9)\lesssim \EE(|\Delta_iW|^\alpha)/v_n^{\alpha(1+\eta)}$, which is bounded above by $\Delta^{\alpha/2}_n/v_n^{\alpha(1+\eta)}$. Therefore, along with Proposition \ref{prop} and \eqref{kernel_sqrt_m}, we have
    \begin{align*}
        B_2&\lesssim \frac{m_n^{1/2}}{\Delta_n^{-1/2} v_n^{2-Y/2}}\Delta_n^{1/4}v_n^{(8-Y)/4}\Delta^{\alpha/8}_n/v_n^{\alpha(1+\eta)/4} \sum_{i=1}^n K_{m_n\Delta_n}(t_{i-1}-\tau)|\chi_{t_{i-1}}|^{Y/4}\EE_{i-1}(|\Delta_iM|^2)^{1/2}\\
        &\lesssim
\Delta_n^{-1/4}v_n^{Y/4}
\left(\frac{\Delta_n^{1/2}}{v_n^{1+\eta}}\right)^{\alpha/4}.%
    \end{align*}
    Since $v_n\gg \Delta_n^{\beta}$, for small enough $\eta>0$, $\frac{1}{2}-\beta(1+\eta)>0$. Then by taking $\alpha$ large enough we  obtain $B_2 = o_P(1)$. $B_1$ can be shown $o_P(1)$ using the similar argument. For $B_3$, by \eqref{KME0}, $ B_3 $ is dominated by
    \begin{align*}
       &\frac{m_n^{1/2}}{\Delta_n^{-1/2} v_n^{2-Y/2}} \sum_{i=1}^n K_{m_n\Delta_n}(t_{i-1}-\tau)\chi_{t_{i-1}}^2\\
       &\qquad \times \EE_{i-1}\big(\Delta_iJ^8\mathbf{1}_{|\chi_{t_{i-1}}\Delta_iJ|\leq4\zeta v_n}\big)^{1/4}\EE_{i-1}\big(|\varepsilon_i|^{\alpha}\big)^{1/4}/v_n^{\alpha(1+\eta)/4}\EE_{i-1}\big(|\Delta_iM|^2\big)^{1/2}\\
        &\lesssim \frac{m_n^{1/2}}{\Delta_n^{-1/2} v_n^{2-Y/2}}\Delta_n^{1/4}v_n^{(8-Y)/4}\frac{\Delta^{(1+\alpha/2)/4}_n}{v_n^{\alpha(1+\eta)/4}} \sum_{i=1}^n K_{m_n\Delta_n}(t_{i-1}-\tau)|\chi_{t_{i-1}}|^{Y/4}\EE_{i-1}\big(|\Delta_iM|^2\big)^{1/2},
    \end{align*}
    which tends to zero in probability by taking $\eta$ small enough and $\alpha$ large enough similar to $B_2$. Next, we analyze the second term in \eqref{diff_indicatorapprox_decomp}. We only  consider $\{v_n\leq|y|\leq v_n(1+v_n^{\eta})\}$ as the other case is similar. Let $a_n=\frac{m_n^{1/2}}{\Delta_n^{-1/2} v_n^{2-Y/2+1+\eta}}$. We need to show that
    $$C_l := a_n \sum_{i=1}^n K_{m_n\Delta_n}(t_{i-1}-\tau) \chi_{t_{i-1}}^2\EE_{i-1}\left[ (\Delta_i J)^2\mathbf{1}_{v_n\leq|\chi_{t_{i-1}}\Delta_iJ|\leq v_n(1+v_n^{\eta})} |\Delta_i M||z_{i,l}| \right]$$
is $o_P(1)$ for $l=1,2,3$. %
By \eqref{kernel_sqrt_m}, H\"older's inequality and Lemma 17 in \cite{BONIECE2024}, for any $p,q>1$ such that $1/p+1/q=1/2$, $C_2$ is dominated by
    \begin{align*}
         &a_n \sum_{i=1}^n K_{m_n\Delta_n}(t_{i-1}-\tau)\chi_{t_{i-1}}^2\\
         &\qquad \times \EE_{i-1}\big(|\Delta_iJ|^{2p}\mathbf{1}_{v_n\leq|\chi_{t_{i-1}}\Delta_iJ|\leq v_n(1+v_n^{\eta})}\big)^{1/p}\EE_{i-1}\big(|\sigma_{t_{i-1}}\Delta_iW|^q\big)^{1/q}\EE_{i-1}(|\Delta_iM|^2)^{1/2}\\
        &\lesssim a_n \sum_{i=1}^n K_{m_n\Delta_n}(t_{i-1}-\tau)|\chi_{t_{i-1}}|^{Y/p}|\sigma_{t_{i-1}}|(\Delta_nv_n^{2p-Y+\eta})^{1/p}\Delta_n^{1/2}\EE_{i-1}(|\Delta_iM|^2)^{1/2}\\
        &\lesssim\frac{(\Delta_nv_n^{2p-Y+\eta})^{1/p}}{v_n^{2-Y/2+1+\eta}}.
    \end{align*}
    Taking $p$ close to $2$ such that $1/p = 1/2-s$ for some small $s>0$, the last line has order $$ \frac{(\Delta_n v_n^{2p-Y+\eta})^{1/p}}   {v_n^{2-Y/2+1+\eta}}= \Delta_n^{1/2-s} v_n^{-1-\eta/2+s(Y-\eta)} \leq \Delta_n^{1/2-s}v_n^{-1-\eta/2}, $$
    which vanishes since $v_n\gg \Delta_n^{\beta}$ and $\Delta_n^{1/2-s-(1+\eta/2)\beta}\ll 1$ given $\eta$ and $s$ are small enough. A similar argument works for $C_1$ and $C_3$ and establishes \eqref{diff_indicator_approx}.
\end{proof}

\subsection{Proof of \eqref{GlobalLocalNoiseOrderPre}}\label{PrfNwEsthPhi}
We give the proof for $p=1$. The general case follows by replacing $v_n$ with $pv_n$, since $p>0$ is fixed and this multiplies $u_n$ by the constant $p^{2-Y/2}$. Let 
 \begin{align*}
     \check{A}^{(d)}_n(v) 
     &:=v^{2-Y}\sum_{i=1}^n \frac{\bar{C}}{2-Y}\Delta_n|\chi_{t_{i-1}}|^{Y} - \Delta_nv^{-Y}\frac{\bar{C}(Y+1)(Y+2)}{2Y} \sum_{i=1}^n\Delta_n|\chi_{t_{i-1}}|^Y\sigma^2_{t_{i-1}}.
 \end{align*}
 In \cite{BONIECE2024supplement} (Lemma 5) it was shown that: 
 \[
   \check{A}_n(v_n)- \check{A}^{(d)}_n(v_n)=o_P\left(\Delta_n^{1/2}v_n^{2-Y}\right)
   +o_P\left(\Delta_n^{3/2}v_n^{-Y}\right),
 \]
 and both terms on the right hand-side are $o_P(v_n^{(4-Y)/2})$. Note that 
 \begin{equation*}
 \begin{aligned}
     \check\Phi_n&:= u_n^{-1}\left(\check{Z}_n(\zeta pv_n)-\check{Z}_n(pv_n) \right)
\end{aligned}
\end{equation*}
can be written as
\begin{align*}
\check\Phi_n&=v_n^{\frac{Y-4}{2}}\sum_{i=1}^{n}\left(\big(\Delta_{i}^{n}X\big)^{2}\,{\bf 1}_{\{v_n<|\Delta_{i}^{n}X|\leq\zeta v_n\}}- \bE_{i-1}\big[\big(\Delta_{i}^{n}X\big)^{2}\,{\bf 1}_{\{v_n<|\Delta_{i}^{n}X|\leq\zeta v_n\}}\big] \right) \\
    &\quad+v_n^{\frac{Y-4}{2}}\,\left(\sum_{i=1}^{n}\bE_{i-1}\big[\big(\Delta_{i}^{n}X\big)^{2}\,{\bf 1}_{\{v_n<|\Delta_{i}^{n}X|\leq\zeta v_n\}}\big]-\left(\check{A}^{(d)}_n(\zeta v_n) -\check{A}^{(d)}_n(v_n)\right)\right) \\
    &\quad+v_n^{\frac{Y-4}{2}}\left(\check{A}^{(d)}_n(\zeta v_n)-\check{A}^{(d)}_n(v_n)-\check{A}_n(\zeta v_n)+\check{A}_n(v_n)\right)\\
    &=: T_0+T_1 + o_P\left(1\right).
\end{align*}
By expansion \eqref{PartExp0_difference},
\begin{align*}
T_1=O_P(\Delta_n^2 v_n^{-4-Y/2})
+o_P(\Delta_n^{-\frac{1}{4}}).
\end{align*}
Since $m_n\sim\theta \Delta_n^{-1/2}$, we obtain $o_P(\Delta_n^{-\frac{1}{4}})=o_P((m_n\Delta_n)^{-1/2})$ and, in light of our assumption $v_n\gg \Delta_n^{\frac{9}{2(8+Y)}}$, we also have $O_P(\Delta_n^2 v_n^{-4-Y/2})=o_P((m_n\Delta_n)^{-1/2})$.

For the term $T_0$, the sum of conditional variances
of each summand satisfies
\begin{align*}
v_n^{Y-4}\sum_{i=1}^n
\var\left(
(\Delta_i^nX)^2
\mathbf{1}_{\{v_n<|\Delta_i^nX|\leq\zeta v_n\}}
\mid\mathcal F_{i-1}\right)&\leq
v_n^{Y-4}\sum_{i=1}^n
\bE_{i-1}\left[
(\Delta_i^nX)^4
\mathbf{1}_{\{v_n<|\Delta_i^nX|\leq\zeta v_n\}}
\right]\\
&\quad=
(\zeta^{4-Y}-1)\Delta_n\sum_{i=1}^n C_{2,i}
+o_P(1)
=O_P(1),
\end{align*}
by \eqref{GNMExp_difference} and boundedness of
$C_{2,i}$ (after localization). Consequently, $T_0=O_P(1)$, and hence $\sqrt{m_n\Delta_n}\,T_0=o_P(1)$.
Together with the preceding bound for $T_1$,
this proves the assertion for $\check\Phi_n$ in
\eqref{GlobalLocalNoiseOrderPre}; the assertion for
$\check\Psi_n$ follows by expressing its second
difference as the difference of two first differences, yielding  \eqref{GlobalLocalNoiseOrderPre}.

\section{Additional proofs for $0<Y<1$}\label{AppFV}

\subsection{Proof of Remark \ref{finitevariation_remark}}\label{AppFVbaseline}

By localization, we may assume that the coefficients and jump sizes are bounded. Since $Y<1$ and $r<Y$, the L\'evy--It\^o decomposition gives
\begin{align}\label{FVincrementbound}
&\EE_{i-1}\left|\Delta_iX-\int_{t_{i-1}}^{t_i}\sigma_s\,dW_s\right|\nonumber\\
&\quad\leq C\Delta_n
+\EE_{i-1}\int_{t_{i-1}}^{t_i}\int|\chi_{s^-}x|N(ds,dx)
+\EE_{i-1}\int_{t_{i-1}}^{t_i}\int|\delta(s,z)|\mathfrak p(ds,dz)
\leq C\Delta_n,
\end{align}
where $N(ds,dx)$ is the jump measure of $J$. Here, we used $\int_{|x|\leq1}|x|\nu(dx)<\infty$ and Assumption \ref{BscAs1}.

For any fixed $a>0$, let
\begin{equation}\label{e:def_Omega_n(a)}
\Omega_n(a):=\left\{\max_{1\leq i\leq n}
\left|\int_{t_{i-1}}^{t_i}\sigma_s\,dW_s\right|
\leq av_n/2\right\}.
\end{equation}
By the BDG and Markov inequalities and boundedness of $\sigma$, for any $p>2$,
\begin{equation}\label{FVcontinuousmax}
\PP(\Omega_n(a)^c) \lesssim v_n^{-p}\sum_{i=1}^n \EE \left |\int_{t_{i-1}}^{t_i} \sigma_s dW_s\right|^p \lesssim \Delta_n^{p/2-1}v_n^{-p},
\end{equation}
which tends to zero when $p(1/2-\beta)>1$. On $\Omega_n(a)$,
\begin{align*}
\left|(\Delta_iX)^2\mathbf{1}_{\{|\Delta_iX|\leq av_n\}}
-\left(\int_{t_{i-1}}^{t_i}\sigma_s\,dW_s\right)^2\right|\leq C_av_n
\left|\Delta_iX-\int_{t_{i-1}}^{t_i}\sigma_s\,dW_s\right|.
\end{align*}
Thus, by \eqref{FVincrementbound},
\begin{align}\label{FVcontinuouscomparison}
&\frac{\sqrt{m_n}}{\bar K_{m_n\Delta_n}(\tau)}
\sum_{i=1}^nK_{m_n\Delta_n}(t_{i-1}-\tau)\nonumber\\
&\qquad\times\EE\left[
\left|(\Delta_iX)^2\mathbf{1}_{\{|\Delta_iX|\leq av_n\}}
-\left(\int_{t_{i-1}}^{t_i}\sigma_s\,dW_s\right)^2\right|
\mathbf{1}_{\Omega_n(a)}\right]\nonumber\\
&\quad\leq C_a\sqrt{m_n}v_n=o(1).
\end{align}
Here $\bar K_{m_n\Delta_n}(\tau)
=\Delta_n\sum_{i=1}^nK_{m_n\Delta_n}(t_{i-1}-\tau)$, and the last convergence follows from \eqref{FVtuning}. Hence,
\begin{equation}\label{FVcontinuous_equivalence}
\sqrt{m_n}\left(
\hat c_{n,\tau}(m_n,av_n)
-\frac{\sum_{i=1}^nK_{m_n\Delta_n}(t_{i-1}-\tau)
\left(\int_{t_{i-1}}^{t_i}\sigma_s\,dW_s\right)^2}
{\bar K_{m_n\Delta_n}(\tau)}
\right)=o_P(1).
\end{equation}
In particular, applying this estimate at $a$ and $b$, for fixed $0<a<b<\infty$, gives
\begin{equation}\label{FVlocaldifferences}
\sqrt{m_n}\big(
\hat c_{n,\tau}(m_n,bv_n)
-\hat c_{n,\tau}(m_n,av_n)
\big)=o_P(1).
\end{equation}

Finally, Theorem~2.1(a) of \cite{Figueroa-Lopez_Wu_2024}, applied to
$\int_0^\cdot\sigma_s\,dW_s$, combined with  \eqref{FVcontinuous_equivalence}, gives the CLT (The normalization does not
affect the limit: by Assumption \ref{asu_kernel} and \eqref{FVtuning},
$
\bar K_{m_n\Delta_n}(\tau)-1=O(m_n^{-1})+o(m_n\Delta_n),$ and $
\sqrt{m_n}\big(\bar K_{m_n\Delta_n}(\tau)-1\big)=o(1). $)
Together with \eqref{FVcontinuous_equivalence}, this proves
\eqref{FVbaselineCLT}.

\subsection{Proof of Remark \ref{finitevariation_GL_remark}}\label{AppFVGL}

We first prove the following analog of the global threshold-difference estimate for $0<Y<1$.

\begin{lemma}\label{FVannularLLN}
Suppose that $0<Y<1$, $v_n\to0$, and $v_n\gg\Delta_n^\beta$ for some $\beta<1/2$. Then, under Assumption \ref{BscAs1}, for any fixed $\zeta>1$,
\begin{equation}\label{FVannularlimit}
\frac{\widehat C_n(\zeta v_n)-\widehat C_n(v_n)}{v_n^{2-Y}}
\xrightarrow{P}
\frac{C_++C_-}{2-Y}(\zeta^{2-Y}-1)
\int_0^1|\chi_s|^Y\,ds.
\end{equation}
\end{lemma}

\begin{proof}
As before, we work under localization. Let $N(ds,dx)$ denote the jump measure of $J$ and $\bar N(ds,dx):=N(ds,dx)-ds\,\nu(dx)$. It suffices to show that the following three terms vanish:
\begin{align}
A_1&:=v_n^{Y-2}\left\{
\sum_{i=1}^n(\Delta_iX)^2
\mathbf{1}_{\{v_n<|\Delta_iX|\leq\zeta v_n\}}\right.\nonumber\\
&\qquad\left.-\int_0^1\int(\chi_sx)^2
\mathbf{1}_{\{v_n<|\chi_sx|\leq\zeta v_n\}}N(ds,dx)
\right\},\label{FVannular_discretization}\\
A_2&:=v_n^{Y-2}\int_0^1\int(\chi_sx)^2
\mathbf{1}_{\{v_n<|\chi_sx|\leq\zeta v_n\}}\bar N(ds,dx),
\label{FVannularmartingale}\\
A_3&:=v_n^{Y-2}\int_0^1\int(\chi_sx)^2
\mathbf{1}_{\{v_n<|\chi_sx|\leq\zeta v_n\}}\nu(dx)\,ds\nonumber\\
&\qquad-\frac{C_++C_-}{2-Y}(\zeta^{2-Y}-1)\int_0^1|\chi_s|^Y\,ds.
\label{FVannularcompensator}
\end{align}
For $A_2$, the isometry and the boundedness of $q$ and $\chi$ give
\[
\EE(A_2^2)
=v_n^{2Y-4}\EE\int_0^1\int(\chi_sx)^4
\mathbf{1}_{\{v_n<|\chi_sx|\leq\zeta v_n\}}\nu(dx)\,ds
\leq Cv_n^Y=o(1).
\]
For $A_3$, a change of variables and $q(x)\to1$ as $x\to0$ yield, for each fixed $s$,
\[
v_n^{Y-2}\int(\chi_sx)^2
\mathbf{1}_{\{v_n<|\chi_sx|\leq\zeta v_n\}}\nu(dx)
\longrightarrow
\frac{C_++C_-}{2-Y}(\zeta^{2-Y}-1)|\chi_s|^Y.
\]
The left-hand side is bounded by $C|\chi_s|^Y$, and is zero if $\chi_s=0$. Dominated convergence then gives $A_3\to0$ almost surely.

It remains to show $A_1=o_P(1)$. Fix a sufficiently small
$\varepsilon>0$ and $\alpha>0$ such that
$\alpha\sup_{s\leq1}|\chi_s|<1/2$ under localization. For each
$i=1,\ldots,n$, define
\begin{align}
\Omega_{i,n}^{(1)}
&:=\left\{\int_{t_{i-1}}^{t_i}\int_{|x|>\alpha v_n}
N(ds,dx)\geq2\right\},\nonumber\\
\Omega_{i,n}^{(2)}
&:=\big(\Omega_{i,n}^{(1)}\big)^c
\cap\left\{\int_{t_{i-1}}^{t_i}\int_{|x|\leq\alpha v_n}
|\chi_{s^-}x|N(ds,dx)>\varepsilon v_n\right\},\nonumber\\
\Omega_{i,n}^{(3)}
&:=\big(\Omega_{i,n}^{(1)}\cup\Omega_{i,n}^{(2)}\big)^c
\cap\left\{\left|\int_{t_{i-1}}^{t_i}\int
\delta(s,z)\mathfrak p(ds,dz)\right|>\varepsilon v_n\right\},\nonumber\\
\Omega_{i,n}^{(0)}
&:=\big(\Omega_{i,n}^{(1)}\cup\Omega_{i,n}^{(2)}
\cup\Omega_{i,n}^{(3)}\big)^c.
\label{FVA1events}
\end{align}
Since $|\chi_sx|<v_n$ whenever $|x|\leq\alpha v_n$, only the jumps
with $|x|>\alpha v_n$ enter the jump integral in $A_1$. By the triangle
inequality,
\begin{equation}\label{FVA1decomposition}
|A_1|\leq A_{1,0}+A_{1,1}+A_{1,2}+A_{1,3},
\end{equation}
where, for $j=0,1,2,3$,
\begin{align*}
A_{1,j}
&:=v_n^{Y-2}\sum_{i=1}^n
\Bigg|(\Delta_iX)^2
\mathbf{1}_{\{v_n<|\Delta_iX|\leq\zeta v_n\}}\\
&\qquad\qquad
-\int_{t_{i-1}}^{t_i}\int
(\chi_sx)^2\mathbf{1}_{\{v_n<|\chi_sx|\leq\zeta v_n\}}
N(ds,dx)\Bigg|\mathbf{1}_{\Omega_{i,n}^{(j)}}.
\end{align*}

We first consider $A_{1,1}$. Each squared term is at most
$\zeta^2v_n^2$, and the number of nonzero summands in the Poisson integral
is at most $\int_{t_{i-1}}^{t_i}\int_{|x|>\alpha v_n}N(ds,dx)$.
Therefore,
\begin{align}
\EE A_{1,1}
&\leq Cv_n^Y\sum_{i=1}^n
\EE\left[
\left(1+\int_{t_{i-1}}^{t_i}\int_{|x|>\alpha v_n}N(ds,dx)\right)
\mathbf{1}_{\Omega_{i,n}^{(1)}}\right]\nonumber\\
&\leq Cv_n^Y n
\left(\Delta_n\int_{|x|>\alpha v_n}\nu(dx)\right)^2
\leq C\Delta_n\alpha^{-2Y}v_n^{-Y}.
\label{FVA1multiple}
\end{align}
The second inequality follows from the bound  $(1+k)\mathbf{1}_{\{k\geq2\}}\leq C k (k-1)$.

On $(\Omega_{i,n}^{(1)})^c$ there is at most
one jump in $(t_{i-1},t_i]$ with $|x|>\alpha v_n$.  Hence,  for $j=2,3$, since
$\Omega_{i,n}^{(j)}\subseteq(\Omega_{i,n}^{(1)})^c$, we have
\begin{align*}
&\Bigg|(\Delta_iX)^2
\mathbf{1}_{\{v_n<|\Delta_iX|\leq\zeta v_n\}}
-\int_{t_{i-1}}^{t_i}\int
(\chi_sx)^2\mathbf{1}_{\{v_n<|\chi_sx|\leq\zeta v_n\}}
N(ds,dx)\Bigg|
\mathbf{1}_{\Omega_{i,n}^{(j)}}\\
&\quad\leq \zeta^2v_n^2
\left(1+\int_{t_{i-1}}^{t_i}\int_{|x|>\alpha v_n}N(ds,dx)\right)
\mathbf{1}_{\Omega_{i,n}^{(j)}}\\
&\quad\leq 2\zeta^2v_n^2\mathbf{1}_{\Omega_{i,n}^{(j)}}.
\end{align*}
By Markov's inequality,
\begin{align}
\EE A_{1,2}
&\leq Cv_n^Y\sum_{i=1}^n
\PP\left(
\int_{t_{i-1}}^{t_i}\int_{|x|\leq\alpha v_n}
|\chi_{s^-}x|N(ds,dx)>\varepsilon v_n\right)\nonumber\\
&\leq \frac{Cv_n^{Y-1}}{\varepsilon}
\EE\int_0^1\int_{|x|\leq\alpha v_n}
|\chi_{s^-}x|\nu(dx)\,ds
\leq C\varepsilon^{-1}\alpha^{1-Y}.
\label{FVA1small}
\end{align}
For $A_{1,3}$, splitting the $\delta$-increment at $\alpha v_n$ and
using Markov's inequality and Assumption \ref{BscAs1}, we obtain
\begin{align}
\EE A_{1,3}
&\leq Cv_n^Y\sum_{i=1}^n
\PP\left(\left|\int_{t_{i-1}}^{t_i}\int
\delta(s,z)\mathfrak p(ds,dz)\right|>\varepsilon v_n\right)\nonumber\\
&\leq Cv_n^Y\sum_{i=1}^n
\PP\left(\int_{t_{i-1}}^{t_i}\int
\mathbf{1}_{\{|\delta(s,z)|>\alpha v_n\}}
\mathfrak p(ds,dz)\geq1\right)\nonumber\\
&\quad+Cv_n^Y\sum_{i=1}^n
\PP\left(\int_{t_{i-1}}^{t_i}\int
|\delta(s,z)|\mathbf{1}_{\{|\delta(s,z)|\leq\alpha v_n\}}
\mathfrak p(ds,dz)>\varepsilon v_n\right)\nonumber\\
&\leq Cv_n^Y\EE\int_0^1\int
\mathbf{1}_{\{|\delta(s,z)|>\alpha v_n\}}\lambda(dz)\,ds\nonumber\\
&\quad+\frac{Cv_n^{Y-1}}{\varepsilon}
\EE\int_0^1\int|\delta(s,z)|
\mathbf{1}_{\{|\delta(s,z)|\leq\alpha v_n\}}\lambda(dz)\,ds\nonumber\\
&\leq C\big(\alpha^{-r}
+\varepsilon^{-1}\alpha^{1-r}\big)v_n^{Y-r}.
\label{FVA1additional}
\end{align}
Here $r<Y$ is as in Assumption \ref{BscAs1}. 

Finally, consider $A_{1,0}$ on $\Omega_n(2\varepsilon)$, as defined in
\eqref{e:def_Omega_n(a)}. Since $Y<1$, we can write
\[
J_t=\left(b-\int_{|x|\leq1}x\,\nu(dx)\right)t
+\int_0^t\int x\,N(ds,dx),
\]
where $b$ is the drift in the L\'evy triplet of $J$. On
$\Omega_{i,n}^{(0)}\cap\Omega_n(2\varepsilon)$,
\begin{align}
\int_{t_{i-1}}^{t_i}\int_{|x|>\alpha v_n}N(ds,dx)\leq1,
\end{align}
and 
\begin{align}
\left|\Delta_iX-
\int_{t_{i-1}}^{t_i}\int_{|x|>\alpha v_n}
\chi_{s^-}xN(ds,dx)\right|&\leq
\left|\int_{t_{i-1}}^{t_i}
\left[\mu_s+\chi_{s^-}
\left(b-\int_{|x|\leq1}x\,\nu(dx)\right)\right]ds\right|
\nonumber\\
&\qquad+
\left|\int_{t_{i-1}}^{t_i}\sigma_s\,dW_s\right|
\nonumber\\
&\qquad+
\int_{t_{i-1}}^{t_i}\int_{|x|\leq\alpha v_n}
|\chi_{s^-}x|N(ds,dx)\nonumber\\
&\qquad+
\left|\int_{t_{i-1}}^{t_i}\int
\delta(s,z)\mathfrak p(ds,dz)\right|\nonumber\\
&\quad\leq C\big(\Delta_n+\varepsilon v_n+\varepsilon v_n+\varepsilon v_n\big)
\leq C\varepsilon v_n.
\label{FVA1single}
\end{align}
For any real $x,y$ with $|x-y|\leq C\varepsilon v_n$, we have
\begin{align}
&\left|x^2\mathbf{1}_{\{v_n<|x|\leq\zeta v_n\}}
-y^2\mathbf{1}_{\{v_n<|y|\leq\zeta v_n\}}\right|\nonumber\\
&\quad\leq |x^2-y^2|\mathbf{1}_{\{v_n<|y|\leq\zeta v_n\}} +x^2\left|
\mathbf{1}_{\{v_n<|x|\leq\zeta v_n\}}
-\mathbf{1}_{\{v_n<|y|\leq\zeta v_n\}}
\right|\nonumber\\
&\quad\leq C\varepsilon v_n^2
\mathbf{1}_{\{v_n<|y|\leq\zeta v_n\}}\nonumber\\
&\qquad+Cv_n^2
\mathbf{1}_{\{(1-C\varepsilon)v_n\leq|y|
\leq(1+C\varepsilon)v_n\}}\nonumber\\
&\qquad+Cv_n^2
\mathbf{1}_{\{(\zeta-C\varepsilon)v_n\leq|y|
\leq(\zeta+C\varepsilon)v_n\}}.
\label{FVA1thresholdbound}
\end{align}

We apply  \eqref{FVA1thresholdbound} with $x=\Delta_iX$ and
$y=\int_{t_{i-1}}^{t_i}\int_{|z|>\alpha v_n}
\chi_{s^-}zN(ds,dz)$. On $\Omega_{i,n}^{(0)}$, the integral defining $y$
contains at most one jump contribution, since $\Omega_{i,n}^{(0)}\subseteq(\Omega_{i,n}^{(1)})^c$.
If $y=0$, all three terms on the right of
\eqref{FVA1thresholdbound} vanish. Hence,
\begin{align}
A_{1,0}\mathbf{1}_{\Omega_n(2\varepsilon)}
&\leq C\varepsilon v_n^Y\int_0^1\int
\mathbf{1}_{\{v_n<|\chi_sx|\leq\zeta v_n\}}N(ds,dx)\nonumber\\
&\quad+Cv_n^Y\int_0^1\int
\mathbf{1}_{\{(1-C\varepsilon)v_n\leq|\chi_sx|
\leq(1+C\varepsilon)v_n\}}N(ds,dx)\nonumber\\
&\quad+Cv_n^Y\int_0^1\int
\mathbf{1}_{\{(\zeta-C\varepsilon)v_n\leq|\chi_sx|
\leq(\zeta+C\varepsilon)v_n\}}
N(ds,dx).
\label{FVA1good}
\end{align}
Since
\begin{align*}
&\int
\mathbf{1}_{\{v_n<|\chi_sx|\leq\zeta v_n\}}\nu(dx)
\leq C|\chi_s|^Yv_n^{-Y},\\
&\int
\mathbf{1}_{\{(1-C\varepsilon)v_n\leq|\chi_sx|
\leq(1+C\varepsilon)v_n\}}\nu(dx)
\leq C\varepsilon|\chi_s|^Yv_n^{-Y},\\
&\int
\mathbf{1}_{\{(\zeta-C\varepsilon)v_n\leq|\chi_sx|
\leq(\zeta+C\varepsilon)v_n\}}
\nu(dx)
\leq C\varepsilon|\chi_s|^Yv_n^{-Y}.
\end{align*}
Thus, taking expectations in \eqref{FVA1good}, we obtain
\begin{equation}\label{FVA1goodexpectation}
\EE\big(A_{1,0}\mathbf{1}_{\Omega_n(2\varepsilon)}\big)
\leq C\varepsilon.
\end{equation}

Combining \eqref{FVA1decomposition}, \eqref{FVA1multiple},
\eqref{FVA1small}, \eqref{FVA1additional}, and
\eqref{FVA1goodexpectation}, we conclude that
\begin{align*}
\EE\big(|A_1|\mathbf{1}_{\Omega_n(2\varepsilon)}\big)
&\leq C\left\{
\Delta_n\alpha^{-2Y}v_n^{-Y}
+\varepsilon^{-1}\alpha^{1-Y}\right.\\
&\qquad\left.
+\big(\alpha^{-r}+\varepsilon^{-1}\alpha^{1-r}\big)v_n^{Y-r}
+\varepsilon\right\}.
\end{align*}
For fixed $\alpha,\varepsilon>0$, the first and third terms vanish
because $\Delta_nv_n^{-Y}\to0$ and $r<Y$. Therefore, for every $d>0$,
\begin{align*}
\limsup_{n\to\infty}\PP(|A_1|>d)
&\leq \limsup_{n\to\infty}\PP(\Omega_n(2\varepsilon)^c)
+\frac1d\limsup_{n\to\infty}
\EE\big(|A_1|\mathbf{1}_{\Omega_n(2\varepsilon)}\big)\\
&\leq \frac{C}{d}
\big(\varepsilon^{-1}\alpha^{1-Y}+\varepsilon\big).
\end{align*}
Taking first $\alpha\to0$ and then $\varepsilon\to0$ proves
$A_1=o_P(1)$. This completes the proof.
\end{proof}

\begin{proof}[Proof of Remark \ref{finitevariation_GL_remark}]
Recall that
\[
\kappa_{1,1}=\frac{C_++C_-}{2-Y}\int_0^1|\chi_s|^Y\,ds.
\]
Since $\chi$ is continuous and $\PP(|\chi_\tau|>0)=1$, we have $\kappa_{1,1}>0$ almost surely. By Lemma \ref{FVannularLLN}, applied at $pv_n$ and $\zeta pv_n$, we can write
\begin{align}\label{FVglobalcoefficients}
\widehat\Upsilon_n(\zeta)
&=\frac{\kappa_{1,1}p^{2-Y}(\zeta^{2-Y}-1)+o_P(1)}
{\kappa_{1,1}p^{2-Y}(\zeta^{2-Y}-1)^2+o_P(1)}
=\Upsilon(\zeta)+o_P(1),\nonumber\\
\frac{\zeta^2\widehat\Upsilon_n(\zeta)}
 {1-(\zeta^2-1)\widehat\Upsilon_n(\zeta)}
&=\frac{1}{\zeta^{-Y}-1}+o_P(1).
\end{align}
For the second equality, we used that
\[
1-(\zeta^2-1)\Upsilon(\zeta)
=\frac{\zeta^{2-Y}-\zeta^2}{\zeta^{2-Y}-1}\neq0.
\]
In particular, both coefficients are $O_P(1)$. By \eqref{debias_formula2b}, \eqref{FVlocaldifferences}
and \eqref{FVglobalcoefficients},
\eqref{FVGLnegligible} holds for $j=1$, also with
$\zeta_2v_n$ in place of $v_n$.
Together with \eqref{FVlocaldifferences}, this gives
$$
\sqrt{m_n}\big(
\tilde c^{(1)}_{n,\tau}(m_n,\zeta_2v_n,\zeta_1)
-\tilde c^{(1)}_{n,\tau}(m_n,v_n,\zeta_1)
\big)=o_P(1).
$$
Using this in \eqref{Twicedebias_formula2b}, together
with \eqref{FVglobalcoefficients}, proves
\eqref{FVGLnegligible} for $j=2$.
The two CLTs then follow from \eqref{FVbaselineCLT}.

\end{proof}

\subsection{Proof of Remark \ref{finitevariation_local_remark}}\label{AppFVlocal}

Suppose that the support of $K$ is contained in $[\ell,u]$, where
$-\infty<\ell\leq0\leq u<\infty$, and let
\[
\Omega_n:=\left\{
\max_{\{i:K_{m_n\Delta_n}(t_{i-1}-\tau)>0\}}
|\Delta_iX|\leq v_n\right\}.
\]
By \eqref{FVcontinuousmax}, \eqref{FVincrementbound}, and Markov's inequality,
\begin{align*}
\PP(\Omega_n^c)
&\leq\PP(\Omega_n(1)^c)
+\sum_{\{i:K_{m_n\Delta_n}(t_{i-1}-\tau)>0\}}
\PP\left(
\left|\Delta_iX-\int_{t_{i-1}}^{t_i}\sigma_s\,dW_s\right|
>v_n/2\right)\\
&\leq o(1)+\frac{C\Delta_n}{v_n}
\sum_{i=1}^n\mathbf{1}_{\{K_{m_n\Delta_n}(t_{i-1}-\tau)>0\}}\\
&\leq o(1)+C\frac{m_n\Delta_n+\Delta_n}{v_n}
=o(1).
\end{align*}
The last convergence follows from $m_n=O(\Delta_n^{-1/2})$ and
$v_n\gg\Delta_n^\beta$, with $\beta<1/2$.
On $\Omega_n$, we have, for every $a\geq1$,
$$
\hat c_{n,\tau}(m_n,av_n)
=\frac{\sum_{i=1}^nK_{m_n\Delta_n}(t_{i-1}-\tau)(\Delta_iX)^2}
{\bar K_{m_n\Delta_n}(\tau)}.
$$
Hence, the differences $\hat c_n(m_n,\zeta v_n)-\hat c_n(m_n, v_n)$ and $\hat c_n(m_n,\zeta^2 v_n)-\hat c_n(m_n, \zeta v_n)$  in \eqref{debias_formula} vanish. 
By the zero-denominator convention in Remark \ref{finitevariation_local_remark}, both correction terms are zero on $\Omega_n$. Since $\PP(\Omega_n)\to1$, this proves \eqref{FVlocalequality}.

\section{Proof of Lemma \ref{lem:terminal_stable_clt}}\label{s:proof_of_lem_clt}
At the end of the proof, we show that the stated hypotheses imply
\begin{equation}
\sum_{i=1}^n
\left\|\EE_{i-1}[\xi_i^n\Delta_iM]\right\|
\xrightarrow{P}0
\label{terminal_clt_all_martingales}
\end{equation}
for every bounded martingale $M$.
We first prove the conclusion of the lemma assuming this condition. 

Also, note that we if we center each increment, replacing $\xi_i^n$
by $\xi_i^n-\EE_{i-1}\xi_i^n$, the conditional covariance is unchanged, and Jensen's inequality gives
$$
\sum_{i=1}^n
\EE_{i-1}\big[\|\xi_i^n-\EE_{i-1}\xi_i^n\|^p\big]
\leq 2^p\sum_{i=1}^n\EE_{i-1}\big[\|\xi_i^n\|^p\big]
\xrightarrow{P}0.
$$
Moreover, for every square-integrable martingale $M$,
$$
\EE_{i-1}\big[(\xi_i^n-\EE_{i-1}\xi_i^n)\Delta_iM\big]
=\EE_{i-1}\big[\xi_i^n\Delta_iM\big],
$$
so both \eqref{terminal_clt_orthogonality} and
\eqref{terminal_clt_all_martingales} are unchanged.
The result for the original array then follows from the
centered case. Hence it it suffices to prove the statement when  $\EE_{i-1}\xi_i^n=0$ (so that $A=0$).

We first prove the result for $d=1$, writing $V$ for the
limiting conditional variance. Set
$$
q_{n,i}:=\EE_{i-1}\big[(\xi_i^n)^2\big],
\qquad
Q_{n,k}:=\sum_{i=1}^k q_{n,i},
\qquad
Q_n:=Q_{n,n}.
$$
Thus, $Q_n\xrightarrow{P}V$ and $Q_n=O_P(1)$.

We may assume that $2<p\leq3$. Indeed, if $p>3$,
conditional H\"older's inequality, followed by
H\"older's inequality for the sum, gives
$$
\sum_{i=1}^n\EE_{i-1}|\xi_i^n|^3
\leq
Q_n^{\frac{p-3}{p-2}}
\left(
\sum_{i=1}^n\EE_{i-1}|\xi_i^n|^p
\right)^{\frac{1}{p-2}}
\xrightarrow{P}0.
$$
We can therefore replace $p$ by $3$ in this case.

\medskip
\noindent
{\bf a) Localization.}
Define
$$
r_{n,i}:=\EE_{i-1}|\xi_i^n|^p,
\qquad
R_{n,k}:=\sum_{i=1}^k r_{n,i},
\qquad
R_n:=R_{n,n}.
$$
Since $R_n\xrightarrow{P}0$, there is a deterministic
sequence $\varepsilon_n\downarrow0$ such that
$$
\PP(R_n>\varepsilon_n)\longrightarrow0.
$$
Fix $K>0$, and set
$$
I_{n,i}^{K}
:=
\mathbf1_{\{Q_{n,i}\leq K,\ R_{n,i}\leq\varepsilon_n\}},
\qquad
\xi_i^{n,K}:=I_{n,i}^{K}\xi_i^n.
$$
Both $Q_{n,i}$ and $R_{n,i}$ are
$\mathcal F_{t_{i-1}}$-measurable. Hence the modified
array $\{\xi_i^{n,K}\}$ is still a martingale difference array. Writing
$$
Q_n^K:=\sum_{i=1}^n
\EE_{i-1}\big[(\xi_i^{n,K})^2\big],
\qquad
R_n^K:=\sum_{i=1}^n
\EE_{i-1}|\xi_i^{n,K}|^p,
$$
we have
\begin{equation}
Q_n^K\leq K,
\qquad
R_n^K\leq\varepsilon_n.
\label{terminal_clt_stopped_bounds}
\end{equation}
H\"older's inequality gives $q_{n,i}\leq r_{n,i}^{2/p}$, and therefore
$$
\max_{1\leq i\leq n}q_{n,i}
\leq R_n^{2/p}\xrightarrow{P}0.
$$
On $\{R_n\leq\varepsilon_n\}$, only the variance cutoff
can stop the array, and
$$
\left|Q_n^K-(Q_n\wedge K)\right|
\leq \max_{1\leq i\leq n}q_{n,i}.
$$
Consequently,
\begin{equation}
Q_n^K\xrightarrow{P}V\wedge K.
\label{terminal_clt_stopped_variance}
\end{equation}

For a bounded martingale $M$, let
$$
A_n^K(M):=
\sum_{i=1}^n
\left|
\EE_{i-1}\big[\xi_i^{n,K}\Delta_iM\big]
\right|.
$$
Since $I_{n,i}^K$ is predictable,
$
\EE_{i-1}\big[\xi_i^{n,K}\Delta_iM\big]
=I_{n,i}^K\EE_{i-1}\big[\xi_i^n\Delta_iM\big].
$
Hence, by \eqref{terminal_clt_all_martingales},
$$
A_n^K(M)
=\sum_{i=1}^n I_{n,i}^K
\left|
\EE_{i-1}\big[\xi_i^n\Delta_iM\big]
\right|
\leq
\sum_{i=1}^n
\left|
\EE_{i-1}\big[\xi_i^n\Delta_iM\big]
\right|
\xrightarrow{P}0.
$$
Also, Cauchy--Schwarz gives
$$
\begin{aligned}
A_n^K(M)
&\leq
\sum_{i=1}^n
\EE_{i-1}\big[(\xi_i^{n,K})^2\big]^{1/2}
\EE_{i-1}\big[(\Delta_iM)^2\big]^{1/2}\\
&\leq
K^{1/2}
\left(
\sum_{i=1}^n\EE_{i-1}\big[(\Delta_iM)^2\big]
\right)^{1/2}.
\end{aligned}
$$
It follows that
$
\EE\big[A_n^K(M)^2\big]
\leq
K\sum_{i=1}^n\EE[(\Delta_iM)^2]
=
K\EE[(M_T-M_0)^2].
$
Thus $A_n^K(M)$ is uniformly integrable, and
\begin{equation}
\EE A_n^K(M)\longrightarrow0.
\label{terminal_clt_orthogonality_L1}
\end{equation}

We now prove the result for this stopped array.
To simplify notation, until part~d) we omit the
superscript $K$ and write $V$ for $V\wedge K$.
In particular, throughout parts~b)--c),
$$
Q_n\leq K,\qquad R_n\leq\varepsilon_n,
\qquad Q_n\xrightarrow{P}V,\qquad 0\leq V\leq K.
$$

\medskip
\noindent
{\bf b) A compensated exponential.}
Fix $u\in\mathbb R$, and define
$$
S_{n,k}:=\sum_{i=1}^k\xi_i^n,
\qquad
L_{n,k}:=
\exp\left(iuS_{n,k}+\frac{u^2}{2}Q_{n,k}\right),
\qquad L_{n,0}=1.
$$
Since $Q_{n,k}\leq K$,
\begin{equation}
|L_{n,k}|\leq e^{u^2K/2}.
\label{terminal_clt_exponential_bound}
\end{equation}

For $2<p\leq3$, Taylor's formula gives
$$
\left|e^{iux}-1-iux+\frac{u^2x^2}{2}\right|
\leq C_u|x|^p,\qquad x\in\mathbb R.
$$
Using $\EE_{i-1}\xi_i^n=0$, we obtain
$$
\EE_{i-1}e^{iu\xi_i^n}
=
1-\frac{u^2}{2}q_{n,i}+\rho_{n,i},
\qquad
|\rho_{n,i}|\leq C_u r_{n,i}.
$$
Since $L_{n,i}=L_{n,i-1}e^{iu\xi_i^n+u^2q_{n,i}/2}$
and both $L_{n,i-1}$ and $q_{n,i}$ are
$\mathcal F_{t_{i-1}}$-measurable,
$$
\begin{aligned}
\EE_{i-1}(L_{n,i}-L_{n,i-1})
&=
L_{n,i-1}
\left(
e^{u^2q_{n,i}/2}\EE_{i-1}e^{iu\xi_i^n}-1
\right)\\
&=
L_{n,i-1}
\left[
e^{u^2q_{n,i}/2}\left(1-\frac{u^2}{2}q_{n,i}\right)-1
+e^{u^2q_{n,i}/2}\rho_{n,i}
\right].
\end{aligned}
$$
For $a\geq0$,
$$
\left|e^a(1-a)-1\right|
=\int_0^a t e^t\,dt
\leq \frac12 e^a a^2.
$$
Applying this with $a=u^2q_{n,i}/2$, and using
$q_{n,i}\leq K$, $|\rho_{n,i}|\leq C_u r_{n,i}$,
and \eqref{terminal_clt_exponential_bound}, gives
$$
\begin{aligned}
\left|
\EE_{i-1}(L_{n,i}-L_{n,i-1})
\right|
&\leq
|L_{n,i-1}|e^{u^2q_{n,i}/2}
\left(\frac{u^4}{8}q_{n,i}^2+C_u r_{n,i}\right)\\
&\leq C_{u,K}\big(q_{n,i}^2+r_{n,i}\big).
\end{aligned}
$$
Since $\sum_{i=1}^n q_{n,i}^2
\leq Q_n\max_i q_{n,i}\leq K\varepsilon_n^{2/p}$,
we obtain
\begin{equation}
\sum_{i=1}^n
\EE\left|
\EE_{i-1}(L_{n,i}-L_{n,i-1})
\right|
\leq
C_{u,K}
\left(K\varepsilon_n^{2/p}+\varepsilon_n\right)
\longrightarrow0.
\label{terminal_clt_exponential_drift}
\end{equation}

\medskip
\noindent
{\bf c) Passing to $\mathcal F$-measurable variables.}
Let $H$ be any bounded, real-valued,
$\mathcal F$-measurable random variable, and set
$$
M_t:=\EE(H\mid\mathcal F_t).
$$
Then $M$ is a bounded martingale.
Write $\Delta_iL_n:=L_{n,i}-L_{n,i-1}$. We have
\begin{equation}\label{e:mart_diff}
M_{t_i}L_{n,i}-M_{t_{i-1}}L_{n,i-1}
=
M_{t_{i-1}}\Delta_iL_n
+L_{n,i-1}\Delta_iM
+\Delta_iM\,\Delta_iL_n.
\end{equation}
The expectation of the second term on the right vanishes.
The sum of the expectations of the first terms on the right in \eqref{e:mart_diff} tends
to zero by \eqref{terminal_clt_exponential_drift}
and the boundedness of $M$.

For the rightmost terms in \eqref{e:mart_diff}, since $\EE_{i-1}\Delta_iM=0$,
\begin{align}
\EE_{i-1}[\Delta_iM\,\Delta_iL_n]
={}&
L_{n,i-1}e^{u^2q_{n,i}/2}
\EE_{i-1}\left[
(e^{iu\xi_i^n}-1)\Delta_iM
\right]\nonumber\\
={}&\nonumber
iuL_{n,i-1}e^{u^2q_{n,i}/2}
\EE_{i-1}[\xi_i^n\Delta_iM]\\
&+
L_{n,i-1}e^{u^2q_{n,i}/2}
\EE_{i-1}\left[
(e^{iu\xi_i^n}-1-iu\xi_i^n)\Delta_iM
\right].\label{e:IuLn}
\end{align}
By \eqref{terminal_clt_orthogonality_L1}, the sum of the absolute expectations of the first terms on the right-hand side of \eqref{e:IuLn} is bounded by
$$
C_{u,K}
\EE\left[
\sum_{i=1}^n
\left|\EE_{i-1}[\xi_i^n\Delta_iM]\right|
\right]
\longrightarrow0.
$$
For the remaining term on the right-hand side of \eqref{e:IuLn},  note that
$$
|e^{iux}-1-iux|^2\leq C_u|x|^p,
\qquad x\in\mathbb R,
$$
because $2<p\leq3$. Cauchy--Schwarz, applied to the
sum and expectation together, yields
$$
\begin{aligned}
&\sum_{i=1}^n
\EE\left[
|e^{iu\xi_i^n}-1-iu\xi_i^n|\,|\Delta_iM|
\right]\\
&\quad\leq
\left(
\sum_{i=1}^n
\EE|e^{iu\xi_i^n}-1-iu\xi_i^n|^2
\right)^{1/2}
\left(
\sum_{i=1}^n\EE(\Delta_iM)^2
\right)^{1/2}\\
&\quad\leq
C_u\varepsilon_n^{1/2}
\EE[(M_T-M_0)^2]^{1/2}
\longrightarrow0.
\end{aligned}
$$
Combining these bounds and telescoping the product
identity gives $\EE[M_TL_{n,n}]-\EE[M_0]\to0$.
Since $L_{n,n}$ is $\mathcal F_T$-measurable,
we have $\EE[M_TL_{n,n}]=\EE[HL_{n,n}]$ and $\EE[M_0]=\EE[H]$. Therefore,
\begin{equation}
\EE\left[
H\exp\left(iuS_{n,n}+\frac{u^2}{2}Q_n\right)
\right]
\longrightarrow\EE[H]
\label{terminal_clt_test_identity}
\end{equation}
for every bounded $\mathcal F$-measurable $H$.

Now, apply \eqref{terminal_clt_test_identity} with
$H e^{-u^2V/2}$ in place of $H$. This is again a
bounded $\mathcal F$-measurable random variable.
Since $Q_n\xrightarrow{P}V$ and $Q_n,V\leq K$,
bounded convergence  gives
$$
\EE\left|
\exp\left(\frac{u^2}{2}(Q_n-V)\right)-1
\right|\longrightarrow0.
$$
Consequently,
\begin{equation}
\EE[H e^{iuS_{n,n}}]
\longrightarrow
\EE[H e^{-u^2V/2}].
\label{terminal_clt_characteristic_identity}
\end{equation}
This proves the required stable convergence for the
stopped array.

\medskip
\noindent
{\bf d) Removing the localization.}
Return to the original notation, and write
$$
S_n:=\sum_{i=1}^n\xi_i^n,
\qquad
S_n^K:=\sum_{i=1}^n\xi_i^{n,K}.
$$
The preceding argument gives, for every fixed $K>0$,
$$
\EE[H \exp(iuS_n^K)]
\longrightarrow
\EE[H \exp(-u^2(V\wedge K)/2)].
$$
Also,
$$
\PP(S_n\neq S_n^K)
\leq
\PP(Q_n>K)+\PP(R_n>\varepsilon_n),
$$
and therefore
$$
\limsup_{n\to\infty}\PP(S_n\neq S_n^K)
\leq\PP(V\geq K).
$$
Since
$|e^{-u^2(V\wedge K)/2}-e^{-u^2V/2}|
\leq\mathbf1_{\{V>K\}}$, it follows that
$$
\begin{aligned}
&\limsup_{n\to\infty}
\left|
\EE[H e^{iuS_n}]
-\EE[H e^{-u^2V/2}]
\right|\\
&\quad\leq
3\|H\|_\infty\PP(V\geq K)
\longrightarrow0\qquad(K\to\infty).
\end{aligned}
$$
Thus \eqref{terminal_clt_characteristic_identity}
holds for the original array.  Hence, by the characteristic-function criterion for stable convergence
(e.g., \citealp[Corollary~3.8]{hausler:luschgy:2015}), this proves the
required $\mathcal F$-stable convergence to $\sqrt V\,Z$.

Finally, for $d>1$, apply the scalar argument above to
$a^\top\xi_i^n$, for each fixed $a\in\mathbb R^d$.
Its limiting conditional variance is $a^\top Ca$,
and the moment condition follows from
\eqref{terminal_clt_moments}. The condition for every
bounded martingale follows from
\eqref{terminal_clt_all_martingales}.
We obtain
$$
\EE\left[
H\exp\left(ia^\top\sum_{i=1}^n\xi_i^n\right)
\right]
\longrightarrow
\EE\left[H e^{-a^\top Ca/2}\right].
$$
This establishes the desired stable convergence when $d>1$.

Finally, it remains to establish \eqref{terminal_clt_all_martingales}
from the original hypotheses \eqref{terminal_clt_orthogonality}. Fix one coordinate. Within this
reduction, write $\xi_i^n$ for that scalar coordinate
and $W$ for its associated Brownian motion.
For any square-integrable martingale $M$, write
$$
A_n(M):=
\sum_{i=1}^n
\left|\EE_{i-1}[\xi_i^n\Delta_iM]\right|,
\qquad
Q_n:=
\sum_{i=1}^n\EE_{i-1}[(\xi_i^n)^2].
$$
For this centered coordinate, $Q_n=O_P(1)$ by the
corresponding diagonal entry of \eqref{terminal_clt_variance}.
 Cauchy--Schwarz gives
\begin{equation}
A_n(M)
\leq
Q_n^{1/2}
\left(
\sum_{i=1}^n\EE_{i-1}[(\Delta_iM)^2]
\right)^{1/2},
\label{terminal_clt_martingale_approx}
\end{equation}
where
$$
\EE\left[
\sum_{i=1}^n\EE_{i-1}[(\Delta_iM)^2]
\right]
=
\EE[(M_T-M_0)^2].
$$
Now suppose that $(M^{(k)})_{k\geq1}$ is a sequence of
square-integrable martingales such that
$$
\EE\left[
\big(M_T-M_0-M_T^{(k)}+M_0^{(k)}\big)^2
\right]\longrightarrow0.
$$
Applying \eqref{terminal_clt_martingale_approx} to
$M-M^{(k)}$ and using Markov's inequality gives,
for any $L>0$ and $\eta>0$,
$$
\begin{aligned}
\PP\big(A_n(M-M^{(k)})>\eta\big)
\leq{}&
\PP(Q_n>L)\\
&+\frac{L}{\eta^2}
\EE\left[
\big(
M_T-M_0-M_T^{(k)}+M_0^{(k)}
\big)^2
\right].
\end{aligned}
$$
Taking $\limsup_{n\to\infty}$, then letting $k\to\infty$
and $L\to\infty$, and using $Q_n=O_P(1)$, we obtain
\begin{equation}
\lim_{k\to\infty}\limsup_{n\to\infty}
\PP\big(A_n(M-M^{(k)})>\eta\big)=0.
\label{terminal_clt_martingale_closure}
\end{equation}
We apply this below with $M^{(k)}=\int_0^\cdot H_s^{(k)}\,dW_s$
to approximate the Brownian integral, and with
$M^{(k)}=N_{\cdot\wedge\tau_k}$ to approximate $N$.

Now let $M$ be bounded. By the Kunita--Watanabe
decomposition,
$$
M_t=M_0+\int_0^t H_s\,dW_s+N_t,
$$
where
$
\EE\int_0^T|H_s|^2\,ds<\infty, $
and $N$ is a square-integrable martingale, with
$N_0=0$, orthogonal to $W$.
We treat the two terms separately.

First suppose that $H$ is a bounded elementary
predictable process:
$$
H_s=
\sum_{j=1}^r H_j\mathbf1_{(s_{j-1},s_j]}(s),
\qquad
0=s_0<\cdots<s_r=T,
$$
where each $H_j$ is bounded and
$\mathcal F_{s_{j-1}}$-measurable.
On every observation interval contained in one
of these elementary intervals,
$$
\EE_{i-1}\left[
\xi_i^n\int_{t_{i-1}}^{t_i}H_s\,dW_s
\right]
=
H_j\EE_{i-1}[\xi_i^n\Delta_iW].
$$
The sum of the absolute values of these terms
tends to zero by \eqref{terminal_clt_orthogonality}
with $M=W$.

Let $\mathcal B_n$ be the set of observation
intervals that contain an elementary partition
point in their interior. There are at most $r-1$
such intervals. Their contribution is bounded by
$$
\begin{aligned}
&\sum_{i\in\mathcal B_n}
\left|
\EE_{i-1}\left[
\xi_i^n\int_{t_{i-1}}^{t_i}H_s\,dW_s
\right]
\right|\\
&\quad\leq
Q_n^{1/2}
\left(
\sum_{i\in\mathcal B_n}
\EE_{i-1}\left[
\left(\int_{t_{i-1}}^{t_i}H_s\,dW_s\right)^2
\right]
\right)^{1/2}\\
&\quad\leq
Q_n^{1/2}\|H\|_\infty
\big((r-1)\Delta_n\big)^{1/2}
\xrightarrow{P}0.
\end{aligned}
$$
Hence $A_n\left(\int_0^\cdot H_s\,dW_s\right)
\xrightarrow{P}0$
for every bounded elementary predictable $H$. For a general $H$ in the decomposition, choose
bounded elementary predictable processes $H^{(k)}$
such that
$$
\EE\int_0^T
|H_s-H_s^{(k)}|^2\,ds\longrightarrow0.
$$
The It\^o isometry and
\eqref{terminal_clt_martingale_closure} extend this
conclusion to the general $H$. It remains to treat $N$. Since the diffusive part of $M_t$ is continuous,
$$
\Delta N_t=\Delta M_t,
\qquad
|\Delta N_t|\leq2\|M\|_\infty.
$$
Let
$$
\tau_k:=\inf\{t\geq0:|N_t|\geq k\},
\qquad
N_t^{(k)}:=N_{t\wedge\tau_k}.
$$
Then $N^{(k)}$ is bounded by
$k+2\|M\|_\infty$ and remains orthogonal to $W$.
Thus \eqref{terminal_clt_orthogonality} gives
$$
A_n(N^{(k)})\xrightarrow{P}0
\qquad\text{for each fixed }k.
$$
Moreover, dominated convergence yields
$$
\EE[(N_T-N_T^{(k)})^2]
=
\EE\big([N]_T-[N]_{T\wedge\tau_k}\big)
\longrightarrow0.
$$
Indeed, the integrand tends to zero almost surely
and is bounded by the integrable random variable
$[N]_T$.
Using \eqref{terminal_clt_martingale_closure},
we obtain $A_n(N)\xrightarrow{P}0$.
Finally,
$$
A_n(M)
\leq
A_n\left(\int_0^\cdot H_s\,dW_s\right)
+A_n(N)
\xrightarrow{P}0.
$$
Returning to vector notation, we obtain, for every bounded
martingale $M$,
$$
\sum_{i=1}^n
\left\|\EE_{i-1}[\xi_i^n\Delta_iM]\right\|
\leq\sum_{j=1}^d\sum_{i=1}^n
\left|\EE_{i-1}[\xi_i^{n,j}\Delta_iM]\right|
\xrightarrow{P}0.
$$
This establishes \eqref{terminal_clt_all_martingales} for every bounded martingale and thus 
completes the proof.

\section*{Acknowledgements}
The first author's research was supported in part by the NSF award DMS-2413558. The second author's research was supported by NSF award DMS-2413557. The authors are grateful to the Associate Editor and two anonymous referees for their multiple corrections and helpful suggestions that help to significantly improve the original manuscript.

\bibliographystyle{imsart-nameyear} %


\end{document}